\documentclass[a4paper,twocolumn,accepted=2026-08-05]{quantumarticle}
\pdfoutput=1

\usepackage{color}
\usepackage{physics}
\usepackage{soul}
\usepackage[print-unity-mantissa=false]{siunitx}
\usepackage{graphicx}
\usepackage{dcolumn}
\usepackage{bm}
\usepackage{hyperref}
\usepackage{amsthm, amsfonts}
\usepackage{thmtools, thm-restate}
\usepackage[numbers,sort&compress]{natbib}

\usepackage{tabularx}

\newcounter{protocol}
\newenvironment{protocol}[1]
  {\par\addvspace{\topsep}
   \noindent
   \tabularx{\linewidth}{@{} X @{}}
    \hline \textbf{Protocol} \refstepcounter{protocol}\textbf{\theprotocol} #1 \\
    \hline}
  { \\
    \hline
   \endtabularx
   \par\addvspace{\topsep}}

\makeatletter 
\renewcommand\onecolumngrid{
\do@columngrid{one}{\@ne}%
\def\set@footnotewidth{\onecolumngrid}
\def\footnoterule{\kern-6pt\hrule width 1.5in\kern6pt}%
}
\makeatother

\newtheorem{definition}{Definition}
\newtheorem{theorem}{Theorem}

\newcommand{\vars}[1]{\mathcal{#1}}
\newcommand{\bigepsilon}{\mathcal{E}}
\newcommand{\leakEC}{\textup{leak}_{\textup{EC}}}
\newcommand{\expE}[1]{\mathbb{E}[#1]}

\newcommand{\PRNG}{\textup{PRNG}}
\newcommand{\TRNG}{\textup{IRNG}}
\newcommand{\RNG}{\textup{RNG}}
\newcommand{\MAC}{\textup{MAC}}
\newcommand{\tol}{\textup{tol}}
\newcommand{\AV}{\textup{AV}}
\newcommand{\BV}{\textup{BV}}
\newcommand{\CV}{\textup{CV}}
\newcommand{\SV}{\textup{SV}}
\newcommand{\SP}{\textup{SP}}
\newcommand{\SA}{\textup{SA}}
\newcommand{\rr}{\textup{r}}
\newcommand{\LB}{\textup{LB}}
\newcommand{\UB}{\textup{UB}}
\newcommand{\PA}{\textup{PA}}
\newcommand{\rob}{\textup{rob}}
\newcommand{\ds}{\textup{ds}}
\newcommand{\EC}{\textup{EC}}
\newcommand{\dett}{\textup{det}}
\newcommand{\secc}{\textup{sec}}
\newcommand{\reall}{\textup{real}}
\newcommand{\ideal}{\textup{ideal}}
\newcommand{\intt}{\textup{int}}
\newcommand{\PNR}{\textup{PNR}}

\newcommand{\trunc}{\textup{trunc}}

\newcommand{\labeltextup}[2]{#1^{\textup{#2}}}
\newcommand{\KOTP}{\labeltextup{K}{mask}}
\newcommand{\kOTP}{\labeltextup{k}{mask}}
\newcommand{\Kh}{\labeltextup{K}{h}}
\newcommand{\kh}{\labeltextup{k}{h}}
\newcommand{\rhoin}{\labeltextup{\rho}{in}}
\newcommand{\rhoout}{\labeltextup{\rho}{out}}
\newcommand{\rhoideal}{\labeltextup{\rho}{ideal}}

\newcommand{\alphaU}{\labeltextup{\alpha}{U}}
\newcommand{\jU}{\labeltextup{j}{U}}
\newcommand{\jOTP}{\labeltextup{j}{mask}}
\newcommand{\Psift}{\labeltextup{P}{sift}}

\newcommand{\labeltextdown}[2]{#1_{\textup{#2}}}
\newcommand{\eph}{\labeltextdown{e}{ph}}
\newcommand{\ebit}{\labeltextdown{e}{bit}}
\newcommand{\Hmin}{\labeltextdown{H}{min}}
\newcommand{\Hmax}{\labeltextdown{H}{max}}
\newcommand{\fsyn}{\labeltextdown{f}{syn}}
\newcommand{\fEC}{\labeltextdown{\gamma}{EC}}
\newcommand{\hateph}{\labeltextdown{\hat{e}}{ph}}
\newcommand{\deltagap}{\labeltextdown{\delta}{gap}}
\newcommand{\hatebit}{\labeltextdown{\hat{e}}{bit}}
\newcommand{\hbin}{\labeltextdown{h}{bin}}
\newcommand{\hPA}{\labeltextdown{h}{PA}}
\newcommand{\FA}{\labeltextdown{F}{A}}
\newcommand{\FB}{\labeltextdown{F}{B}}
\newcommand{\FS}{\labeltextdown{F}{S}}
\newcommand{\fA}{\labeltextdown{f}{A}}
\newcommand{\fB}{\labeltextdown{f}{B}}
\newcommand{\KA}{\labeltextdown{K}{A}}
\newcommand{\KB}{\labeltextdown{K}{B}}
\newcommand{\kA}{\labeltextdown{k}{A}}
\newcommand{\kB}{\labeltextdown{k}{B}}
\newcommand{\DPE}{\labeltextdown{D}{PE}}
\newcommand{\DC}{\labeltextdown{D}{C}}
\newcommand{\tildeFA}{\labeltextdown{\tilde{F}}{A}}
\newcommand{\tildeFB}{\labeltextdown{\tilde{F}}{B}}
\newcommand{\tildeFS}{\labeltextdown{\tilde{F}}{S}}
\newcommand{\tildeDPE}{\labeltextdown{\tilde{D}}{PE}}
\newcommand{\epsrob}{\labeltextdown{\varepsilon}{rob}}
\newcommand{\epsKS}{\labeltextdown{\varepsilon}{KS}}

\newcommand{\epsMS}{\labeltextdown{\varepsilon}{MS}}
\newcommand{\epsEA}{\labeltextdown{\varepsilon}{EA}}
\newcommand{\epsEAf}{\labeltextdown{\varepsilon}{EA,f}}
\newcommand{\epsEAaf}{\labeltextdown{\varepsilon}{EA,af}}
\newcommand{\epsph}{\labeltextdown{\varepsilon}{ph}}

\newcommand{\epsbit}{\labeltextdown{\varepsilon}{bit}}
\newcommand{\epssec}{\labeltextdown{\varepsilon}{sec}}
\newcommand{\epsSP}{\labeltextdown{\varepsilon}{SP}}
\newcommand{\epsds}{\labeltextdown{\varepsilon}{ds}}
\newcommand{\epssm}{\labeltextdown{\varepsilon}{sm}}
\newcommand{\epsSO}{\labeltextdown{\varepsilon}{SO}}
\newcommand{\epsguess}{\labeltextdown{\varepsilon}{guess}}
\newcommand{\epsCV}{\labeltextdown{\varepsilon}{CV}}
\newcommand{\epsmatch}{\labeltextdown{\varepsilon}{match}}
\newcommand{\Sout}{\labeltextdown{\vars{S}}{out}}
\newcommand{\Sin}{\labeltextdown{\vars{S}}{in}}

\newcommand{\OmegaPE}{\labeltextdown{\Omega}{PE}}
\newcommand{\tildeOmegaPE}{\labeltextdown{\tilde{\Omega}}{PE}}
\newcommand{\pguess}{\labeltextdown{p}{guess}}

\newcommand{\ephtol}{e_{\textup{ph},\textup{tol}}}
\newcommand{\epssecint}{\varepsilon_{\textup{sec},\textup{int}}}
\newcommand{\tildethetapbetain}{\tilde{\theta}_{\beta}^{'\textup{in}}}
\newcommand{\tildethetabetain}{\tilde{\theta}_{\beta}^{\textup{in}}}

\newcommand{\ebittol}{e_{\textup{bit},\textup{tol}}}
\newcommand{\epsMACa}{\varepsilon_{\textup{MAC},1}}
\newcommand{\epsMACb}{\varepsilon_{\textup{MAC},2}}
\newcommand{\epsserfa}{\varepsilon_{\textup{serf},1}}
\newcommand{\epsserfb}{\varepsilon_{\textup{serf},2}}
\newcommand{\epssma}{\varepsilon_{\textup{sm},1}}
\newcommand{\epssmb}{\varepsilon_{\textup{sm},2}}
\newcommand{\epssmc}{\varepsilon_{\textup{sm},3}}
\newcommand{\fsyndec}{f_{\textup{ec}}}
\newcommand{\rhoidealint}{\rho^{\textup{ideal},\textup{int}}}

\newcommand{\rhoideall}[1]{\rho^{\textup{ideal},#1}}
\newcommand{\rhoidealintt}[1]{\rho^{\textup{ideal},\textup{int},#1}}
\newcommand{\rhooutt}[1]{\rho^{\textup{out},#1}}
\newcommand{\rhopoutt}[1]{\rho^{'\textup{out},#1}}
\newcommand{\rhoinn}[1]{\rho^{\textup{in},#1}}
\newcommand{\ebitt}[1]{e_{\textup{bit},#1}}
\newcommand{\hatebitt}[1]{\hat{e}_{\textup{bit},#1}}

\title{Quantum Authenticated Key Expansion with Key Recycling}

\author{Wen Yu Kon}
\email{wenyu.kon@jpmchase.com}
\affiliation{Global Technology Applied Research, JPMorganChase}
 
\author{Jefferson Chu}
\affiliation{Global Technology Applied Research, JPMorganChase}

\author{Kevin Han Yong Loh}
\affiliation{Global Technology Applied Research, JPMorganChase}

\author{Obada Alia}
\affiliation{Global Technology Applied Research, JPMorganChase}

\author{Omar Amer}
\affiliation{Global Technology Applied Research, JPMorganChase}

\author{Marco Pistoia}
\affiliation{Global Technology Applied Research, JPMorganChase}

\author{Kaushik Chakraborty}
\email{kaushik.chakraborty@jpmchase.com}
\affiliation{Global Technology Applied Research, JPMorganChase}

\author{Charles Lim}
\affiliation{Global Technology Applied Research, JPMorganChase}

\begin{document}

\maketitle

\begin{abstract}
Data privacy and authentication are two main security requirements for remote access and cloud services.
While QKD has been explored to address data privacy concerns, oftentimes its use is separate from the client authentication protocol despite implicitly providing authentication.
Here, we present a \emph{quantum authentication key expansion} (QAKE) protocol that (1) integrates both authentication and key expansion within a single protocol, and (2) provides key recycling property -- allowing all authentication keys to be reused.
We analyse the security of the protocol in a QAKE framework adapted from a classical authentication key exchange (AKE) framework, providing separate security conditions for authentication and data privacy.
We experimentally implemented the protocol with appropriate post-selection.
Additional results on the security of pseudorandom basis generation in QAKE and decoy state BB84 are provided.
\end{abstract}

\section{Introduction}
\label{sec:Introduction}
In this distributed digital era, remote access and cloud services are becoming increasingly popular, where a client device can access another device or a server, at any time, and from anywhere.
Such services typically has two main security requirements: authentication and data privacy.
To restrict access to these devices/servers for security or compliance reasons, such as in banking, healthcare, or government, it is necessary for the identity of the client to be authenticated.
Data privacy, on the other hand, ensures that the data transferred between the client and server, which may contain highly sensitive business or personal information (e.g. while looking at a remote desktop display of client details), is private from any eavesdroppers.
Fig.~\ref{fig:Application} presents a scenario where such a central server can identify its end-users via authentication and aid in the key generation between distant end-users.\\

\begin{figure*}
    \centering
    \includegraphics[width=.75\linewidth]{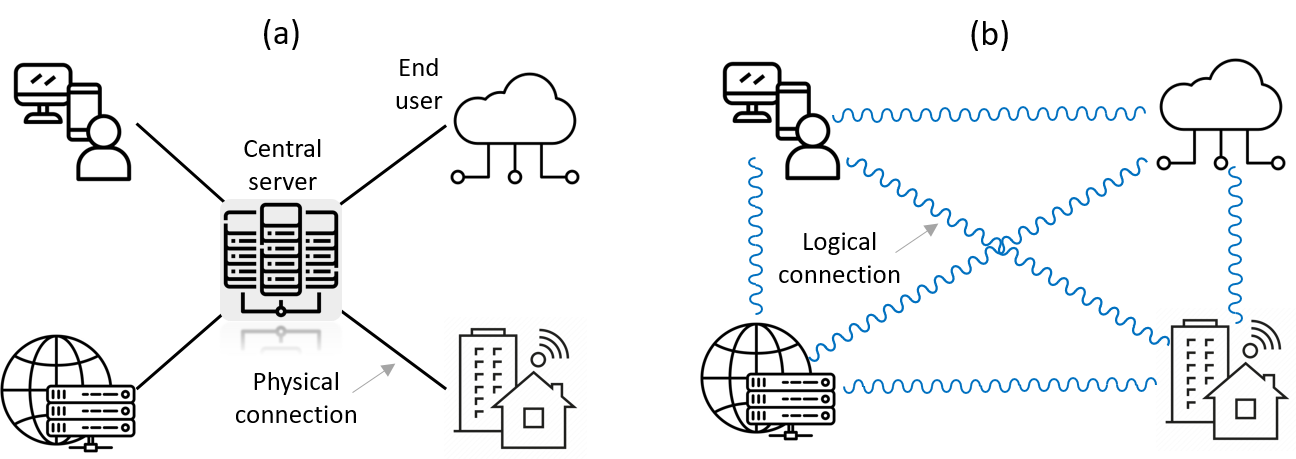}
    \caption{\label{fig:Application} Schematic for end user key generation via a central server in a star network topology. (a) End users share a direct physical connection to the central server, and can establish keys with it via the authenticated key exchange protocol. (b) By performing operations on the keys, the central server is able aid the end users establish shared keys between them without the end users sharing a direct link.}
\end{figure*}

Many current remote access and cloud services establishes a \emph{virtual private networks} (VPNs) by means of internet protocols like IPsec.
Such protocols rely on either \textit{public key infrastructure} (PKI) or symmetric key cryptography to provide the necessary security requirements.
PKI offers convenient key exchange without pre-shared secrets, but older schemes such as RSA~\cite{RSA1978} are vulnerable to quantum attacks~\cite{Shor1994}. 
While post-quantum PKI schemes such as lattice-based cryptography that have been standardized by NIST~\cite{FIPS203} are believed to be quantum-safe, PKI remains susceptible to \emph{harvest-now-decrypt-later} attacks.
On the other hand, symmetric key cryptography offers stronger long-term data privacy, but requires pre-shared secrets keys that have to be regularly refreshed after a fixed number of uses or duration~\cite{NIST_Key_Management}, posing a practical challenge.\\

To address the need for key refresh in symmetric key cryptography, internet service providers explored the use of \emph{quantum key distribution} (QKD) for key expansion. 
In this process, a separate authentication protocol is utilised to authenticate the client of the remote access service before utilising QKD to expand secret keys shared between the client and server.
Therefore, we propose a \emph{quantum authenticated key expansion} (QAKE) protocol that improves efficiency by combining all entity and message authentication steps into two authentication steps: one from Alice to Bob and one from Bob to Alice. 
This integrated approach provides both client-server authentication and key expansion within a single protocol session.\\

Another useful property for authentication to have in practice is key recycling~\cite{Fehr2017_SID}, where all secret keys (e.g. authentication keys) remain secure if the protocol succeeds.
In the context of cloud services, it eliminates the need for cloud service providers to update the clients' authentication key across all its cloud servers after every session, saving significant communication resources.
Authentication schemes with key recycling have been proposed, but they either do not recycle the full authentication key (require recycling of an OTP key, which we call authentication masking key)~\cite{Carter1977_WegmanCarter,Portmann2014_Authentication}, or do not account for experimental imperfections such as loss and multi-photon events~\cite{Portmann2017_Authentication,Fehr2017_SID,Leermakers2019_QKRNoise}.
The QAKE proposed here is designed with key recycling in mind, using the corrected bitstring, which is expected to have high min-entropy when the protocol passes, to mask the authentication key and authentication masking keys and allow them to be reused.\\

To analyse both the authentication and key exchange tasks in QAKE, we adapt the classical authenticated key exchange (AKE) framework~\cite{Guilhem2020_AKE}, where separate security conditions for authentication and key exchange are presented.
This highlights the role of various protocol components in providing data privacy and authentication, and allows for comparison with the case of two separate protocols -- one for key exchange and one for authentication.\\

We demonstrated the protocol using commercial QKD devices.
The raw data collected from the QKD devices are post-processed offline through a software stack.
This includes error correction by adapting the LDPC IP cores from Xilinx~\cite{LDPC_IP}, and privacy amplification using Toeplitz hashing.
With a detection rate of about \SI{1e-4}{} and a quantum bit error rate (QBER) of about \SI{2}{\percent}, the QAKE protocol is able to provide a key rate of about \SI{1e-5}{}.

\section{Quantum Authenticated Key Exchange}
\label{sec:Results}

\subsection{Message Authentication in QKD}

In QKD, message authentication, which can double as entity authentication, typically relies on information-theoretic message authentication protocols (such as Wegman-Carter~\cite{Carter1977_WegmanCarter}) to authenticate every message sent between the parties.
These message authentication protocols have composable security and can thus be used in QKD securely.
Oftentimes, the QKD keys generated are then utilized for the next session's authentication.\\

However, a direct utilization of the typical protocol design above may not be desirable in the practical setting.
While authenticating every message may simplify the security analysis, it might lead to significant authentication costs due to the shared randomness required, thereby reducing key rates.
This issue may be exacerbated by for instance the use of Cascade~\cite{Brassard1994_Cascade} as a means for error correction since it requires multiple rounds of communication, each of which may need to be authenticated if we naively apply message authentication to every communication steps.
This is addressed in Ref.~\cite{Kiktenko2020_Lightweight}, where authentication is performed only in the last two communication steps, thereby reducing the authentication costs.\\

Using QKD keys for future authentication sessions requires updating the authentication keys after each use.
If the authentication keys can remain static, this would enable a wider range of applications.
For instance, in applications where a clients access services at multiple locations (e.g. ATM), authentication keys may be stored in multiple regional databases. 
Updating the authentication keys in all databases after every authentication session would be operationally inefficient.

\subsection{QAKE Protocol}

\begin{figure}
    \centering
    \includegraphics[width=0.45\textwidth]{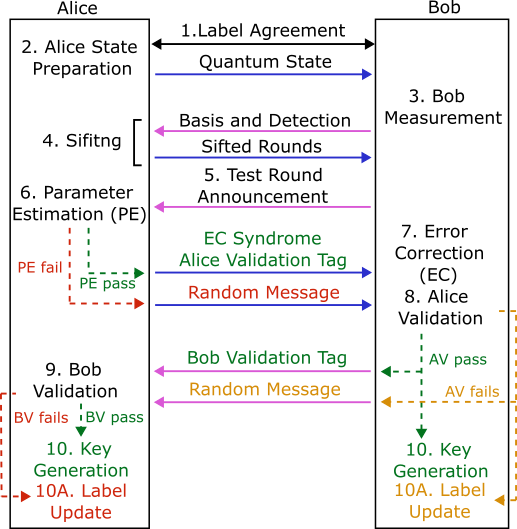}
    \caption{Summary of the QAKE protocol proposed. The protocol begins with both parties establishing the label to use, followed by Alice sending quantum states to Bob. Bob announces his basis choice and detected rounds and Alice responds with the sifted rounds. Bob announces the test round results, for which Alice can use for parameters estimation. If this passes, error correction and validation of Alice's identity would occur followed by validation of Bob's identity. Finally, the two parties independently decide whether to generate secret keys or update their label.}
    \label{fig:RoundEfficientQAKE}
\end{figure}

As such, the QAKE protocol focuses on incorporating the entity authentication check into the QKD process in an efficient manner, namely: (1) reducing authentication costs, (2) removing the need for key updating, (3) combining error verification with authentication, (4) secure reusing of privacy amplification seed.\\

To reduce communication costs, the QAKE protocol follows the proposal in Ref.~\cite{Kiktenko2020_Lightweight} and performs authentication only twice, once to authenticate Alice (and communication transcript) and once to authenticate Bob.
Since the order of certain steps in QKD is crucial for security, most notably that Bob's measurement should occur before the basis is announced, careful design is necessary.
For instance, we require that Bob announces his basis choice before Alice announces hers for sifting.
This prevents Eve from storing the quantum states and measuring them when the basis is announced since that requires Eve to commit to a response of Bob's measurement basis choice to Alice before Bob has actually selected the basis.\\

To remove the need for key updating, the idea of key recycling or reusing of authentication keys has been proposed.
Early proposals of key recycling~\cite{Carter1977_WegmanCarter,Portmann2014_Authentication} requires pre-sharing of authentication masking keys (one for each authentication) to act as a one-time pad for the authentication tag, which may be impractical in practice for lightweight devices (e.g. ATM cards).
Later proposals~\cite{Portmann2017_Authentication,Fehr2017_SID,Leermakers2019_QKRNoise} that only require updating of authentication masking keys when authentication fails require additional quantum communication separate from the QKD process for message authentication, but the protocols proposed do not fully account for experimental imperfections (loss, errors, multi-photon events etc).
Here, we propose the utilization of the corrected bitstring in QKD, where min-entropy can be guaranteed by the parameter estimation in QKD, as a means to protect
the authentication key and authentication masking key when the protocol succeeds.
Besides straightforwardly accounting for experimental imperfections via the QKD protocol, combining the quantum process for both protocols can reduce finite-size effects, since there need not be separate parameter estimation for both the process in QKD and authentication of Ref.~\cite{Fehr2017_SID,Leermakers2019_QKRNoise}.\\

We also provide further improvements by having the message authentication scheme double as error verification, since the corrected bitstring is hashed there.
This reduces the number of shared randomness between Alice and Bob that is utilized in the post-processing of the QKD protocol.\\

The final modification made is to allow the privacy amplification seed $R$ to be pre-shared between Alice and Bob and reused for future privacy amplification.
Unlike typical QKD protocols where Alice (or Bob) selects a random seed and publicly announces it, in this case, Alice and Bob do not announce $R$.
As such, reusing $R$ does not impact security since for quantum-proof strong extractors, the keys generated are close to independent on $R$.
This use of a pre-shared random $R$ could simplify practical implementation since the \emph{field programmable gate array} (FPGA) used for privacy amplification can be pre-programmed with the seed and also reduce communication and authentication costs since $R$ does not need to sent from Alice to Bob.\\

To illustrate the effects of key recycling and reusing of privacy amplification seed $R$, we need to consider that the protocol can be repeated multiple times.
We call each complete run of the protocol as a session, and assume that we can have up to $n_{sess}$ sessions of the protocol.
In each session, Alice and Bob exchange $n_{rnd}$ quantum states and we call each quantum state exchange as a single round (like QKD).\\

\onecolumngrid

\begin{protocol}{Quantum Authenticated Key Exchange}
\textit{Goal.} Alice and Bob authenticates one another, and performs key exchange.
\begin{enumerate}
    \item \textbf{Label Agreement}: Alice and Bob exchange $\alpha$ and $\alpha'$. Alice (resp. Bob) sends $\alpha$ (resp. $\alpha'$) and receives $\alpha_\rr'$ (resp. $\alpha_\rr$), which results in a label choice $\beta=\max\{\alpha,\alpha_\rr'\}$ (resp. $\beta'=\max\{\alpha_\rr,\alpha'\}$).
    \item \textbf{Alice State Preparation}: Alice randomly chooses a $n_{rnd}$-bit basis string $\theta$, a $n_{rnd}$-bit string $X$ and a $n_{rnd}$-trit string $v\in\{0,1,2\}^{n_{rnd}}$ according to probability distribution $p_v$. She then prepares $n_{rnd}$ quantum systems $Q_i$ in phase-randomised coherent BB84 states $\left\{\rho_{Q_i}^{\theta_i,X_i,\mu_{v_i}}\right\}_{i\in[1,n_{rnd}]}$, with basis $\theta_i$, bit value $X_i$, and intensity $\mu_{V_i}$.
    \item \textbf{Bob Measurement}: Alice sends $Q^{n_{rnd}}$ to Bob, who measures subsystems $Q_i$ using a randomly chosen basis $\theta'\in\{0,1\}^{n_{rnd}}$, and records outcome $X_i'$. If Bob detects no clicks, he declares $X_i'=\perp$. If Bob detects multiple clicks, he randomly selects $x_i'\in\{0,1\}$.
    \item \textbf{Sifting}: Bob records the detection rounds, $P=\{i:X_i'\neq\perp\}$, and announces $P$ and $\theta_{P}'=\{\theta_i':i\in P\}$ to Alice. Alice computes the sifted rounds $\Psift=\{i:i\in P_{\rr},\theta_{i,\rr}'=\theta_i\}$, and forwards it to Bob.
    \item \textbf{Test Round Announcement}: Bob randomly splits the sifted rounds into $\Psift_1$ and $\Psift_2$, with $\abs{\Psift_1}=\lceil f_{P_1}\abs{\Psift}\rceil$, where $f_{P_1}$ is some pre-determined fraction of rounds. Bob announces $\Psift_1$, $\Psift_2$ and $X_{\Psift_1}'$. 
    \item \textbf{Parameter Estimation}: Alice checks if: (1) the number of test rounds matches the expected number $\abs{\Psift_{1,\rr}}=\lceil f_{P_1}\abs{\Psift}\rceil$, (2) the number of sifted rounds exceeds a threshold, $\abs{\Psift_{\rr}}\geq \Psift_{\tol}$, (3) the lower bound on single-photon events in sets $\Psift_{2,\rr}$ and $\Psift_{1,\rr}$ (estimated via decoy state analysis) exceeds a threshold, $\hat{N}^{\LB}_{\Psift_{2,\rr},1}\geq N_{\Psift_{2,\rr},1}^{\tol}$, $\hat{N}^{\LB}_{\Psift_{1,\rr},1}\geq N_{\Psift_{1,\rr},1}^{\tol}$, (4) the upper bound on single-photon bit-error rate in set $\Psift_{1,\rr}$ (estimated via decoy-state analysis) falls below a threshold, $\hatebitt{\Psift_{1,\rr},1}^{\UB}\leq \ebitt{1,\tol}$, and (5) the upper bound on bit error rate in set $\Psift_{1,\rr}$ (estimated via Serfling bound) falls below a threshold, $\hatebitt{\Psift_{2,\rr}}^{\UB}\leq \ebitt{\tol}$. If all checks pass, Alice sets $\DPE=1$, otherwise she sets $\DPE=0$.
    \item \textbf{Error Correction}: If $\DPE=1$, Alice generates a syndrome $syn=\fsyn(X_{\Psift_{2,\rr}})$ of length $l_{syn}$ and forwards it to Bob, otherwise, Alice sends a random string of length $l_{syn}$ to Bob. Bob receives the syndrome $syn_\rr$ and computes the corrected bit string $\hat{X}_{\Psift_2,\rr}=\fsyndec(X_{\Psift_2}',syn_\rr)$.
    \item \textbf{Alice Validation}: If $\DPE=1$, Alice generates a tag $T_{\AV}=h_1(\Kh_1,X_{\Psift_{2,\rr}}||P_{\rr}||\theta_{P,\rr}'||\Psift||X_{\Psift_1,\rr}'||\Psift_{1,\rr}||\Psift_{2,\rr}||syn)\oplus \KOTP_{1,\beta}$ and forwards it to Bob, otherwise, Alice sends a random string of the same length as the tag. Bob receives tag $T_{\AV,\rr}$ and generates verification tag $\tilde{T}_{\AV}=h_1(\Kh_1,\hat{X}_{\Psift_{2,\rr}}||P||\theta_P'||\Psift_{\rr}||X_{\Psift_1}'||\Psift_1||\Psift_2||syn_\rr)\oplus \KOTP_{1,\beta'}$ and checks if $T_{\AV,\rr}=\tilde{T}_{\AV}$. If the tags matches, Bob validates Alice and output $\FB=1$, otherwise, he sets $\FB=0$.
    \item \textbf{Bob Validation}: 
    If $\FB=1$, Bob generates a tag, $T_{\BV}=h_2(K_2,\hat{X}_{\Psift_{2,\rr}})$, and sends $T_{\BV}$ to Alice. If $\FB=0$, Bob sends a random string of the same length instead. Alice receives the tag $T_{\BV,\rr}$, and computes the verification tag, $\tilde{T}_{\BV}=h_2(K_2,X_{\Psift_{2,\rr}})$. If $T_{\BV,\rr}=\tilde{T}_{\BV}$, Alice validates Bob, $D_{\BV}=1$, otherwise, she sets $D_{\BV}=0$.
    \item \textbf{Secret Key Generation and Label Update}: Alice decides whether to accept the session based on her parameter estimation and validation of Bob, i.e. $\FA=\DPE\land D_{\BV}$. If Alice (resp. Bob) accepts the session and decides to perform key generation, $\FA=1$ (resp. $\FB=1$), she (resp. he) performs privacy amplification $\KA=\hPA(R,X_{\Psift_{2,\rr}})$ (resp. $\KB=\hPA(R,\hat{X}_{\Psift_{2,\rr}})$) and sets $\alpha=\beta$ (resp. $\alpha'=\beta'$), where $\KA\KB$ are the cryptographically secure keys that can be used for other purposes. If key generation is not performed, the labels are updated, i.e. if $\FA=0$, Alice updates her label $\alpha=\beta+1$ and if $\FB=0$, Bob updates his label $\alpha'=\beta'+1$. 
\end{enumerate}
\end{protocol}

\twocolumngrid

The full $n_{sess}$-session protocol begins with two parties, Alice and Bob, sharing secrets $S$ containing: (1) authentication keys $\Kh_1$ and $K_2$, (2) privacy amplification seed $R$, and (3) authentication masking key with $m$ separate keys $\KOTP_1=\{\KOTP_{1,j}\}_{j=1,\cdots,m}$.
Alice and Bob each has a label, $\alpha$ and $\alpha'$ respectively, that notes the secrets to utilise for the round, thereby allowing secrets to be replaced (by never using them again) when authentication fails.
They also publicly share a $\epsMACa$-almost XOR 2-universal hash function $h_1$ (uses authentication key $K_1^h)$, a $\epsMACb$-almost strongly 2-universal hash function $h_2$ (uses authentication key $K_2$) to generate the authentication tags $T_{AV}$ and $T_{BV}$ respectively~\cite{Rogaway2002_BucketHashing,Portmann2014_Authentication,Stinson1994_UniversalHashing}, and a 2-universal hash function $\hPA$ (uses privacy amplification seed $R$) for privacy amplification to generate secret keys~\cite{Portmann2014_Authentication,Stinson1994_UniversalHashing}.
Alice and Bob also agrees on a fraction of rounds $f_{P_1}$ to use as test rounds for parameter estimation and an error correction protocol which includes syndrome generation function $\fsyn$ (to generate a syndrome of length $l_{syn}$) and error correction function $\fsyndec$ (to correct the message using the syndrome). 
The protocol steps are summarised in Fig.~\ref{fig:RoundEfficientQAKE}, with additional details provided in protocol 1 below. 
We note here that the subscript $\rr$ on any random variables represent the received variable, e.g. $\alpha_{\rr}$ is received by Bob when $\alpha$ is sent by Alice, and $\land$ refers to the AND gate.
Each protocol session is able to be triggered by either Alice and Bob at any time (e.g. when an ATM card is plugged into an ATM) by for instance using a handshake protocol without authentication.
During the session, Alice and Bob update the parameters $\FA$ and $\FB$ respectively to bit 0 or 1, which indicates whether they wish to accept the session and generate a key.
The values of $\FA$ and $\FB$ are initialised as $\phi$ and is always reset independently by Alice and Bob to $\FA=\FB=\phi$ at the end of each session if they participated in the session, but we analyse the security of the protocol at the point before this reset occurs.
Note that if Alice and Bob did not participate in a session, their value of $\FA$ and $\FB$ respectively would remain as $\phi$.

\section{Protocol Security}
\label{sec:AKEProtocolSec}

QKD, while achieving the task of authenticated key exchange (AKE)~\cite{Mosca2013_QKDAKE}, separates the message authentication and key exchange tasks completely, and combines their security using composability arguments.
This makes it difficult to analyse the QAKE protocol directly in the QKD security framework since (1) message authentication is only used in the final two communication steps, which may allow the adversary to alter the ordering of prior steps, and (2) the message authentication steps are also used to check for matching bit strings between Alice and Bob (correctness condition in QKD). 
Therefore, we introduce a QAKE security definition which includes both authentication and key exchange, similar to those discussed in Ref.~\cite{Portmann2022_Security}.

\subsection{Multi-Session Security and Reduction to Single-session Security}

Let $\rho_0$ be the initial quantum state, where the shared secrets $S$ are generated, and $\vars{P}_i$ be the $i$-th session of the protocol. 
For simplicity, we assume that the sessions are sequential, i.e. one session must end before a new session can begin, and we can express $i$ sessions of the protocol as $\vars{P}^i=\vars{P}_i\circ\cdots\circ\vars{P}_1$.
For each session, we can consider an overall input state $\rhoin$ shared across all parties that acts as an input to the channel $\vars{P}_i$ and an output state $\rhoout=\vars{P}_i(\rhoin)$ shared by all parties at the end of each session. 
Note that the output state also acts as the input state to the next session.
We define the security of the multi-session protocol by requiring the output state to be close to an ideal output state for every session, i.e. trace distance between output state and ideal output state is $\Delta(\vars{P}^i(\rho_0),\rhoout)\leq\varepsilon$ for all $i=1,\cdots,n_{sess}$, where $\rhoout\in \Sout'$ is an ideal output state satisfying QAKE security conditions (to be defined later).
Since each session accepts the secrets (such as authentication key and privacy amplification seeds) as inputs, there could be input states where the secrets are correlated to the adversary.
Therefore, the protocol cannot be composable in the usual sense -- that $\vars{P}_i$ is $\varepsilon$-close to an ideal functionality.
However, assuming that these secrets are not additionally leaked between sessions via side-channels, we can restrict the input states to an ideal set of input states $\vars{S}_{in}$ (to be defined later) that results in the output state being $\varepsilon$-close to an ideal output state\footnote{Note that QKD implicitly requires such an assumption when authentication keys are carried between sessions, since any leakage of authentication keys outside of the QKD protocol can compromise QKD security. Additionally, if one can guarantee the input states are ideal (like QKD), it may be possible to define an equivalent protocol $\vars{P}_i'$ that is composable, where the secrets are generated within each protocol session.}.\\

We can reduce the multi-session analysis to a single-session security with appropriate definitions of $\Sin$ and $\Sout'$, where $\Sout'$ contains an additional condition that the privacy for the secrets that are utilised in the next session be maintained.
If the protocol maps $\vars{S}_{in}$ to $\Sout'$, and $\Sout'\subseteq \Sin$, an inductive argument can be made that the protocol is secure -- similar to composability arguments, except with restricted input states.
This is presented formally as
\begin{theorem}
\label{thm:Multi_to_Single_Red}
Suppose we have the conditions:
\begin{enumerate}
    \item The initial state is an ideal input state $\rho_0\in\Sin$.
    \item There exists an ideal output state which preserves any necessary secrets that is $\epssecint$-close to the output state of a single session $\vars{P}_i$, i.e. for any $\sigma\in \Sin$, there exists a state $\rho\in \Sout'$ such that
    \begin{equation*}
        \Delta(\vars{P}_i(\sigma),\rho)\leq\epssecint.
    \end{equation*}
    \item The set of ideal output states is a subset of the set of ideal input states $\Sout'\subseteq\Sin$.
\end{enumerate}
Then, for all $i\in[1,n_{sess}]$, there exists $\rho_i\in \Sout'$ such that
\begin{equation*}
    \Delta(\vars{P}^i(\rho_0),\rho_i)\leq i\epssecint.
\end{equation*}
\end{theorem}
\begin{proof} See Appendix~\ref{app:Multi_to_Single_Red}.
\end{proof}
The theorem reduces the security analysis to the analysis of single-session protocol.

\subsection{Security Conditions}

To analyse the single-session security, we define the set of ideal input states $\vars{S}_{in}$ and output states $\vars{S}_{out}'$.
We make two assumptions here, that the generated secret keys can be leaked to the adversary between sessions (e.g. if the keys are used for encryption), and that the bitstring $X$ is no longer accessible to any party at the end of a session except what is announced or leaked during the session.
We note that the latter assumption is also implicit in device-dependent QKD, where the bit values are freshly sampled in each session\footnote{For any correlations between bit values of different sessions, i.e. source correlation, it can be accounted in the same way as QKD, such as introducing a leakage channel.}.\\

For the ideal input state, we observe that the privacy amplification seed $R$ and authentication key $K_2$ should always be secret between sessions, while authentication masking keys $\KOTP_{1,i}$ that will no longer be utilised (i.e. for indices $i<\alpha,\alpha'$) can be traced out.
However, the secrets $(\Kh_1,\KOTP_{1,i})$ cannot in general be guaranteed to remain private since Alice may reveal tag $T_{\AV}$ to the adversary without an authentication check occurring prior to Alice's validation step.
This would leak some information on $(\Kh_1,\KOTP_{1,i})$, which we quantify by allowing the adversary in later sessions to generate $T_{\AV}$ for any chosen message $M$ prior to that session.
As such, the ideal input state can be defined by $\rho^{in}=\sum_{jj'}p_{jj'}\rhoinn{jj'}_{LSE}$,
\begin{widetext}
\begin{equation}
    \rhoinn{jj'}_{LSE}=\dyad{jj'}_{\alpha\alpha'}\otimes\tau_{K_2R\KOTP_{1,j}\cdots\KOTP_{1,m}}\otimes\bigepsilon_{K_1^h\KOTP_{1,j-1}E_{j-1}\rightarrow E}
    \circ\cdots\circ\bigepsilon_{K_1^h\KOTP_{1,j'}E_{j'}\rightarrow E_{j'+1}}(\tau_{\Kh_1\KOTP_{1,j'}\cdots \KOTP_{1,j-1}}\otimes\rho_{E_{j'}}),
\end{equation}
\end{widetext}
where $\tau_A=\sum_{a\in\vars{A}}\frac{1}{\abs{\vars{A}}}\dyad{a}$ is a uniformly distributed state of variable $A$, $\bigepsilon_{K_1^h\KOTP_{1,i}E_i\rightarrow E_{i+1}}$ are quantum channels that allow the adversary to generate $T_{\AV}$ for any chosen message, $L=\alpha\alpha'$ are the labels and $S=K_1^hK_2R\KOTP_{1,j}\cdots \KOTP_{1,m}$ refer to the secrets.\\

The set of ideal output states is defined based on the ideal functionality given the ideal input states, which are required to satisfy a mix of QKD and authentication conditions as well as the requirement for $\vars{S}_{out}'\subseteq\vars{S}_{in}$.
As such, we adapt a recently proposed classical authenticated key exchange (AKE) framework~\cite{Guilhem2020_AKE}, where the tasks of authentication and key secrecy are defined and their relation made explicit.
Let $\rho_{B}^a$ be the conditional state of system $B$ conditioned on event $A=a$, and $\rho_{B\land A=a}=p_a\rho_B^a$ be the sub-normalised conditional state.
This gives us a set of five security conditions (including one from $\vars{S}_{out}'\subseteq\vars{S}_{in}$), given by
\begin{enumerate}
    \item \textbf{Robustness}: In the absence of any adversary, the authentication should pass with high probability, i.e. $\Pr[\FA=\FB=1]\geq 1-\epsrob$.
    \item \textbf{Explicit Entity Authentication}: When Alice accepts the QAKE session, $\FA=1$, Bob must have accepted as well, $\FB=1$. This condition on Alice, termed full (f) explicit entity authentication, requires that $\Pr[\FA=1,\FB\neq1]\leq\epsEAf$. A similar condition, termed almost-full (af) explicit entity authentication, requires that when Bob accepts the session, $\FB=1$, Alice must be involved in the session, i.e. $\FA\neq\phi$. As such, the condition requires that $\Pr[\FA=\phi,\FB=1]\leq\epsEAaf$, with $\epsEA=\epsEAaf+\epsEAf$ (Recall $\FA=\phi$ if Alice did not participate in the session).
    \item \textbf{Match Security}: When Alice and Bob accepts the authentication, they will generate the same key, i.e. $\Pr[\KA\neq\KB,\FA=1,\FB=1]\leq\epsMS$.
    \item \textbf{Key Secrecy}: When keys are generated by Bob, they should remain uniform and secret from the adversary, i.e. $\rho_{K_BLSE\land(F_A,F_B)\in\{01,11\}}$ is $\epsKS$-close in trace distance to $\tau_{K_B}\otimes\rho_{LSE\land(F_A,F_B)\in\{01,11\}}$.
    \item \textbf{Shared Secrets Privacy}: The privacy of the shared secrets $S$ should be maintained as necessary, which is achieved by the output state being $\epsSP$-close to an ideal input state in $\vars{S}_{in}$. More formally, $\sum_{(\fA,\fB)\not\in\vars{S}_{\SP}}p_{\fA\fB}\rho^{\fA,\fB}_{LSE}$ is $\epsSP$-close in trace distance to $\sum_{(\fA,\fB)\not\in\vars{S}_{\SP}}p_{\fA\fB}\sigma^{\fA,\fB}_{LSE}$, where $\sigma^{f_A,f_B}_{LSE}\in\vars{S}_{in}$, $p_{\fA\fB}$ is the probability of $(F_A,F_B)=(\fA,\fB)$ and $\vars{S}_{SP}=\{0\phi,\phi0,10\}$. 
\end{enumerate}
The property of explicit entity authentication highlights an observation, also discussed in the context of QKD~\cite{Portmann2014_Authentication,Kiktenko2020_Lightweight,Portmann2022_Security}, where the abort decision may not be shared by both Alice and Bob, and one party could choose to generate keys while the other does not.
This condition is imposed on Alice and Bob differently.
Since Alice is the receiver of the final authentication tag, when Bob accepts and sends this final tag, he is unable to guarantee that Alice will accept since this final tag can be modified by the adversary, as noted also in Ref.~\cite{Kiktenko2020_Lightweight}.
We note that the match security condition matches the correctness condition in QKD and the key secrecy matches the secrecy condition in QKD.
For shared secrets privacy, the events $(F_A,F_B)=(0\phi),(\phi0),(10)$ are excluded since these events are accounted for in explicit entity authentication.\\

The steps to generalize the AKE framework to the quantum regime could be of independent interest, since similar the classical AKE framework could be generalized to other quantum variants.
Examples include quantum anonymous key exchange where authentication is made one-way, or QAKE with computationally secure message authentication.
The generalization of the AKE framework to QAKE can be found in Appendix~\ref{app:AKESecDefn}.\\

The set of ideal output states would thus be those satisfying conditions 2 to 5 perfectly.
The explicit entity authentication condition will define the combinations of $F_AF_B$ present in the ideal output state, with the match security requiring $K_A=K_B$ and key secrecy enforcing independence of the keys from $E$.
The shared secrets privacy requirement enforces that the output state forms a mixture of $\sigma^{f_Af_B}_{LSE}\in\vars{S}_{in}$.
Therefore, the set of ideal output state $\vars{S}_{out}'=\{\rhoidealint\}$ contains states of the form
\begin{widetext}
\begin{equation}
\label{eqn:QAKE_ideal_output_state}
\begin{split}
    &\rhoidealint=\dyad{\perp\perp}_{\KA\KB}\otimes\left[p_{0\phi}\dyad{0\phi}_{\FA\FB}\otimes \sigma^{0\phi}_{LSE}+p_{\phi0}\dyad{\phi0}_{\FA\FB}\otimes\sigma^{\phi0}_{LSE}+p_{00}\dyad{00}_{\FA\FB}\otimes \sigma^{00}_{LSE}\right]\\
    &+p_{01}\dyad{01,\perp}_{\FA\FB\KA}\otimes\tau_{\KB}\otimes\sigma^{01}_{LSE}+p_{11}\dyad{11}_{\FA\FB}\otimes\tilde{\tau}_{\KA\KB}\otimes\sigma^{11}_{LSE},
\end{split}
\end{equation}
\end{widetext}
where $\tilde{\tau}_{K_AK_B}=\frac{1}{\abs{\vars{K}}}\sum_{k\in\vars{K}}\dyad{kk}_{K_AK_B}$ is a uniformly distributed state of variable $K_A$ and $K_B$, with $K_A=K_B$.
The single-session security is thus given by the statement: for any $\sigma_{LSE}\in\vars{S}_{in}$, there exists $\rho_{F_AF_BK_AK_BLSE}\in\vars{S}_{out}'$ such that $\Delta(\vars{P}(\sigma_{LSE}),\rho_{F_AF_BK_AK_BLSE})\leq\varepsilon_{sec,int}$.\\

To simplify the security analysis, we can reduce the trace distance to the security conditions, similarly to  how QKD security can be reduced to correctness and secrecy,
\begin{theorem}
\label{thm:eps_sec_components}
    Consider a protocol $\vars{P}$ with an ideal input state $\sigma_{LSE}\in\Sin$. Then, overall protocol security is satisfied when the security parameter is
    \begin{equation*}
        \epssecint=\epsEA+\epsMS+\epsKS+\epsSP.
    \end{equation*}
\end{theorem}
\begin{proof}
See Appendix~\ref{app:Sec_Para_Compute}.
\end{proof}

\subsection{Single-session Security}

The defined set of input and output states satisfy conditions 1 and 3 in Thm.~\ref{thm:Multi_to_Single_Red} trivially, while the remaining condition is demonstrated in the theorem below, where we assume an ideal state preparation and fair sampling (transmission $\eta$ is basis-independent).
A sketch of the security proof is provided below, with the full security proof presented in Appendix~\ref{app:AKEProtocolSec_OverallSec}.
\begin{restatable}[]{theorem}{AKEProtMainProof}\label{thm:AKEProtocolMain}
    The QAKE protocol $\vars{P}$ is secure with security parameters 
    \begin{equation*}
    \begin{gathered}
        \epsEA=\epsMACa+2^{-\Psift_{\tol}[1+\hbin(f_{P_1})]+\log_2(\Psift_{\tol}+1)}+\frac{1}{\abs{\vars{T}_{\BV}}}\\
        \epsMS=\epsMACb\\
        \epsKS=2\sqrt{2\epsserfa}+\frac{1}{2}\times 2^{-\frac{1}{2}[H'-l_{\KB}]}+\varepsilon_{\vars{P}'}\\
        \epsSP=\varepsilon_{\SP,01,11}'+\varepsilon_{\vars{P}'},
    \end{gathered}
    \end{equation*}
    where $l_{K_B}$ is the length of $K_B$, $\varepsilon_{serf,1}$ is a penalty term associated with the use of modified Serfling bound~\cite{Tomamichel2012_BB84FiniteKey,Curty2014_DecoyMDIQKD} to estimate the single-photon bit error rate in $\Psift_2$ from that of $\Psift_1$. The penalty associated with shifting to an idealised protocol $\vars{P}'$ is
    \begin{gather*}
        \varepsilon_{\vars{P}'}=2\left(\epsMACa+\epsMACb+\epsds+\epsSO\right)\\
        \epsSO=2^{-\Psift_{tol}}+2^{-N_{\Psift_{2,1}}^{\tol}}+2^{-\Psift_{\tol}\hbin(f_{P_1})+\log_2(\Psift_{\tol}+1)},
    \end{gather*}
    the shared secrets privacy parameter for the idealised protocol when $\FB=1$ is
    \begin{gather*}
        \varepsilon_{\SP,01,11}'=4\sqrt{2\epsserfa}+2(\varepsilon_2+\varepsilon_3)\\
        +\varepsilon_{\SP,\MAC,1}+\varepsilon_{\SP,\MAC,2}\\
        \varepsilon_{\SP,\MAC,1}=\sqrt{(\abs{\vars{T}_{\AV}}\epsMACa-1)+2^{\log_2\left(\frac{2+\varepsilon_3}{\varepsilon_3\abs{\vars{T}_{\BV}}}\right)-H'}}\\
        \varepsilon_{\SP,\MAC,2}=\sqrt{(\abs{\vars{T}_{\BV}}\epsMACb-1)+2^{\log_2\left(\frac{2}{\varepsilon_2}+1\right)-H'}},
    \end{gather*}
    the smooth min-entropy term is
    \begin{equation*}
        H'=N_{\Psift_2,1}^{\tol}[1-\hbin(\ephtol')]-2-\log_2\abs{\vars{T}_{\AV}}\abs{\vars{T}_{\BV}}-\leakEC,
    \end{equation*}
    $\varepsilon_2$ and $\varepsilon_3$ are parameters to optimise over, $\abs{\vars{T}_{AV}}$ and $\abs{\vars{T}_{BV}}$ are the size (dimension) of the respective tag, $\Psift_{\tol}$ is the minimum tolerated number of sifted rounds, $\epsds$ accounts for errors associated with decoy state analysis, $N_{\Psift_{2,1}}^{\tol}$ is the minimum tolerated number of single-photon rounds in set $\Psift_2$, and $\ephtol'$ is the maximum tolerated single-photon phase error rate in set $\Psift_2$, and $\leakEC$ is the information leakage due to error correction.
\end{restatable}
\begin{proof}[Proof Sketch] 
The explicit entity authentication security condition relies on the fact that when $\FA=\phi$ (resp. $\FB\neq1$), the tag $T_{\AV}$ (resp. $T_{\BV}$) is not generated.
As such, the adversary has to guess a correct message-tag pair to trigger $F_B=1$ (resp. $F_A=1$).
By the nature of the strong 2-universal hash function, the probability guessing a correct message-tag pair is small.\\

For match security, like the correctness condition of QKD, the security is based on error verification. Bob's tag $T_{\BV}$ involves the corrected bitstring $\hat{X}_{\Psift_{2,\rr}}$, which has to be checked against tag $\tilde{T}_{\BV}$ involving $X_{\Psift_{2,\rr}}$.
If the bitstrings do not match (and thus the keys are mismatched as well), the probability of matching tags are small.\\

For key secrecy, message authentication provided by the authentication steps allows the security condition to be reduced to that of QKD with ``idealised" message authentication and where protocol step ordering is obeyed.
The proof
then follows from
that of QKD, where we seek to lower bound the smooth min-entropy $\Hmin^{\epssma}(\hat{X}_{\Psift_{2,\rr}}|\beta\Kh_1K_2\KOTP_{1,\beta}E)$.\\

Shared secrets privacy requires that necessary secrets remain private.
The main protection on $K_2$ and $(\Kh_1,\KOTP_{1,\beta})$ when the protocol succeeds and they are reused is based on the value of $\hat{X}_{\Psift_2}$ included as part of the message while generating tags $T_{\BV}$ and $T_{\AV}$.
Based on the ideal key privacy property~\cite{Fehr2017_SID} along with quantum leftover hash lemma~\cite{Tomamichel2011_QLHL}, the overall hash function $h_1$ with mask and $h_2$ are $\epsMACa$ and $\epsMACb$ strong extractors respectively, thereby preserving the secrecy of the authentication and mask keys $K_2$ and $(\Kh_1,\KOTP_{1,\beta})$.
\end{proof}

The final condition to consider is the robustness of the protocol.
The main sources of non-acceptance of a session would be either failure of the parameter estimation checks or error correction failure.
We can use concentration bounds~\cite{Zubkov2013_SVBound} to compute the probability of failure of the parameter estimation checks, which we label as $\varepsilon_{\ds,\rob}$ for terms associated with the decoy state checks, $\varepsilon_{\bigepsilon,P_1,\rob}$ for the bit error check in set $P_1$, and $\varepsilon_{\rob,P}$ for terms associated with bounding $\abs{P}$ (bounds $P^{\tol}\leq\abs{P}\leq P^{\UB}$).
The error correction failures can either be due to the bit error rate in set $P_2$ exceeding some tolerance bound, or failure to correct for errors even within the tolerance bound.
In the former case, knowing the bit error rate in set $P_1$ falls below $\ebitt{1,\tol}$, the probability that the error rate in set $P_2$ exceeding some tolerance value $\ebittol'$ can be computed based on the modified Serfling bound~\cite{Tomamichel2012_BB84FiniteKey,Curty2014_DecoyMDIQKD}, with failure probability $\epsserfb$.
The failure probability of the latter case depends on the error correction protocol utilised, and we define the error probability as $\varepsilon_{\EC}$ for total communication $l_{syn}$-bits during the error correction phase.
As such, the overall robustness can be quantified by $\epsrob$,
\begin{equation}
    \epsrob= \varepsilon_{\rob,P}+\epsserfb+\varepsilon_{\ds,\rob}+\varepsilon_{\bigepsilon,P_1,\rob}+\varepsilon_{\EC}.
\end{equation}

\section{Numerical Analysis}

We analyse the performance of the QAKE protocol by simulating the length of keys generated based on Thm.~\ref{thm:AKEProtocolMain}. 
We assume a simple experimental model with Alice preparing decoy BB84 states, sending it through a channel with loss $\eta$, and Bob performing measurement in a random basis, i.e. with detection probability $p_{\dett,\mu_j}=p_{\mu_j}(1-e^{-\eta\mu_j})$.
Additionally, Bob's detector is assumed to have zero dark counts. 
We fix $\mu_0=0.45$, $\mu_1=0.225$, $\mu_2=0$ as the decoy intensities, $p_{\mu_0}=0.2$, $p_{\mu_1}=0.6$ and $p_{\mu_2}=0.2$ as their respective probabilities.
We also fix the length of the authentication tags to be 80 bits each, with $\epsMACa=\epsMACb=2^{-80}$, i.e. the authentication is 2-universal, allowing us to use the tighter bound for $\epsSP$ shown in Thm.~\ref{thm:QAKESecProofSSFB1}.
The error correction is assumed to utilise \SI{30}{\percent} of the sifted bits.\\

The simulation is then performed by rearranging the result in Thm.~\ref{thm:AKEProtocolMain}, and optimising $l_{\KB}$ over the splitting ratio (also size of test set) $f_{P_1}$, robustness and secrecy parameters (components of $\varepsilon_{rob}$ and $\varepsilon_{sec,int}$) in Matlab, with other parameters fixed as described in the results.
We choose the security parameter of each session of the protocol as $\epssec=$ \SI{1e-15}{}, which yields an overall protocol security of \SI{1e-6}{} when we allow it to run up to \SI{1e9}{} sessions.
The robustness parameter is chosen at $\epsrob=$ \SI{1e-10}{}, and each session involves $n_{rnd}=$ \SI{1e10}{} signals being sent by Alice to Bob.
The concentration bound used for robustness parameters is the tight bound on binomial distribution~\cite{Zubkov2013_SVBound} due to i.i.d. state preparation noting independence of the detection probability on basis selection by fair sampling assumption.
The concentration bound used for estimation of the expectation values, $\expE{N_{\mu_v,P_i}}$, from the observed values, $N_{\mu_v,P_i}$, associated with decoy state estimation (errors form part of decoy state related parameters $\epsds$) is Kato's bound~\cite{Kato2020_KatoIneq,Curras2021_KatoIneqPara,Zhang2023}.
The simulation results are shown in Fig.~\ref{fig:Key_Rate_Simulation_Main} for various QBER values $Q$.\\

In terms of key generation rate, QAKE has a slightly higher key rate than standard QKD, though the extent of improvement depends on the length of the authentication tags used and the number of communication rounds.
If we use the QAKE protocol and revert it to a standard QKD protocol, i.e. using authenticated communication channel and having error verification step, one would need to generate seven authentication tags (2 for label agreement, 2 for sifting, 1 for test round announcement, 1 for error correction, 1 for error verification) instead of two, leading to an additional 400 bits of key utilized.
Therefore, QAKE may provide a significant advantage when key generation rates are low, such as when block sizes are small (e.g. satellite QKD) or loss is high, or when there is a large number of communication rounds (e.g. for Cascade) -- since more bits have to be sacrificed for message authentication.

\begin{figure}[ht]
    \centering
    \includegraphics[width=\linewidth]{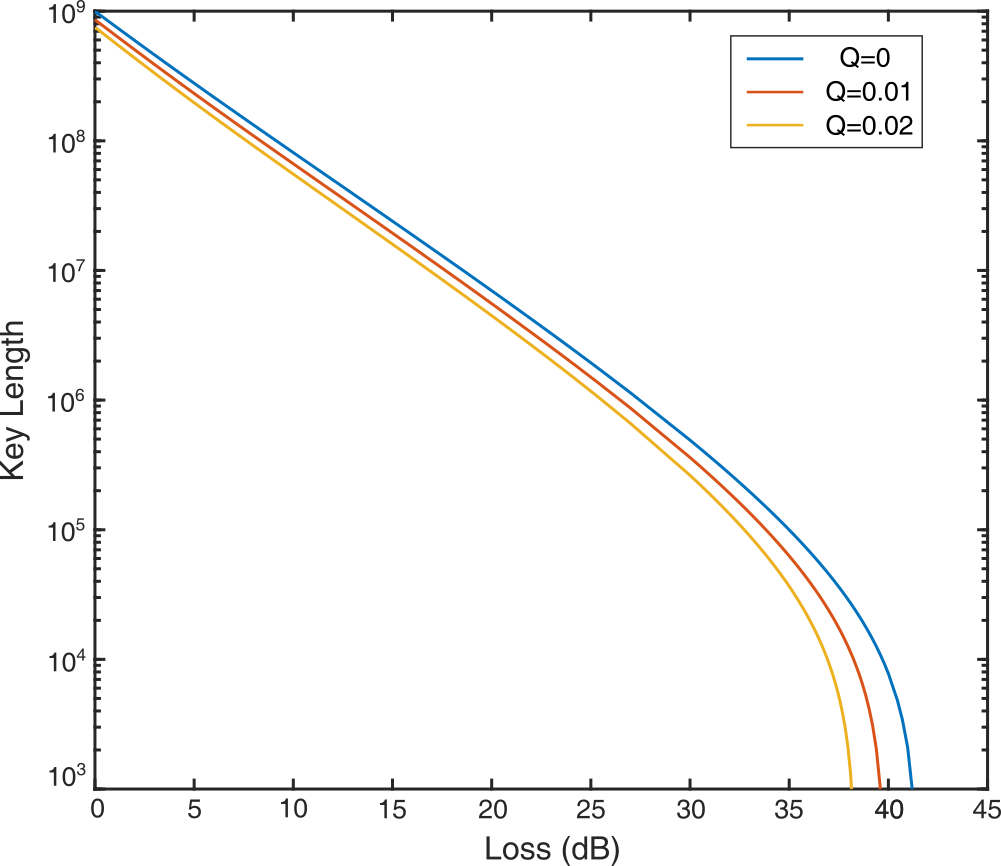}
    \caption{Simulated key length for QAKE protocol at various signal loss values (including detector loss), various QBER values $Q$, with $n_{rnd}=$ \SI{2e10}{} signals sent, $\epsrob=$ \SI{1e-10}{}, and $\epssec=$ \SI{1e-15}{}. The authentication tag used is 80-bits, i.e. $\abs{T_{\AV}}=\abs{T_{\BV}}=80$, and we assume the error correction to utilise \SI{30}{\percent} of the sifted bits. The system parameters are chosen to be $\mu_0=0.45$, $\mu_1=0.225$, $\mu_2=0$ as the decoy intensities, $p_{\mu_0}=0.2$, $p_{\mu_1}=0.6$ and $p_{\mu_2}=0.2$ as their respective probabilities, and we assume the system to have zero dark counts.}
    \label{fig:Key_Rate_Simulation_Main}
\end{figure}

\section{Experimental Validation}

\subsection{Experimental Setup}

We validated our results by implementing the QAKE protocol on the ID Quantique Clavis XGR QKD systems, with the setup shown in Fig.~\ref{Exp_setup}.
We connected Alice and Bob in a back-to-back configuration with an optical attenuator of 10~dB emulating a fiber distance of \SI{45}{\kilo\metre} (\SI{0.22}{\decibel\per\kilo\metre}). 
We also configured the intensity $\mu_0=0.45$ and the probability of Alice selecting each basis to about \SI{50}{\percent}. 
The QKD systems provide QKD raw data, which includes the encoded qubit values from the transmitter (Alice) and the measured qubit values from the receiver (Bob) during the key exchange process. 
We collected the raw data using an Application Programming Interface (API) command in a text format, which was utilised to estimate the QBER and detection rate.\\ 

\begin{figure}[!ht]
    \centering
    \includegraphics[width=0.7\linewidth]{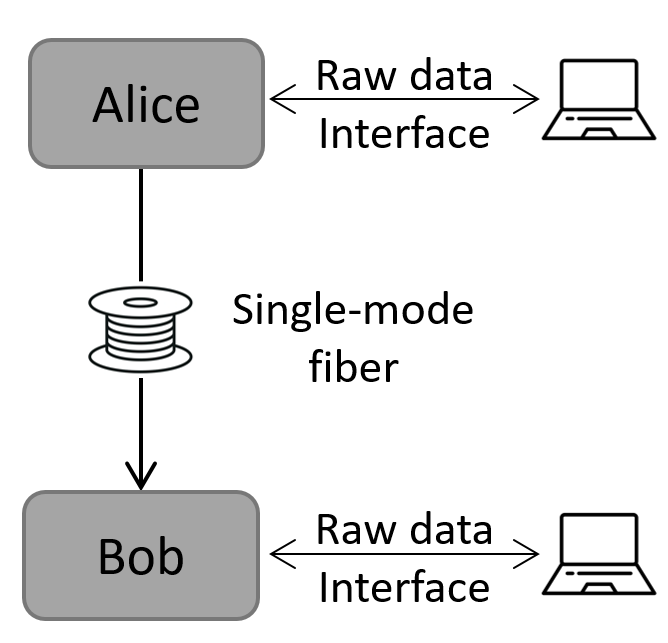}
    \caption{\label{Exp_setup}Experimental setup implementing the protocol, with \SI{10}{\decibel} of channel loss from a \SI{45}{\kilo\meter} fiber (emulated with an optical attenuator). The quantum state preparation and measurement are performed by the ID Quantique Clavis XGR QKD systems, which outputs the relevant raw data (including state prepared and measurement outcomes). Post-selection is performed on the raw data to match the protocol requirements of having balanced (\SI{50}{\percent}) detector basis choice. The remaining steps of the protocols are then simulated (including error correction, authentication tag generation and privacy amplification) using the post-selected data.}
\end{figure}

Due to the lack of a direct interface between the QKD system's raw data output and our FPGA, the post-processing of the raw data was performed offline (not in real-time).
Firstly, to balance the basis detection rates, we performed post-selection of the data such that the detection probability in each basis is about \SI{50}{\percent}.
Secondly, sifting is performed to ensure that Alice and Bob has matching basis.
The data is then utilised to run the protocol in a software stack.
This includes the use of Toeplitz hashing as the XOR two-universal hash function and as the the two-universal hash function respectively for the authentication step and privacy amplification steps, and the use of a two-way \emph{low-density parity-check} (LDPC) code for error correction.
Details of the LDPC and Toeplitz hashing implementation are presented in the subsequent subsections.\\

We utilise similar parameters for authentication tag size as the numerical simulation, while the number of bits communicated for error correction is set to a maximum of \SI{30}{\percent} of the size of set $P_2$ (containing bits used for key generation) based on LDPC performance.

\subsection{Toeplitz Hashing}

Toeplitz hashing is implemented on a Virtex Ultrascale+ XCVU9p FPGA, running at a clock frequency of 332 MHz. 
Due to the limitation of resources in FPGA, we decomposed the large Toeplitz matrix evenly into small submatrix multiplication steps~\cite{Zhang2016_Toeplitz}. 
We designed pipelined modules for Toeplitz submatrix construction, matrix multiplication and intermediate accumulation for processing 32 input bits from the FPGA interface on every clock cycle. 
The matrix construction module consists of a 32-bit wide FIFO that stores the large Toeplitz matrix, and a shift register that stores the submatrix. 
On every clock cycle, the submatrix is constructed by reading 32-bits from the FIFO and using it to shifting the register. 
Matrix multiplication is realized by bitwise-AND and XOR-reduction. 
An accumulation register performs bitwise-XOR for matrix multiplication outputs on each clock cycle, to generate the final output bits.

\subsection{Error Correction}

For error correction, we implemented a two-way \emph{low-density parity-check} (LDPC) code on the same FPGA by adapting LDPC IP cores from Xilinx~\cite{LDPC_IP}.
We used the following settings on the LDPC IP: base matrix choice $\text{bg}=0$, syndrome length control $\text{mb}=5$, codeword length controls $z_{\text{set}}=5$ and $z_j=5$, normalisation factor $\text{sc}_{\text{idx}}=14$, and scaling $\text{llr}_{\text{scale}}=14$.
This results in a codeword size of 1188 bytes and syndrome size of 220 bytes.
The LDPC IP core is designed for standard error correction usage (e.g. in 5G), where the syndrome is also sent through the channel and experiences the same error rate as the data, unlike in QAKE where the syndrome is assumed to be error-free.
Let $X$ be a 1188-byte data string ($X'$ being Bob's data string), $X_1$ and $X_2$ be the first 968 bytes and last 220 bytes of data, and $syn$ be the syndrome. The adapted protocol we use is as follows:
\begin{enumerate}
    \item For every 1188 bytes of data, Alice generates a 220-byte syndrome from the first 968 bytes of data, $X_1$, and XOR the syndrome with $X_2$, i.e. $syn'=syn\oplus X_1$.
    \item Alice computes a 16-byte (128-bit) hash of the data string using SHA-256 (second half of the hash is discarded), $T=\text{SHA256}(X)$.
    \item Alice sends $syn'$ and $T$ to Bob.
    \item Bob performs XOR of his last 220 bytes of data with the syndrome, i.e. $syn''=syn'\oplus X_2'$. Note this is equivalent to introducing the error bits $X_2'\oplus X_2$ onto the syndrome, $syn''$.
    \item Bob performs error correction with $X_1'$ and $syn''$ using the LDPC IP core, and obtains a corrected data string $\hat{X}_1$ and corrected syndrome $\hat{syn}$.
    \item Bob computes the correction on the second part of his data string, with $\hat{X}_2=X_2'\oplus\hat{syn}\oplus syn'$.
    \item Bob computes the hash of the corrected data string using SHA-256, $T'=\text{SHA256}(\hat{X})$, and compares it with the hash sent by Alice.
    \item If the hash does not match, Bob will inform Alice to increase the syndrome length by 44 bytes, and send that additional 44 bytes to Bob.
    \item Alice and Bob repeat steps 4 to 8 until the hash matches. If the syndrome length longer than the maximum transmission unit (MTU) size the FPGA can support (mb = 10, syndrome = 440 bytes, maximum of 6 rounds), Alice will simply reveal all 1188 bytes of data to Bob instead.
\end{enumerate}
Based on our tests for a bit error rate of \SI{2}{\percent}, the LDPC code can almost always succeed with a sacrifice of at most \SI{30}{\percent} of the sifted bitstring.

\begin{table*}[!ht]
\centering
\caption{\label{table:ExptResults} Experimental results for three runs of the protocol with $n_{rnd}$ number of signals sent (after post-selection) for each run. The detection probability per signal, $p_{\dett}$, and average bit error rate is computed from the measurement results. The key length is optimised over the various epsilon terms, including errors associated with decoy state method, $\epsds$, and the modified Serfling bound error for single-photon bit error rate $\epsserfa$, with the optimal values presented here, for secrecy parameter $\epssec=$ \SI{1e-15}{}.}
\begin{tabular}{cccccccc}
    Run & $n_{rnd}$ & $p_{\dett}$ & QBER & Key Length & $\epsds$ & $\epsserfa$ \\
    \hline
    1 & \SI{2.04e10}{} & \SI{1.03e-4}{} & 0.0196 & \SI{1.77e5}{} & \SI{1.30e-16}{} & \SI{3.17e-33}{} \\
    2 & \SI{2.36e10}{} & \SI{1.03e-4}{} & 0.0202 & \SI{2.13e5}{} & \SI{1.29e-16}{} & \SI{3.24e-33}{} \\
    3 & \SI{2.13e10}{} & \SI{1.03e-4}{} & 0.0203 & \SI{2.13e5}{} & \SI{1.30e-16}{} & \SI{3.20e-34}{} \\
    \hline
\end{tabular}
\end{table*}

\subsection{Experimental Results}

We run the protocol a total of three separate times, and fix a testing probability of \SI{25}{\percent}, which is close to the optimal testing probability at the same transmission range in the numerical simulation. 
Table~\ref{table:ExptResults} show the results obtained for the various runs. 
All three runs recorded around \SI{2e10}{} signals sent after post-selection, with an average bit error rate of \SI{2}{\percent} and key length of \SI{2e5}{}. 
This is roughly in line with the simulated results for an ideal setting presented in Fig.~\ref{fig:Key_Rate_Simulation_Main}.

\section{Discussion}
\label{sec:Discussion}
We presented a QAKE protocol that incorporates entity authentication with key expansion.
It compresses the authentication steps to the last two communication steps to reduce authentication cost.
The protocol is also designed to allow for key recycling, where the authentication keys and masking keys need not be refreshed after each protocol session.
Key recycling can be useful for applications since communication overhead from the sharing of the refreshed secrets by the server (e.g. between all remote access servers) can be removed when the protocol passes.\\

\textbf{Open Problems:} There are some limitations to the presented protocol and security analysis, which can open up areas that can be explored in the future.\\

An interesting area to explore is the expansion of QAKE or QKD to achieve other tasks within the classical AKE framework.
For instance, classical AKE can allow for anonymous key exchange by not requiring mutual authentication for all parties (perhaps necessary only to authenticate the client but not for the client to authenticate the server), or allow for client-server authentication (without key generation).
Such expansion of tasks may require protocol changes that can be informed by the modifications within the classical AKE framework, which is well-studied by the classical cryptography community.\\

It is also worthwhile to have a more comprehensive AKE framework that incorporates both classical and quantum AKE.
This brings the analysis closer to a practical setting, where both classical authentication and key exchange methods may live with QAKE and other quantum protocols.
Analysing the interplay between these protocols and their respective guarantees can provide insight into how these diverse systems can be incorporated in practice, and allow us to understand the types of attacks and risks we should consider in QKD or QAKE design.
Moreover, such an endeavour can build a common language with classical cryptography that would make QKD or similar quantum cryptographic protocol security more accessible.\\

\textbf{Additional Results:} Pre-sharing of basis information between Alice and Bob can avoid penalties to key rate from sifting, but renders the protocol inefficient due to the large size of $\theta$ required.
The most practical solution to this problem is to use a shorter pre-shared basis seed and generate a longer string as the shared basis from a quantum-safe \textit{pseudorandom number generator} (PRNG) like AES~\cite{NIST2001_AES}.
In Appendix~\ref{app:PRNG_QAKE}, we present a QAKE protocol that is adapted from the authentication protocol by Fehr et. al.~\cite{Fehr2017_SID}, which utilises pre-shared pseudorandom basis, and we present a formal analysis of the use of pseudorandom basis in decoy state BB84 in Appendix~\ref{app:PRNGBasisChoiceProof}.
We note that the secret keys generated remains information-theoretic secure, provided that the adversary is computationally bounded during the protocol itself.
Inspired by the authentication achieved by Fehr et. al.~\cite{Fehr2017_SID} with a single communication step, we also developed a two-step (one challenge and one response) client authentication protocol that is secure in the practical setting (with channel loss, noise and multi-photon events).
This result is presented in Appendix~\ref{app:CAProtocol}.

\begin{acknowledgments}
This paper was prepared for informational purposes by the Global Technology Applied Research center of JPMorgan Chase \& Co. This paper is not a product of the Research Department of JPMorgan Chase \& Co. or its affiliates. Neither JPMorgan Chase \& Co. nor any of its affiliates makes any explicit or implied representation or warranty and none of them accept any liability in connection with this paper, including, without limitation, with respect to the completeness, accuracy, or reliability of the information contained herein and the potential legal, compliance, tax, or accounting effects thereof. This document is not intended as investment research or investment advice, or as a recommendation, offer, or solicitation for the purchase or sale of any security, financial instrument, financial product or service, or to be used in any way for evaluating the merits of participating in any transaction.
All contributions by J.C., K.H.Y.L., O. Alia, M.P. and C.L. were made while working at JPMorgan Chase.
\end{acknowledgments}

\section*{Author Contributions}

W.Y.K. and K.C. conceptualized the work; W.Y.K. and K.C. performed the formal analysis; J.C. and K.H.Y.L. designed and implemented the software; O. Alia and J.C. performed the experiment and data collection; C.L. and M.P. provided overall supervision of the project; all authors contributed to the drafting and reviewing of the manuscript.

\bibliographystyle{plainnat}
\bibliography{SID_Biblio}

\begin{thebibliography}{48}
\providecommand{\natexlab}[1]{#1}
\providecommand{\url}[1]{\texttt{#1}}
\expandafter\ifx\csname urlstyle\endcsname\relax
  \providecommand{\doi}[1]{doi: #1}\else
  \providecommand{\doi}{doi: \begingroup \urlstyle{rm}\Url}\fi

\bibitem[{Advanced Micro Devices}()]{LDPC_IP}
{Advanced Micro Devices}.
\newblock ldpc encoder/decoder.
\newblock URL \url{https://www.xilinx.com/products/intellectual-property/ef-di-ldpc-enc-dec.html}.

\bibitem[Barker(2020)]{NIST_Key_Management}
E.~Barker.
\newblock {Recommendation for Key Management: Part 1 – General}.
\newblock May 2020.
\newblock \doi{10.6028/nist.sp.800-57pt1r5}.

\bibitem[Barker and Kelsey(2015)]{NIST_PRNG}
E.~Barker and J.~Kelsey.
\newblock {Recommendation for Random Number Generation Using Deterministic Random Bit Generators}.
\newblock June 2015.
\newblock \doi{10.6028/nist.sp.800-90ar1}.

\bibitem[Bellare et~al.(1997)Bellare, Desai, Jokipii, and Rogaway]{Bellare1997_PRNG}
M.~Bellare, A.~Desai, E.~Jokipii, and P.~Rogaway.
\newblock A concrete security treatment of symmetric encryption.
\newblock In \emph{Proceedings 38th Annual Symposium on Foundations of Computer Science}, pages 394--403, 1997.
\newblock \doi{10.1109/SFCS.1997.646128}.

\bibitem[Brassard and Salvail(1994)]{Brassard1994_Cascade}
G.~Brassard and L.~Salvail.
\newblock Secret-key reconciliation by public discussion.
\newblock In Tor Helleseth, editor, \emph{Advances in Cryptology --- EUROCRYPT '93}, pages 410--423, Berlin, Heidelberg, 1994. Springer Berlin Heidelberg.
\newblock ISBN 978-3-540-48285-7.
\newblock \doi{10.1007/3-540-48285-7_35}.

\bibitem[Brzuska et~al.(2011)Brzuska, Fischlin, Warinschi, and Williams]{Brzuska2011_BR}
C.~Brzuska, M.~Fischlin, B.~Warinschi, and S.~C. Williams.
\newblock Composability of bellare-rogaway key exchange protocols.
\newblock In \emph{Proceedings of the 18th ACM Conference on Computer and Communications Security}, CCS '11, page 51–62, New York, NY, USA, 2011. Association for Computing Machinery.
\newblock ISBN 9781450309486.
\newblock \doi{10.1145/2046707.2046716}.

\bibitem[Canetti(2000)]{Canetti2000_Composable}
R.~Canetti.
\newblock Security and composition of multiparty cryptographic protocols.
\newblock \emph{J. Cryptology}, 13:\penalty0 143–202, 2000.
\newblock \doi{10.1007/s001459910006}.

\bibitem[Canetti(2001)]{Canetti2001_Composable}
R.~Canetti.
\newblock Universally composable security: A new paradigm for cryptographic protocols.
\newblock In \emph{Proceedings of the 42nd IEEE Symposium on Foundations of Computer Science}, FOCS '01, page 136, USA, 2001. IEEE Computer Society.
\newblock ISBN 0769513905.
\newblock \doi{10.1109/sfcs.2001.959888}.

\bibitem[Carter and Wegman(1979)]{Carter1977_WegmanCarter}
J.~L. Carter and M.~N. Wegman.
\newblock Universal classes of hash functions.
\newblock \emph{Journal of Computer and System Sciences}, 18:\penalty0 143--154, 1979.
\newblock \doi{10.1016/0022-0000(79)90044-8}.

\bibitem[Cover and Thomas(2005)]{Cover2005_InfTheory}
T.~M. Cover and J.~A. Thomas.
\newblock \emph{Information Theory and Statistics}, chapter~11, pages 347--408.
\newblock John Wiley \& Sons, Ltd, 2005.
\newblock ISBN 9780471748823.
\newblock \doi{10.1002/047174882X.ch11}.

\bibitem[Currás-Lorenzo et~al.(2021)Currás-Lorenzo, Navarrete, Azuma, Kato, Curty, and Razavi]{Curras2021_KatoIneqPara}
G.~Currás-Lorenzo, Á. Navarrete, K.~Azuma, G.~Kato, M.~Curty, and M.~Razavi.
\newblock Tight finite-key security for twin-field quantum key distribution.
\newblock \emph{npj Quantum Information}, 7:\penalty0 22, 2021.
\newblock \doi{10.1038/s41534-020-00345-3}.

\bibitem[Curty et~al.(2014)Curty, Xu, Cui, Lim, Tamaki, and Lo]{Curty2014_DecoyMDIQKD}
M.~Curty, F.~Xu, W.~Cui, C.~C.~W. Lim, K.~Tamaki, and H.-K. Lo.
\newblock Finite-key analysis for measurement-device-independent quantum key distribution.
\newblock \emph{Nature Communications}, 5:\penalty0 3732, 2014.
\newblock \doi{10.1038/ncomms4732}.

\bibitem[de~Saint~Guilhem et~al.(2020)de~Saint~Guilhem, Fischlin, and Warinschi]{Guilhem2020_AKE}
C.~D. de~Saint~Guilhem, M.~Fischlin, and B.~Warinschi.
\newblock Authentication in key-exchange: Definitions, relations and composition.
\newblock In \emph{2020 IEEE 33rd Computer Security Foundations Symposium (CSF)}, pages 288--303, 2020.
\newblock \doi{10.1109/CSF49147.2020.00028}.

\bibitem[Dodis et~al.(2013)Dodis, Pointcheval, Ruhault, Vergniaud, and Wichs]{Dodis2013_PRNG}
Y.~Dodis, D.~Pointcheval, S.~Ruhault, D.~Vergniaud, and D.~Wichs.
\newblock Security analysis of pseudo-random number generators with input: /dev/random is not robust.
\newblock In \emph{Proceedings of the 2013 ACM SIGSAC Conference on Computer \& Communications Security}, CCS '13, page 647–658, New York, NY, USA, 2013. Association for Computing Machinery.
\newblock ISBN 9781450324779.
\newblock \doi{10.1145/2508859.2516653}.

\bibitem[Doosti et~al.(2021)Doosti, Delavar, Kashefi, and Arapinis]{Doosti2021_QUnforgeability}
M.~Doosti, M.~Delavar, E.~Kashefi, and M.~Arapinis.
\newblock A unified framework for quantum unforgeability, 2021.
\newblock arXiv:2103.13994.

\bibitem[Dworkin et~al.(2001)Dworkin, Barker, Nechvatal, Foti, Bassham, Roback, and Dray]{NIST2001_AES}
M.~Dworkin, E.~Barker, J.~Nechvatal, J.~Foti, L.~Bassham, E.~Roback, and J.~Dray.
\newblock Advanced encryption standard (aes).
\newblock 2001-11-26 2001.
\newblock \doi{10.6028/NIST.FIPS.197}.

\bibitem[Fehr and Salvail(2017)]{Fehr2017_SID}
S.~Fehr and L.~Salvail.
\newblock Quantum authentication and encryption with key recycling.
\newblock In Jean-S{\'e}bastien Coron and Jesper~Buus Nielsen, editors, \emph{Advances in Cryptology -- EUROCRYPT 2017}, pages 311--338, Cham, 2017. Springer International Publishing.
\newblock ISBN 978-3-319-56617-7.
\newblock \doi{10.1007/978-3-319-56617-7_11}.

\bibitem[Hoang and Shen(2020)]{Hoang2020_PRNG}
V.~T. Hoang and Y.~Shen.
\newblock Security analysis of nist ctr-drbg.
\newblock In \emph{Advances in Cryptology – CRYPTO 2020}, volume 12170. Springer, Cham, 2020.
\newblock ISBN 978-3-030-56784-2.
\newblock \doi{10.1007/978-3-030-56784-2_8}.

\bibitem[Kato(2020)]{Kato2020_KatoIneq}
G.~Kato.
\newblock Concentration inequality using unconfirmed knowledge, 2020.
\newblock arXiv:2002.04357.

\bibitem[Katz and Lindell(2014)]{Katz2014_Crypto}
J.~Katz and Y.~Lindell.
\newblock \emph{Introduction to Modern Cryptography}.
\newblock Chapman \& Hall/CRC, 2014.
\newblock ISBN 1466570261.
\newblock \doi{doi.org/10.1201/9781351133036}.

\bibitem[Kiktenko et~al.(2020)Kiktenko, Malyshev, Gavreev, Bozhedarov, Pozhar, Anufriev, and Fedorov]{Kiktenko2020_Lightweight}
E.~O. Kiktenko, A.~O. Malyshev, M.~A. Gavreev, A.~A. Bozhedarov, N.y~O. Pozhar, M.~N. Anufriev, and A.~K. Fedorov.
\newblock Lightweight authentication for quantum key distribution.
\newblock \emph{IEEE Transactions on Information Theory}, 66\penalty0 (10):\penalty0 6354--6368, 2020.
\newblock \doi{10.1109/TIT.2020.2989459}.

\bibitem[Konig et~al.(2009)Konig, Renner, and Schaffner]{Konig2009_OprMinMax}
R.~Konig, R.~Renner, and C.~Schaffner.
\newblock The operational meaning of min- and max-entropy.
\newblock \emph{IEEE Transactions on Information Theory}, 55\penalty0 (9):\penalty0 4337--4347, 2009.
\newblock \doi{10.1109/TIT.2009.2025545}.

\bibitem[Leermakers and Skoric(2019)]{Leermakers2019_QKRNoise}
D.~Leermakers and B.~Skoric.
\newblock Security proof for quantum key recycling with noise.
\newblock \emph{Quantum Information \& Computation}, 19:\penalty0 0913--0934, 2019.
\newblock \doi{10.26421/QIC19.11-12-2}.

\bibitem[Lim et~al.(2014)Lim, Curty, Walenta, Xu, and Zbinden]{Lim2014_DecoyQKD}
C.~C.~W. Lim, M.~Curty, N.~Walenta, F.~Xu, and H.~Zbinden.
\newblock Concise security bounds for practical decoy-state quantum key distribution.
\newblock \emph{Physical Review A}, 89:\penalty0 022307, Feb 2014.
\newblock \doi{10.1103/PhysRevA.89.022307}.

\bibitem[Mitzenmacher and Upfal(2017)]{Mitzenmacher2017_Prob}
M.~Mitzenmacher and E.~Upfal.
\newblock \emph{Probability and Computing: Randomization and Probabilistic Techniques in Algorithms and Data Analysis}.
\newblock Cambridge University Press, USA, 2nd edition, 2017.
\newblock ISBN 110715488X.
\newblock \doi{10.1017/CBO9780511813603}.

\bibitem[Mosca et~al.(2013)Mosca, Stebila, and Ustao{\u{g}}lu]{Mosca2013_QKDAKE}
M.~Mosca, D.~Stebila, and B.~Ustao{\u{g}}lu.
\newblock Quantum key distribution in the classical authenticated key exchange framework.
\newblock In P.~Gaborit, editor, \emph{Post-Quantum Cryptography}, pages 136--154, Berlin, Heidelberg, 2013. Springer Berlin Heidelberg.
\newblock ISBN 978-3-642-38616-9.
\newblock URL \url{https://link.springer.com/chapter/10.1007/978-3-642-38616-9_9}.

\bibitem[{National Institute of Standards and Technology}(2024)]{FIPS203}
{National Institute of Standards and Technology}.
\newblock Module-lattice-based key-encapsulation mechanism standard.
\newblock \penalty0 (FIPS 203), 2024.
\newblock \doi{10.6028/nist.fips.203}.

\bibitem[Portmann(2014)]{Portmann2014_Authentication}
C.~Portmann.
\newblock Key recycling in authentication.
\newblock \emph{IEEE Transactions on Information Theory}, 60\penalty0 (7):\penalty0 4383--4396, 2014.
\newblock \doi{10.1109/TIT.2014.2317312}.

\bibitem[Portmann(2017)]{Portmann2017_Authentication}
C.~Portmann.
\newblock Quantum authentication with key recycling.
\newblock In \emph{Advances in Cryptology -- EUROCRYPT 2017}, pages 339--368, Cham, 2017. Springer International Publishing.
\newblock ISBN 978-3-319-56617-7.
\newblock \doi{10.1007/978-3-319-56617-7_12}.

\bibitem[Portmann and Renner(2022)]{Portmann2022_Security}
C.~Portmann and R.~Renner.
\newblock Security in quantum cryptography.
\newblock \emph{Reviews of Modern Physics}, 94:\penalty0 025008, Jun 2022.
\newblock \doi{10.1103/RevModPhys.94.025008}.

\bibitem[Price et~al.(2021)Price, Rarity, and Erven]{Price2021_DDOSPRNGEncryptBasis}
A.~B. Price, John~G. Rarity, and C.~Erven.
\newblock A quantum key distribution protocol for rapid denial of service detection.
\newblock \emph{EPJ Quantum Technology}, 7:\penalty0 8, 2021.
\newblock \doi{10.1140/epjqt/s40507-020-00084-6}.

\bibitem[Renes and Renner(2012)]{Renes2012_MaxEnt}
J.~M. Renes and R.~Renner.
\newblock One-shot classical data compression with quantum side information and the distillation of common randomness or secret keys.
\newblock \emph{IEEE Transactions on Information Theory}, 58\penalty0 (3):\penalty0 1985--1991, 2012.
\newblock \doi{10.1109/TIT.2011.2177589}.

\bibitem[Rivest et~al.(1978)Rivest, Shamir, and Adleman]{RSA1978}
R.~L. Rivest, A.~Shamir, and L.~Adleman.
\newblock A method for obtaining digital signatures and public-key cryptosystems.
\newblock \emph{Communications of the ACM}, 21\penalty0 (2):\penalty0 120–126, feb 1978.
\newblock ISSN 0001-0782.
\newblock \doi{10.1145/359340.359342}.

\bibitem[Rogaway(1999)]{Rogaway2002_BucketHashing}
P.~Rogaway.
\newblock Bucket hashing and its application to fast message authentication.
\newblock \emph{Journal of Cryptology}, 12:\penalty0 91--115, 1999.
\newblock \doi{10.1007/PL00003822}.

\bibitem[Shor(1994)]{Shor1994}
P.W. Shor.
\newblock Algorithms for quantum computation: discrete logarithms and factoring.
\newblock In \emph{Proceedings 35th Annual Symposium on Foundations of Computer Science}, pages 124--134, 1994.
\newblock \doi{10.1109/SFCS.1994.365700}.

\bibitem[Stinson(1994)]{Stinson1994_UniversalHashing}
D.R. Stinson.
\newblock Universal hashing and authentication codes.
\newblock \emph{Designs, Codes and Cryptography}, 4:\penalty0 369--380, 1994.
\newblock \doi{10.1007/BF01388651}.

\bibitem[Tomamichel(2015)]{Tomamichel2015_ITBook}
M.~Tomamichel.
\newblock \emph{Quantum Information Processing with Finite Resources}.
\newblock Springer Briefs in Mathematical Physics. Springer Cham, 2015.
\newblock \doi{10.1007/978-3-319-21891-5}.

\bibitem[Tomamichel and Leverrier(2017)]{Tomamichel2017_QKDProof}
M.~Tomamichel and A.~Leverrier.
\newblock A largely self-contained and complete security proof for quantum key distribution.
\newblock \emph{Quantum}, 1:\penalty0 14, 2017.
\newblock \doi{10.22331/q-2017-07-14-14}.

\bibitem[Tomamichel and Renner(2011)]{Tomamichel2011_EUR}
M.~Tomamichel and R.~Renner.
\newblock Uncertainty relation for smooth entropies.
\newblock \emph{Physical Review Letters}, 106:\penalty0 110506, Mar 2011.
\newblock \doi{10.1103/PhysRevLett.106.110506}.

\bibitem[Tomamichel et~al.(2010)Tomamichel, Colbeck, and Renner]{Tomamichel2010_MaxMinDuality}
M.~Tomamichel, R.~Colbeck, and R.~Renner.
\newblock Duality between smooth min- and max-entropies.
\newblock \emph{IEEE Transactions on Information Theory}, 56\penalty0 (9):\penalty0 4674--4681, 2010.
\newblock \doi{10.1109/TIT.2010.2054130}.

\bibitem[Tomamichel et~al.(2011)Tomamichel, Schaffner, Smith, and Renner]{Tomamichel2011_QLHL}
M.~Tomamichel, C.~Schaffner, A.~Smith, and R.~Renner.
\newblock Leftover hashing against quantum side information.
\newblock \emph{IEEE Transactions on Information Theory}, 57\penalty0 (8):\penalty0 5524--5535, 2011.
\newblock \doi{10.1109/TIT.2011.2158473}.

\bibitem[Tomamichel et~al.(2012)Tomamichel, Lim, Gisin, and Renner]{Tomamichel2012_BB84FiniteKey}
M.~Tomamichel, C.~C.~W. Lim, N.~Gisin, and R.~Renner.
\newblock Tight finite-key analysis for quantum cryptography.
\newblock \emph{Nature Communications}, 3\penalty0 (1):\penalty0 634, 2012.
\newblock \doi{10.1038/ncomms1631}.

\bibitem[Trushechkin et~al.(2018)Trushechkin, Tregubov, Kiktenko, Kurochkin, and Fedorov]{Trushechkin2018_PRNGQKD}
A.~S. Trushechkin, P.~A. Tregubov, E.~O. Kiktenko, Y.~V. Kurochkin, and A.~K. Fedorov.
\newblock Quantum-key-distribution protocol with pseudorandom bases.
\newblock \emph{Physical Review A}, 97:\penalty0 012311, Jan 2018.
\newblock \doi{10.1103/PhysRevA.97.012311}.

\bibitem[Xue and Calabretta(2022)]{xue2022nanosecond}
X.~Xue and Ni. Calabretta.
\newblock Nanosecond optical switching and control system for data center networks.
\newblock \emph{Nature Communications}, 13\penalty0 (1):\penalty0 2257, 2022.
\newblock \doi{10.1038/s41467-022-29913-1}.

\bibitem[Yeo et~al.(2010)Yeo, Xu, Liaw, Wang, Wang, and Cheng]{yeo2010448}
Y.-K. Yeo, Z.~Xu, C.-Y. Liaw, D.~Wang, Y.~Wang, and T.-H. Cheng.
\newblock A 448$\times$ 448 optical cross-connect for high-performance computers and multi-terabit/s routers.
\newblock In \emph{Optical Fiber Communication Conference}, page OMP6. Optica Publishing Group, 2010.
\newblock \doi{10.1364/ofc.2010.omp6}.

\bibitem[Zhang et~al.(2023)Zhang, Zeng, Ye, Lo, and Ma]{Zhang2023}
X.~Zhang, P.~Zeng, T.~Ye, H.-K. Lo, and X.~Ma.
\newblock Quantum complementarity approach to device-independent security.
\newblock \emph{Physical Review Letters}, 131:\penalty0 140801, Oct 2023.
\newblock \doi{10.1103/PhysRevLett.131.140801}.

\bibitem[Zhang et~al.(2016)Zhang, Nie, Liang, and Zhang]{Zhang2016_Toeplitz}
Xiaoguang Zhang, You-Qi Nie, Hao Liang, and Jun Zhang.
\newblock Fpga implementation of toeplitz hashing extractor for real time post-processing of raw random numbers.
\newblock In \emph{2016 IEEE-NPSS Real Time Conference (RT)}, pages 1--5, 2016.
\newblock \doi{10.1109/RTC.2016.7543094}.

\bibitem[Zubkov and Serov(2013)]{Zubkov2013_SVBound}
A.~M. Zubkov and A.~A. Serov.
\newblock A complete proof of universal inequalities for the distribution function of the binomial law.
\newblock \emph{Theory of Probability \& Its Applications}, 57\penalty0 (3):\penalty0 539--544, 2013.
\newblock \doi{10.1137/S0040585X97986138}.

\end{thebibliography}

\newpage

\appendix

\onecolumngrid

\section{Preliminaries}
\label{app:Prelim}
\subsection{Quantum Systems}

The state of a generic quantum system $A$ can be represented by a density matrix $\rho_{A}$.
For a classical random variable $Y$ that takes on values $y\in\vars{Y}$ according to probability distribution $p_y$, it can be expressed as a quantum state $\rho_Y=\sum_{y\in\vars{Y}}p_y\tilde{\Pi}_Y^y$,
where $\tilde{\Pi}_Y^y=\dyad{y}_Y$ for simplicity.
Two special classical states are the uniformly distributed system, with $p_y=\frac{1}{\abs{\vars{Y}}}$ represented by $\tau_Y$, and the uniformly distributed systems with matching values, $\tilde{\tau}_{YY'}$, where $p_y=\frac{1}{\abs{\vars{Y}}}$ and $y=y'$ for all $y,y'\in\vars{Y}$.
A classical-quantum state with classical random variable $A$ and quantum subsystem $B$ can be expressed as
\begin{equation*}
    \rho_{AB}=\sum_{a}p_a\tilde{\Pi}_A^a\otimes\rho_B^a,
\end{equation*}
where $\rho_B^a$ (or sometimes $\rho_{B|A=a}$) is the quantum state of subsystem $B$ conditioned on event $A=a$.
We further define the sub-normalized conditional state $\rho_{B\land A=a}$ as $p_a\rho_B^a$.
We label a quantum channel that maps inputs $I$ to output $O$, with fixed inputs $a$ by $\bigepsilon_{I\rightarrow O}^a$.
The symbol $\bigepsilon$ is used to label quantum channels controlled by an adversary while letters $\vars{A}$, $\vars{B}$, $\vars{C}$ are used to label quantum channels corresponding to actions by honest or hypothetical parties.
We label a quantum measurement on system $A$ with outcome $X=x$ by the general measurement operator $F_A^x$, with corresponding positive operator-value measure (POVM) $\Pi_A^x=(O_A^x)^{\dagger}O_A^x$.
For projections onto an orthonormal basis, the POVMs (also general measurement operators) are expressed as projectors, $\{\tilde{\Pi}_{Y}^y\}_{y\in\vars{Y}}$.

\subsection{Quantum Information Theory}
\label{sec:Prelim_QInfo}

We use the generalised trace distance measure to measure the distinguishability of two (possibly sub-normalized) quantum states $\rho$ and $\sigma$~\cite{Tomamichel2015_ITBook},
\begin{equation}
    \Delta(\rho,\sigma):=\frac{1}{2}\norm{\rho-\sigma}_1+\frac{1}{2}\abs{Tr(\rho-\sigma)},
\end{equation}
where $\norm{\cdot}_1$ is the trace norm.
Another useful distance measure is the purified distance, 

\begin{equation}
    P(\rho,\sigma)=\sqrt{1-\left[\Tr\abs{\sqrt{\rho}\sqrt{\sigma}}+\sqrt{(1-\Tr\rho)(1-\Tr\sigma)}\right]^2},
\end{equation}

which can be bounded by the trace distance, $P(\rho,\sigma)\leq\sqrt{2\Delta(\rho,\sigma)}$.
The min-entropy is defined as~\cite{Tomamichel2015_ITBook}
\begin{equation}
    \Hmin(A|E)_{\rho_{AE}}:=\sup\{\lambda\in\mathbb{R},\sigma_B:\rho_{AE}\leq 2^{-\lambda}\mathbb{I}_A\otimes\sigma_B\},
\end{equation}
with max-entropy defined as a min-entropy through the duality relation, $\Hmax(A|B)_{\rho}:=-\Hmin(A|E)_{\rho}$ for pure $\rho_{ABE}$.
The smooth min- and max-entropy are defined as
\begin{equation}
\begin{gathered}
    \Hmin^{\varepsilon}(A|E)_{\rho_{AE}}:=\max_{\substack{\rho'_{AE}:\Tr[\rho_{AE}']\leq 1,\\ P(\rho_{AE},\rho_{AE}')\leq\varepsilon}}\Hmin(A|E)_{\rho_{AE}'}\\
    \Hmax^{\varepsilon}(A|E)_{\rho_{AE}}:=\min_{\substack{\rho'_{AE}:\Tr[\rho_{AE}']\leq 1,\\ P(\rho_{AE},\rho_{AE}')\leq\varepsilon}}\Hmax(A|E)_{\rho_{AE}'}
\end{gathered}
\end{equation}

\subsection{Two-Universal Hash Functions}

For information-theoretic authentication schemes constructed from universal hash functions, a client generates a valid tag from the hash function with an authentication key and the message, which is then verified by a server with the same hash function and authentication key.
We note that there are alternative frameworks of authentication with quantum access to the hash function based on unforgeability~\cite{Doosti2021_QUnforgeability}, but we assume here that only classical access to the hash function is allowed, given that the adversary in our proposed protocols has no direct access to the hash function.
Utilising a $\varepsilon$-almost strong 2-universal hash function prevents message tampering,
\begin{definition}[$\varepsilon$-almost strong 2-universal hash function~\cite{Stinson1994_UniversalHashing,Portmann2014_Authentication}]
\label{defn:AS2U}
    A family of hash functions $\{h(k,.):\vars{M}\rightarrow\vars{T}\}_{k\in\vars{K}}$ is $\varepsilon$-almost strongly 2-universal if for all messages $m_1,m_2\in\vars{M}$ and $m_1\neq m_2$, and all $t_1,t_2\in\vars{T}$,
    \begin{equation*}
        \frac{1}{\abs{\vars{K}}}\sum_{k\in\vars{K}}\Pr[h(k,m_1)=t_1\land h(k,m_2)=t_2]\leq\frac{\varepsilon}{\abs{\vars{T}}},
    \end{equation*}
    and
    \begin{equation*}
        \frac{1}{\abs{\vars{K}}}\sum_{k\in\vars{K}}\Pr[h(k,m)=t]=\frac{1}{\abs{\vars{T}}}.
    \end{equation*}
\end{definition}
An $\varepsilon$-almost strong 2-universal hash function $h$ can be formed from an $\varepsilon$-almost XOR 2-universal hash function $h'$ where the average probability of $h'(k,m_1)\oplus h'(k,m_2)=t$ for any $t$ is bounded by $\varepsilon$~\cite{Rogaway2002_BucketHashing}, and a authentication masking key (used as a one-time pad), $h(k,m)=h'(\kh,m)\oplus \kOTP$~\cite{Carter1977_WegmanCarter,Portmann2014_Authentication}.
A weaker notion of universality can be defined, where an $\varepsilon$-almost 2-universal hash function refers to one where the average probability of $h(m_1)=h(m_2)$ is bounded by $\varepsilon$~\cite{Tomamichel2011_QLHL}.\\

It is clear that if the authentication masking key is never used again after the hash is performed, $\Kh$ remains secure and tag $T$ appears as a random string to the adversary,
\begin{theorem}
\label{thm:OTPWegCarSec}
    Let $h'$ be a $\varepsilon$-almost XOR 2-universal hash function with authentication key $K^h$, and $h(k,m)=h'(\kh,m)\oplus \kOTP$, with a fully random masking string $\KOTP$. Consider an adversary which was provided with a valid message-tag pair $(m,t)$, and the honest party which discards (trace away) $\KOTP$ (i.e. $\KOTP$ acts as a one-time pad). The seed $\Kh$ remains private, in the sense that
    \begin{equation*}
        \rho_{\Kh MTE}=\tau_{\Kh}\otimes\tau_T\otimes\rho_{ME}.
    \end{equation*}
\end{theorem}

A property of $\varepsilon$-almost 2-universal hash functions that is useful is that they are strong extractors.
This is demonstrated from the quantum leftover hash lemma (QLHL)~\cite{Tomamichel2011_QLHL,Tomamichel2017_QKDProof},
\begin{theorem}[Quantum leftover hash lemma]
    \label{thm:Modified_QLHL}
    Let $h$ be a $\varepsilon$-almost 2-universal hash function with key $K$. Consider a protocol which begins with state $\tau_K\otimes\rho_{ME}$, and the hash is computed as $Y=h(K,M)$.
    Then, for any $\epssm\geq0$ and $\varepsilon'>0$,
    \begin{equation*}
        \Delta(\rho_{YKE},\tau_{Y}\otimes\tau_K\otimes\rho_E)\leq 2\epssm+2\varepsilon'        +\sqrt{2^l\varepsilon-1+ 2^{l-\Hmin^{\epssm}(M|E)+\log_2\left(\frac{2}{\varepsilon^{'2}}+1\right)}},
    \end{equation*}
    where $l=\abs{Y}$ is the length of tag $Y$.
    We call the hash function a $\delta$-strong extractor, where $\delta$ matches the RHS of the inequality.
    If $h$ is instead 2-universal, then
    \begin{equation*}
        \Delta(\rho_{YKE},\tau_{Y}\otimes\tau_K\otimes\rho_E)\leq 2\epssm+\frac{1}{2}\times 2^{-\frac{1}{2}[\Hmin^{\epssm}(M|E)-l]}.
    \end{equation*}
\end{theorem}

\subsection{Pseudorandom Number Generation}

The security of cryptographically-secure PRNGs are defined by the probability that any distinguishing protocol with limited resources can distinguish between the PRNG output and an ideal random number generator (IRNG) output~\cite{Katz2014_Crypto}.
We note that this security is in general composable when considering computational security.~\cite{Canetti2000_Composable,Canetti2001_Composable}.
We can consider quantum adversaries in general since the QAKE protocol involves quantum steps. 
Since we are interested in computing key generation rates, we adopt a more concrete security definition with resources $t$~\cite{Bellare1997_PRNG,Dodis2013_PRNG} (e.g. time) instead of presenting security against quantum polynomial-time (QPT) adversaries\footnote{QPT adversaries are adversaries with polynomial access to both classical and quantum resources (e.g. polynomial storage and computational power)}.
\begin{definition}[Quantum-secure PRNG]
\label{defn:QSRPNG}
    A pseudorandom number generator (PRNG) is a function $G^{\PRNG}:\vars{K}\rightarrow \vars{Y}$ that takes as input a seed $k\in\vars{K}$ with length $l$ and outputs a bit string $y\in\vars{Y}$ with length $n$, where $l < n$. 
    A $(l,n)$-standard PRNG is $(t,\varepsilon)$ quantum-secure if for any quantum algorithm $\vars{A}_D^{t}$ constrained by some resources $t$ with output $D$, the PRNG equipped with a random seed $G^{\PRNG}_{\vars{K}}$ is indistinguishable from an ideal random number generator $G^{\TRNG}$ that samples random numbers $r$ uniformly at random from $\vars{Y}$,
    \begin{equation*}
        \abs{\Pr[\vars{A}^{G^{\PRNG}_{\vars{K}},t}_{D}=1]-\Pr[\vars{A}^{G^{\TRNG},t}_{D}=1]}\leq \varepsilon.
    \end{equation*}
\end{definition}

One important property of the PRNG is the difficulty in guessing the seed of the PRNG when given partial information of the output, which can only succeed with a small probability $\epsguess$, as summarised below.

\begin{theorem}
\label{thm:PRNGGuess}
    Consider a scenario where a random number $y\in\{0,1\}^n$ is prepared from a $(t',\varepsilon_{\PRNG})$-quantum-secure PRNG with a random seed $k$ of length $l$ with security parameter $\varepsilon_{\PRNG}$, i.e. $y=G^{\PRNG}(k)$, and a $n$-bit bit string $z$, with $\Pr[z_i=1]=p_1$, for $i=1,\cdots,n$, with $n>2l$. Suppose an adversary is provided with $y_{z=1}:=\{y_i:z_i=1\}$, $G^{\PRNG}$ and $z$, and is tasked to guess the seed $k$. When $t'\geq t+t_{SG,PR}$, the probability that any quantum adversary with resource upper bounded by $t$ can correctly guess $k$ is bounded by
    \begin{equation*}
        \Pr[\hat{k}=k]\leq\epsguess
    \end{equation*}
    where $\epsguess \leq \varepsilon_{\PRNG}+2^{-n+l}$, and $t_{SG,PR}$ is the total resources required for generating string $z$ and generating $y$ from a master key using the PRNG.
\end{theorem}
\begin{proof}
We start with noting that correctly guessing the seed allows one to guess all the outcomes $y$ from the PRNG. If there is an algorithm which can guess the seed from the set $y_{z=1}$, it can be used to distinguish whether a string $y \in \{0,1\}^n$ is being generated from a PRNG or from an IRNG. Therefore, it will violate the security of the $(t',\varepsilon_{\PRNG})$-PRNG.  

Suppose for contradiction that an algorithm $\vars{A}$ with resources $t$ is able to guess the seed $k$ correctly with probability at least $\epsguess>2^{-l}$ (note $\epsguess\leq 2^{-l}$ is trivially achievable with random guess), i.e.,
\begin{equation} 
\Pr[\vars{A}(y_{z=1})=k] > \epsguess.
\end{equation}

We construct an algorithm $\vars{\hat A}$ that distinguishes $G^{\PRNG}(k)$ from $G^{\TRNG}$ as follows:
\begin{enumerate}
    \item Upon receiving a string $y \in \{0,1\}^n$, it generates a string $z = z_1\ldots z_n$, where $z_i \in \{0,1\}$, and $\Pr [z_i = 1] = p_1$, for all $i \in \{1, \ldots ,n\}$. 
    \item From $z$, and $y$, it constructs the set $y_{z=1}$, and feed it to the algorithm $\vars{A}$ to guess the seed. Suppose $\vars{A}(y_{z=1}) = \hat{k}$.
    \item If $G^{\PRNG}(\hat{k}) = y$, the algorithm returns $1$, indicating $y$ is being generated via a PRNG. 
    \item Otherwise, the algorithm returns $0$, indicating $y$ is sampled via an IRNG.
\end{enumerate}
The additional resources required for the generation of string $z$ and generation of $y$ from the master key guess $\hat{k}$ is defined as $t_{SG,PR}$.
As such, $\hat{\vars{A}}$ requires $t+t_{SG,PR}$ resources.\\

When the protocol is implemented on the PRNG, it is clear that if the guess of $\hat{k}$ is correct, i.e. $\hat{k}=k$, the algorithm returns 1.
Therefore, 
\begin{equation*}
    \Pr[\vars{\hat A}^{G^{\PRNG},t+t_{SG,PR}} = 1]>\epsguess.
\end{equation*}
When the protocol is implemented on an IRNG, the range of $G^{PRNG}$ is limited by the size of the seed, $2^l$.
For any $y$ outside the range of $G^{PRNG}$, it is clear that $G^{PRNG}(k)\neq y$ for any $k$, and $\vars{\hat A}^{G^{\TRNG},t+t_{SG,PR}}$ returns 0.
Therefore, for uniformly generated $y\in\{0,1\}^n$,
\begin{equation*}
    \Pr [\vars{\hat A}^{G^{\TRNG},t+t_{SG,PR}} = 1]\leq 2^{-n+l}.
\end{equation*}
Given the condition that $n>2l$ and $\epsguess>2^{-l}$, the IRNG term is smaller and the ability of $\vars{\hat{A}}$ to distinguish between the PRNG and IRNG is then
\begin{equation*}
    \abs{\Pr[\vars{\hat A}^{G^{\PRNG},t+t_{SG,PR}} = 1] - \Pr [\vars{\hat A}^{G^{\TRNG},t+t_{SG,PR}} = 1]} > \epsguess -2^{-n+l}.    
\end{equation*}
According to the definition of the PRNG, $\abs{\Pr[\vars{\hat A}^{G^{\PRNG},t'}_D = 1] - \Pr [\vars{\hat A}^{G^{\TRNG},t'}_D = 1]} \leq \varepsilon_{\PRNG}$.
Since $t'\geq t+t_{SG,PR}$, the bound holds in particular for algorithm $\vars{\hat A}^{G^{\RNG},t+t_{SG,PR}}$.
Therefore, when $\epsguess -2^{-n+l} =\varepsilon_{\PRNG}$, we encounter a contradiction and $\Pr[\hat{k}=k]\leq\epsguess=\varepsilon_{\PRNG}+2^{-n+l}$.
\end{proof}

\section{Adapting AKE Framework to QAKE Security}
\label{app:AKESecDefn}
We seek to incorporate standard QKD with authentication security in a QAKE security framework, which utilises the same terminology and tools as a classical AKE framework.
Here, we present only the main points of the classical AKE framework relevant in the context of QKD, and we refer the reader to Ref.~\cite{Guilhem2020_AKE} for more details on the classical AKE framework.

\subsection{Security Model}
\label{app:General_Security_Model}

We begin by first examining the security model of QAKE, based off the QKD security model and authentication requirements before linking the model to the classical AKE framework.
QAKE adopts a security model where the target key generation partner is between honest Alice and Bob, with an adversary (or eavesdropper) present in their communication channel that could seek to impersonate either party or steal the shared keys between Alice and Bob.
The protocol can be generalised to multiple parties, with any pre-shared secrets used for authentication being shared pairwise.
In general, the protocol involve input registers $S$ representing the pre-shared secrets (e.g. authentication keys), $L$ for any labels (e.g. to indicate the secret to use), with $E$ labelling any side-information of the adversary.
After the protocol, the output would include $\FA\FB$, which represent the respective parties' choice: (1) $F=\phi$ representing that it is not involved in the session, (2) $F=0$ representing that it chose not to authenticate the other party, and (3) $F=1$ representing that it chose to authenticate the other party.
It would also output keys $\KA$ and $\KB$ respectively if $\FA=1$ and $\FB=1$ (Note that this can in principle be decoupled -- i.e. parties can choose to authenticate, but not generate keys, but we couple them for simplicity).
Otherwise, it would output $\perp$ signalling that no secure keys are generated.\\

The classical AKE framework~\cite{Guilhem2020_AKE} is broader in its model, examining multiple parties, each with its own unique identity.
While classical AKE do not require all parties to provide information for other parties to authenticate its identity (to provide for instances where anonymous key generation is acceptable), we require mutual authentication for the QAKE protocol, i.e. setting the set of parties that need to provide authentication information to be the set of all parties.
During each session of the protocol, each party can set up local sessions (an instance that is present locally which runs the protocol steps and stores the outputs, e.g. in QKD, this could be Alice's QKD box running Alice's steps in a single protocol session and storing the outputs), labelled by $l=(i,j,k)$, with its identity $i$ (e.g. Alice), intended peer $j$ (e.g. Bob), and session number $k$.
The parties are assumed to hold public-private key pairs $(\text{pk}_i,\text{sk}_i)$ for authentication purposes.
For QAKE, we generalize to allow for the use of symmetric keys by having parties hold ``matching keys" $\text{mk}_i$ instead, which can either contain $\text{sk}_i$ for symmetric keys or $\text{pk}_i$ for public keys\footnote{We note that this generalization has an impact on certain attacks. For instance, the classical AKE model captures the key-compromise impersonation attack, where an adversary that knows the secret key of Alice attempts to impersonate Bob and get Alice to accept the authentication. The leakage of $\text{sk}_A$ alone would not cause this compromise, but with the use of symmetric keys, the protocol would no longer provide resistance to key-compromise impersonation. However, since this attack is not relevant to the QAKE of interest, where both Alice and Bob are assumed to be honest, we leave the account of such changes to future work.}.
We note also that the corresponding peers should hold the corresponding matching key $\text{mk}_i$ that allows them to authenticate party $i$.\\

During the protocol, the adversary is present in the network between all parties, and is able to interfere with the communication, e.g. delay, redirect or alter messages.
In addition, the adversary is in general allowed to call on the session to reveal the session key and corrupt any party, which forces the reveal of the party's secret keys $\text{sk}_i$.
We note that the security conditions (e.g. key secrecy) will impose conditions (e.g. party not corrupted) on the use of such attacks.
For QAKE, since Alice and Bob are assumed to be honest, the reveal of the session key, along with corrupting of Alice and Bob is disallowed.\\

During the protocol or at the end of the protocol, various outputs or information would be present with each party's local sessions.
These are:
\begin{enumerate}
    \item Session acceptance: Indicates if a party has accepted or rejected the session. If a party is not involved, the session acceptance is maintained as $\phi$ as initialized. This matches exactly the choice $F_AF_B$ of the QAKE protocol.
    \item Keys: The session outputs keys, which is set to $\perp$ unless the session is accepted -- matching $K_AK_B$ in QAKE.
    \item Session identifier: Initialized as $\perp$, and changed to other values when the session is accepted. This value is meant to identify local sessions that are partners -- local sessions with same session identifier are partners (note that they have to both accept the session). This value is implicitly part of QAKE protocols, where Alice's and Bob's generated local sessions are partners. To explicitly include the identifier, one can simply modify the local session label $L$ to an identifier, e.g. setting identifier as $(i,j,k)$ if $i<j$ and $(j,i,k)$ otherwise, or include a random string that is sent and authenticated during the protocol. We take the first case for simplicity, and also note that for the QAKE protocol, these partnering local sessions are always between two different parties (Alice and Bob).
    \item Entity confirmation identifier: Indicates the local sessions that eventually would partner, i.e. similar to session identifier, but can be set earlier than protocol acceptance stage. For the QAKE protocol, we can set these values as $(i,j,k)$ (or $(j,i,k)$) after the first message is sent or arrives at the party.\footnote{The choice of the confirmation identifier should coincide with when the identity of the peer (e.g. Bob) is made known to the the party (e.g. Alice). In classical AKE and QAKE, a pre-specified peer model is assumed, where each local session knows its intended partner's identity. As such, these information should be exchange at the start of the protocol in practice.}
    \item Key confirmation identifier: Indicates the local sessions that eventually would generate the same key, and this can be chosen similarly to the entity confirmation identifier for QAKE.
\end{enumerate}

\subsection{Security Conditions}
\label{app:General_Security_Cond}

With the input and output registers defined, let us formally introduce the security definitions.
We note that the classical AKE framework has multiple security conditions that are related and corresponds to desirable properties. 
However, we choose only the most suitable, with many of the other properties being derivable from these conditions.\\

The first condition, robustness, is not part of the classical AKE framework, but is important nonetheless.
It necessitates that the protocol can succeed when the adversary is absent, and the protocol is not trivially rejecting all authentication to always guarantee security.
\begin{definition}[$\epsrob$-robustness]
    A QAKE protocol is $\epsrob$-robust if it passes with high probability in the absence of any adversary, i.e. 
    \begin{equation*}
        \Pr[\FA=\FB=1]\geq 1-\epsrob.
    \end{equation*}
\end{definition}

The second condition, explicit entity authentication, addresses the security associated with mutual authentication between Alice and Bob.
It comes in two variants: full explicit entity authentication and almost-full explicit entity authentication.\\

We impose full explicit entity authentication on Alice, requiring that Bob will generate a partnering local session when Alice accepts.
More formally, the definition in Ref.~\cite{Guilhem2020_AKE} presents a predicate that states that for all local sessions $l$, when the local session accepts, there exists a partner local session $l'$, and this local session belongs to the intended peer of the local session\footnote{We note there is a separate subtle implication that any partner to $l$ must belong to the intended peer defined in explicit entity authentication. This means that the predicate does not preclude the possibility that two local sessions can be established on Bob that partners the same Alice. Since Alice and Bob are considered honest in QAKE, this can be simply prevented by having them to individually check that they do not have local sessions with the same session identifiers.}.
This includes a simplification noting that for QAKE, mutual authentication is expected, and both Alice and Bob are assumed to be honest.
As such, we have the implication that $F_A=1\implies F_B=1$, which the security condition requires to be true with high probability.\\

The condition on almost-full explicit entity authentication is imposed on Bob, who sends the final authentication message to Alice.
As such, he is unable to guarantee that at the end of the protocol, Alice would generate an accepting local session (necessary for session identifier generation and partnering) since the adversary can interfere with the final authentication message transmission.
More formally, the definition in Ref.~\cite{Guilhem2020_AKE} presents a predicate that states that for all local sessions $l$, when the local session accepts, there exists a local session $l'$ with the same entity confirmation identifier, which is a partner to $l$ if $l'$ generates a session identifier, and this partner must be the peer of $l$.
This includes the same simplification that mutual authentication is expected, and both Alice and Bob are assumed to be honest.
As such, we have the implication that $F_B=1\implies F_A=0,1$, where Alice's local session either rejects the session (same entity confirmation identifier, but not partners) or accepts the session (partners)\footnote{There is similar subtle implication as per full explicit entity authentication where the definition does not prevent multiple accepting local sessions to be established on Alice, though that can be addressed similarly since Alice is honest.}.\\

The security of both conditions requires that the probability of the predicates being false is small, i.e. the probability of ``bad events" is small.
We note that the security is defined relative to a computationally bounded adversary, with the small probability given as negligible relative to a security parameter.
We can generalise here (and in the following security definitions) to a general unbounded adversary, and define some small parameter $\varepsilon$ to quantify the ``negligible" probability.
As such, we define the explicit entity authentication as
\begin{definition}[($\epsEAf$,$\epsEAaf$)- explicit entity authentication]
    A QAKE protocol has $\epsEAf$-full explicit entity authentication for Alice if the authentication fails with high probability when Bob is not accepting, i.e.
    \begin{equation*}
        \Pr[\FA=1,\FB=0]+\Pr[\FA=1,\FB=\phi]\leq\epsEAf,
    \end{equation*}
    and it has $\epsEAaf$-almost full explicit entity authentication for Bob if the authentication fails with high probability when Alice is not involved in the session, i.e. \begin{equation*}
        \Pr[\FA=\phi,\FB=1]\leq\epsEAaf.
    \end{equation*}
\end{definition}
Note that we label $\epsEA=\epsEAf+\epsEAaf$ for brevity.\\

The third security condition is termed Match-security, though in the context of the QAKE protocol, it reduces to the correctness condition.
This security condition mainly elevates the entity authentication security condition to provide key authentication as well -- where Alice and Bob can guarantee that the other party will generate the same key confirmation identifier and keys if both local sessions are accepting\footnote{There is a separate condition termed key-match soundness that is required to elevate explicit entity authentication to explicit key authentication. This condition is proven in classical AKE framework based on match secrecy and key secrecy conditions.}.
Formally, Ref.~\cite{Guilhem2020_AKE} defines the Match-predicate as four conditions: (1) partner local sessions generate the same key, (2) partner local sessions generate same key confirmation identifier, (3) at most two local sessions can generate the same session identifier, and (4) local sessions with the same key confirmation identifier would generate the same key if both local sessions are accepting.
The four conditions, in the context of QAKE, can be simplified:
\begin{enumerate}
    \item This condition can be summarised as when partnering local sessions are accepting (in QAKE this means that they are key generating as well), they will generate the same keys, i.e. $F_A=F_B=1\implies K_A=K_B$.
    \item Since we set key confirmation identifier and session identifier in the same way for QAKE, this condition is always true.
    \item This condition is always guaranteed in QAKE since the choice of session identifier as $(i,j,k)$ (or $(j,i,k)$), along with honest Alice and Bob, and the fact that only one pair of local sessions have same set of parties $(i,j)$ and same index $k$, means that only two local sessions can generate the same session identifier.
    \item Equivalent to condition 1 since key confirmation identifier and session identifier are set in the same way.
\end{enumerate}
Condition 1 in the context of QAKE therefore matches the correctness condition of QKD, while the remaining are either always true or reduce to condition 1.
Therefore, we can define the Match-security as
\begin{definition}[$\epsMS$-Match security]
    A QAKE protocol has $\epsMS$-Match security if
    \begin{equation*}
        \Pr[\KA\neq \KB,\FA=1,\FB=1]\leq\epsMS.
    \end{equation*}
\end{definition}

The final condition is key secrecy (also termed BR-secrecy), which requires that any keys generated be secret from the adversary. 
Ref.~\cite{Guilhem2020_AKE} defines the secrecy as a secrecy game, being the ability for an adversary to distinguish between an ideal key sampled from the key distribution and the actual key.
Certain secrecy freshness conditions are necessary for the game, such as the owner and peer being honest, the session key not directly revealed to the adversary (e.g. announced), and that there is expected mutual authentication, which are all part of the assumptions of QAKE, i.e. the freshness condition is satisfied.
More formally, the secrecy condition is defined by~\cite{Guilhem2020_AKE} 
\begin{equation}
    \abs{\Pr[\vars{A}_D^{\reall}=1]-\Pr[\vars{A}_D^{\ideal}=1]}=\text{negl},
\end{equation}
where a distinguisher $\vars{A}$ makes the guess $D$ on whether the keys are the actual key ($D=1$) or an ideal key $(D=0)$, guess $D$ is set as $\perp$ if the local session is not accepting~\cite{Brzuska2011_BR}, and negl representing neglible in some security parameter.
Noting that when $D=\perp$, the ideal and real cases are identical, and with the upgrade from a computationally bounded adversary to an unbounded adversary with quantum capabilities, the secrecy condition can be redefined in terms of trace distance, representing the ability to distinguish between the ideal and real keys.
Note that we explicitly take only the cases of $(\FA,\FB)$ being $(0,1)$ and $(1,1)$, since the probability of $\FB=1$ when $\FA=\phi$ is small, and accounted for in entity authentication.
As such, we define key secrecy as
\begin{definition}[$\epsKS$-key secrecy]
    A QAKE protocol has $\epsKS$-key secrecy if 
    \begin{equation*}
        \Delta(\rho_{\FA\FB\KB LSE\land (\FA,\FB)\in\{01,11\}},\tau_{\KB}\otimes\rho_{\FA\FB LSE\land (\FA,\FB)\in\{01,11\}})\leq\epsKS.
    \end{equation*}
\end{definition}

\section{Security Analysis of QAKE Protocol}
\label{app:AKEProtocolSec}
\subsection{Multi-Session to Single-Session Security Reduction}
\label{app:Multi_to_Single_Red}

Here, we prove an tighter version of Thm.~\ref{thm:Multi_to_Single_Red} using an inductive proof, noting that the final session need not keep the secrets private, where we define the set of ideal states for this final session as $\Sout$.
\begin{theorem}
\label{thm:Multi_to_Single_Red_Extension}
Suppose we have the conditions:
\begin{enumerate}
    \item The initial state is an ideal input state $\rho_0\in\Sin$.
    \item There exists an ideal output state which preserves any necessary secrets that is $\epssecint$-close to the output state of a single session $\vars{P}_i$, i.e. for any $\sigma\in \Sin$, there exists a state $\rho\in \Sout'$ such that
    \begin{equation*}
        \Delta(\vars{P}_i(\sigma),\rho)\leq\epssecint.
    \end{equation*}
    \item The set of ideal output states is a subset of the set of ideal input states $\Sout'\subseteq\Sin$.
    \item There exists an ideal output state that is $\epssec$-close to the output state of a single session $\vars{P}_i$, i.e. for any $\sigma\in\Sin$, there exists a state $\rho\in \Sout$ such that
    \begin{equation*}
        \Delta(\vars{P}_i(\sigma),\rho)\leq\epssec.
    \end{equation*}
\end{enumerate}
Then, for all $i\in[1,n_{sess}]$, there exists $\rho_i\in \Sout$ such that 
\begin{equation*}
    \Delta(\vars{P}^i(\rho_0),\rho_i)\leq (i-1)\epssecint+\epssec.
\end{equation*}
\end{theorem}

\begin{proof}
The proof is by induction. 
Let us consider a statement that for any $i\in[1,n_{sess}-1]$, there exists $\rho_{i}\in\Sout'$ such that
\begin{equation}
    \Delta(\vars{P}^{i}(\rho_0),\rho_{i})\leq i\epssecint.
\end{equation}
For $i=1$, the statement is trivially true from conditions 1 and 2.
Suppose the statement is true for index $j$, i.e. there exists $\rho_j\in\Sout'$ such that $\Delta(\vars{P}^j(\rho_0),\rho_j)\leq j\epssecint$.
From condition 3, we know that $\rho_j\in \Sout'\subseteq \Sin$ is a valid input state as well.
As such, by condition 2, there exists a state $\rho_{j+1}\in \Sout'$ such that
\begin{equation}
    \Delta(\vars{P}(\rho_j),\rho_{j+1})\leq\epssecint.
\end{equation}
We can expand the LHS term for the $(j+1)$-th index, where there exists $\rho_{j+1}\in \Sout'$ such that
\begin{equation}
\begin{split}
    &\Delta(\vars{P}^{j+1}(\rho_0),\rho_{j+1})\\
    \leq&\Delta(\vars{P}^{j+1}(\rho_0),\vars{P}(\rho_j))+\Delta(\vars{P}(\rho_j),\rho_{j+1})\\
    \leq&\Delta(\vars{P}^j(\rho_0),\rho_j)+\epssecint\\
    \leq& (j+1)\epssecint,
\end{split}
\end{equation}
where the first line applies the triangle inequality, and the second line notes that the CPTP map $\vars{P}$ cannot increase trace distance.
Since the statement is true for $i=1$, and the statement being true for $j$ implies that it is true for $j+1$, by mathematical induction, the statement is true for any $i\in [1,n_{sess}-1]$.\\

To prove the statement of Thm.~\ref{thm:Multi_to_Single_Red}, we first note that since $\rho_i\in\Sout'\subseteq\Sin$, condition 4 states that there exists a state $\rho_{i+1}\in\Sout$ such that
\begin{equation}
    \Delta(\vars{P}(\rho_i),\rho_{i+1})\leq\epssec.
\end{equation}
As such, we can use the triangle inequality to arrive at
\begin{equation}
\begin{split}
    \Delta(\vars{P}^{i+1}(\rho_0),\rho_{i+1})\leq&\Delta(\vars{P}^{i+1}(\rho_0),\vars{P}(\rho_m))+\Delta(\vars{P}(\rho_m),\rho_{i+1})\\
    \leq&\Delta(\vars{P}^{i}(\rho_0),\rho_i)+\epssec\\
    \leq& i\epssecint+\epssec
\end{split}
\end{equation}
for any $i\in[1,n_{sess}]$.
The same statement is true for $i=0$ by condition 4, and thus we can arrive at the theorem.
\end{proof}

Condition 2 requires an additional condition for the set of ideal states, which we express as $\rhoidealint\in\Sout'$, to match $\Sin$, where the necessary secrets $S$ are kept private.
Crucially, the main difference between $\rhoidealint\in\Sout'$ and $\rhoideal\in\Sout$ is the form of the $LSE$ subsystem.
In general, $\rhoideal$ has $LSE$ subsystems that can take any form, including one matching the actual output state while the form of $\rhoidealint$ has $LSE$ subsystems that are of a specific form (depending on the protocol requirement).
This forms the rationale behind the definition of the shared secrets privacy condition\footnote{In classical AKE~\cite{Guilhem2020_AKE}, this requirement of maintaining secret values from one session to the next is part of the parameter crypt.
Here, we make the requirement explicit with $\epsSP$ distance from full secrecy.}:
\begin{definition}[$\epsSP$-Shared Secrets Privacy]
    A QAKE protocol has $\epsSP$-shared secrets privacy if
    \begin{equation*}
        \Delta\left(\sum_{(\fA,\fB)\not\in\vars{S}_{\SP}}p_{\fA\fB}\rho^{\fA\fB}_{LSE},\sum_{(\fA,\fB)\not\in\vars{S}_{\SP}}p_{\fA\fB}\sigma_{LSE}^{\fA\fB}\right)\leq\epsSP,
    \end{equation*}
    where $p_{\fA\fB}$ is the probability of the output state recording $(\FA,\FB)=(\fA,\fB)$, $\vars{S}_{\SP}=\{\phi1,1\phi,10\}$ and $\sigma_{LSE}^{\fA\fB}\in\vars{S}_{in}$.
\end{definition}

\subsection{Single-Session Security Conditions}
\label{app:Sec_Para_Compute}

Recall that the set of ideal input state is defined by $\vars{S}_{in}=\{\rhoin=\sum_{jj'}p_{jj'}\rhoinn{jj'}_{LSE}\}$, where
\begin{equation*}
\begin{split}
    \rhoinn{jj'}_{LSE}=&\dyad{jj'}_{\alpha\alpha'}\otimes\tau_{K_2R\KOTP_{1,j}\cdots\KOTP_{1,m}}\otimes\bigepsilon_{K_1^h\KOTP_{1,j-1}E_{j-1}\rightarrow E}\circ\cdots\circ\bigepsilon_{K_1^h\KOTP_{1,j'}E_{j'}\rightarrow E_{j'+1}}(\tau_{\Kh_1\KOTP_{1,j'}\cdots \KOTP_{1,j-1}}\otimes\rho_{E_{j'}}),
\end{split}
\end{equation*}
where $L=\alpha\alpha'$ are the labels and $S$ refer to the secrets.
We note here an observation that when a partial trace of $\KOTP_{K_{1,j'}}$ is applied to the state, it effectively removes the corresponding channel $\bigepsilon_{K_1^h\KOTP_{1,j'}E_{j'}\rightarrow E_{j'+1}}$, with appropriate relabelling of $E_{j'}$ as $E_{j'+1}$.\footnote{More formally, the partial trace can be shifted to before the corresponding channel $\bigepsilon_{K_1^h\KOTP_{1,j'}E_{j'}\rightarrow E_{j'+1}}$, which describes the adversary querying the oracle, followed by $\KOTP_{1,j'}$ being discarded. Since $\KOTP_{1,j'}$ acts as an one-time pad on the tag $T_{\AV}$, Thm.~\ref{thm:OTPWegCarSec} implies that $T_{\AV}$ would appear random to the adversary. As such, this is equivalent to the adversary randomly sampling $T_{\AV}$ and using that value instead. As such, the process can be described by simply $\bigepsilon_{E_{j'}\rightarrow E_{j'+1}}$, an internal update of the adversary's internal state.}
Similarly, recall that the set of ideal output state is given by $\vars{S}_{out}'=\{\rhoidealint\}$, where
\begin{equation*}
\begin{split}
    &\rhoidealint=\dyad{\perp\perp}_{\KA\KB}\otimes\left[p_{0\phi}\dyad{0\phi}_{\FA\FB}\otimes \sigma^{0\phi}_{LSE}+p_{\phi0}\dyad{\phi0}_{\FA\FB}\otimes\sigma^{\phi0}_{LSE}+p_{00}\dyad{00}_{\FA\FB}\otimes \sigma^{00}_{LSE}\right]\\
    &+p_{01}\dyad{01,\perp}_{\FA\FB\KA}\otimes\tau_{\KB}\otimes\sigma^{01}_{LSE}+p_{11}\dyad{11}_{\FA\FB}\otimes\tilde{\tau}_{\KA\KB}\otimes\sigma^{11}_{LSE},
\end{split}
\end{equation*}
for $\sigma_{LSE}^{\fA\fB}\in\vars{S}_{in}$.
We note that for $\vars{S}_{out}$, the ideal output state would be of a similar form,
\begin{equation}
\begin{split}
    \rhoideal=&\dyad{\perp\perp}_{\KA\KB}\otimes\left[p_{0\phi}\dyad{0\phi}_{\FA\FB}\otimes\rho^{0\phi}_{LSE}+p_{\phi0}\dyad{\phi0}_{\FA\FB}\otimes\rho^{\phi0}_{LSE}+p_{00}\dyad{00}_{\FA\FB}\otimes \rho^{00}_{LSE}\right]\\
    &+p_{01}\dyad{01,\perp}_{\FA\FB\KA}\otimes\tau_{\KB}\otimes\rho^{01}_{LSE}+p_{11}\dyad{11}_{\FA\FB}\otimes\tilde{\tau}_{\KA\KB}\otimes\rho^{11}_{LSE}.
\end{split}
\end{equation}
With the output state descriptions, and noting that $\KA\KB\FA\FB$ can be incorporated into $E$ for the next session, we can observe that the output state in $\Sout'$ is a linear combination of $\sigma_{LSE}^{\alpha\alpha'}$, matching the general input state and thus satisfying condition 3 in Thm.~\ref{thm:Multi_to_Single_Red}.
Condition 1 in the same theorem is trivially satisfied since all secrets begin as private with index $\alpha=\alpha'=1$.
Therefore, we can focus on analysing security conditions 2 and 4.
The security definition is given by the statement: for any $\sigma_{LSE}\in\vars{S}_{in}$, there exists $\rhoidealint\in\vars{S}_{out}'$ and $\rhoideal\in\vars{S}_{out}$ such that
\begin{equation}
\begin{gathered}
    \Delta(\vars{P}(\sigma_{LSE}),\rhoidealint_{F_AF_BK_AK_BLSE})\leq\epssecint\\
    \Delta(\vars{P}(\sigma_{LSE}),\rhoideal_{F_AF_BK_AK_BLSE})\leq\epssec
\end{gathered}
\end{equation}
Here, we prove the generalized version of Thm.~\ref{thm:eps_sec_components}, which allows the security analysis to be reduced to various security conditions.
\begin{theorem}
    Consider a protocol $\vars{P}$ with an ideal input state $\sigma_{LSE}\in\Sin$. Then, overall protocol security is satisfied when the intermediate session security parameter is
    \begin{equation*}
        \epssecint=\epsEA+\epsMS+\epsKS+\epsSP,
    \end{equation*}
    and the final session security parameter is
    \begin{equation*}
        \epssec=\epsEA+\epsMS+\epsKS.
    \end{equation*}
\end{theorem}
\begin{proof}
For any input state $\sigma_{LSE}\in\vars{S}_{in}$, let $\rho_{F_AF_BK_AK_BLSE}=\vars{P}(\sigma_{LSE})$ be the output state.
We seek find an ideal output state $\rhoideal_{F_AF_BK_AK_BLSE}$ that is close to the actual output state with probabilities $p_{\FA\FB}^{\ideal}$, where we note the additional superscript $\ideal$ to distinguish it from the actual output distribution $p_{\FA\FB}^{\reall}$ of $\rho_{F_AF_BK_AK_BLSE}$.
We define
\begin{equation}
    p_{\FA\FB}^{\ideal}=\begin{cases}
        p_{\phi0}^{\reall}+p_{\phi1}^{\reall} & \FA=\phi,\FB=0\\
        p_{0\phi}^{\reall}+p_{1\phi}^{\reall} & \FA=0,\FB=\phi\\
        p_{00}^{\reall}+p_{10}^{\reall} & \FA=0,\FB=0\\
        p_{\FA\FB}^{\reall} & Otherwise
    \end{cases},
\end{equation}
while the $\rhoideall{\fA\fB}_{LSE}=\rho_{LSE}^{\fA\fB}$ of $\rho_{F_AF_BK_AK_BLSE}$ for $(\FA,\FB)\not\in\{\phi1,1\phi,10\}$.
With this definition, we can split the trace distance for $\epssec$-security into components, first extracting the components where $\FA$ or $\FB$ is $\phi$,
\begin{equation}
\begin{split}
    \Delta(\rho,\rhoideal)\leq &\Delta(\rho_{\land (\FA,\FB)\in\{01,11\}},\rhoideal_{\land (\FA,\FB)\in\{01,11\}})+\Delta(\rho_{\land \FA=\phi},\rhoideal_{\land \FA=\phi})+\Delta(\rho_{\land \FB=\phi},\rhoideal_{\land \FB=\phi})\\
    &+\Delta(\rho_{\land (\FA,\FB)\in\{00,10\}},\rhoideal_{\land (\FA,\FB)\in\{00,10\}}),
\end{split}
\end{equation}
where the subsystem labels $\KA\KB\FA\FB LSE$ are dropped for brevity.\\

We now show that the latter three trace distances are bounded by the explicit entity authentication security condition.
Expanding the first term explicitly, we have
\begin{equation}
\begin{split}
    &\Delta(\rho_{\land \FA=\phi},\rhoideal_{\land \FA=\phi})\\
    \leq&p_{\phi0}^{\reall}\Delta(\rho_{|\FA=\phi,\FB=0},\rhoideal_{| \FA=\phi,\FB=0})+p_{\phi1}^{\reall}\Delta(\rho_{|\FA=\phi,\FB=1},\rhoideal_{|\FA=\phi,\FB=0})\\
    \leq& \Delta(\dyad{\perp\perp\phi0}_{\KA\KB\FA\FB}\otimes\rho_{LSE}^{\phi0},\dyad{\perp\perp\phi0}_{\KA\KB\FA\FB}\otimes\rho_{LSE}^{\phi0})+p_{\phi1}^{\reall}\\
    =&p_{\phi1}^{\reall},
\end{split}
\end{equation}
where we note that the trace distance in the third line is 0.
A similar argument can be made for the second term for $\FB=\phi$, resulting in an upper bound of $p_{1\phi}^{\reall}$.
The final term gives
\begin{equation}
\begin{split}
    &\Delta(\rho_{\land (\FA,\FB)\in\{00,10\}},\rhoideal_{\land (\FA,\FB)\in\{00,10\}})\\
    \leq&\Delta(\dyad{\perp\perp00}_{\KA\KB\FA\FB}\otimes\rho_{LSE}^{00},\dyad{\perp\perp00}_{\KA\KB\FA\FB}\otimes\rho_{LSE}^{00})+p_{10}^{\reall}\\
    =&p_{10}^{\reall}.
\end{split}
\end{equation}
As such, the three trace distances are upper bounded by
\begin{equation}
    \Pr[F_A=1,F_B=0]+\Pr[F_A=1,F_B=\phi]+\Pr[F_A=\phi,F_B=1]\leq\epsEA.
\end{equation}
This accounts for the explicit entity authentication security condition contribution to the security.\\

The second security condition to extract is match security.
To isolate its contribution, we need to first introduce an intermediate state $\tilde{\rho}$, for which the value of $\KA$ when $\FA=1,\FB=1$ is replaced by $\KB$, i.e. $\tilde{\rho}=\vars{M}_{\KB\rightarrow \KA}(\rho)$, where $\vars{M}_{\KB\rightarrow \KA}=\dyad{11}_{\FA\FB}\otimes\left(\sum_{k}\dyad{kk}_{\KA\KB}\circ\Tr_{\KA}\right)$.
As such, we can apply the triangle inequality to expand
\begin{equation}
\begin{split}
    &\Delta(\rho_{\land (\FA,\FB)\in\{01,11\}},\rhoideal_{\land (\FA,\FB)\in\{01,11\}})\\
    \leq&\Delta\left(\tilde{\rho}_{\land (\FA,\FB)\in\{01,11\}},\rhoideal_{\land(\FA,\FB)\in\{01,11\}}\right)+\Delta\left(\rho_{\land \FA=\FB=1},\tilde{\rho}_{\land \FA=\FB=1}\right),
\end{split}
\end{equation}
noting that $\rho_{\land\FA=0,\FB=1}$ is identical to $\tilde{\rho}_{\land\FA=0,\FB=1}$ since $\vars{M}_{\KB\rightarrow\KA}$ does not alter the input state.
We now focus on the latter term, where we can expand the trace distance 
\begin{equation}
\begin{split}
    &\Delta\left(\rho_{\land \FA=\FB=1},\tilde{\rho}_{\land \FA=\FB=1}\right)\\
    =&p_{11}^{\reall}\Delta\left(\sum_{\kA\kB}p_{\kA\kB|11}\dyad{\kA\kB}_{K_AK_B}\otimes\rho^{11,\kA\kB}_{LSE},\sum_{\kA\kB}p_{\kA\kB|11}\dyad{\kA\kA}_{K_AK_B}\otimes\rho^{11,\kA\kB}_{LSE}\right)\\
    =&\frac{p_{11}^{real}}{2}\norm{\sum_{\kA\neq k_B}p_{\kA\kB|11}(\dyad{\kA\kB}_{K_AK_B}-\dyad{k_Bk_B}_{K_AK_B})\otimes\rho^{11,\kA\kB}_{LSE}}_1\\
    \leq&\frac{p_{11}^{\reall}}{2}\sum_{\kA\neq \kB}p_{\kA\kB|11}\norm{\dyad{\kA\kB}_{K_AK_B}-\dyad{\kB\kB}_{K_AK_B}}_1\norm{\rho^{11,\kA\kB}_{LSE}}_1\\
    =&\Pr[\KA\neq \KB,\FA=1,\FB=1]\\
    \leq&\epsMS.
\end{split}
\end{equation}
This accounts for the contribution of match security.\\

The final term can be reduced to the key secrecy condition.
We first note that when $\FA=\FB=1$, whenever the key $\KA$ is generated, it will be identical to $\KB$ in both $\tilde{\rho}$ and $\rhoideal$.
When $\FA=0,\FB=1$ instead, $\KA$ is simply $\perp$.
As such, we can define a CPTP map to get a value of $\KA$ based on $\FA$ and $\KB$.
Since CPTP maps cannot increase trace distance, we can remove the $\KA$ subsystem from the trace distance.
Expanding the final term, we get
\begin{equation}
\begin{split}
    &\Delta\left(\tilde{\rho}_{\land(\FA,\FB)\in\{01,11\}},\rhoideal_{\land (\FA,\FB)\in\{01,11\}}\right)\\
    \leq&\Delta\left(\sum_{(\fA,\fB)\in\vars{S}'}p_{\fA\fB}^{\reall}\dyad{\fA\fB}\otimes\rho_{\KB LSE}^{\fA\fB},\sum_{(\fA,\fB)\in\vars{S}'}p_{\fA\fB}^{\reall}\dyad{\fA\fB}\otimes\tau_{\KB}\otimes\rho_{LSE}^{\fA\fB}\right)\\
    =&\Delta\left(\rho_{\KB\FA\FB LSE\land(\FA,\FB)\in\{01,11\}},\tau_{\KB}\otimes\rho_{\KB\FA\FB LSE\land(\FA,\FB)\in\{01,11\}}\right)\\
    \leq&\epsKS
\end{split}
\end{equation}
where $\vars{S}'$ refers to the set $\{01,11\}$.
Combining the results, we have that
\begin{equation}
    \Delta(\rho,\rhoideal)\leq\epsEA+\epsMS+\epsKS,
\end{equation}
which we can use to define $\epssec$.\\

In the case of the intermediate sessions,
we can find an $\rhoidealint_{\KA\KB\FA\FB LSE}\in\vars{S}_{out}'$ with the same probability distribution $p_{\FA\FB}^{\ideal}$, but with $\sigma_{LSE}^{\fA\fB}\in\vars{S}_{in}$.
We first note that if we choose an ideal state $\rhoideall{1}$ that is the same as $\rhoideal$ except that $\rhoideall{\fA,\fB}_{LSE}=\frac{p^{\reall}_{\fA\fB}}{p^{\reall}_{\fA\fB}+p^{\reall}_{\fA'\fB'}}\rho_{LSE}^{\fA\fB}+\frac{p^{\reall}_{\fA'\fB'}}{p^{\reall}_{\fA\fB}+p^{\reall}_{\fA'\fB'}}\sigma_{LSE}^{\fA'\fB'}$ for $(\fA,\fB,\fA',\fB')\in\{(\phi,0,\phi,1),(0,\phi,1,\phi),(0,0,1,0)\}$ and any $\sigma_{LSE}^{\fA'\fB'}\in\vars{S}_{in}$, the analysis above remains unchanged and by the triangle inequality,
\begin{equation}
    \Delta(\rho,\rhoidealint)\leq\epsEA+\epsMS+\epsKS+\Delta(\rhoideall{1},\rhoidealint).
\end{equation}
The trace distance can be expanded as
\begin{equation}
\begin{split}
    &\Delta(\rhoideall{1},\rhoidealint)\\
    \leq&\Delta\left(\sum_{(\fA,\fB)\notin\vars{S}_{\SP}}p_{\fA\fB}^{\reall}\sigma_{\KA\KB}^{\fA\fB}\otimes\rho_{LSE}^{\fA\fB},\sum_{(\fA,\fB)\notin\vars{S}_{\SP}}p_{\fA\fB}^{\reall}\sigma_{\KA\KB}^{\fA\fB}\otimes\sigma_{LSE}^{\fA\fB}\right)\\
    \leq&\Delta\left(\sum_{(\fA,\fB)\notin\vars{S}_{\SP}}p_{\fA\fB}^{\reall}\rho_{LSE}^{\fA\fB},\sum_{(\fA,\fB)\notin\vars{S}_{\SP}}p_{\fA\fB}^{\reall}\sigma_{LSE}^{\fA\fB}\right)\\
    \leq&\epsSP,
\end{split}
\end{equation}
where $\sigma_{\KA\KB}^{\fA\fB}$ is the ideal form of the keys (either $\perp$ or uniform) depending on $\fA,\fB$, and $\sigma_{\KA\KB}^{\fA\fB}$ is identical in both $\rhoideal$ and $\rhoidealint$.
The second inequality stems from reversing the CPTP map of selecting $\KA\KB$ from $\FA\FB$, which is implicitly part of $E$ since these values can be public in general.
Combining the results, we get that
\begin{equation}
    \Delta(\rho,\rhoidealint)\leq\epsEA+\epsMS+\epsKS+\epsSP,
\end{equation}
which we can use to define $\epssecint$.
\end{proof}

\subsection{Single-Session Security Analysis}
\label{app:AKEProtocolSec_OverallSec}

Since the protocol $\vars{P}$ is a linear map, we can break the output state into components based on the input state, i.e. $\rhooutt{jj'}=\vars{P}(\rhoinn{jj'}_{LSE})$, and the examine their trace distances separately,
\begin{equation}
    \Delta(\rhoout,\rhoidealint)\leq\sum_{jj'}p_{jj'}\Delta(\vars{P}(\rhooutt{jj'}),\rhoidealintt{jj'}).
\end{equation}
As such, it is sufficient to analyse the single-session security for a single component,
\begin{equation}
\begin{gathered}
    \Delta(\rhooutt{jj'},\rhoidealintt{jj'})\leq\varepsilon_{\secc,\intt,jj'}\\\Delta(\rhooutt{jj'},\rhoideall{jj'})\leq\varepsilon_{\secc,jj'},
\end{gathered}
\end{equation}
with overall security $\epssec=\max_{jj'}\varepsilon_{\secc,jj'}$ and $\epssecint=\max_{jj'}\varepsilon_{\secc,\intt,jj'}$.
Furthermore, noting that there are restrictions of labels $\beta$ and $\beta'$ when the ideal input state has label $\alpha$ and $\alpha'$, the ideal state for a single component can be further restricted to states of the form
\begin{equation}
\label{eqn:QAKE_ideal_output_state_Full}
\begin{split}
    \rhoidealintt{jj'}=&\dyad{\perp\perp}_{\KA\KB}\otimes\left[\dyad{0\phi}_{\FA\FB}\otimes\sum_{\tilde{j}\geq j}p_{0\phi,\tilde{j}}\sigma^{0\phi,\tilde{j}+1,j'}_{LSE}\right.+\dyad{\phi0}_{\FA\FB}\otimes\sum_{\tilde{j}'\geq j'}p_{\phi 0,\tilde{j}'}\sigma^{\phi0,j,\tilde{j}'+1}_{LSE}\\
    &\left.+\dyad{00}_{\FA\FB}\otimes \sum_{\tilde{j}\geq j,\tilde{j}'\geq j'}p_{00,\tilde{j}\tilde{j}'}\sigma^{00,\tilde{j}+1,\tilde{j}'+1}_{LSE}\right]+\dyad{01,\perp}_{\FA\FB\KA}\otimes\tau_{\KB}\otimes\sum_{\tilde{j}\geq\max\{j,j'\}}p_{01,\tilde{j}}\sigma^{01,\tilde{j}+1,\tilde{j}}_{LSE}\\
    &+\dyad{11}_{\FA\FB}\otimes\tilde{\tau}_{\KA\KB}\otimes\sum_{\tilde{j}\geq\max\{j,j'\}}p_{11,\tilde{j}}\sigma^{11,\tilde{j}\tilde{j}}_{LSE},
\end{split}
\end{equation}
where $\sigma^{\fA\fB,jj'}_{LSE}\in\vars{S}_{in}$, $\tilde{j}$ and $\tilde{j}'$ indicates $\beta=\tilde{j}$ and $\beta'=\tilde{j}'$, with the summation of $\tilde{j}$ or $\tilde{j}'$ stemming from the ability for an adversary to force Alice and Bob to adopt a larger index $\beta$ and $\beta'$ during the session by attacking the label agreement step.
As such, when we perform the security analysis, we focus on a single component ideal input state and ideal output states with the form above.\\

The single-session security can be summarised as Thm.~\ref{thm:AKEProtocolMain}, which is restated below.
\AKEProtMainProof*

\begin{figure*}
    \centering
    \includegraphics[width=0.8\linewidth]{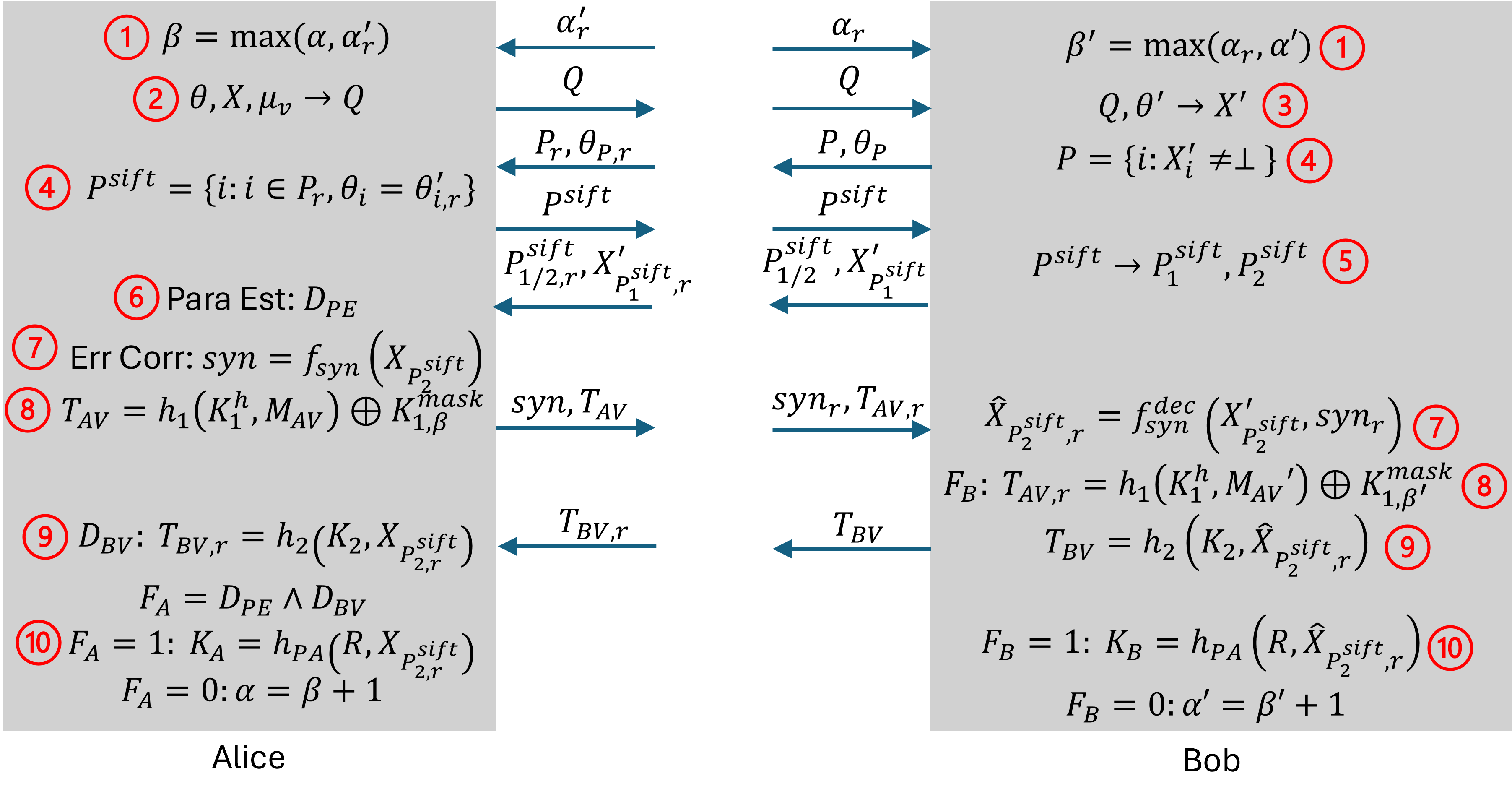}
    \caption{Summary of Protocol 1. $M_{\AV}$ labels the messages to be authenticated for which Alice hashes and $M_{\AV}'$ labels the corresponding messages hashed by Bob to authenticated these messages.}
    \label{fig:QAKE_Summary}
\end{figure*}

\begin{proof}[Sketch Proof]
To aid in the proof, we present a figure summarising the protocol in Fig.~\ref{fig:QAKE_Summary}.
Thm.~\ref{thm:QAKESecProofEA}
presents the proof of explicit entity authentication.
When Alice is not involved, $\FA=\phi$, even with access to an oracle giving one message-tag pair for $T_{AV}$, the $\epsMACa$-almost 2-universal hash function ensures that $\FB=0$ since it is difficult for the adversary to find a second message-tag pair.
On the other hand, when $\FB=\phi$ or $\FB=0$, no $T_{BV}$ is generated by Bob and thus by $\epsMACb$-almost 2-universal property of $h_2$, Alice would not authenticate Bob, $\FA=0$.\\

Thm.~\ref{thm:QAKESecProofMS}
bounds the match security parameter, which is identical to the correctness property of QKD.
This is guaranteed by the error verification using hash function $h_2$.\\

The authentication steps guarantee not only entity authentication, but also message authentication.
This ensures ideal features, such as: (1) labels match, $\beta=\beta'$, (2) step ordering is obeyed, (3) messages match, (4) tags unaltered.
In Thm.~\ref{thm:QAKE_ideal_protocol_replace}, we show that we can convert the authentication checks into ideal checks for these features without much penalty.
Furthermore, in Thm.~\ref{thm:QAKE_ideal_DPE_replace}, we also reduce the parameter estimation to an ideal parameter check.
These modification reduces the analysis for key secrecy to a QKD analysis -- where these ideal features occur, and we prove key secrecy in Thm.~\ref{thm:QAKESecProofMITMMain}.\\

Finally, shared secrets privacy accounts for the fact that in the worst case when keys are not generated $\FB\neq 1$, either no secrets is leaked when authentication fails or at most $T_{AV}$ is announced.
Thm.~\ref{thm:QAKESecProofSSphi0} and Thm.~\ref{thm:QAKESecProofSSFA0} shows that this leads to an ideal output state in $\vars{S}_{out}'$.
When keys are generated, i.e. parameter estimation passes, Thm.~\ref{thm:QAKESecProofSSFB1} shows that the min-entropy of $X_{P_2^{sift}}$ ensures that authentication tags $T_{AV}$ and $T_{BV}$ leak minimal information on the authentication and masking key, while the strong extractor property for key generation ensures the privacy amplification key $R$ is kept secret.
\end{proof}

We note that while the overall security parameters $\epssec$ and $\epssecint$ can be determined by a sum of the respective security parameters, the bound may be tightened due to overlapping considerations for different security conditions.

\subsubsection{Explicit Entity Authentication}
\label{app:QAKE_explicit_EA}

Here, we consider the first security condition of explicit entity authentication, which consists of Alice's full explicit entity authentication and Bob's almost-full explicit entity authentication.
We note that the events of $\FA=\phi$, $\FB=\phi$ and $\FA,\FB\neq\phi$ are mutually exclusive and the adversary is assumed to select one 
attack at the start of the protocol.
\begin{theorem}
\label{thm:QAKESecProofEA}
Consider the QAKE protocol $\vars{P}$, the full explicit entity authentication security parameter is
\begin{equation*}
    \epsEAf=\frac{\Pr[\FB\neq 1]}{\abs{\vars{T}_{\BV}}}.
\end{equation*}
and the almost-full explicit entity authentication security parameter is
\begin{equation*}
    \epsEAaf=\Pr[\FA=\phi](\epsMACa+2^{-\Psift_{\tol}[1+\hbin(f_{P_1})]+\log_2(\Psift_{\tol}+1)}).
\end{equation*}
\end{theorem}
\begin{proof}
We begin with the case of full explicit entity authentication, which we recall bounds the probability of events $(\FA,\FB)\in\{10,1\phi\}$.
To obtain $\FA=1$, Alice needs to at least successfully validate Bob by checking if $\tilde{T}_{\BV}=T_{\BV,\rr}$, i.e. $\Pr[\FA=1|\FB=\phi]\leq\Pr[\tilde{T}_{\BV}=T_{\BV,\rr}|\FB=\phi]$.
Since $K_2$ remains private from the adversary without Bob sending $T_{\BV}$, we can use the uniformity of the strong 2-universal hash function to have
\begin{equation}
    \Pr[\FA=1|\FB=\phi]\leq\Pr[h_2(K_2,M)=T_{\BV,\rr}]=\frac{1}{\abs{\vars{T}_{\BV}}}
\end{equation}
for any $T_{\BV,\rr}$.
A similar argument holds when conditioned on  
$\FB=0$ since a random string is sent in place of a valid tag.
As such, full explicit entity authentication security parameter is  $\epsEAf=\frac{\Pr[\FB\neq 1]}{\abs{\vars{T}_{\BV}}}$.\\

For the case of almost-full explicit entity authentication, Alice does not participate in the protocol, $\FA=\phi$.
When $\alpha>\beta'$, the adversary may have partial information on $(\Kh_1,\KOTP_{1,\beta'})$ from a previous session via a tag generated with a chosen message input $M$.
As such, the two events that can lead to $\FB=1$, or successful validation $\tilde{T}_{\AV}=T_{\AV,\rr}$, are either: (1) message $M$ generating the tag matches the message $M_{\AV}'$ in the current session, or (2) guessing the correct $T_{\AV,\rr}$ to send to Bob for a mismatch message, $M\neq M_{\AV}'$.
As such, we can bound
\begin{equation}
\begin{split}
    \Pr[\FB=1|\FA=\phi]\leq&\Pr[h_1(\Kh_1,M_{\AV}')\oplus \KOTP_{1,\beta'}=T_{\AV,\rr}|h_1(\Kh_1,M)\oplus \KOTP_{1,\beta'}=T_{\AV},M\neq M_{\AV}']+\pguess(M_{\AV}')\\
    \leq&\epsMACa+2^{-\Psift_{\tol}[1+\hbin(f_{P_1})]+\log_2(\Psift_{\tol}+1)},
\end{split}
\end{equation}
where the $\epsMACa$-almost strong 2-universal hash function guarantees that the adversary cannot successfully guess a second message-tag pair with probability higher than $\epsMACa$ and the guessing of $M_{\AV}'$ for $\FA=\phi$ reduces to the ability to guess $\theta_P'$ and $\Psift_1$.
As such, the almost-full explicit entity authentication security parameter is $\epsEAaf=\Pr[\FA=\phi](\epsMACa+2^{-\Psift_{\tol}[1+\hbin(f_{P_1})]+\log_2(\Psift_{\tol}+1)})$.
\end{proof}

\subsubsection{Match Security}
\label{app:QAKE_Match_Sec}

The match security condition is similar to the correctness condition of QKD, and the security guarantee can be provided during Bob validation, where the string $X_{\Psift_2,\rr}$ and $\hat{X}_{\Psift_2,\rr}$ are matched via a hash.

\begin{theorem}
\label{thm:QAKESecProofMS}
Consider the QAKE protocol $\vars{P}$, the match security parameter is
\begin{equation*}
    \epsMS=\epsMACb.
\end{equation*}
\end{theorem}
\begin{proof}
Since the keys generated $\KA=\hPA(R,X_{\Psift_{2,\rr}})$ and $\KB=\hPA(R,\hat{X}_{\Psift_{2,\rr}})$ with the same privacy amplification seed $R$ would match when $X_{\Psift_{2,\rr}}=\hat{X}_{\Psift_{2,\rr}}$, i.e. $\{X_{\Psift_{2,\rr}}=\hat{X}_{\Psift_{2,\rr}}\}\implies\{\KA=\KB\}$, the converse gives
\begin{equation}
    \Pr[\KA\neq \KB,\FA=\FB=1]\leq\Pr[X_{\Psift_{2,\rr}}\neq\hat{X}_{\Psift_{2,\rr}},\FA=\FB=1].
\end{equation}
The event $\FA=1$ contains the event 
$\tilde{T}_{\BV}=T_{\BV,\rr}$ (part of $D_{\BV}=1$), allowing us to upper bound by
\begin{equation}
    \Pr[h_2(K_2,X_{\Psift_{2,\rr}})=T_{\BV,\rr}|X_{\Psift_{2,\rr}}\neq\hat{X}_{\Psift_{2,\rr}},h_2(K_2,\hat{X}_{\Psift_{2,\rr}})=T_{\BV}],
\end{equation}
which describes the probability that an adversary can provide a valid tag to a message when given the tag to a different message.
Since $K_2$ is private from the adversary, by the $\epsMACb$-almost strong 2-universal property of the hash,
$\Pr[\KA\neq \KB,\FA=\FB=1]\leq\epsMACb$.
\end{proof}

\subsubsection{Idealised Parameter Estimation}
\label{app:ideal_DPE_switch}

To simplify analysis, we introduced an ``idealised" version of parameter estimation related to the decoy state analysis.
We define an idealised version of the the check $\DPE$ as
\begin{equation}
    \tildeDPE=\begin{cases}
        1 & \substack{N_{\Psift_{2,\rr},1}\geq N_{\Psift_2,1}^{\tol},\, N_{\Psift_{1,\rr},1}\geq N_{\Psift_1,1}^{\tol},\,\DPE=1\\
        \frac{wt(X_{\Psift_{1,\rr},1}\oplus X_{\Psift_1,\rr}'[\{\Psift_{1,\rr},1\}])}{N_{\Psift_{1,\rr},1}}\leq \ebitt{1,\tol}}  \\
        0 & otherwise
    \end{cases},
\end{equation}
where $X_{\Psift_1,\rr}'[\{\Psift_{1,\rr},1\}]$ refers to the bit values of the bitstring $X_{\Psift_1,\rr}'$ at indices where $i\in \Psift_{1,\rr},$ and a single photon is sent, and $wt(S)$ gives the Hamming weight of a string $S$, i.e. the final condition means the actual single photon events in $\Psift_{1,\rr}$ gives a bit error value upper bounded by $\ebitt{1,\tol}$.
Here, we added the three desired conditions (actual single-photon detection/error fall within the bounds).\\

\begin{theorem}
\label{thm:QAKE_ideal_DPE_replace}
    Let $\vars{P}''$ be the protocol when $D_{PE}$ is replaced with $\tilde{D}_{PE}$. Then, for any $\sigma_{LSE}\in\vars{S}_{in}$,
    \begin{equation*}
        \Delta(\vars{P}(\sigma_{LSE}),\vars{P}''(\sigma_{LSE}))\leq2\epsds,
    \end{equation*}
    where $\epsds$ accounts for errors associated with decoy state analysis.
\end{theorem}

\begin{proof}
The replacement of $\DPE$ with $\tildeDPE$ would result in a difference only when $\DPE=1$ but $\tildeDPE=0$.
The probability of this event is tied to decoy state analysis.
From decoy state analysis~\cite{Lim2014_DecoyQKD}, the single-photon quantities of interest can always be bounded by
\begin{align*}
    N_{\Psift_i,1}\geq &\frac{p_1\mu_0}{(\mu_1-\mu_2)(\mu_0-\mu_1-\mu_2)}\times\left[\frac{e^{\mu_1}}{p_{\mu_1}}\left(1-\frac{\mu_2(\mu_1+\mu_2)}{\mu_0^2}\right)\mathbb{E}[N_{\Psift_i,\mu_1}]-\frac{(\mu_1^2-\mu_2^2)e^{\mu_0}}{\mu_0^2p_{\mu_0}}\mathbb{E}[N_{\Psift_i,\mu_0}]\right.\\
    &\left.-\frac{e^{\mu_2}}{p_{\mu_2}}\left(1-\frac{\mu_1(\mu_1+\mu_2)}{\mu_0^2}\right)\mathbb{E}[N_{\Psift_i,\mu_2}]\right]\\
    \bigepsilon_{\Psift_1,1}\leq & \frac{p_1}{\mu_1-\mu_2}\left(\frac{e^{\mu_1}}{p_{\mu_1}}\mathbb{E}[N_{\bigepsilon,\Psift_1,\mu_1}]-\frac{e^{\mu_2}}{p_{\mu_2}}\mathbb{E}[N_{\bigepsilon,\Psift_1,\mu_2}] \right),
\end{align*}
where $\bigepsilon_{\Psift_1,1}$ is the number of error bits in the single photon events in set $\Psift_1$, and $p_1$ is the probability of single-photon signals sent by the source.
The estimated quantities $\hat{N}_{\Psift_i,1}^{\LB}$ and $\hatebitt{\Psift_1,1}^{\UB}$ can be computed similarly from the above bounds, but utilises appropriate bounds on the expectation values $\mathbb{E}[N_{\Psift_i,\mu_j}]$ and $\mathbb{E}[N_{\bigepsilon,\Psift_i,\mu_j}]$ derived from the values observed during the protocol.
For instance, $\hat{N}_{P_i,1}^{\LB}$ utilises the estimated lower bound $\mathbb{E}[N_{\Psift_i,\mu_1}^{\LB}]$ instead of $\mathbb{E}[N_{\Psift_i,\mu_1}]$ (which it does not have access to), estimated using concentration bounds (we use Kato's inequality~\cite{Kato2020_KatoIneq,Curras2021_KatoIneqPara} as the concentration bound) from the observed $N_{\Psift_i,\mu_1}$ value.
The probability of the estimated lower bound $\mathbb{E}[N_{\Psift_i,\mu_1}^{\LB}]$ exceeding the true expectation value $\mathbb{E}[N_{\Psift_i,\mu_1}]$ is quantified by $\varepsilon_{N_{\Psift_i,\mu_1}}$.
As such, when these concentration bounds are respected, we have that $N_{\Psift_i,1}\geq\hat{N}_{\Psift_i,1}^{\LB}$, which implies that the probability that $N_{\Psift_i,1}<N_{\Psift_i,1}^{\tol}$ while $\hat{N}_{\Psift_i,1}^{\LB}\geq N_{\Psift_i,1}^{\tol}$ is bounded by the events where the concentration bounds are violated.
As such, the probability of $\DPE=1$ but $\tildeDPE=0$ is the sum of all such concentration bound violation events and is labelled $\epsds$.
Therefore, the trace distance on the output states gains a $2\epsds$ penalty when we replace the decision $\DPE$ with its idealised version $\tildeDPE$.\\
\end{proof}

Applying Thm.~\ref{thm:QAKE_ideal_DPE_replace} replaces the parameter check with an ideal check of the single qubit parameters.
As such, the later analysis also applies to a QAKE protocol where single-qubit BB84 states are prepared and single qubit bit error and number of rounds can be observed directly.

\subsubsection{Idealised Authentication}
\label{app:ideal_version_switch}

The validation steps in the protocol also serve as message authentication steps.
To further simplify the analysis, we introduce an ``idealised" version of message authentication checks, replacing the authentication decisions $\FB$ and $D_{\BV}$ with idealised versions.
The purposes of authentication checks are twofold; They serve to ensure that the message transmitted between Alice and Bob match, and that the order of the communication steps are obeyed.
As such, when we define the idealised authentication checks, these are the conditions that we impose.
Crucially, there are several events that we desire in an ideal check which are provided with high probability:
\begin{enumerate}
    \item Matching label agreement, $\beta=\beta'$, is desirable to ensure the label agreement step is obeyed. Its security is provided by the fact that the masking key $\KOTP_{1,\beta}$ choice is label-dependent, and mismatching masking key would lead to failed authentication.
    \item Parameter estimation passing, $\DPE=1$, is meant to indicate to Bob that parameter estimation has passed. Its security is guaranteed by Alice's decision to not send a valid $T_{\AV}$ when $\DPE=0$.
    \item Alice's tag generation occurs before Alice's tag is validated, which indicates the ordering of Alice's Validation step is obeyed, i.e. the event $\Omega_{5\rightarrow 6}$, where $\Omega_{i\rightarrow j}$ labels that the $i$-th step occur before the $j$-th step. Note here that the steps refers a set of action by either Alice or Bob, receiving an input and providing and output from the communication channel. For QAKE, there are a total of seven steps.
    \item Matching message, $M_{\AV}=M_{\AV}'$, where $M_{\AV}=(X_{\Psift_{2,\rr}},P_{\rr},\theta_{P,\rr}',\Psift,X_{\Psift_1,\rr}',\Psift_{1,\rr},\Psift_{2,\rr},syn)$ and $M_{\AV}'=(\hat{X}_{\Psift_{2,\rr}},P,\theta_P',\Psift_{\rr},X_{\Psift_1}',\Psift_1,\Psift_2,syn_\rr)$, which also contains the error verification check ($X_{\Psift_{2,\rr}}=\hat{X}_{\Psift_{2,\rr}})$.
    \item Tags unaltered, $T_{\AV}=T_{\AV,\rr}$.
    \item Protocol steps before Alice's validation are obeyed, namely the event $\{\Omega_{i\rightarrow i+1}\}_{i\in[2,4]}$. We note that the first step, Alice's state preparation, would WLOG occur before step 2 (Bob responding with $P$ and basis $\theta_P$)\footnote{If the state preparation occurs after step 2, the state is always equivalent to one where the state preparation occurs before step 2 but with the adversary not interacting with the quantum state sent by Alice prior to step 2. This is because the state preparation does not require any input from the channel.}.
\end{enumerate}
As such, we label a replacement decision $\tildeFB$, which is made at the same time as $\FB$, but with separate checks\footnote{The checks can be considered to be performed by a hypothetical third-party, though such a party need not be physically present in the protocol.}.
For validation of Alice, this replacement made is 
\begin{equation}
    \tildeFB=\begin{cases}
        1 & \substack{\beta=\beta',\,\DPE=1,\,T_{\AV}=T_{\AV,\rr},\\M_{\AV}=M_{\AV}',\Omega_{2\rightarrow3\rightarrow4\rightarrow5\rightarrow 6}}\\
        0 & otherwise
    \end{cases},
\end{equation}
where $\Omega_{2\rightarrow3\rightarrow4\rightarrow5\rightarrow6}$ is the event where the step ordering is $2\rightarrow3\rightarrow4\rightarrow5\rightarrow6$.
The list of steps are:
\begin{enumerate}
    \item Alice's state preparation and sending of quantum state.
    \item Bob's measurement and reply of $P$ and $\theta_P'$.
    \item Alice's sifting and response $\Psift$.
    \item Bob's choice of test round and announcement of $\Psift_1$, $\Psift_2$ and $X_{\Psift_1}'$.
    \item Alice's parameter estimation and validation tag generation, sending syndrome $syn$ and tag $T_{\AV}$ to Bob.
    \item Bob performs Alice's validation, and responds with his own validation tag $T_{\BV}$.
    \item Alice validates Bob's tag.
\end{enumerate}

The replacement of $D_{\BV}$ is simpler, with 
\begin{equation}
    \tilde{D}_{\BV}=\begin{cases}
        1 & \tildeFB=1,\,T_{\BV}=T_{\BV,\rr},\,\Omega_{6\rightarrow 7}\\
        0 & otherwise
    \end{cases},
\end{equation}
where $\Omega_{6\rightarrow 7}$ indicates that Bob sends $T_{\BV}$ before Alice is requested to verify, and $\tildeFA:=\tilde{D}_{\BV}\land \DPE$.\\

\begin{theorem}
\label{thm:QAKE_ideal_protocol_replace}
    Let $\vars{P}$ be the protocol and $\vars{P}'$ be the idealised protocol (with $\DPE$, $\FB$, $D_{\BV}$ replaced with $\tildeDPE$, $\tildeFB$ and $\tilde{D}_{\BV}$). Then, for any $\sigma_{LSE}\in\vars{S}_{in}$,
    \begin{equation*}
        \Delta(\vars{P}(\sigma_{LSE}),\vars{P}'(\sigma_{LSE}))\leq 2(\epsMACa+\epsSO)+2\epsMACb+2\epsds,
    \end{equation*}
    where $\epsSO=2^{-\Psift_{tol}}+2^{-N_{\Psift_{2,1}}^{\tol}}+2^{-\Psift_{\tol}\hbin(f_{P_1})+\log_2(\Psift_{\tol}+1)}$.
\end{theorem}

\begin{proof}
From Thm.~\ref{thm:QAKE_ideal_DPE_replace}, the switch to $\tildeDPE$ introduces a $2\epsds$ penalty.
As such, we can consider only the changes in $\FB$ and $D_{\BV}$, and combine the results using triangle inequality at the end.\\

For the switch
from $\FB$ to $\tildeFB$ for scenarios where $(K_1^h,\KOTP_{1,\beta'})$ is private (where $\beta'\geq\alpha$),
the replacement results in a penalty of $2\Pr[\FB=1,\tildeFB=0]$.
This probability can be computed,
\begin{equation}
\label{eqn:QAKE_DAV_change}
\begin{split}
    &\Pr[\FB=1,\tildeFB=0]\\
    \leq&\Pr[\FB=1,(\Omega^c_{5\rightarrow6}\lor\DPE=0)]+\Pr[\FB,\DPE=1,\Omega_{5\rightarrow6},\beta\neq\beta']\\
    &+\Pr[\FB,\DPE=1,\Omega_{5\rightarrow6},\beta=\beta',M_{\AV}\neq M_{\AV}']\\
    &+\Pr[\FB,\DPE=1,\Omega_{5\rightarrow6},\beta=\beta',M_{\AV}= M_{\AV}',T_{\AV}\neq T_{\AV,r}]\\
    &+\Pr[\FB,\DPE=1,\Omega_{5\rightarrow6},\beta=\beta',M_{\AV}= M_{\AV}',T_{\AV}=T_{\AV,r},\Omega^c_{1\rightarrow2\rightarrow3\rightarrow4}]\\
    \leq&\frac{\Pr[\Omega^c_{5\rightarrow6}\lor\DPE=0]}{\abs{\vars{T}_{\AV}}}+\frac{\Pr[\DPE=1,\Omega_{5\rightarrow6},\beta\neq\beta']}{\abs{\vars{T}_{\AV}}}+\Pr[\DPE=1,\Omega_{5\rightarrow6},\beta=\beta',M_{\AV}\neq M_{\AV}']\epsMACa\\
    &+\Pr[M_{\AV}= M_{\AV}',\Omega^c_{2\rightarrow3}]+\Pr[M_{\AV}= M_{\AV}',\Omega_{2\rightarrow3},\Omega^c_{3\rightarrow4}]+\Pr[M_{\AV}= M_{\AV}',\Omega_{2\rightarrow3\rightarrow4},\Omega^c_{4\rightarrow5}]\\
    \leq&\epsMACa+\pguess(\theta_P')+\pguess(\theta_P)+\pguess(\Psift_1),
\end{split}
\end{equation}
where $\pguess(Z)$ is the probability of the adversary guessing the parameter $Z$ successfully before $Z$ is announced.
The first inequality splits the ``bad events" into different cases.
The second inequality uses the fact that when $\Omega_{5\rightarrow 6}^c$ or $\DPE=0$, no valid tag is generated before Bob checks Alice's authentication tag.
Therefore, by the property of the strong 2-universal hash function, the probability of $\FB=1$ is no better than a random guess of the tag.
When $\beta\neq\beta'$, since authentication key $K_{1,\beta}$ is random and uncorrelated to $K_{1,\beta'}$, the adversary can do no better than make a random guess of the tag, which results in a $\frac{1}{\abs{\vars{T}_{\AV}}}$ penalty.
As for message mismatch, the property of strong 2-universal hash function ensures that the probability of generating a correct tag for $M_{\AV}'$ is bounded by $\epsMACa$.
The final inequality combines the first three terms to upper bound it by $\epsMACa$, while events $\Omega_{i\rightarrow j}^c$ with matching messages indicate that the adversary made a correct guess of the message sent in step $i$, which we label $\pguess(Z)$ for message $Z$.
For step 2 (Bob response after quantum state measurement), the adversary has to guess $\theta_P'$, while for step 3, this quantity is $\theta_P$ (equivalent to guessing $\Psift$ with knowledge of $\theta_P'$ from step 2), and in step 4, the quantity is the choice of test round $\Psift_1$ from $\Psift$.\\

We can bound the guessing probability explicitly.
The guessing probability of $\theta_P'$ is $2^{-\abs{P}}\leq2^{-\Psift_{\tol}}$ since the basis are randomly chosen, and $\abs{P}\leq\Psift_{\tol}$.
The probability of guessing $\theta_P$ can be bounded by the probability of guessing Alice's basis choice for the $N_{\Psift_{2,1}}^{\tol}$ single-photon rounds, where the quantum state cannot provide any information on the basis.
Finally, the probability of guessing the test round choice is~\cite{Cover2005_InfTheory} 
\begin{equation}
    \begin{pmatrix}\abs{\Psift}\\  f_{\Psift_1}\abs{\Psift}\end{pmatrix}^{-1}\leq2^{-\abs{\Psift}\hbin(f_{\Psift_1})+\log_2(\abs{\Psift}+1)}\leq 2^{-\Psift_{tol}\hbin(f_{P_1})+\log_2(\Psift_{tol}+1)},
\end{equation}
where we drop the ceiling function ($\lceil\rceil$) for simplicity and note that the second term is decreasing in $\abs{P}$ for any chosen number of test rounds (at least 1 test and key round).
The overall penalty to satisfy step order can be summarised as $\epsSO=2^{-\Psift_{tol}}+2^{-N_{\Psift_{2,1}}^{\tol}}+2^{-\Psift_{\tol}\hbin(f_{P_1})+\log_2(\Psift_{\tol}+1)}$.\\

To make the same switch for cases when $(\Kh_1,\KOTP_{1,\beta'})$ is accessible via a single oracle access ($\bigepsilon_{\Kh_1,\KOTP_{1,\beta'}E_i\rightarrow E_{i+1}}$), we first note that such cases require $\beta'<\alpha$ and by extension $\beta'<\beta$, i.e. $\tildeFB$ should always return 0 since $\beta$ and $\beta'$ are mismatched.
We first note that $T_{\AV}$ generated with private $\KOTP_{1,\beta}$ provides no advantage to the adversary in passing $\FB$.
As such, we can follow the argument in Thm.~\ref{thm:QAKESecProofEA} to show that the probability of $\FB=1$ is upper bounded by
the probability of $\FB=1$ when $\FA=\phi$.
Combining both results, we can thus replace $\FB$ by $\tildeFB$ in general in this protocol with a penalty of $2(\epsMACa+\epsSO)$.\\

By similar arguments, the replacement of $D_{\BV}$ to $\tilde{D}_{\BV}$, which involves mainly Bob's validation step, yields an error of up to $\epsMACb$.
As such, these changes result in a
penalty to the trace distance of $2\epsMACb$.
\end{proof}

\subsubsection{Key Secrecy}

Key secrecy analyses the two specific cases where Bob generates a key, $(\FA,\FB)\in\{01,11\}$. 
Intuitively, since the protocol is similar to decoy-state BB84 with modifications only in its authentication steps, the key secrecy security condition relies heavily on the secrecy condition of BB84.
We can summarise the result as a theorem,
\begin{theorem}
\label{thm:QAKESecProofMITMMain}
Consider the QAKE protocol $\vars{P}$. Then, the key secrecy security parameter is
\begin{equation*}
    \epsKS=\epsKS'+2\left(\epsds+\epsMACa+\epsSO+\epsMACb\right)
\end{equation*}
where
\begin{gather*}
    \epsKS'=2\sqrt{2\epsserfa}+\frac{1}{2}\times 2^{-\frac{1}{2}[N_{P_2,1}^{\tol}-N_{P_2,1}^{\tol}\hbin(\ephtol')-2-\log_2\abs{\vars{T}_{\AV}}\abs{\vars{T}_{\BV}}-\leakEC-l_{\KB}]}
\end{gather*}
is the key secrecy of the idealised protocol $\vars{P}'$.
\end{theorem}
\begin{proof}
We begin by performing a swap to an ``idealised" version of the protocol $\vars{P}'$ by replacing $\FB$, $D_{\BV}$ and $\DPE$.
By Thm.~\ref{thm:QAKE_ideal_protocol_replace}, this gives
\begin{multline}
\label{eqn:QAKE_switch_to_ideal_D}
    \Delta(\rhooutt{jj'}_{\FA\FB\KB LSE\land(\FA,\FB)\in\{01,11\}},\tau_{\KB}\otimes\rhooutt{jj'}_{\FA\FB LSE\land(\FA,\FB)\in\{01,11\}})\\
    \leq\Delta(\rhopoutt{jj'}_{\tildeFA\tildeFB\KB LSE\land(\tildeFA,\tildeFB)\in\{01,11\}},\tau_{\KB}\otimes\rhopoutt{jj'}_{\tildeFA\tildeFB LSE\land(\tildeFA,\tildeFB)\in\{01,11\}})\\
    +2\left(\epsds+\epsMACa+\epsSO+\epsMACb\right),
\end{multline}
where $\rho'$ label the output state when an idealised version of the protocol is utilised.
We first combine 01 and 11 cases into a single one by arguing that the trace distance is independent of the final communication step using Thm.~\ref{thm:QAKEKeySecIdealisedReplace}.
This is followed by analysing the trace distance shown in Thm.~\ref{thm:QAKEKeySecIdealisedProtocol}.
Combining the results, we obtain the security parameter stated in the theorem.
\end{proof}

To simplify the analysis, we combine the cases of $(\tildeFA,\tildeFB)$ being 01 and 11 into a single condition on $\tildeFB=1$ since only the key $\KB$ is of interest.
\begin{theorem}
\label{thm:QAKEKeySecIdealisedReplace}
    Consider the QAKE protocol $\vars{P}'$. Then, we have
    \begin{equation*}
        \Delta(\rhopoutt{jj'}_{\tildeFA\tildeFB\KB LSE\land(\tildeFA,\tildeFB)\in\{01,11\}},\tau_{\KB}\otimes\rhopoutt{jj'}_{\tildeFA\tildeFB LSE\land(\tildeFA,\tildeFB)\in\{01,11\}})\leq \Delta(\rhopoutt{jj'}_{\KB L'S'E\land \tildeFB=1},\tau_{\KB}\otimes\rhopoutt{jj'}_{L'S'E\land \tildeFB=1}),
    \end{equation*}
    where the states on RHS has subsystems $L'S'$, indicating the labels and secrets correspond to that when the final step of Alice has been removed.
\end{theorem}
\begin{proof}
We first note that since $\tildeFB=1$, we have $\tildeDPE=1$ and the decision of $\tildeFA$ depends solely on $\tilde{D}_{\BV}$ generated in the final step.
Moreover, since $T_{\BV}\neq\perp$, the $\tildeFA$ condition depends solely on $T_{\BV}=T_{\BV,\rr}$, i.e. whether the adversary altered $T_{\BV}$ transmission from Bob to Alice.
As such, the final step can be re-expressed as a 
CPTP map that maps $T_{\BV}$ and $T_{\BV,\rr}$ (part of $E$) to $\tildeFA$ followed by generating secret $S$ from $S'$.
Since CPTP maps cannot increase trace distance, we can reverse the map (i.e. reversing Alice's decision step and subsequent index update, which WLOG is the final step of the protocol), yielding
\begin{equation}
    \Delta(\rhopoutt{jj'}_{\FA\FB\KB LSE\land(\tildeFA,\tildeFB)\in\{01,11\}},\tau_{\KB}\otimes\rhopoutt{jj'}_{\tildeFA\tildeFB LSE\land(\tildeFA,\tildeFB)\in\{01,11\}})\leq \Delta(\rhopoutt{jj'}_{\KB L'S'E\land \tildeFB=1},\tau_{\KB}\otimes\rhopoutt{jj'}_{L'S'E\land \tildeFB=1}),
\end{equation}
where the $L'S'E$ indicates that the state of the subsystem before Alice's final protocol step.
This aligns with our understanding that $\tildeFA$ outcome has no impact on the privacy of $\KB$ since it does not reveal additional information.
\end{proof}

With the simplification, we proceed to prove the key secrecy security parameter for the QAKE protocol with idealised parameter estimation and authentication checks,
\begin{theorem}
\label{thm:QAKEKeySecIdealisedProtocol}
    Consider the QAKE protocol $\vars{P}'$. Then, we have that 
    \begin{multline*}
        \Delta(\rhopoutt{jj'}_{\KB L'S'E\land \tildeFB=1},\tau_{\KB}\otimes\rhopoutt{jj'}_{L'S'E\land \tildeFB=1})\leq 2\sqrt{2\epsserfa}+\frac{1}{2}\times 2^{-\frac{1}{2}[N_{P_2,1}^{\tol}-N_{P_2,1}^{\tol}\hbin(\ephtol')-2-\log_2\abs{\vars{T}_{\AV}}\abs{\vars{T}_{\BV}}-\leakEC-l_{\KB}]},
    \end{multline*}
    where $\ephtol'$ is the error tolerance inclusive of the correction due to Serfling bound.
\end{theorem}

\begin{proof}
Let $\tildeOmegaPE$ be the event where parameter estimation checks $N_{\Psift_1,1}\geq N_{\Psift_1,1}^{\tol}$, $N_{\Psift_2,1}\geq N_{\Psift_2,1}^{\tol}$, $\frac{wt(X_{\Psift_1,1}\oplus X_{\Psift_{1,1}}')}{N_{\Psift_1,1}}\leq \ebitt{1,\tol}$ are successful, for which the actual $\Psift_1$ and $\Psift_2$ sets are utilised instead of the received $\Psift_{1,\rr}$ and $\Psift_{2,\rr}$ in the $\tildeDPE=1$ condition.
Due to the matching message condition in idealised authentication check, when $\tildeFB=1$, $\rho_{\land \tildeFB=1\land\OmegaPE}=\rho_{\land \tildeFB=1\land\tildeOmegaPE}$.
Noting that $\hPA$ is 2-universal, we bound the trace distance using the quantum leftover hash lemma~\cite{Tomamichel2017_QKDProof,Tomamichel2011_QLHL},
\begin{equation}
\begin{split}
    \Delta(\rhopoutt{jj'}_{\KB L'S'E''\land \tildeFB=1},\tau_{\KB}\otimes\rhopoutt{jj'}_{ L'S'E\land \tildeFB=1})
    \leq&\Delta(\rhopoutt{jj'}_{\KB L'S'E\land \tildeFB=1\land\tildeOmegaPE},\tau_{\KB}\otimes\rhopoutt{jj'}_{L'S'E\land \tildeFB=1\land\tildeOmegaPE})\\
    \leq& \frac{1}{2}\left\{4\epssma+2^{-\frac{1}{2}\left[\Hmin^{\epssma}(\hat{X}_{\Psift_{2,\rr}}|L'S''E)_{\rhopoutt{jj'}_{\hat{X}_{P_{2,\rr}}L'S'E\land \tildeFB=1\land\tildeOmegaPE}}-l_{\KB}\right]}\right\},
\end{split}
\end{equation}
where $S''=S'\setminus R$.
Focusing only on the terms relevant to the current session, we can express the smooth min-entropy as
\begin{equation}
    \Hmin^{\epssma}(\hat{X}_{\Psift_{2,\rr}}|L'S'E)_{\rhopoutt{jj'}_{\hat{X}_{\Psift_{2,\rr}}L'S'E\land \tildeFB=1\land\tildeOmegaPE}}=\Hmin^{\epssma}(\hat{X}_{\Psift_{2,\rr}}| \Kh_1K_2\KOTP_{1,\beta}E)_{\rhopoutt{jj'}_{\hat{X}_{\Psift_{2,\rr}}\Kh_1K_2\KOTP_{1,\beta}E\land \tildeFB=1\land\tildeOmegaPE}},
\end{equation}
where we note the labels are implicitly contained in $E$.
Before simplifying the min-entropy, we first list the important conditions required for $\tildeFB=1$ (inclusive of $\tildeDPE=1$).
These are
\begin{enumerate}
    \item $\beta=\beta'$: The index chosen for the protocol session matches.
    \item $M_{\AV}=M_{\AV}'$: Classical messages exchanged between Alice and Bob before Alice Validation step matches.
    \item $\Omega_{2\rightarrow3\rightarrow4\rightarrow5\rightarrow6}$: Protocol step order is respected from steps 2 to 6.
    \item $\tildeOmegaPE$: Standard parameter estimation checks.
\end{enumerate}
With the conditions, we can simplify the smooth min-entropy term similarly to QKD, using Thm.~10 of Ref.~\cite{Tomamichel2017_QKDProof} to remove $\land\Omega$ conditions when necessary,
\begin{equation}
\begin{split}
    &\Hmin^{\epssma}(\hat{X}_{\Psift_{2,\rr}}|\Kh_1K_2\KOTP_{1,\beta}E)_{\rhopoutt{jj'}_{\hat{X}_{\Psift_{2,\rr}}\Kh_1K_2\KOTP_{1,\beta}E\land \tildeFB=1\land\tildeOmegaPE}}\\
    =&\Hmin^{\epssma}(X_{\Psift_{2}}|\Kh_1K_2\KOTP_{1,\beta}T_{\BV}{\tildeFB}E_6)_{\rhopoutt{jj'}_{ X_{\Psift_{2}}\Kh_1K_2\KOTP_{1,\beta}T_{\BV}{\tildeFB}E_6\land\beta=\beta'\land M_{\AV}=M_{\AV}'\land\Omega_{2\rightarrow3\rightarrow4\rightarrow5\rightarrow6}\land\tildeOmegaPE}}\\
    \geq&\Hmin^{\epssma}(X_{\Psift_2}|\Kh_1\KOTP_{1,\beta}T_{\AV}\tildeDPE synE_5)_{\rhopoutt{jj'}_{ X_{\Psift_{2}}\Kh_1\KOTP_{1,\beta}T_{\AV}\tildeDPE synE_5\land\beta=\beta'\land M_{\AV}=M_{\AV}'\land\Omega_{2\rightarrow3\rightarrow4\rightarrow5}\land\tildeOmegaPE}}-1-\log_2\abs{\vars{T}_{\BV}}\\
    \geq&\Hmin^{\epssma}(X_{\Psift_2}| \Psift_1\Psift_2X_{\Psift_1}'E_4)_{\rhopoutt{jj'}_{ X_{\Psift_{2}}\Psift_1\Psift_2X_{\Psift_1}'E_4\land M_{\AV}=M_{\AV}'\land\Omega_{2\rightarrow3\rightarrow4}\land\tildeOmegaPE}}-2-\log_2\abs{\vars{T}_{\BV}}\abs{\vars{T}_{\AV}}-\leakEC\\
    \geq&\Hmin^{\epssma}(X_{\Psift_2}|\Psift_1\Psift_2X_{\Psift_1}'\theta_P\theta_P'PE_1)_{\rhopoutt{jj'}_{ X_{\Psift_{2}}\Psift_1\Psift_2X_{\Psift_1}'\theta_P\theta_P'PE_1\land\tildeOmegaPE}}-2-\log_2\abs{\vars{T}_{\BV}}\abs{\vars{T}_{\AV}}-\leakEC\\
\end{split}
\end{equation}
The first equality uses the fact that Bob's corrected bit string matches Alice's and that $\Psift_2=\Psift_{2,\rr}$, and expands $E=E_5T_{\BV}\tildeFB$, corresponding to output of the sixth step.
The second inequality notes the removal of the classical $T_{\BV}$ and $\tildeFB$ using the chain rule of smooth min-entropy for classical systems~\cite{Tomamichel2015_ITBook}, followed by removing $K_2$, noting that it is independent of the rest of the state without $T_{\BV}$.
Furthermore, the CPTP map describing Eve taking the outputs of step 5 and preparing inputs of step 6 is reversed.
The third inequality removes the classical $T_{\AV}$, $\tildeDPE$ and $syn$ using the chain rule of smooth min-entropy, followed by removal of $K_1^h$ and $\KOTP_{1,\beta}$ since they are independent of the state after removal of $T_{\AV}$.
Furthermore, the CPTP map describing Eve taking the outputs of step 4 and preparing the inputs of step 5 is reversed.
The final inequality reverses the remaining CPTP map to arrive at the state after step 1.\\

The quantum state $\rhopoutt{jj'}_{ X_{\Psift_{2}}\Psift_1\Psift_2X_{\Psift_1}'\theta_P\theta_P'PE_1\land\tildeOmegaPE}$ can be viewed as being generated from a standard decoy-state BB84.
As such, assuming ideal state preparation and fair sampling, we can bound the min-entropy by~\cite{Lim2014_DecoyQKD,Tomamichel2017_QKDProof}, 
\begin{equation}
    \Hmin^{\epssma}(X_{\Psift_2}|\Psift_1\Psift_2X_{\Psift_1}'\theta_P\theta_P'PE_1)_{\rhopoutt{jj'}_{ X_{\Psift_{2}}\Psift_1\Psift_2X_{\Psift_1}'\theta_P\theta_P'PE_1\land\tildeOmegaPE}}\geq N_{\Psift_2,1}^{\tol}[1-\hbin(\ephtol')],
\end{equation}
where $\epssma\geq \sqrt{2\epsserfa}$, $N_{\Psift_2,1}^{\tol}$ is a lower bound on the number of single-photon events in set $\Psift_2$, with phase error tolerance
\begin{equation}
    \ephtol'=\ebitt{1,\tol}+g(N_{\Psift_1,1}^{\tol},N_{\Psift_2,1}^{\tol},\epsserfa),
\end{equation}
and $g(k,n,\varepsilon)=\sqrt{\frac{(n+k)(k+1)}{2nk^2}\ln\frac{1}{\varepsilon}}$ stemming from the modified Serfling bound estimating the single-photon bit error rate in the key generation rounds of $\Psift_2$ from that of $\Psift_1$~\cite{Tomamichel2012_BB84FiniteKey,Curty2014_DecoyMDIQKD}.
\end{proof}

\subsubsection{Shared Secrets Privacy}

The final security parameter to examine is the shared secrets privacy condition, which can be summarised as

\begin{theorem}
    Consider the QAKE protocol $\vars{P}$. Then, the shared secrets privacy security condition satisfies
    \begin{equation*}
        \epsSP=\varepsilon_{\SP,01,11}'+2\left(\epsds+\epsMACa+\epsSO+\epsMACb\right)
    \end{equation*}
    where
    \begin{gather*}
        \varepsilon_{\SP,01,11}'=4\sqrt{2\epsserfa}+2(\varepsilon_2+\varepsilon_3)+\varepsilon_{\SP,\MAC,1}+\varepsilon_{\SP,\MAC,2},\\
        \varepsilon_{\SP,\MAC,1}=\sqrt{(\abs{\vars{T}_{\AV}}\epsMACa-1)+2^{\log_2\left(\frac{2}{\varepsilon_3}+1\right)+\leakEC+2+\log_2\abs{\vars{T}_{\AV}}-N_{\Psift_2,1}^{\tol}[1-\hbin(\ephtol')]}}\\
        \varepsilon_{\SP,\MAC,2}=\sqrt{(\abs{\vars{T}_{\BV}}\epsMACb-1)+2^{\log_2\left(\frac{2}{\varepsilon_2}+1\right)+\leakEC+2+\log_2\abs{\vars{T}_{\AV}}\abs{\vars{T}_{\BV}}-N_{\Psift_2,1}^{\tol}[1-\hbin(\ephtol')]}}
    \end{gather*}
    is the shared secrets privacy parameter for idealised protocol $\vars{P}'$ for cases where $\FB=1$. 
\end{theorem}
\begin{proof}
We begin by performing the swap to an ``idealised" version of the protocol $\vars{P}'$ following Thm.~\ref{thm:QAKE_ideal_protocol_replace}, 
with replacement of $\FB$, $D_{\BV}$ and $\DPE$, which incurs a penalty of $2(\epsds+\epsMACa+\epsMACb+\epsSO)$.
We then analyse the components of the shared secrets privacy trace distance separately, based on the $(\tildeFA,\tildeFB)$ values.
For the $( \tildeFA,\tildeFB)=(\phi,0)$ case, the switch to ``idealised" checks means none of the secrets are utilised and thus are not leaked, as presented in Thm.~\ref{thm:QAKESecProofSSphi0}.
For the case of $(\tildeFA,\tildeFB)\in\{0\phi,00\}$, which is analysed together in Thm.~\ref{thm:QAKESecProofSSFA0}, the main leakage concern is that from sending $T_{\AV}$.
There, we provide a clear reason for the necessity of the channels $\bigepsilon_{K_1^h\KOTP_{1,i}E_i\rightarrow E_{i+1}}$ -- to quantify the leakage of information when $T_{\AV}$ is sent, but where $\KOTP_{1,i}$ is still required for future sessions.
Finally, the case of $(\tildeFA,\tildeFB)\in\{01,11\}$ can be analysed in a similar way to the key secrecy condition.
The privacy of secrets $\Kh_1\KOTP_{1,i}K_2$ is maintained via the strong extractor property and QLHL~\cite{Fehr2017_SID,Tomamichel2011_QLHL}, where the adversary's uncertainty of $\hat{X}_{\Psift_2,\rr}$ provides masking for the keys of the 2-universal hash function, as shown in Thm.~\ref{thm:QAKESecProofSSFB1}.
Combining all the results, we get the statement presented in the theorem.
\end{proof}

We analyse the condition with the ideal output state from Eqn.~\ref{eqn:QAKE_ideal_output_state_Full}, and we recall that an ideal state in $\vars{S}_{in}$ has the form
\begin{equation}
\begin{split}
    \rhoinn{jj'}_{LSE}=\dyad{jj'}_{\alpha\alpha'}\otimes\tau_{K_2R\KOTP_{1,j}\cdots\KOTP_{1,m}}\otimes\bigepsilon_{K_1^h\KOTP_{1,j-1}E_{j-1}\rightarrow E}\circ\cdots\circ\bigepsilon_{K_1^h\KOTP_{1,j'}E_{j'}\rightarrow E_{j'+1}}(\tau_{\Kh_1\KOTP_{1,j'}\cdots \KOTP_{1,j-1}}\otimes\rho_{E_{j'}}).
\end{split}
\end{equation}
We begin with the case of $(\tildeFA,\tildeFB)=(\phi,0)$,
\begin{theorem}
    \label{thm:QAKESecProofSSphi0}
    Consider the QAKE protocol $\vars{P}'$. Then, there exists $\sigma_{LSE}^{\phi0,jj'}\in\vars{S}_{in}$ such that
    \begin{equation*}
        \Delta\left(p_{\phi0}^{\reall}\rhopoutt{jj'}_{LSE|(\tildeFA,\tildeFB)=\phi0},\sum_{\tilde{j}'\geq j'}p_{\phi0,\tilde{j}}\sigma^{\phi0,j,\tilde{j}'+1}_{LSE}\right)=0.
    \end{equation*}
\end{theorem}
\begin{proof}
The goal here and in subsequent theorems is to prove that secrets that should not be leaked remains private.
The trace distance can be expressed as
\begin{equation}
    \Delta\left(p_{\phi0}^{\reall}\rho^{'\text{out},jj'}_{LSE|(\tildeFA,\tildeFB)=\phi0},\sum_{\tilde{j}'\geq j'}p_{\phi0,\tilde{j}}\sigma^{\phi0,j,\tilde{j}'+1}_{LSE}\right)\leq\sum_{\tilde{j}'\geq j'}p_{\phi0,\tilde{j}} \Delta\left(\rho^{'\text{out},jj'}_{LSE|(\tildeFA,\tildeFB)=\phi0,\beta'=\tilde{j}'},\sigma^{\phi0,j,\tilde{j}'+1}_{LSE}\right),
\end{equation}
for $\tilde{j}'\geq j'$, noting that the choice of $\beta'<j'$ is not possible (by definition of $\beta'$). 
We can thus analyse the trace distance separately for each $\tilde{j}'$ selection.\\

Here, we focus on an arbitrary $\tilde{j}'$ for analysis.
We first note that with the change to idealised checks, the protocol does not utilise any of the shared secrets when $\tildeFA=\phi$ and thus would not leak them.
As such, the only operation on the secrets is the tracing away of $\KOTP_{1,i}$ for $i\leq\min\{j,\tilde{j}'+1\}$, which will no longer be used.
For input state with
$\rhoinn{jj'}_{LSE}\in\vars{S}_{in}$, tracing away $\KOTP_{1,i}$ would lead to removal of its corresponding channel $\bigepsilon_{K_1^h\KOTP_{1,i}E_i\rightarrow E_{i+1}}$ for $i\in[\min\{j,j'\},\in\{j,\tilde{j}'+1\}]$, resulting in a state of the form $\rhoinn{j,\tilde{j}'+1}_{LSE}$ above. 
Therefore, the output state is ideal.
\end{proof}

We can consider the cases of $(\tildeFA,\tildeFB)=(0,\phi)$ and $\tildeFA=\tildeFB=0$ cases together, since they share a common event $\tildeFA=0$.
\begin{theorem}
\label{thm:QAKESecProofSSFA0}
    Consider the QAKE protocol $\vars{P}'$. Then, there exists $\sigma_{LSE}^{0\phi,jj'},\sigma_{LSE}^{00,jj'}\in\vars{S}_{in}$ such that
    \begin{equation*}
    \begin{split}
        &\Delta\left(p_{0\phi}^{\reall}\rhopoutt{jj'}_{LSE|(\tildeFA,\tildeFB)=0\phi}+p_{00}^{\reall}\rhopoutt{jj'}_{LSE|(\tildeFA,\tildeFB)=00}, \sum_{\tilde{j}\geq j}p_{0\phi,\tilde{j}}\sigma^{0\phi,\tilde{j}+1,j'}_{LSE}+\sum_{\tilde{j}\geq j,\tilde{j}'\geq j'}p_{00,\tilde{j}\tilde{j}'}\sigma^{00,\tilde{j}+1,\tilde{j}'+1}_{LSE}\right)=0.
    \end{split}
    \end{equation*}
\end{theorem}
\begin{proof}
We can split the
trace distance according to the $\tilde{j}$ and $\tilde{j}'$ values, 
\begin{equation}
    \sum_{\tilde{j}\geq j}p_{0\phi,\tilde{j}}\Delta\left(\rho^{'\text{out},jj'}_{LSE|(\tildeFA,\tildeFB)=0\phi,\beta=\tilde{j}}, \sigma^{0\phi,\tilde{j}+1,j'}_{LSE}\right)+\sum_{\tilde{j}\geq j,\tilde{j}'\geq j'}p_{00,\tilde{j}\tilde{j}'}\Delta\left(\rho^{'\text{out},jj'}_{LSE|(\tildeFA,\tildeFB)=00,(\beta,\beta')=\tilde{j}\tilde{j}'},\sigma^{00,\tilde{j}+1,\tilde{j}'+1}_{LSE}\right),
\end{equation}
and evaluate the individual terms.
We note that when $\tildeFB\neq 1$,
$K_2$ and $R$ remain private since they are not used.\\

We begin with studying the terms of the $(\tildeFA,\tildeFB)=(0,\phi)$ case, which we have to split further into two sub-cases: (1) $\tilde{j}+1\leq j'$ and (2) $\tilde{j}+1>j'$.
Since Alice sends tag $T_{\AV}$ irrespective of $\tildeFA$, some information of $(K_1^h,\KOTP_{1,\tilde{j}})$ can be leaked, though how it manifests in the ideal state depends on the sub-case.
In sub-case (1), we note that $\KOTP_{1,i}$ for any $i\leq\tilde{j}$ will not be utilised in the future and can be traced out.
This includes tracing away $\KOTP_{1,\tilde{j}}$, which by Thm.~\ref{thm:OTPWegCarSec}, guarantees that $K_1^h$ remains private, which leads to an output state of the form $\rhoinn{\tilde{j}+1,j'}_{LSE}$.
In sub-case (2), $\KOTP_{1,\tilde{j}}$ may be utilised in future sessions.
Since the only information of the secrets learnt in the protocol session is via $T_{\AV}$, we can consider a worse case where the adversary has a single access to an oracle where it can select any message $M$ and obtain a corresponding tag, i.e. $\bigepsilon_{K_1^hK_{1,\tilde{j}}E_{\tilde{j}}\rightarrow E}$.
This leads to an additional channel in the state, thereby matching the form of $\rhoinn{\tilde{j}+1,j'}_{LSE}$ when $\tilde{j}+1>j'$.
Combining both arguments and selecting the appropriate $\sigma_{LSE}^{0\phi,\tilde{j}+1,j}$ states, the trace distance contribution when $(\tildeFA,\tildeFB)=(0,\phi)$ is 0.\\

In the $(\tildeFA,\tildeFB)=(0,0)$ case, $T_{\AV}$ announcement remains the main concern.
Similar to the previous case, when $\tilde{j}+1>\tilde{j}'+1$, the leakage of $(\Kh_1,\KOTP_{1,\tilde{j}})$ results in an additional channel $\bigepsilon_{K_1^hK_{1,\tilde{j}}E_{\tilde{j}}\rightarrow E}$ in the output state, while for $\tilde{j}+1\leq\tilde{j}'+1$, the tracing out of $\KOTP_{1,\tilde{j}}$ ensures that $\Kh_1$ remains private.
In addition, $\KOTP_{1,i}$ for $i\leq\min\{\tilde{j},\tilde{j}'\}$ are traced out.
The resulting output state would therefore match $\rhoinn{\tilde{j}+1,\tilde{j}'+1}_{LSE}$, yielding a trace distance of 0.

\end{proof}

Here, we analyse the final two cases of $(\tildeFA,\tildeFB)\in\{01,11\}$ jointly since they are similar.
\begin{theorem}
\label{thm:QAKESecProofSSFB1}
    Consider the QAKE protocol $\vars{P}'$. Then, there exists $\sigma^{01,jj'}_{LSE},\sigma^{11,jj'}_{LSE}\in\vars{S}_{in}$ such that
    \begin{equation*}
    \begin{split}
        \Delta\left(p_{01}^{\reall}\rhopoutt{jj'}_{LSE|(\tildeFA,\tildeFB)=01}+p_{11}^{\reall}\rhopoutt{jj'}_{LSE|(\tildeFA,\tildeFB)=11}, \sum_{\tilde{j}\geq \max\{j,j'\}}p_{01,\tilde{j}}\sigma^{01,\tilde{j}+1,\tilde{j}}_{LSE}+\sum_{\tilde{j}\geq \max\{j,j'\}}p_{11,\tilde{j}}\sigma^{11,\tilde{j},\tilde{j}}_{LSE}\right)\leq\varepsilon_{\SP,01,11}'
    \end{split}
    \end{equation*}
    where
    \begin{equation*}
    \begin{split}
        \varepsilon_{\SP,01,11}'=&4\sqrt{2\epsserfa}+2(\varepsilon_2+\varepsilon_3)+\sqrt{(\abs{\vars{T}_{\BV}}\epsMACb-1)+2^{\log_2\left(\frac{2}{\varepsilon_2}+1\right)+\leakEC+2+\log_2\abs{\vars{T}_{\AV}}\abs{\vars{T}_{\BV}}-N_{\Psift_2,1}^{\tol}[1-\hbin(\ephtol')]}}\\
        &+\sqrt{(\abs{\vars{T}_{\AV}}\epsMACa-1)+2^{\log_2\left(\frac{2}{\varepsilon_3}+1\right)+\leakEC+2+\log_2\abs{\vars{T}_{\AV}}-N_{\Psift_2,1}^{\tol}[1-\hbin(\ephtol')]}},
    \end{split}
    \end{equation*}
    where $\varepsilon_1$ and $\varepsilon_2$ are parameters to optimise over. If the hash functions are 2-universal, then
    \begin{gather*}
        \varepsilon_{\SP,01,11}'=4\sqrt{2\epsserfa}+\frac{1}{2}\times 2^{-\frac{1}{2}[N_{P_2,1}^{\tol}-N_{P_2,1}^{\tol}\hbin(\ephtol')-\leakEC-2-\log_2\abs{\vars{T}_{\AV}}\abs{\vars{T}_{\BV}}]}\left(1+2^{-\frac{1}{2}\log_2\abs{\vars{T}_{\AV}}}\right).
    \end{gather*}
\end{theorem}
\begin{proof}
We can explicitly express the form of the output state and selected $\sigma_{LSE}^{01,jj'},\sigma_{LSE}^{11,jj'}$ states,
\begin{equation}
\begin{split}
    p_{01}^{\reall}\rhopoutt{jj'}_{LSE|(\tildeFA,\tildeFB)=01}=&\sum_{\tilde{j}\geq j_{max}}\dyad{\tilde{j}+1,\tilde{j}}_{\alpha\alpha'}\otimes\rhopoutt{jj'}_{\Kh_1K_2R\KOTP_{1,[\tilde{j},m]}E\land \tildeFB=1\land \tildeFA=0\land\beta=\tilde{j}}\\
    p_{11}^{\reall}\rhopoutt{jj'}_{LSE|(\tildeFA,\tildeFB)=11}=&\sum_{\tilde{j}\geq j_{\max}} \dyad{\tilde{j},\tilde{j}}_{\alpha\alpha'}\otimes\rhopoutt{jj'}_{\Kh_1K_2R\KOTP_{1,[\tilde{j},m]}E\land \tildeFB=1\land \tildeFA=1\land\beta=\tilde{j}}\\
    \sum_{\tilde{j}\geq \max\{j,j'\}}p_{01,\tilde{j}}\sigma^{01,\tilde{j}+1,\tilde{j}}_{LSE}=&\sum_{\tilde{j}\geq j_{max}}\dyad{\tilde{j}+1,\tilde{j}}_{\alpha\alpha'}\otimes\tau_{\Kh_1K_2R\KOTP_{1,[\tilde{j},m]}}\otimes\rhopoutt{jj'}_{E\land \tildeFB=1\land \tildeFA=0\land\beta=\tilde{j}}\\
    \sum_{\tilde{j}\geq \max\{j,j'\}}p_{11,\tilde{j}}\sigma^{11,\tilde{j},\tilde{j}}_{LSE}=&\sum_{\tilde{j}\geq j_{max}} \dyad{\tilde{j},\tilde{j}}_{\alpha\alpha'}\otimes\tau_{\Kh_1K_2R\KOTP_{1,[\tilde{j},m]}}\otimes\rhopoutt{jj'}_{E\land \tildeFB=1\land \tildeFA=1\land\beta=\tilde{j}}.
\end{split}
\end{equation}
We note that $\tildeFA$ and $\tildeFB$ is implicitly stored in $E$, and can be utilised to compute $\alpha\alpha'$.
As such, we can express the overall trace distance as $\Delta(\rhopoutt{jj'}_{\Kh_1K_2\KOTP_{1,\tilde{j}}E\land \tildeFB=1},\tau_{\Kh_1K_2\KOTP_{1,\tilde{j}}}\otimes \rhopoutt{jj'}_{E\land \tildeFB=1})$, where $R$ and $\KOTP_{1,[\tilde{j}+1,m]}$ are removed since they are unused in this protocol session, and thus remain private.\\

We simplify the trace distance by first introducing an intermediate state with only $K_2$ being private, 
\begin{equation}
\begin{split}
    &\Delta(\rhopoutt{jj'}_{\Kh_1K_2\KOTP_{1,\tilde{j}}E\land \tildeFB=1},\tau_{\Kh_1K_2\KOTP_{1,\tilde{j}}}\otimes \rhopoutt{jj'}_{E\land \tildeFB=1})\\
    \leq&\Delta(\rhopoutt{jj'}_{\Kh_1K_2\KOTP_{1,\tilde{j}}E\land \tildeFB=1},\tau_{K_2}\otimes \rhopoutt{jj'}_{\Kh_1\KOTP_{1,\tilde{j}}E\land \tildeFB=1})+\Delta(\rhopoutt{jj'}_{\Kh_1\KOTP_{1,\tilde{j}}E\land \tildeFB=1},\tau_{\Kh_1\KOTP_{1,\tilde{j}}}\otimes \rhopoutt{jj'}_{E\land \tildeFB=1}).
\end{split}
\end{equation}
To simplify the trace distance, we first note that $h_1(\Kh_1,X)\oplus \KOTP_1$ is an $\epsMACa$-almost strong 2-universal hash function with keys $(\Kh_1,\KOTP_1)$.
Therefore, we can apply Thm.~\ref{thm:Modified_QLHL},
\begin{equation}
\begin{split}
    &\Delta(\rhopoutt{jj'}_{\Kh_1K_2\KOTP_{1,\tilde{j}}E\land \tildeFB=1},\tau_{\Kh_1K_2\KOTP_{1,\tilde{j}}}\otimes \rhopoutt{jj'}_{E\land \tildeFB=1})\leq 2(\epssmb+\epssmc+\varepsilon_2+\varepsilon_3)\\
    &+\sqrt{(\abs{\vars{T}_{\BV}}\epsMACb-1)+2^{\log_2\abs{\vars{T}_{\BV}}+\log_2\left(\frac{2}{\varepsilon_2}+1\right)-\Hmin^{\epssmb}(\hat{X}_{\Psift_2,\rr}| \Kh_1\KOTP_{1,\beta}E)_{\land\tildeOmegaPE\land \tildeFB=1}}}\\
    &+\sqrt{(\abs{\vars{T}_{\AV}}\epsMACa-1)+2^{\log_2\abs{\vars{T}_{\AV}}+\log_2\left(\frac{2}{\varepsilon_3}+1\right)-\Hmin^{\epssmc}(\hat{X}_{\Psift_2,\rr}| E)_{\land\tildeOmegaPE\land \tildeFB=1}}}
\end{split}
\end{equation}
where we label $\land\tildeOmegaPE$ on smooth min-entropy as evaluation on the corresponding state when event $\tildeOmegaPE$ occurs.
We can now follow a similar analysis as in the proof of Thm.~\ref{thm:QAKESecProofMITMMain}, with the main difference being the lack of conditioning on some secrets.
For the first smooth min-entropy with smoothing parameter $\epssmb$, the conditioning on $K_2$ is not present.
Therefore, $T_{\BV}$, as the output of a $\epsMACb$-almost strongly 2-universal hash function, is uniform and independent of the input, importantly $X_{\Psift_2}$.
As such, it can be removed without using the min-entropy chain rule, i.e. without incurring a $\log_2\abs{\vars{T}_{\BV}}$ penalty, resulting in
\begin{equation}
    \Hmin^{\epssmb}(\hat{X}_{\Psift_2,\rr}| \Kh_1\KOTP_{1,\beta}E)_{\rhopoutt{jj'}_{\hat{X}_{\Psift_2,\rr}\Kh_1\KOTP_{1,\beta}E\land\tildeOmegaPE\land \tildeFB=1}}\geq N_{\Psift_2,1}^{\tol}-N_{\Psift_2,1}^{\tol}\hbin(\ephtol')-\leakEC-\log_2\abs{\vars{T}_{\AV}}-2.
\end{equation}
In the second term, both the conditioning on $K_2$ and $\Kh_1\KOTP_{1,\beta}$ are absent.
Similarly, we can argue that $T_{\AV}$ and $T_{\BV}$ are uniform and independent of $X_{P_2}$ since the respective keys of the almost 2-universal hash functions are not part of the conditioning.
As such, we have
\begin{equation}
    \Hmin^{\epssmc}(\hat{X}_{\Psift_2,\rr}| E)_{\rhopoutt{jj'}_{\hat{X}_{\Psift_2,\rr}E\land\tildeOmegaPE\land \tildeFB=1}}\geq N_{\Psift_2,1}^{\tol}-N_{\Psift_2,1}^{\tol}\hbin(\ephtol')-\leakEC-2.
\end{equation}
Combining the results, the trace distance is bounded by
\begin{equation}
\begin{split}
    &\Delta(\rhopoutt{jj'}_{\Kh_1K_2\KOTP_{1,\tilde{j}}E\land \FB=1},\tau_{\Kh_1K_2\KOTP_{1,\tilde{j}}}\otimes \rhopoutt{jj'}_{E\land \FB=1})\leq 4\sqrt{2\epsserfa}+2(\varepsilon_2+\varepsilon_3)\\
    &+\sqrt{(\abs{\vars{T}_{\BV}}\epsMACb-1)+2^{\log_2\left(\frac{2}{\varepsilon_2}+1\right)+\leakEC+2+\log_2\abs{\vars{T}_{\AV}}\abs{\vars{T}_{\BV}}-N_{\Psift_2,1}^{\tol}[1-\hbin(\ephtol')]}}\\
    &+\sqrt{(\abs{\vars{T}_{\AV}}\epsMACa-1)+2^{\log_2\left(\frac{2}{\varepsilon_3}+1\right)+\leakEC+2+\log_2\abs{\vars{T}_{\AV}}-N_{\Psift_2,1}^{\tol}[1-\hbin(\ephtol')]}}.
\end{split}
\end{equation}
We can repeat the analysis with two-universal hash function using QLHL to obtain the second set of bounds in the theorem.
\end{proof}

\section{QAKE with Pseudorandom Basis Choice}
\label{app:PRNG_QAKE}
The QAKE protocol, similar to many QKD protocols, suffers from key rate loss due to sifting.
One method to avoid sifting is to pre-share basis information and having Bob measure in the same basis as Alice.
However, this introduces two issues which renders the protocol inefficient: (1) Alice and Bob needs to pre-share a basis $\theta$ that is large to account for signal loss and finite-size effects (we require on the order of \SI{1e10}{} signals in our experiment), (2) the pre-shared basis has to  be refreshed every session via key updating due to leakage from multi-photon events.
The most practical solution to the problem is to use a shorter pre-shared master key and generate a larger key from a quantum-safe \textit{pseudorandom number generator} (PRNG) like AES~\cite{NIST2001_AES}.
We note that even without pre-sharing the basis, using PRNG still confers the advantage of reducing the length of the random string required to generate the basis.
We note that there are similar proposals in QKD where the basis information is generated or encrypted using a PRNG~\cite{Trushechkin2018_PRNGQKD,Price2021_DDOSPRNGEncryptBasis}.
\subsection{Protocol}

We begin with two parties, Alice and Bob, sharing secrets $S$: (1) shared basis seed $\tilde{\theta}=\{\tilde{\theta}_j\}_{j=1,\cdots,m}$, where $\tilde{\theta}_i\in\{0,1\}^{l_{\tilde{\theta}}}$, (2) authentication keys $\Kh_1$ and $K_2$, (3) privacy amplification seed $R$, and (4) authentication masking key $\KOTP_1=\{\KOTP_{1,j}\}_{j=1,\cdots,m}$.
We note that the authentication keys and privacy amplification seed always remain secret and unchanged.
Alice and Bob each has a label, $\alpha$ and $\alpha'$ respectively, that notes the secrets to utilise for the session (i.e. which $\tilde{\theta}_i$ and $K_{1,i}^{OTP}$ to utilise), thereby allowing secrets to be replaced (by never using them again) when authentication fails.
They also publicly share a $\epsMACa$-almost XOR 2-universal hash function $h_1$ (authentication key $K_1^h)$, a $\epsMACb$-almost strongly 2-universal hash function $h_2$ (authentication key $K_2$) to generate the authentication tags, and a 2-universal hash function $\hPA$ for privacy amplification to generate secret keys.
Alice and Bob also agrees on an error correction protocol, with syndrome generation function $\fsyn$, error correction function $\fsyndec$ and an error correction efficiency of $\fEC$.
They also agree on the use of a $(t',\varepsilon_{\PRNG})$-quantum secure PRNG, $G^{\PRNG}$.
We note here that the subscript $\rr$ on any random variables represent the received variable, e.g. $\alpha_{\rr}$ is received by Bob when $\alpha$ is sent by Alice, since the adversary could alter the messages sent through the classical communication channel.
The protocol schematic is shown in Fig.~\ref{fig:RoundEfficientAKE}, and is detailed below.\\ 

We note here that we choose the PRNG with $t'\geq 2(t+t_{MG})+t_{ex}+t_{SG,PR}$, for an adversary assumed to be limited by resource $t$ during the protocol (e.g. limited run time).
The resource $t_{MG}$ and $t_{ex}$ relate to the security analysis of the protocol with PRNG, used to prove that the gap between bit and phase error is small in Appendix~\ref{app:EUR_with_PRNG}.
The resource $t_{SG,PR}$ relates to the limited ability of the adversary to guess the basis seed during before it is announced, as shown in Appendix~\ref{app:Prelim}.
We note additionally that the PRNG master key is much smaller than the basis string length, importantly $2\abs{\tilde{\theta}_{\beta}}<n_{rnd}$.

\begin{protocol}{Quantum Authenticated Key Exchange with Pseudorandom Basis Selection}
\textit{Goal.} Alice and Bob authenticates one another, and performs key exchange.
\begin{enumerate}
    \item \textbf{Label Agreement}: Alice and Bob exchange $\alpha$ and $\alpha'$. Alice (resp. Bob) sends $\alpha$ (resp. $\alpha'$) and receives $\alpha_\rr'$ (resp. $\alpha_\rr$), which results in an label choice $\beta=\max\{\alpha,\alpha_\rr'\}$ (resp. $\beta'=\max\{\alpha_\rr,\alpha'\}$).
    \item \textbf{Alice State Preparation}: Alice generates a $n_{rnd}$-bit basis string with a PRNG using the basis seed, $\theta_{\beta}=G^{\PRNG}(\tilde{\theta}_{\beta})$, randomly chooses a $n_{rnd}$-bit string $X$ and a $n_{rnd}$-trit string $v\in\{0,1,2\}^{n_{rnd}}$ according to probability distribution $p_v$. She then prepares $n_{rnd}$ quantum systems $Q_i$ in phase-randomised coherent BB84 states $\left\{\rho_{Q_i}^{\theta_i,X_i,\mu_{v_i}}\right\}_{i\in[1,n_{rnd}]}$, with basis $\theta_i$, bit value $X_i$, and intensity $\mu_{v_i}$.
    \item \textbf{Bob Measurement}: Alice sends $Q^{n_{rnd}}$ to Bob, who measures subsystems $Q_i$ using basis $\theta_{\beta'}'=G^{\PRNG}(\tilde{\theta}_{\beta'}')$, and records outcome $X_i'$. If Bob detects no clicks, he declares $X_i'=\perp$. If Bob detects multiple clicks, he randomly selects $X_i'\in\{0,1\}$.
    \item \textbf{Test Round Announcement}: Bob records the detection rounds, $P=\{i:X_i'\neq\perp\}$, and randomly splits it into $P_1$ and $P_2$, with $\abs{P_1}=\lceil f_{P_1}\abs{P}\rceil$, where $f_{P_1}$ is some pre-determined fraction of rounds for parameter estimation. Bob announces $P_1$, $P_2$, $\tilde{\theta}_{\beta'}'$ and $X_{P_1}'$. 
    \item \textbf{Parameter Estimation}: Alice checks if: (1) the number of test rounds matches the expected number $\abs{P_{1,\rr}}=\lceil f_{P_1}\abs{P}\rceil$, (2) the basis seed matches, $\tilde{\theta}_{\beta',\rr}'=\tilde{\theta}_{\beta}$, (3) the lower bound on single-photon events in sets $P_{2,\rr}$ and $P_{1,\rr}$ (estimated via decoy state analysis) exceeds a threshold, $\hat{N}^{\LB}_{P_{2,\rr},1}\geq N_{P_{2,\rr},1}^{\tol}$, $\hat{N}^{\LB}_{P_{1,\rr},1}\geq N_{P_{1,\rr},1}^{\tol}$, (4) the upper bound on single-photon bit-error rate in set $P_{1,\rr}$ (estimated via decoy-state analysis) falls below a threshold, $\hatebitt{P_{1,\rr},1}^{\UB}\leq \ebitt{1,\tol}$, and (5) the upper bound on bit error rate in set $P_{2,\rr}$ (estimated via Serfling bound) falls below a threshold, $\hatebitt{P_{2,\rr}}^{\UB}\leq \ebitt{\tol}$. If all checks pass, Alice sets $\DPE=1$, otherwise she sets $\DPE=0$.
    \item \textbf{Error Correction}: If $\DPE=1$, Alice generates a syndrome $syn=\fsyn(X_{P_{2,\rr}})$ of length $l_{syn}$ and forwards it to Bob, otherwise, Alice sends a random string of length $l_{syn}$ to Bob. Bob receives the syndrome $syn_\rr$ and computes the corrected bit string $\hat{X}_{P_2,\rr}=\fsyndec(X_{P_2}',syn_\rr)$.
    \item \textbf{Alice Validation}: If $\DPE=1$, Alice generates a tag $T_{\AV}=h_1(\Kh_1,X_{P_1,\rr}'||P_{1,\rr}||P_{2,\rr}||X_{P_{2,\rr}}||syn)\oplus \KOTP_{1,\beta}$ and forwards it to Bob, otherwise, Alice sends a random string of the same length as the tag. Bob receives tag $T_{\AV,\rr}$ and generates verification tag $\tilde{T}_{\AV}=h_1(\Kh_1,X_{P_1}'||P_1||P_2||\hat{X}_{P_{2,\rr}}||syn_\rr)\oplus \KOTP_{1,\beta'}$ and checks if $T_{\AV,\rr}=\tilde{T}_{\AV}$. If the tags matches, Bob validates Alice and output $\FB=1$, otherwise, he sets $\FB=0$.
    \item \textbf{Bob Validation}: If $\FB=1$, Bob generates a tag, $T_{\BV}=h_2(K_2,\hat{X}_{P_{2,\rr}})$, and sends $T_{\BV}$ to Alice. If $\FB=0$, Bob sends a random string of the same length as the tag instead. Alice receives the tag $T_{\BV,\rr}$, and computes the verification tag, $\tilde{T}_{\BV}=h_2(K_2,X_{P_{2,\rr}})$. If $T_{\BV,\rr}=\tilde{T}_{\BV}$, Alice validates Bob, $D_{\BV}=1$, otherwise, she sets $D_{\BV}=0$.
    \item \textbf{Secret Key Generation and Label Update}: Alice decides whether to accept the session based on her parameter estimation and validation of Bob, i.e. $\FA=\DPE\land D_{\BV}$. If Alice (resp. Bob) decides to perform key generation, $\FA=1$ (resp. $\FB=1$), she (resp. he) performs privacy amplification $\KA||\tilde{\theta}_{\beta}=\hPA(R,X_{P_{2,\rr}})$ (resp. $\KB||\tilde{\theta}_{\beta'}'=\hPA(R,\hat{X}_{P_{2,\rr}})$) and sets $\alpha=\beta$ (resp. $\alpha'=\beta'$), where the basis seed $\tilde{\theta}$ is updated and $\KA\KB$ are the cryptographically secure keys that can be used for other purposes. If key generation is not performed, the labels are updated, i.e. if $\FA=0$, Alice updates her label $\alpha=\beta+1$ and if $\FB=0$, Bob updates his label $\alpha'=\beta'+1$.
\end{enumerate}
\end{protocol}

\subsection{Overall Protocol Security}
\label{app:PRNGAKEProtocolSec_OverallSec}

\begin{figure}[!ht]
    \centering
    \includegraphics[width=0.55\textwidth]{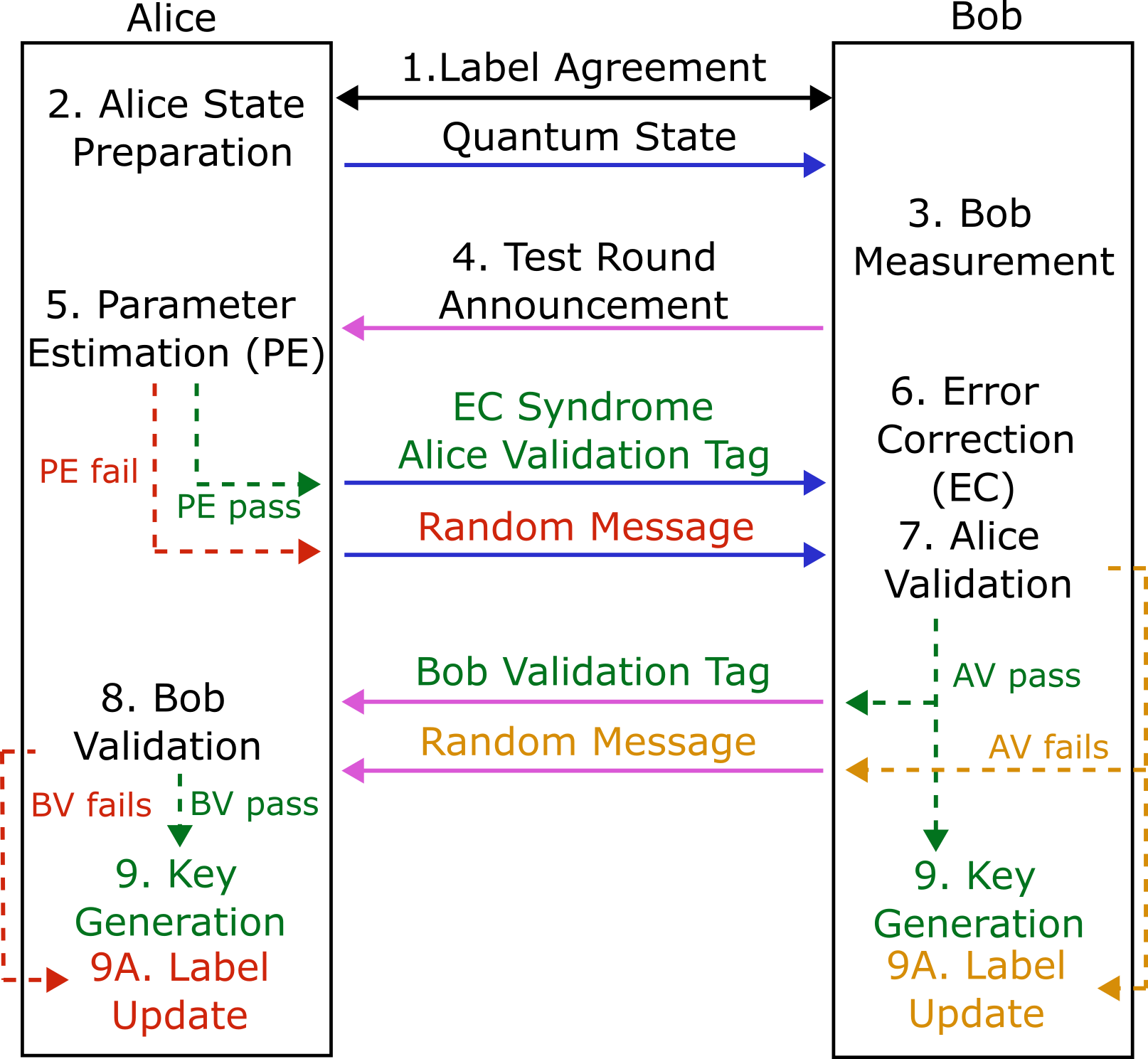}
    \caption{Summary of the QAKE with pseudorandom basis protocol proposed. The protocol begins with both parties establishing the label to use, followed by Alice sending quantum states to Bob. Bob announces the test round results, which Alice can then use for parameter estimation. If this passes, error correction and validation of Alice's identity would occur followed by validation of Bob's identity. Finally, the two parties independently decide whether to generate secret keys or update their label (signalling a protocol failure).}
    \label{fig:RoundEfficientAKE}
\end{figure}

As per the security analysis of the original QAKE protocol, the focus of the security analysis here would be for a single-session, for which we have to first define the set of ideal input and output states.
The labels remain $\alpha$ and $\alpha'$, with the shared secrets $S$ having additionally basis seed $\tilde{\theta}_i$.
The basis seed $\tilde{\theta}_i$ used to generate the shared basis can potentially be leaked after utilisation in a protocol session, so we let basis seeds for indices $i\leq\alpha_{\max}-1$ ($\alpha_{\max}=\max\{\alpha,\alpha'\}$) be leaked.
Authentication masking keys behave similarly, with keys corresponding to indices $i<\alpha_{\min}-1$ ($\alpha_{\min}=\min\{\alpha,\alpha'\}$) no longer being used.
Unlike the original QAKE protocol, authentication key $\Kh_1$ will always remain secure, alongside authentication key $K_2$ and privacy amplification seed $R$.
This difference stems from Bob's reply of the basis seed $\tilde{\theta}_{\beta'}'$, which serves as a method for Alice to authenticate Bob, preventing $T_{\AV}$ from being announced when this check fails.
As such, we define the ideal input state $\rhoin=\sum_{jj'}p_{jj'}\rhoinn{jj'}_{LSE}$, where
\begin{equation}
\begin{split}
    \rhoinn{jj'}_{LSE}=&\dyad{jj'}_{\alpha\alpha'}\otimes\tau_{\Kh_1K_2R\KOTP_{1,j_{\min}}\cdots \KOTP_{1,m}}\otimes\tilde{\tau}_{\tilde{\theta}_{j_{\max}}\cdots\tilde{\theta}_{m}\tilde{\theta}_{j_{\max}}'\cdots\tilde{\theta}_{m}'}\otimes\rho_{E},
\end{split}
\end{equation}
where $L=\alpha\alpha'$ and $S$ refer to the secrets.\\

The ideal output state in the intermediate sessions contains the secret key variables $\KA\KB$, the decision labels $\FA\FB$ which also doubles as authentication success/failure indicators, and $LSE$.
Following the argument in Sec.~\ref{app:AKEProtocolSec_OverallSec}, we can analyse the trace distance corresponding to each input state separately.
The QAKE security conditions then dictate that the ideal intermediate output state is of the form
\begin{equation}
\begin{split}
    \rhoidealintt{jj'}=&\dyad{\perp\perp}_{\KA\KB}\otimes\left[\dyad{0\phi}_{\FA\FB}\otimes\sum_{\tilde{j}\geq j}p_{0\phi,\tilde{j}}\sigma^{0\phi,\tilde{j}+1,j'}_{LSE}\right.+\dyad{\phi0}_{\FA\FB}\otimes\sum_{\tilde{j}'\geq j'}p_{\phi 0,\tilde{j}'}\sigma^{\phi0,j,\tilde{j}'+1}_{LSE}\\
    &\left.+\dyad{00}_{\FA\FB}\otimes \sum_{\tilde{j}\geq j,\tilde{j}'\geq j'}p_{00,\tilde{j}\tilde{j}'}\sigma^{00,\tilde{j}+1,\tilde{j}'+1}_{LSE}\right]+\dyad{01,\perp}_{\FA\FB\KA}\otimes\tau_{\KB}\otimes\sum_{\tilde{j}\geq\max\{j,j'\}}p_{01,\tilde{j}}\sigma^{01,\tilde{j}+1,\tilde{j}}_{LSE}\\
    &+\dyad{11}_{\FA\FB}\otimes\tilde{\tau}_{\KA\KB}\otimes\sum_{\tilde{j}\geq\max\{j,j'\}}p_{11,\tilde{j}}\sigma^{11,\tilde{j}\tilde{j}}_{LSE},
\end{split}
\end{equation}
where $\sigma_{LSE}^{\fA\fB,jj'}\in\vars{S}_{in}$, matching Eqn.~\eqref{eqn:QAKE_ideal_output_state_Full}.
Conditions 1 and 3 in Thm.~\ref{thm:Multi_to_Single_Red} are trivially satisfied by observation, allowing us to focus on the single-session security for a single component, with overall security $\epssec$.\\

We can summarise the security of the protocol as
\begin{theorem}
\label{thm:PRNGQAKE_Main}
    Consider a QAKE protocol with PRNG basis choice $\vars{P}$, where the adversary is limited by resource $t$ and the PRNG selected is $(t',\varepsilon_{\PRNG})$-quantum-secure with $t'\geq 2(t+t_{MG})+t_{ex}$ and $t'\geq t+t_{SG,PR}$. The protocol is $\left(\frac{1}{\abs{\vars{T}_{\AV}}}+\frac{1}{\abs{\vars{T}_{\BV}}},\epsMACb,\epsKS'+\varepsilon_{\vars{P}'},\varepsilon_{SP,10,11}'+\varepsilon_{\hat{\theta},\PRNG}+\varepsilon_{\vars{P}'}\right)$-secure.
    The penalty associated with shifting to an idealised protocol $\vars{P}'$ is
    \begin{equation*}
        \varepsilon_{\vars{P}'}=2(\epsMACa+\varepsilon_{\hat{\theta},\PRNG}+\epsMACb+\epsds),
    \end{equation*}
    where $\epsds$ accounts for errors associated with decoy state analysis and $\varepsilon_{\hat{\theta},\PRNG}$ is the error associated with guessing the basis generation seed $\tilde{\theta}$.
    The key secrecy parameter for the idealised protocol is
    \begin{gather*}
        \epsKS'=4\epssm+2^{-\frac{1}{2}[N_{P_2,1}^{\tol}-N_{P_2,1}^{\tol}\hbin(\ephtol')-2-\log_2\abs{\vars{T}_{\AV}}\abs{\vars{T}_{\BV}}-\leakEC-l_{\KB}-l_{\tilde{\theta}}]},        
    \end{gather*}
    where $\ephtol'$ is the phase error tolerance inclusive of the correction due to Serfling bound (with error $\epsserfa$), $\abs{\vars{T}_{\AV}}$ and $\abs{\vars{T}_{\BV}}$ are the size of the respective tags and $\leakEC$ is the information leakage from error correction.
    The shared secrets privacy parameter for the idealised protocol when $\FB=1$ is
    \begin{gather*}
        \varepsilon_{\SP,10,11}'=10\epssm+4(\varepsilon_2+\varepsilon_3)+\varepsilon_{\SP,\PA}+\varepsilon_{\SP,\MAC,1}+\varepsilon_{\SP,\MAC,2},\\
        \varepsilon_{\SP,\PA}=2^{-\frac{1}{2}[N_{P_2,1}^{\tol}-N_{P_2,1}^{\tol}\hbin(\ephtol')-\leakEC-\log_2\abs{\vars{T}_{\AV}}\abs{\vars{T}_{\BV}}-l_{\KB}-l_{\tilde{\theta}}]}\\
        \varepsilon_{\SP,\MAC,1}=\sqrt{(\abs{\vars{T}_{\AV}}\epsMACa-1)+2^{\log_2\left(\frac{2}{\varepsilon_3}+1\right)+\leakEC+2+\log_2\abs{\vars{T}_{\AV}}-N_{P_2,1}^{\tol}[1-\hbin(\ephtol')]}}\\
        \varepsilon_{\SP,\MAC,2}=\sqrt{(\abs{\vars{T}_{\BV}}\epsMACb-1)+2^{\log_2\left(\frac{2}{\varepsilon_2}+1\right)+\leakEC+2+\log_2\abs{\vars{T}_{\AV}}\abs{\vars{T}_{\BV}}-N_{P_2,1}^{\tol}[1-\hbin(\ephtol')]}}
    \end{gather*}
    the associated smoothing parameter is
    \begin{equation*}
        \epssm=\sqrt{2(2\varepsilon_{\TRNG}+\varepsilon_{\PRNG}+\epsserfa)},
    \end{equation*}
    and where   $\varepsilon_{\TRNG}$ is the statistical error associated with estimating the bit-phase error gap with an ideal RNG, $N_{P_2,1}^{\tol}$ is the minimum tolerated number of single-photon rounds in set $P_2$, and $\varepsilon_2$ and $\varepsilon_3$ are parameters to optimise over.
\end{theorem}

\begin{proof}[Sketch Proof]
Thm.~\ref{thm:PRNGQAKESecProofEA} and Thm.~\ref{thm:AKESecProofKC} bounds the explicit entity authentication and match security respectively.
They follow closely the proofs of Thm.~\ref{thm:QAKESecProofEA} and Thm.~\ref{thm:QAKESecProofMS}, with the main difference being that no oracle access to generate $T_{\AV}$.\\

Similar to QAKE security, the authentication steps can be used to guarantee ideal features.
In Thm.~\ref{thm:PRNG_QAKE_ideal_protocol_replace}, we show that $\DPE$, $\FB$ and $D_{\BV}$ can be changed to their idealised versions, with some small penalty.
This reduces the analysis for key secrecy to a QKD analysis in Thm.~\ref{thm:AKESecProofMITMMain}, and the corresponding QKD scheme with pseudorandom basis choice is analysed in Appendix~\ref{app:PRNGBasisChoiceProof}.\\

Similar to the original QAKE, shared secrets privacy when keys are not generated result in no secrets being leaked or the secret leaked, namely the basis generation seed $\tilde{\theta}_{\beta}$, is no longer utilised due to labels being updated.
Notably, the use of basis generation seed $\tilde{\theta}_{\beta'}'$ as a response provide additional security since Alice do not reveal $T_{\AV}$ when the basis generation seed it receives does not match her own.
This results in an output state in $\vars{S}_{out}'$, as presented in Thm.~\ref{thm:AKESecProofSSphi0} and Thm.~\ref{thm:AKESecProofSSFA0}.
When keys are generated, the parameter estimation passing ensures $X_{P_2}$ has sufficient min-entropy such that $T_{\AV}$ and $T_{\BV}$ leak minimal information on the authentication and masking key while also refreshing the basis generation seed $\tilde{\theta}_{\beta}$.
This is presented in Thm.~\ref{thm:AKESecProofSSFB1}.
\end{proof}

We note that while the overall security parameters $\epssec$ and $\epssecint$ can be determined by a sum of the respective security parameters, the bound can be tightened due to overlapping considerations for different security conditions.
For instance, the secret key generation and label update hashes $X_{P_2,\rr}$ and outputs both the key $\KA$ and the updated basis string $\tilde{\theta}_{\beta}$.
As such, in the analysis of both key secrecy and shared secrets privacy, the same security consideration with penalty term relating to smooth min-entropy of $X_{P_2,\rr}$ is included.

\subsubsection{Explicit Entity Authentication}

The explicit entity authentication security condition looks at the probability of $( \FA,\FB)\in\{(1,\phi),(\phi,1),(1,0)\}$.
\begin{theorem}
\label{thm:PRNGQAKESecProofEA}
Consider the QAKE protocol $\vars{P}$ with adversary resource limit $t$ and a $(t',\varepsilon_{\PRNG})$-quantum-secure PRNG.
The full explicit entity authentication security parameter is
\begin{equation*}
    \epsEAf=\frac{\Pr[\FB\neq1]}{\abs{\vars{T}_{\BV}}}.
\end{equation*}
and the almost-full explicit entity authentication security parameter is
\begin{equation*}
    \epsEAaf=\frac{\Pr[\FA=\phi]}{\abs{\vars{T}_{\AV}}}.
\end{equation*}
\end{theorem}
\begin{proof}
We begin with the case of full explicit authentication, which is defined with $\Pr[\FA=1,\FB=\phi]+\Pr[\FA=1,\FB=0]$.
Using the same analysis as in Thm.~\ref{thm:QAKESecProofEA}, we can use the fact that $K_2$ remains private to show that $\epsEAf=\frac{\Pr[\FB\neq1]}{\abs{\vars{T}_{\BV}}}$.\\

For almost-full explicit authentication, Alice does not participate in the protocol, $\FA=\phi$.
We can follow the analysis of Thm.~\ref{thm:QAKESecProofEA}, but with some simplification.
To obtain a result of $\FB=1$, Bob needs to successfully validate Alice by checking if $\tilde{T}_{\AV}=T_{\AV,\rr}$, where $T_{\AV,\rr}$ has to be chosen by the adversary.
Suppose a worse case where a message $M$ can be known to the adversary.
Since $(\Kh_1,\KOTP_{1,\beta'})$ remains private from the adversary, by the uniformity property of a strong 2-universal hash function, we have that
\begin{equation}
    \Pr[\FB=1|\FA=\phi]\leq\Pr[h_1(\Kh_1,M)\oplus \KOTP_{1,\beta'}=T_{\AV,\rr}]=\frac{1}{\abs{\vars{T}_{\AV}}}.
\end{equation}
As such, $\Pr[\FA=\phi,\FB=1]=\frac{\Pr[\FA=\phi]}{\abs{\vars{T}_{\AV}}}$.
\end{proof}

\subsubsection{Match Security}

\begin{theorem}
\label{thm:AKESecProofKC}
Consider the QAKE protocol $\vars{P}$ with adversary resource limit $t$ and a $(t',\varepsilon_{\PRNG})$-quantum-secure PRNG.
The match security parameter is
\begin{equation*}
    \epsMS=\epsMACb.
\end{equation*}
\end{theorem}
\begin{proof}
Same as Thm.~\ref{thm:QAKESecProofMS}.
\end{proof}

\subsubsection{Idealised Parameter Estimation and Authentication}
\label{app:PRNG_QAKE_idealised_prot}

To simplify the analysis for some security conditions, we similarly introduce ``idealised" versions of the parameter estimation and authentication checks.
We first note that the steps are given by:
\begin{enumerate}
    \item Alice's state preparation and sending of quantum state.
    \item Bob's measurement and test round announcement, with reply $P_1$, $P_2$, $\tilde{\theta}_{\beta'}'$ and $X_{P_1}'$.
    \item Alice's parameter estimation and validation tag generation, sending syndrome $syn$ and tag $T_{\AV}$ to Bob.
    \item Bob performs Alice's validation, and and responds with his own validation tag $T_{\BV}$.
    \item Alice validates Bob's tag. 
\end{enumerate}
The replacement of the decoy state parameter estimation $\DPE$ to $\tildeDPE$ and Bob's validation $D_{\BV}$ to $\tilde{D}_{\BV}$ can be performed as described in Sec.~\ref{app:ideal_DPE_switch} and Sec.~\ref{app:ideal_version_switch} (note $\tilde{D}_{\BV}$ has $\Omega_{4\rightarrow 5}$ instead).
The replacement for Alice's validation $\FB$ to $\tildeFB$ is modified since the steps of the protocol are different, and the announcement of $\tilde{\theta}_{\beta}$ can provide additional authentication ability.
We define instead
\begin{equation}
    \tildeFB=\begin{cases}
        1 & \substack{\beta=\beta',\,\DPE=1,\,T_{\AV}=T_{\AV,\rr},\\
        M_{\AV}=M_{\AV}',\,\Omega_{2\rightarrow3\rightarrow4}}\\
        0 & otherwise
    \end{cases}.
\end{equation}

\begin{theorem}
\label{thm:PRNG_QAKE_ideal_protocol_replace}
    Let $\vars{P}$ be the protocol and $\vars{P}'$ be the idealised protocol (with $\DPE$, $\FB$, $D_{\BV}$ replaced with $\tildeDPE$, $\tildeFB$ and $\tilde{D}_{\BV}$). Then, for any $\sigma_{LSE}\in\vars{S}_{in}$,
    \begin{equation*}
        \Delta(\vars{P}(\sigma_{LSE}),\vars{P}'(\sigma_{LSE}))\leq 2(\epsMACa+\varepsilon_{\hat{\theta},PRNG})+2\epsMACb+2\epsds,
    \end{equation*}
    where $\varepsilon_{\hat{\theta},PRNG}$ is the probability of guessing the basis generation seed.
\end{theorem}

\begin{proof}
Since $(\Kh_1,\KOTP_{1,\beta'})$ is always private, we can follow the same argument in Thm.~\ref{thm:QAKE_ideal_protocol_replace} to show 
\begin{equation}
\begin{split}
    \Pr[\FB=1,\tildeFB=0]\leq&\epsMACa+\Pr[\FB=1,\DPE=1,\Omega_{3\rightarrow 4},\beta=\beta',M_{\AV}=M_{\AV}',T_{\AV}=T_{\AV,\rr},\Omega_{2\rightarrow 3}^c]\\
    \leq&\epsMACa+\Pr[\DPE=1,\beta=\beta',\Omega_{2\rightarrow 3}^c]\\
    \leq&\epsMACa+\pguess(\tilde{\theta}_{\beta}),
\end{split}
\end{equation}
where the first inequality uses the expansion in Eqn.~\eqref{eqn:QAKE_DAV_change} with the first three terms bounded by $\epsMACa$, while the remaining term relates to step ordering.
The third inequality simplifies the probability to a guessing probability since
$\tilde{\theta}_{\beta}$ has to be matched for $\DPE=1$, and this has to take place for $\beta=\beta'$ ($\tilde{\theta}_{\beta'}$ not leaked to adversary prior to the session) and before Bob announces $\tilde{\theta}_{\beta}$ in step 2.
The guessing probability of $\tilde{\theta}_{\beta}$ can be computed, since it is only utilised to generate the quantum state in the first session.
For a given choice of parameter with $p_{\geq 2}$ probability of multi-photon events which we assume in the worst case would provide $\theta_i$ to the adversary, then given a $(t',\varepsilon_{\PRNG})$-quantum secure PRNG with $t'\geq t+t_{SG,PR}$ and $2l_{\tilde{\theta}}<n_{rnd}$, by Thm.~\ref{thm:PRNGGuess}, the probability of guessing $\tilde{\theta}_{\beta}$ is bounded by $\varepsilon_{\hat{\theta},PRNG}=\varepsilon_{\PRNG}+2^{-n_{rnd}+l_{\tilde{\theta}}}$, where $n_{rnd}$ is the number of signals sent.
As such, these changes result in a similar penalty to the trace distance, with a penalty of $2(\epsMACa+\varepsilon_{\hat{\theta},PRNG})+2\epsMACb+2\epsds$.
\end{proof}

\subsubsection{Key Secrecy}

Key secrecy analyses the two specific cases where $\KB$ generates a key, $(\FA,\FB)\in\{01,11\}$. 
Instead of having the security relying on the secrecy condition of decoy-state BB84 like in the original QAKE protocol, it relies on the secrecy condition of decoy-state BB84 with shared pseudorandom basis in Appendix~\ref{app:BB84_PRNG_Secrecy}.
\begin{theorem}
\label{thm:AKESecProofMITMMain}
Consider the QAKE protocol $\vars{P}$ with adversary resource limit $t$ and a $(t',\varepsilon_{\PRNG})$-quantum-secure PRNG.
Then, the key secrecy security parameter is
\begin{equation}
    \epsKS=\epsKS'+2(\epsMACa+\varepsilon_{\hat{\theta},\PRNG})+2\epsMACb+2\epsds
\end{equation}
where
\begin{gather*}
    \epsKS'=4\epssm+2^{-\frac{1}{2}[N_{P_2,1}^{\tol}-N_{P_2,1}^{\tol}\hbin(\ephtol')-2-\log_2\abs{\tilde{T}_{\AV}}\abs{\tilde{T}_{\BV}}-\leakEC-l_{\KB}-l_{\tilde{\theta}}]}\\
    \epssm=\sqrt{2(2\varepsilon_{\TRNG}+\varepsilon_{\PRNG}+\epsserfa)}
\end{gather*}
is the key secrecy of the idealised protocol $\vars{P}'$,  $\varepsilon_{\hat{\theta},\PRNG}$ is the error associated with guessing the basis generation seed $\tilde{\theta}$, $\ephtol'$ is the error tolerance inclusive of the correction due to Serfling bound.
\end{theorem}
\begin{proof}
From Thm.~\ref{thm:PRNG_QAKE_ideal_protocol_replace}, we can swap to an idealised version of the protocol $\vars{P}'$ with penalty $2(\epsds+\epsMACa+\varepsilon_{\hat{\theta},\PRNG}+\epsMACb)$.
Following the argument of Thm.~\ref{thm:QAKEKeySecIdealisedReplace}, by reversing the final step of Alice, we can combine the cases of $(\tilde{F}_A,\tilde{F}_B)=(0,1),(1,1)$ into a single case, giving
\begin{equation}
    \Delta(\rhopoutt{jj'}_{\tildeFA\tildeFB\KB LSE\land(\tildeFA,\tildeFB)\in\{01,11\}},\tau_{\KB}\otimes\rhopoutt{jj'}_{\tildeFA\tildeFB LSE\land(\tildeFA,\tildeFB)\in\{01,11\}})\leq \Delta(\rhopoutt{jj'}_{\KB L'S'E\land \tildeFB=1},\tau_{\KB}\otimes\rhopoutt{jj'}_{L'S'E\land \tildeFB=1}).
\end{equation}
With the removal of $\tildeFA$, we can provide a key secrecy security parameter for $\vars{P}'$.\\

Let us first note that $\tilde{\theta}_{\beta}'$ is refreshed during key generation, and we label the original value before refresh as $\tildethetapbetain$ for clarity.
We follow a similar analysis as Thm.~\ref{thm:QAKEKeySecIdealisedProtocol} and we first define $\tildeOmegaPE$ to be the event where parameter estimation checks $N_{P_1,1}\geq N_{P_1,1}^{\tol}$, $N_{P_2,1}\geq N_{P_2,1}^{\tol}$, $\frac{wt(X_{P_1,1}\oplus X_{P_{1},1}')}{N_{P_1,1}}\leq \ebitt{1,\tol}$ are successful, for which the actual $P_1$ and $P_2$ sets are utilised.
Due to the condition of $(P_1,P_2,X_{P_1}')=(P_{1,\rr},P_{2,\rr},X_{P_1,\rr}')$, when $\FB=1$, $\rho_{\land \FB=1\land\OmegaPE}=\rho_{\land \FB=1\land\tildeOmegaPE}$.
Let us define $\tilde{L}'=L'\setminus\tilde{\theta}_{\beta}$, where the refreshed $\tilde{\theta}_{\beta}$ has to be uniformly independent as well (part of shared secret privacy).
Noting that $\hPA$ is 2-universal, we bound the trace distance using the quantum leftover hash lemma~\cite{Tomamichel2017_QKDProof,Tomamichel2011_QLHL},
\begin{equation}
\begin{split}
    &\Delta(\rhopoutt{jj'}_{\KB L'S'E'\land \tildeFB=1},\tau_{\KB}\otimes\rhopoutt{jj'}_{L'S'E'\land \tildeFB=1})\\
    \leq&2p_{\tildeOmegaPE}\Delta(\rhopoutt{jj'}_{\KB\tilde{\theta}_{\beta}\tilde{L}'S'E\land \tildeFB=1|\tildeOmegaPE},\tau_{\KB\tilde{\theta}_{\beta}}\otimes\rhopoutt{jj'}_{\tilde{L}'S'E\land \tildeFB=1|\tildeOmegaPE})\\
    \leq& p_{\tildeOmegaPE}\left\{4\epssma+2^{-\frac{1}{2}[\Hmin^{\epssma}(\hat{X}_{P_{2,\rr}}|L'S''E)_{\rhopoutt{jj'}_{\hat{X}_{P_{2,\rr}}L'S''E\land \tildeFB=1|\tildeOmegaPE}}-l_{\KB}-l_{\tilde{\theta}}]}\right\},
\end{split}
\end{equation}
where $S''=S'\setminus R$, and the first inequality uses the property $\Delta(\rho_{ABC},\tau_A\otimes\rho_{BC})\leq2\Delta(\rho_{ABC},\tau_{AB}\otimes\rho_C)$.
$L'$ is simply equal values of $\beta$ and $\beta'$ since $\tildeFB=1$ leaves it unchanged and the step updating $\beta$ (Alice's final step) has been removed.
$S''$ includes $\Kh_1$, $K_2$, $\{\tilde{\theta}_i\}_i$, and $\{\KOTP_{1,i}\}_i$.
We note that any $\tilde{\theta}_i$ and $\KOTP_{1,i}$ secrets for $i<\beta$ along with $\tildethetabetain$ would be part of $E'$ (either announced or not required to remain secret) and those for 
$i>\beta$ would be independent of $\hat{X}_{P_2,\rr}$ since they are not utilised in the protocol yet.
As such, following the arguments in the proof of Thm.~\ref{thm:QAKEKeySecIdealisedProtocol}, we can simplify
\begin{equation}
\begin{split}
    &\Hmin^{\epssma}(\hat{X}_{P_{2,\rr}}|L'S''E)_{\rhopoutt{jj'}_{\hat{X}_{P_{2,\rr}}L'S''E\land \tildeFB=1|\tildeOmegaPE}}\\
    =&\Hmin^{\epssma}(\hat{X}_{P_{2,\rr}}|\beta \Kh_1K_2\KOTP_{1,\beta}E)_{\rhopoutt{jj'}_{\hat{X}_{P_{2,\rr}}\Kh_1K_2\KOTP_{1,\beta}E\land \tildeFB=1|\tildeOmegaPE}}\\
    \geq&\Hmin^{\epssma'}(X_{P_2}|\beta P_1P_2X_{P_1}'\tildethetabetain E_2)_{\rhopoutt{jj'}_{\beta X_{P_2}P_1P_2X_{P_1}'\tildethetabetain E_2|\tildeOmegaPE}}-2-\log_2\abs{\vars{T}_{\BV}}\abs{\vars{T}_{\AV}}-\leakEC.
\end{split}
\end{equation}

The quantum state $\rhopoutt{jj'}_{\beta X_{P_2}E_2P_1P_2X_{P_1}'\tildethetabetain|\tildeOmegaPE}$ can be viewed as being generated from a protocol matching that in Thm.~\ref{thm:BB84PRNG_MinEnt}.
Since the chosen PRNG has $t'\geq 2(t+t_{MG})+t_{ex}$, we can apply Thm.~\ref{thm:BB84PRNG_MinEnt} and bound the min-entropy of interest,
\begin{equation}
    \Hmin^{\epssma'}(X_{P_2}|\beta P_1P_2\tilde{\theta}_{\beta}X_{P_1}'E_2)_{\rhopoutt{jj'}_{\beta P_1P_2\tilde{\theta}_{\beta}X_{P_1}'E_2|\tildeOmegaPE}}\geq N_{P_2,1}^{\tol}-N_{P_2,1}^{\tol}\hbin(\ephtol'),
\end{equation}
where $\epssma\geq \sqrt{2\epsph}$, 
\begin{equation}
    \epsph=2\varepsilon_{\TRNG}+\frac{\varepsilon_{\PRNG}}{p_{\tildeOmegaPE}^2}+\epsserfa
\end{equation}
and phase error tolerance
\begin{equation}
    \ephtol'=\ebitt{1,\tol}+g(N_{P_1,1}^{\tol},N_{P_2,1}^{\tol},\epsserfa)+2\sqrt{\frac{2}{N_{P_2,1}^{\tol}}\ln\frac{1}{\varepsilon_{\TRNG}}},
\end{equation}
from the EUR analysis. Combining the results, we have that $\epssma=\sqrt{2\left(2\varepsilon_{\TRNG}+\frac{\varepsilon_{\PRNG}}{p_{\tildeOmegaPE}^2}+\epsserfa\right)}$ and
\begin{equation}
    \Hmin^{\epssma}(\hat{X}_{P_{2,\rr}}|\beta \Kh_1K_2\KOTP_{1,\beta}E)_{\rhopoutt{jj'}_{\land \FB=1|\tildeOmegaPE}}\geq N_{P_2,1}^{\tol}-N_{P_2,1}^{\tol}\hbin(\ephtol')-\leakEC-\log_2\abs{\vars{T}_{\AV}}\abs{\vars{T}_{\BV}}-2.
\end{equation}
\end{proof}

\subsubsection{Shared Secrets Privacy}

The final security parameter
is the shared secrets privacy condition, which can be summarised as
\begin{theorem}
    Consider the QAKE protocol $\vars{P}$ with adversary resource limit $t$ and a $(t',\varepsilon_{\PRNG})$-quantum-secure PRNG.
    Then, the shared secrets privacy security condition satisfies
    \begin{equation*}
        \epsSP=\varepsilon_{\SP,10,11}'+\varepsilon_{\hat{\theta},\PRNG}+2(\epsMACa+\varepsilon_{\hat{\theta},\PRNG}+\epsMACb+\epsds),
    \end{equation*}
    where
    \begin{gather*}
        \varepsilon_{\SP,10,11}'=10\epssm+4(\varepsilon_2+\varepsilon_3)+\varepsilon_{\SP,\PA}+\varepsilon_{\SP,\MAC,1}+\varepsilon_{\SP,\MAC,2},\\
        \epssm=\sqrt{2(2\varepsilon_{\TRNG}+\varepsilon_{\PRNG}+\epsserfa)}\\
        \varepsilon_{\SP,\PA}=2^{-\frac{1}{2}[N_{P_2,1}^{\tol}-N_{P_2,1}^{\tol}\hbin(\ephtol')-\leakEC-\log_2\abs{\vars{T}_{\AV}}\abs{\vars{T}_{\BV}}-l_{\KB}-l_{\tilde{\theta}}]}\\
        \varepsilon_{\SP,\MAC,1}=\sqrt{(\abs{\vars{T}_{\AV}}\epsMACa-1)+2^{\log_2\left(\frac{2}{\varepsilon_3}+1\right)+\leakEC+2+\log_2\abs{\vars{T}_{\AV}}-N_{P_2,1}^{\tol}[1-\hbin(\ephtol')]}}\\
        \varepsilon_{\SP,\MAC,2}=\sqrt{(\abs{\vars{T}_{\BV}}\epsMACb-1)+2^{\log_2\left(\frac{2}{\varepsilon_2}+1\right)+\leakEC+2+\log_2\abs{\vars{T}_{\AV}}\abs{\vars{T}_{\BV}}-N_{P_2,1}^{\tol}[1-\hbin(\ephtol')]}}
    \end{gather*}
    is the shared secrets privacy parameter for idealised protocol $\vars{P}'$ for cases where $\FB=1$.
\end{theorem}
\begin{proof}
For simplicity, we perform a swap to an ``idealised" version of the protocol $\vars{P}'$ by replacing $\FB$, $D_{\BV}$ and $\DPE$, resulting in a trace distance penalty of $2(\epsds+\epsMACa+\varepsilon_{\hat{\theta},\PRNG}+\epsMACb)$. 
To simplify, we shall examine the components of the sum over $(\tildeFA,\tildeFB)$ separately, noting the exclusion of $(\phi,1)$, $(1,\phi)$ and $(1,0)$ events.
Thm.~\ref{thm:AKESecProofSSphi0} shows that the trace distance is 0 in the $(\phi,0)$ case.
In cases where $\FA=0$, Thm.~\ref{thm:AKESecProofSSFA0} shows that the trace distance is bounded by $\varepsilon_{\hat{\theta},\PRNG}$.
The final theorem, Thm.~\ref{thm:AKESecProofSSFB1}, shows the trace distance bound for the remaining cases where $\FB=1$.
When combined, the trace distance bound is as presented in the theorem.
\end{proof}

We recall that the input state in $\vars{S}_in$ is of the form 
\begin{equation}
    \rhoinn{jj'}_{LSE}=\dyad{jj'}_{\alpha\alpha'}\otimes\tau_{\Kh_1K_2R\KOTP_{1,j_{\min}}\cdots \KOTP_{1,m}}\otimes\tilde{\tau}_{\tilde{\theta}_{\j_{\max}}\cdots\tilde{\theta}_{m}\tilde{\theta}_{j_{\max}}'\cdots\tilde{\theta}_{m}'}\otimes\rho_{E},
\end{equation}
with $j_{\min}=\min\{j,j'\}$ and $j_{\max}=\max\{j,j'\}$.
For the case of $(\tildeFA,\FB)=(\phi,0)$, the trace distance is simply 0,
\begin{theorem}
    \label{thm:AKESecProofSSphi0}
    Consider the QAKE protocol $\vars{P}'$. Then, there exists $\sigma_{LSE}^{\phi0,jj'}\in\vars{S}_{in}$ such that
    \begin{equation*}
        \Delta\left(p_{\phi0}^{\reall}\rhopoutt{jj'}_{LSE|(\tildeFA,\tildeFB)=\phi0},\sum_{\tilde{j}'\geq j'}p_{\phi0,\tilde{j}}\sigma^{\phi0,j,\tilde{j}'+1}_{LSE}\right)=0.
    \end{equation*}
\end{theorem}
\begin{proof}
We start by splitting the trace distance based on $\tilde{j}'$,
\begin{equation}
    \Delta\left(p_{\phi0}^{\reall}\rhopoutt{jj'}_{LSE|(\tildeFA,\tildeFB)=\phi0},\sum_{\tilde{j}'\geq j'}p_{\phi0,\tilde{j}}\sigma^{\phi0,j,\tilde{j}'+1}_{LSE}\right)\leq\sum_{\tilde{j}'\geq j'}p_{\phi0,\tilde{j}} \Delta\left(\rhopoutt{jj'}_{LSE|(\tildeFA,\tildeFB)=\phi0,\beta'=\tilde{j}'},\sigma^{\phi0,j,\tilde{j}'+1}_{LSE}\right),
\end{equation}
noting that $\beta'<j'$ is not possible (by definition of $\beta'$), allowing us to analyse the trace distance separately for each $\tilde{j}'$ selection.\\

When $\tildeFA=\phi$, $\tildeFB=0$, the only steps that are relevant are Bob's measurement, test round announcement and Alice validation.
Since Alice validation failure stems from simply $T_{\AV}=\perp$, these steps do not utilise shared secrets $K_2$, $\Kh_1$, $\KOTP_{1,\tilde{j}'}$, and $R$, which remain private from the adversary.
We note that the protocol would implicitly trace out $\KOTP_{[\min\{j-1,j'-1\},\min\{j-1,\tilde{j}'\}]}$ since these would no longer be available for use at later steps, and what remains to show is that $\tilde{\theta}_{[\tilde{j}_{\max},m]}\tilde{\theta}_{[\tilde{j}_{\max},m]}'$ remains private.
It is clear that the basis generation seed $\tilde{\theta}_{\tilde{j}'}'$ is utilised during Bob's measurement step.
However, since Bob's index is updated to $\tilde{j}'+1$, leakage (even fully) of $\tilde{\theta}_{\tilde{j}'}'$ would still allow it to match the form of the ideal state 
since these are allowed to be part of $E$ when $\tilde{j}_{\max}>\tilde{j}'$.
Therefore, the output state can match the ideal output state (i.e. there exists $\sigma^{\phi0,j,\tilde{j}'+1}_{LSE}\in\vars{S}_{in}$ which matches the output state), and the resulting trace distance is zero.
\end{proof}

The cases of $(\tildeFA,\tildeFB)=(0,\phi)$ and $\tildeFA=\tildeFB=0$ cases are considered together.
In both cases, the risk of leakage is due to the possibility of announcement of the tag $T_{\AV}$ by Alice, though this occurs with low probability.
\begin{theorem}
\label{thm:AKESecProofSSFA0}
    Consider the QAKE protocol $\vars{P}'$ with adversary resource limit $t$ and a $(t',\varepsilon_{\PRNG})$-quantum-secure PRNG.
    Then, there exists $\sigma_{LSE}^{0\phi,jj'},\sigma_{LSE}^{00,jj'}\in\vars{S}_{in}$ such that 
    \begin{equation*}
    \Delta\left(p_{0\phi}^{\reall}\rhopoutt{jj'}_{LSE|(\tildeFA,\tildeFB)=0\phi}+p_{00}^{\reall}\rhopoutt{jj'}_{LSE|(\tildeFA,\tildeFB)=00}, \sum_{\tilde{j}\geq j}p_{0\phi,\tilde{j}}\sigma^{0\phi,\tilde{j}+1,j'}_{LSE}+\sum_{\tilde{j}\geq j,\tilde{j}'\geq j'}p_{00,\tilde{j}\tilde{j}'}\sigma^{00,\tilde{j}+1,\tilde{j}'+1}_{LSE}\right)\leq\varepsilon_{\hat{\theta},\PRNG}.
    \end{equation*}
\end{theorem}
\begin{proof}
We can split the trace distance according to the $\tilde{j}$ and $\tilde{j}'$ values, 
\begin{equation*}
    \sum_{\tilde{j}\geq j}p_{0\phi,\tilde{j}}\Delta\left(\rhopoutt{jj'}_{LSE|(\tildeFA,\tildeFB)=0\phi,\beta=\tilde{j}}, \sigma^{0\phi,\tilde{j}+1,j'}_{LSE}\right)+\sum_{\tilde{j}\geq j,\tilde{j}'\geq j'}p_{00,\tilde{j}\tilde{j}'}\Delta\left(\rhopoutt{jj'}_{LSE|(\tildeFA,\tildeFB)=00,(\beta,\beta')=\tilde{j}\tilde{j}'},\sigma^{00,\tilde{j}+1,\tilde{j}'+1}_{LSE}\right),
\end{equation*}
and evaluate the individual terms instead.
For any $\tilde{j}$ and $\tilde{j}'$, the secrets $R$ and $K_2$ remain secure since they are not utilised when $\tildeFA=0$ and $\tildeFB\neq 1$, which leaves the secrecy of $\Kh_1$, $\KOTP_1$ and $\tilde{\theta}$ to examine.\\

Due to possible differences in the adversary's information (based on input state) when $\tilde{j}$ and $\tilde{j}'$ are mismatched, we further split the analysis into three scenarios: (1) $\tilde{j}'>\tilde{j}$, (2) $\tilde{j}'<\tilde{j}$ and (3) $\tilde{j}'=\tilde{j}$.
We note that $\tilde{j}'=j'$ in the case where $\tildeFB=\phi$.\\

In case (1), it is clear that any attempt to learn $\tilde{\theta}_{\tilde{j}'}$ would not impact shared secrets privacy since it is either not utilised (when $\tildeFB=\phi$) or no longer have to be private (when $\tildeFB=0$, since updated index is $\tilde{j}'+1$).
As such, what remains is $\Kh_1$ and $\KOTP_{1,\tilde{j}}$, which if $\tildeDPE=0$, would not be leaked since $T_{\AV}$ is not announced.
However, a worse case one can consider when $\tilde{j}'>\tilde{j}$ is that $\tilde{\theta}_{\tilde{j}}$ can be known to the adversary (input state does not guarantee privacy of $\tilde{\theta}_{\tilde{j}}$ when $\tilde{j}<j'$).
Having the basis information allows the adversary to easily pass parameter estimation, $\tildeDPE=1$, by simulating an honest Bob that measures with basis generated from basis seed $\tilde{\theta}_{\tilde{j}}$.
In such a scenario, the syndrome $syn$ and tag $T_{\AV}$ would be announced by Alice, which may compromise shared secrets $\Kh_1$ and $\KOTP_{1,\tilde{j}}$.
Since $\tildeFA=0$, and $\tilde{j}<\tilde{j}'$, the updated index at the end of the protocol gives $j_{\min}=\tilde{j}+1$, which implies that $\KOTP_{1,\tilde{j}}$ is traced away.
From Thm.~\ref{thm:OTPWegCarSec}, and noting that $\tilde{D}_{\BV}$ has no relation to $\Kh_1$, the shared secret $\Kh_1$ remains secure from the adversary.
As such, the trace distance in this case is simply 0 since the output state is ideal.\\

In case (2), similarly, any attempt to learn $\tilde{\theta}_{\tilde{j}}$ is fruitless since it no longer has to be private at the end of the protocol when $\tildeFA=0$.
As such, the focus is on $\Kh_1$ and $\KOTP_{1,\tilde{j}}$.
Here, we notice that unlike the first case, $\tilde{\theta}_{\tilde{j}}$ is inaccessible to the adversary at the start of the protocol, as guaranteed by the input state since $\tilde{j}\geq j\geq j_{\max}$.
Without prior knowledge of the basis generation seed, the adversary would have little ability to pass the parameter estimation step.
More concretely, it is upper bounded by the probability of providing a matching $\tilde{\theta}_{\tilde{j}}$ from the generated quantum state alone
\begin{equation}
    \Pr[\tildeDPE=1]\leq \Pr[\hat{\theta}=\tilde{\theta}],
\end{equation}
which in turn is small.
Given that the adversary has access to at most $t$ resources, by the $(t',\varepsilon_{\PRNG})$-quantum secure PRNG property, with $t'\geq t+t_{SG,PR}$, and that $2l_{\tilde{\theta}}<n_{rnd}$, Thm.~\ref{thm:PRNGGuess} gives
\begin{equation}
    \Pr[\hat{\theta}=\tilde{\theta}]\leq\varepsilon_{\hat{\theta},\PRNG},
\end{equation}
where $\varepsilon_{\hat{\theta},\PRNG}$ matches $\epsguess$ in the theorem, and noting that only multi-photon events in the quantum states generated by Alice, $Q$, can provide information about the pseudorandom basis choice.
Additionally, since $\tildeDPE=0$ leads to an ideal state, the trace distance term when case (2) occurs is upper bounded by $\varepsilon_{\hat{\theta},\PRNG}$.\\

In the final case, the $\tildeFB=0$ and $\tildeFB=\phi$ events have to be considered separately.
When $\tildeFB=\phi$, Bob is not involved in the protocol and thus $\tilde{\theta}_{\tilde{j}}$ is not accessible to the adversary during the protocol, except from the quantum state $Q$ sent by Alice.
As such, we can follow the argument of case (2), to obtain a bound of $\varepsilon_{\hat{\theta},\PRNG}$.
When $\tildeFB=0$, Bob's involvement in the protocol and subsequent announcement of $\tilde{\theta}_{\tilde{j}'}$ can allow the adversary to pass the parameter estimation checks.
Here, since both Alice and Bob's indices are updated and $j_{\min}=\tilde{j}+1$, we can follow the argument of case (1) to get zero trace distance.\\

Combining the contribution for all $\tilde{j}$ and $\tilde{j}'$, the overall trace distance can be bounded by $\varepsilon_{\hat{\theta},\PRNG}$.
\end{proof}

Here, we analyse the final two cases of $(\tildeFA,\tildeFB)\in\{01,11\}$ jointly since they are similar.
\begin{theorem}
\label{thm:AKESecProofSSFB1}
    Consider the QAKE protocol $\vars{P}'$ with adversary resource limit $t$ and a $(t',\varepsilon_{\PRNG})$-quantum-secure PRNG.
    Then, there exists $\sigma_{LSE}^{01,jj'},\sigma_{LSE}^{11,jj'}\in\vars{S}_{in}$ such that
    \begin{equation*}
    \begin{split}
        &\Delta\left(p_{01}^{\reall}\rhopoutt{jj'}_{LSE|(\tildeFA,\tildeFB)=01}, \sum_{\tilde{j}\geq \max\{j,j'\}}p_{01,\tilde{j}}\sigma^{01,\tilde{j}+1,\tilde{j}}_{LSE}\right)+\Delta\left(p_{11}^{\reall}\rhopoutt{jj'}_{LSE|(\tildeFA,\tildeFB)=11}, \sum_{\tilde{j}\geq \max\{j,j'\}}p_{11,\tilde{j}}\sigma^{11,\tilde{j},\tilde{j}}_{LSE}\right)\leq\varepsilon_{\SP,01,11}',
    \end{split}
    \end{equation*}
    where
    \begin{equation*}
    \begin{split}
        \varepsilon_{\SP,01,11}'=&10\sqrt{2(2\varepsilon_{\TRNG}+\varepsilon_{\PRNG}+\epsserfa)}+4(\varepsilon_2+\varepsilon_3)+2^{-\frac{1}{2}[N_{P_2,1}^{\tol}-N_{P_2,1}^{\tol}\hbin(\ephtol')-\leakEC-\log_2\abs{\vars{T}_{\AV}}\abs{\vars{T}_{\BV}}-l_{\KB}-l_{\tilde{\theta}}]}\\
        &+\sqrt{(\abs{\vars{T}_{\AV}}\epsMACa-1)+2^{\log_2\left(\frac{2}{\varepsilon_3}+1\right)+\leakEC+2+\log_2\abs{\vars{T}_{\AV}}-N_{P_2,1}^{\tol}[1-\hbin(\ephtol')]}}\\
        &+\sqrt{(\abs{\vars{T}_{\BV}}\epsMACb-1)+2^{\log_2\left(\frac{2}{\varepsilon_2}+1\right)+\leakEC+2+\log_2\abs{\vars{T}_{\AV}}\abs{\vars{T}_{\BV}}-N_{P_2,1}^{\tol}[1-\hbin(\ephtol')]}},
    \end{split}
    \end{equation*}
    $\varepsilon_{\hat{\theta},\PRNG}$ is the error associated with guessing the basis generation seed $\tilde{\theta}$, $\ephtol'$ is the error tolerance inclusive of the correction due to Serfling bound.
    If 2-universal hash functions are used, then
    \begin{equation*}
    \begin{split}
        \varepsilon_{\SP,01,11}'=&10\sqrt{2(2\varepsilon_{\TRNG}+\varepsilon_{\PRNG}+\epsserfa)}+2^{-\frac{1}{2}[N_{P_2,1}^{\tol}-N_{P_2,1}^{\tol}\hbin(\ephtol')-\leakEC-\log_2\abs{\vars{T}_{\AV}}\abs{\vars{T}_{\BV}}-l_{\KB}-l_{\tilde{\theta}}]}\\
        &+2^{-\frac{1}{2}[N_{P_2,1}^{\tol}-N_{P_2,1}^{\tol}\hbin(\ephtol')-\leakEC-2-\log_2\abs{\vars{T}_{\AV}}\abs{\vars{T}_{\BV}}]}\left(1+2^{-\frac{1}{2}\log_2\abs{\vars{T}_{\BV}}}\right).
    \end{split}
    \end{equation*}
\end{theorem}
\begin{proof}
We begin by explicitly expressing the form of the output state and selected $\sigma_{LSE}^{01,jj'}$ and $\sigma_{LSE}^{11,jj'}$ states,
\begin{equation}
\begin{split}
    p_{01}^{\reall}\rhopoutt{jj'}_{LSE|(\tildeFA,\tildeFB)=01}=&\sum_{\tilde{j}\geq j_{max}}\dyad{\tilde{j}+1,\tilde{j}}_{\alpha\alpha'}\otimes\rhopoutt{jj'}_{\Kh_1K_2R\KOTP_{1,[\tilde{j},m]}\tilde{\theta}_{[\tilde{j}+1,m]}\tilde{\theta}_{[\tilde{j}+1,m]}'E\land \tildeFB=1\land \tildeFA=0\land\beta=\tilde{j}}\\
    p_{11}^{\reall}\rhopoutt{jj'}_{LSE|(\tildeFA,\tildeFB)=11}=&\sum_{\tilde{j}\geq j_{\max}} \dyad{\tilde{j},\tilde{j}}_{\alpha\alpha'}\otimes\rhopoutt{jj'}_{\Kh_1K_2R\KOTP_{1,[\tilde{j},m]}\tilde{\theta}_{[\tilde{j},m]}\tilde{\theta}_{[\tilde{j},m]}'\land \tildeFB=1\land \tildeFA=1\land\beta=\tilde{j}}\\
    \sum_{\tilde{j}\geq \max\{j,j'\}}p_{01,\tilde{j}}\sigma^{01,\tilde{j}+1,\tilde{j}}_{LSE}=&\sum_{\tilde{j}\geq j_{max}}\dyad{\tilde{j}+1,\tilde{j}}_{\alpha\alpha'}\otimes\tau_{\Kh_1K_2R\KOTP_{1,[\tilde{j},m]}}\otimes\tilde{\tau}_{\tilde{\theta}_{[\tilde{j}+1,m]}\tilde{\theta}_{[\tilde{j}+1,m]}'}\otimes\rhopoutt{jj'}_{E\land \tildeFB=1\land \tildeFA=0\land\beta=\tilde{j}}\\
    \sum_{\tilde{j}\geq \max\{j,j'\}}p_{11,\tilde{j}}\sigma^{11,\tilde{j}\tilde{j}}_{LSE}=&\sum_{\tilde{j}\geq j_{max}} \dyad{\tilde{j},\tilde{j}}_{\alpha\alpha'}\otimes\tau_{\Kh_1K_2R\KOTP_{1,[\tilde{j},m]}}\otimes\tilde{\tau}_{\tilde{\theta}_{[\tilde{j},m]}\tilde{\theta}_{[\tilde{j},m]}'}\otimes\rhopoutt{jj'}_{E\land \tildeFB=1\land \tildeFA=1\land\beta=\tilde{j}},
\end{split}
\end{equation}
where the variables are explicitly listed, with suitable tracing out of $\KOTP_{1,[1,\tilde{j}-1]}$, noting that $\beta=\beta'$ by $\tildeFB=1$ condition.

When $\tildeFB=1$, we are guaranteed that the corrected bit string $\hat{X}_{P_{2,\rr}}$ matches exactly Alice's bit string $X_{P_{2,\rr}}$.
Since the privacy amplification seed $R$ matches as well, when $\tildeFA=1$, the updated basis seed matches as well, $\tilde{\theta}_{\tilde{j}}=\tilde{\theta}_{\tilde{j}}'$.
As such, we can write a CPTP map that copies $\tilde{\theta}_{\tilde{j}}'$ to $\tilde{\theta}_{\tilde{j}}$.
Furthermore, we note that $\alpha\alpha'$ can be computed from $\beta$ alone (implicitly part of $E$) since $\tildeFA$ and $\tildeFB$ are individually fixed in both trace distance.
As such, we can simplify the trace distance using the strong convexity property, the fact that CPTP maps cannot increase trace distance, and removing common terms in the protocol (namely $\KOTP_{1,[\tilde{j}+1,m]}$ and $\tilde{\theta}_{[\tilde{j}+1,m]}'$ and $R$ when $\tildeFA=0$),
\begin{equation}
\begin{split}
    \Delta\left(p_{01}^{\reall}\rhopoutt{jj'}_{LSE|(\tildeFA,\tildeFB)=01}, \sum_{\tilde{j}\geq \max\{j,j'\}}p_{01,\tilde{j}}\sigma^{01,\tilde{j}+1,\tilde{j}}_{LSE}\right)
    \leq&\Delta\left(\rhopoutt{jj'}_{\Kh_1K_2\KOTP_{1,\beta}E\land \tildeFB=1\land \tildeFA=0},\tau_{\Kh_1K_2\KOTP_{1,\beta}}\otimes\rhopoutt{jj'}_{ E\land \tildeFB=1\land \tildeFA=0}\right)\\
    \leq&\Delta\left(\rhopoutt{jj'}_{ \Kh_1K_2\KOTP_{1,\beta}E\land \tildeFB=1},\tau_{\Kh_1K_2\KOTP_{1,\beta}}\otimes\rhopoutt{jj'}_{ E\land \tildeFB=1}\right)\\
    \Delta\left(p_{11}^{\reall}\rhopoutt{jj'}_{LSE|(\tildeFA,\tildeFB)=11}, \sum_{\tilde{j}\geq \max\{j,j'\}}p_{11,\tilde{j}}\sigma^{11,\tilde{j}\tilde{j}}_{LSE}\right)   \leq&\Delta\left(\rhopoutt{jj'}_{\Kh_1K_2R\KOTP_{1,\beta}\tilde{\theta}_{\beta}'E\land \tildeFB=1\land \tildeFA=1}, \tau_{ \Kh_1K_2R\KOTP_{1,\beta}\tilde{\theta}_{\beta}'}\otimes\rhopoutt{jj'}_{E\land \tildeFB=1\land \tildeFA=1}\right)\\
    \leq&\Delta\left(\rhopoutt{jj'}_{ \Kh_1K_2R\KOTP_{1,\beta}\tilde{\theta}_{\beta}'E\land \tildeFB=1}, \tau_{ \Kh_1K_2R\KOTP_{1,\beta}\tilde{\theta}_{\beta}'}\otimes\rhopoutt{jj'}_{E\land \tildeFB=1}\right),
\end{split}
\end{equation}
where we further remove the $\tildeFA$ conditions respectively.\\

From these trace distances, we can follow a similar analysis as in the proof of Thm.~\ref{thm:QAKEKeySecIdealisedProtocol}.
We follow the same label of the original basis seed before refresh as $\tildethetapbetain$, with the same condition $\tildeOmegaPE$.
We can split the trace distance for the $\tildeFA=0$ case, and simplify each trace distance based on the quantum leftover hash lemma~\cite{Tomamichel2011_QLHL,Tomamichel2017_QKDProof} and almost two-universal property of the hash function~\cite{Portmann2014_Authentication} (similar to proof of Thm.~\ref{thm:QAKESecProofSSFB1}),
\begin{equation}
\begin{split}
    &\Delta\left(\rhopoutt{jj'}_{\beta \Kh_1K_2\KOTP_{1,\beta}E\land \tildeFB=1},\tau_{\Kh_1K_2\KOTP_{1,\beta}}\otimes\rhopoutt{jj'}_{\beta E\land \tildeFB=1}\right)\\
    \leq& p_{\tildeOmegaPE}\left[\Delta\left(\rhopoutt{jj'}_{\beta \Kh_1K_2\KOTP_{1,\beta}E\land \tildeFB=1|\tildeOmegaPE},\tau_{K_2}\otimes\rhopoutt{jj'}_{\beta \Kh_1\KOTP_{1,\beta}E\land \tildeFB=1|\tildeOmegaPE}\right)\right.\\
    &\left.+\Delta\left(\rhopoutt{jj'}_{\beta \Kh_1\KOTP_{1,\beta}E\land \tildeFB=1},\tau_{\Kh_1\KOTP_{1,\beta}}\otimes\rhopoutt{jj'}_{\beta E\land \tildeFB=1|\tildeOmegaPE}\right)\right]\\
    \leq&p_{\tildeOmegaPE}\left[2(\epssmb+\epssmc+\varepsilon_2+\varepsilon_3)+\sqrt{(\abs{\vars{T}_{\AV}}\epsMACa-1)+2^{\log_2\abs{\vars{T}_{\AV}}+\log_2\left(\frac{2}{\varepsilon_3}+1\right)-\Hmin^{\epssmc}(\hat{X}_{P_2,\rr}| \beta E)_{\land \tildeFB=1|\tildeOmegaPE}}}\right.\\
    &\left.+\sqrt{(\abs{\vars{T}_{\BV}}\epsMACb-1)+2^{\log_2\abs{\vars{T}_{\BV}}+\log_2\left(\frac{2}{\varepsilon_2}+1\right)-\Hmin^{\epssmb}(\hat{X}_{P_2,\rr}| \beta\Kh_1\KOTP_{1,\beta}E)_{\land \tildeFB=1|\tildeOmegaPE}}}\right]\\
\end{split}
\end{equation}
and for the $\tildeFA=1$ case,
\begin{equation}
\begin{split}
    &\Delta\left(\rhopoutt{jj'}_{\beta \Kh_1K_2R\KOTP_{1,\beta}\tilde{\theta}_{\beta}'E\land \tildeFB=1},\tau_{\beta \Kh_1K_2R\KOTP_{1,\beta}\tilde{\theta}_{\beta}'}\otimes\rhopoutt{jj'}_{E\land \tildeFB=1}\right)\\
    \leq& p_{\tildeOmegaPE}\left[\Delta\left(\rhopoutt{jj'}_{\beta \Kh_1K_2R\KOTP_{1,\beta}\tilde{\theta}_{\beta}'\KB E\land \tildeFB=1|\tildeOmegaPE},\tau_{R\tilde{\theta}_{\beta}'\KB}\otimes\rhopoutt{jj'}_{\beta \Kh_1K_2\KOTP_{1,\beta}E\land \tildeFB=1|\tildeOmegaPE}\right)\right.\\
    &+\Delta\left(\rhopoutt{jj'}_{\beta \Kh_1K_2\KOTP_{1,\beta}E\land \tildeFB=1|\tildeOmegaPE},\tau_{K_2}\otimes\rhopoutt{jj'}_{\beta \Kh_1\KOTP_{1,\beta}E\land \tildeFB=1|\tildeOmegaPE}\right)\\
    &\left.+\Delta\left(\rhopoutt{jj'}_{\beta \Kh_1\KOTP_{1,\beta}E\land \tildeFB=1},\tau_{\Kh_1\KOTP_{1,\beta}}\otimes\rhopoutt{jj'}_{\beta E\land \tildeFB=1|\tildeOmegaPE}\right)\right]\\
    \leq&p_{\tildeOmegaPE}\left\{2(\epssma+\epssmb+\epssmc+\varepsilon_2+\varepsilon_3)+\frac{1}{2}\times 2^{-\frac{1}{2}\left[\Hmin^{\epssma}(\hat{X}_{P_{2,\rr}}|\beta \Kh_1K_2\KOTP_{1,\beta}E)_{\land \tildeFB=1|\tildeOmegaPE}-l_{\KB}-l_{\tilde{\theta}}\right]}\right.\\
    &\left.+\sqrt{(\abs{\vars{T}_{\AV}}\epsMACa-1)+2^{\log_2\abs{\vars{T}_{\AV}}+\log_2\left(\frac{2}{\varepsilon_3}+1\right)-\Hmin^{\epssmc}(\hat{X}_{P_2,\rr}| \beta E)_{\land \tildeFB=1|\tildeOmegaPE}}}\right.\\
    &\left.+\sqrt{(\abs{\vars{T}_{\BV}}\epsMACb-1)+2^{\log_2\abs{\vars{T}_{\BV}}+\log_2\left(\frac{2}{\varepsilon_2}+1\right)-\Hmin^{\epssmb}(\hat{X}_{P_2,\rr}| \beta\Kh_1\KOTP_{1,\beta}E)_{\land \tildeFB=1|\tildeOmegaPE}}}\right\},
\end{split}
\end{equation}
where we label $\land \tildeFB=1|\tildeOmegaPE$ as the evaluation of the min-entropy on the corresponding state conditioned on event $\tildeOmegaPE$ and when $\tildeFB=1$, and we reintroduced $\KB$ since it is part of the output of the final hash function.
Following the proof of Thm.~\ref{thm:QAKEKeySecIdealisedProtocol}, the smooth min-entropy term can be bounded by
\begin{equation}
    \Hmin^{\epssma}(\hat{X}_{P_{2,\rr}}|\beta \Kh_1K_2\KOTP_{1,\beta}E)_{\rhopoutt{jj'}_{\beta \Kh_1K_2\KOTP_{1,\beta}E\land \tildeFB=1|\tildeOmegaPE}}\geq N_{P_2,1}^{\tol}-N_{P_2,1}^{\tol}\hbin(\ephtol')-\leakEC-\log_2\abs{\vars{T}_{\AV}}\abs{\vars{T}_{\BV}}-2.
\end{equation}
The second term and third terms can follow a similar analysis, with the main difference being the lack of conditioning of some secrets.
In the second term, the conditioning on $K_2$ is not present. 
Therefore, $T_{\BV}$, as the output of a $\epsMACb$-almost strongly 2-universal hash function, is uniform and independent of the input, importantly $X_{P_2}$.
As such, it can be removed without using the min-entropy chain rule, i.e. without incurring a $\log_2\abs{\vars{T}_{\BV}}$ penalty, resulting in
\begin{equation}
    \Hmin^{\epssmb}(\hat{X}_{P_2,\rr}|\beta \Kh_1\KOTP_{1,\beta}E)_{\rhopoutt{jj'}_{\beta\hat{X}_{P_2,\rr}\Kh_1\KOTP_{1,\beta}E\land \tildeFB=1|\tildeOmegaPE}}\geq N_{P_2,1}^{\tol}-N_{P_2,1}^{\tol}\hbin(\ephtol')-\leakEC-\log_2\abs{\vars{T}_{\AV}}-2.
\end{equation}
In the third term, both the conditioning on $K_2$ and $\Kh_1\KOTP_{1,\beta}$ are absent.
Similarly, we can argue that $T_{\AV}$ and $T_{\BV}$ are uniform and independent of $X_{P_2}$ since the respective seeds of the almost 2-universal hash functions are not part of the conditioning.
As such, we have
\begin{equation}
    \Hmin^{\epssmc}(\hat{X}_{P_2,\rr}|\beta E)_{\rhopoutt{jj'}_{\beta\hat{X}_{P_2,\rr}E\land \tildeFB=1|\tildeOmegaPE}}\geq N_{P_2,1}^{\tol}-N_{P_2,1}^{\tol}\hbin(\ephtol')-\leakEC-2.
\end{equation}
Combining the results, we recover the theorem.
We can repeat the analysis with two-universal hash function to obtain the second set of bounds in the theorem.
\end{proof}

\subsection{Numerical Analysis}

Before we proceed with the numerical simulation, we briefly introduce the robustness.
The main sources of failure are the parameter estimation checks during the protocol or error correction failure, with overall robustness similarly quantified as original QAKE protocol,
\begin{equation}
    \epsrob= \varepsilon_{\rob,P}+\epsserfb+\varepsilon_{\ds,\rob}+\varepsilon_{\bigepsilon,P_1,\rob}+\varepsilon_{\EC},
\end{equation}
where $\varepsilon_{\rob,P}$ is associated with $P^{\tol}\leq\abs{P}\leq P^{\UB}$ bounds, $\epsserfb$ is the error associated with modified Serfling bound~\cite{Tomamichel2012_BB84FiniteKey,Curty2014_DecoyMDIQKD} for estimating bit error rate in $P_2$, $\varepsilon_{\ds,\rob}$ is associated with the decoy state, $\varepsilon_{\bigepsilon,P_1,\rob}$ is associated with the bit error tolerance in set $P_1$, and $\varepsilon_{\EC}$ is simply the failure rate of error correction.
The concentration bound used for robustness parameters is the tight bound on binomial distribution~\cite{Zubkov2013_SVBound} due to i.i.d. state preparation.
The concentration bound used for estimation of the expectation values, $\expE{N_{\mu_v,P_i}}$, from the observed values, $N_{\mu_v,P_i}$, associated with decoy state estimation (errors form part of decoy state $\epsds$ in $\epssecint$) is Kato's bound~\cite{Kato2020_KatoIneq,Curras2021_KatoIneqPara,Zhang2023}.\\

We analyse the performance of the QAKE protocol by simulating the length of keys $l_{K_B}$ generated based on Thm.~\ref{thm:PRNGQAKE_Main}. 
We assume a simple experimental model with Alice preparing decoy BB84 states, sending it through a channel with loss $\eta$, and Bob performing measurement in the same basis as Alice.
Bob's detector is assumed to have zero dark counts, and the overall experiment is assumed to have \SI{2}{\percent} QBER, with the following detection probability, $\Pr[\dett,\mu_j]=p_{\mu_j}(1-e^{-\eta\mu_j})$, where we also fix $\mu_0=0.45$, $\mu_1=0.225$, $\mu_2=0$ as the decoy intensity, $p_{\mu_0}=0.2$, $p_{\mu_1}=0.6$ and $p_{\mu_2}=0.2$ as their respective probabilities.
The PRNG is assumed to use a 256-bit basis generation seed $\tilde{\theta}_{\beta}$.
Given the adversary's resource $t$ is limited by the duration of each protocol session (less than a minute), and additional resources $t_{ex}$, $t_{MG}$ and $t_{SG,PR}$ are simple steps that can be performed within the protocol duration, it is safe to argue that the PRNG is secure.
In this case, we assume a 128-bit security, $\varepsilon_{\PRNG}=2^{-128}$, which should be a conservative estimate of the security of the PRNG for the run time of the full protocol.
We note that such quantum-secure PRNG can be constructed from well-known protocols, such as AES~\cite{NIST2001_AES,Hoang2020_PRNG}.
In the simulation, we also fix the length of the authentication tags as 80 bits each, and that $\epsMACa=\epsMACb=2^{-80}$, allowing us to use the tighter bound presented in Thm.~\ref{thm:PRNGQAKE_Main} with 2-universal hash function.
We set the error correction efficiency to be $\fEC=1.2$.\\

The security parameter of the each session of the protocol is fixed at $\epssecint=$ \SI{1e-15}{}, which yields an overall protocol security of \SI{1e-6}{} when we allow it to run up to \SI{1e9}{} sessions.
The robustness parameter is fixed at $\epsrob=$ \SI{1e-10}{}, and each sessions involves $n_{rnd}=$ \SI{1e10}{} signals being sent by Alice to Bob.
The simulation is then performed by rearranging the result in Thm.~\ref{thm:AKEProtocolMain}, and optimising $l_{\KB}$ over the splitting ratio (also size of test set) $f_{P_1}$, robustness and secrecy parameters (components of $\varepsilon_{rob}$ and $\epssec$) on Matlab, with other parameters fixed as described in the results in Fig.~\ref{fig:SimulationResult}.\\

\begin{figure}[!ht]
    \centering
    \includegraphics[width=0.7\linewidth]{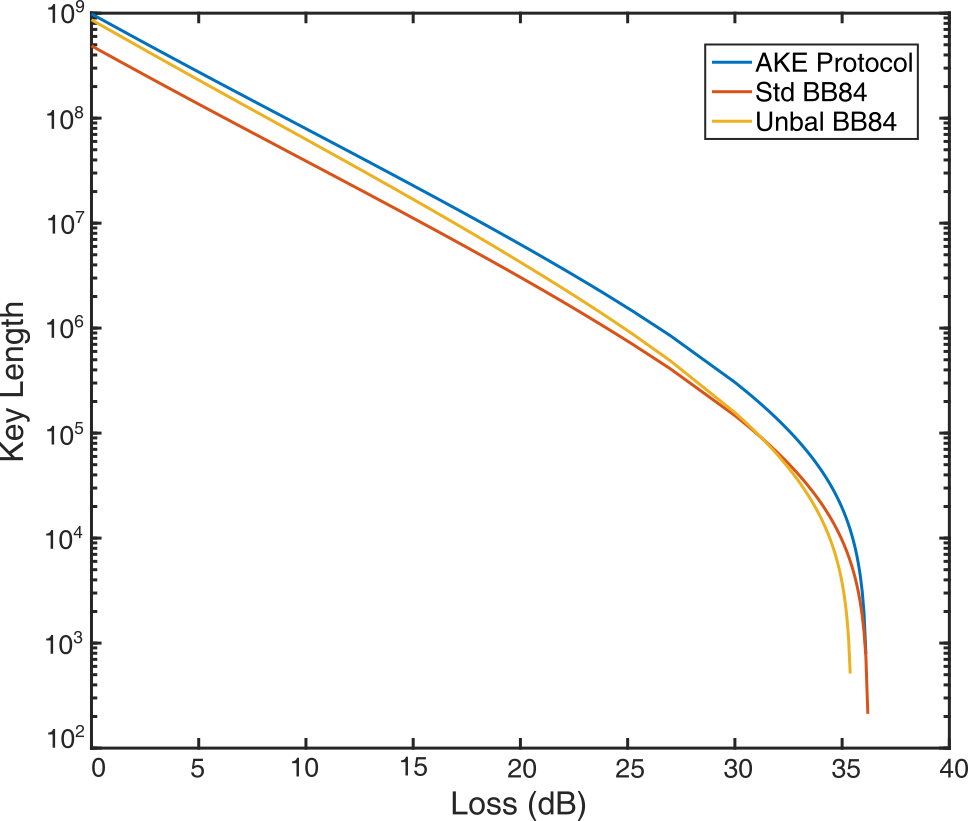}
    \caption{\label{fig:SimulationResult}Simulated key length for QAKE with pseudorandom basis protocol at various signal loss values (including detector loss etc), standard BB84 protocol and unbalanced BB84 protocol, with $n_{rnd}=$ \SI{1e10}{} signals sent, $\epsrob=$ \SI{1e-10}{}, and $\epssec=$ \SI{1e-15}{}. The PRNG parameter is assumed to be $\varepsilon_{\PRNG}=2^{-128}$, with $l_{\tilde{\theta}}=256$. The authentication tag used is 80-bits, i.e. $\abs{\vars{T}_{\AV}}=\abs{\vars{T}_{\BV}}=80$, and we assume the error correction efficiency to be $\fEC=1.2$. The system parameters are chosen to be $\mu_0=0.45$, $\mu_1=0.225$, $\mu_2=0$ as the decoy intensity, $p_{\mu_0}=0.2$, $p_{\mu_1}=0.6$ and $p_{\mu_2}=0.2$ as the respective probabilities, and we assume the system to have zero dark counts and a QBER of \SI{2}{\percent}.}
\end{figure}

To study the effect of the use of PRNG, we compare the results with the keys generated in a standard BB84 scheme~\cite{Tomamichel2017_QKDProof} and a BB84 scheme with unbalanced basis choices~\cite{Lim2014_DecoyQKD}.
To make the comparison fair, we utilise the same analysis results, but remove the penalties associated with PRNG, and make corresponding changes to the sifting and testing probability.
The standard BB84 scheme~\cite{Tomamichel2017_QKDProof} we compare to uses the same experimental model, but due to sifting, has an additional sifting factor resulting in
\begin{equation*}
    \Pr[\det,\text{sift},\mu_j]=\frac{p_{\mu_j}}{2}(1-e^{-\eta\mu_j}).
\end{equation*}
For the BB84 scheme with unbalanced basis choice~\cite{Lim2014_DecoyQKD}, the probability of selecting $X$ basis, $p_X$, is varied, and all rounds in the $X$ basis is used as the test round.
As such, the resulting model statistics has
\begin{gather*}
    \Pr[\det,\text{test},\mu_j]=p_X^2p_{\mu_j}(1-e^{-\eta\mu_j})\\
    \Pr[\det,\text{key},\mu_j]=(1-p_X)^2p_{\mu_j}(1-e^{-\eta\mu_j}),
\end{gather*}
while the test rounds in the earlier schemes have only an $f_{P_1}$ factor to the detected rate.
For both BB84 cases, due to lack of PRNG use, we simply set $\varepsilon_{\PRNG}=\varepsilon_{\hat{\theta},\PRNG}=\varepsilon_{\TRNG}=0$, with no basis seed to refresh, $l_{\tilde{\theta}}=0$, phase error
\begin{equation*}
    \ephtol'=\ebitt{1,\tol}+g(N_{P_1,1}^{\tol},N_{P_2,1}^{\tol},\epsserfa).
\end{equation*}
They are simulated similarly by optimising $l_{\KB}$ with the respective changes on Matlab.
The results are shown in Fig.~\ref{fig:SimulationResult}, where the protocol with PRNG generated basis has an advantage, mostly stemming from the lack of sifting.
This shows up as a doubling of key rate of BB84, and an increasingly better key rate relative to BB84 with unbalanced basis as channel loss increases. 
There is a crossover at low key lengths, likely due to the impact of the 256-bit basis seed that has to be refreshed.

\subsection{Discussion}

We note that compared to the original QAKE protocol, when PRNG is utilised, it loses the key recycling property.
It only has a weaker key refreshing property, where most shared secrets can be recycled when the protocol passes, except for the basis seed $\tilde{\theta}$, which needs to be refreshed every session.
If one extends the assumption on the PRNG to remain secure for time $t_{\secc}\gg t'$ and $m_{\secc}$ sessions of protocol execution, one can recover the key recycling property by removing Bob's announcement of $\tilde{\theta}'_{\beta'}$ and using the random selection of test round to ensure that the adversary cannot successfully report a test round choice to Alice before it has received that information from Bob (i.e. Bob measured the quantum systems).
The same basis generation seed could be used with a different public input (e.g. counter) to the PRNG for basis generation, thereby removing the need to refresh basis seed $\tilde{\theta}$ when the protocol passes, allowing for full key recycling.
We note that $\tilde{\theta}$ would still have to be updated once $t_{\secc}$ has passed or when $n_{sess}$ sessions has reached, and one could refer to NIST recommendations for the refresh rate for such symmetric keys~\cite{NIST_PRNG}.\\

One technical challenge to implement PRNG basis choice is the requirement of an active basis choice at the receiver to guarantee measurement in the same basis and the key rate advantage.
This requires an extremely fast optical switch at the receiver to switch between the X and Z basis. 
A pulsed laser with a \SI{500}{\mega\hertz} repetition rate would require a nanosecond optical switch.
Nanosecond optical switches are commercially available with a low insertion loss of \SI{0.7}{\decibel}, across the whole optical communication band.
They have been demonstrated in data center networks~\cite{xue2022nanosecond} and dynamic optical switches with nanoseconds switching speed and large inputs/outputs (448 $\times$ 448) have also been previously demonstrated~\cite{yeo2010448}.
With the availability of nanosecond optical switches with low insertion loss, it may be feasible for high-speed implementation of the QAKE protocol with active basis selection.\\

\section{Decoy-State BB84 with Pseudorandom Basis}
\label{app:PRNGBasisChoiceProof}
\subsection{Motivation}

The QAKE protocol with pseudorandom basis has key secrecy parameter dependent on the secrecy condition of decoy-state BB84 with shared pseudorandom basis.
Here, we study such a protocol, where the main changes made relative to decoy-state BB84 are:
\begin{enumerate}
    \item Alice and Bob pre-share a basis generation seed $\tilde{\theta}_{\beta}$.
    \item During the protocol, Alice and Bob would use a PRNG $G^{\PRNG}$ to generate the basis to prepare and measure the quantum state respectively.
    \item The sifting step is removed since Alice and Bob will always prepare and measure in the same basis.
\end{enumerate}
Due to the removal of the sifting step, decoy state BB84 with shared pseudorandom basis is expected to gain a factor of 2 advantage in key generation rate.
Similar protocols have been of interest in the literature, where basis information is generated or encrypted using a PRNG~\cite{Trushechkin2018_PRNGQKD,Price2021_DDOSPRNGEncryptBasis}.
However, a complete proof of security of the use of PRNG for basis selection remains elusive in the literature, with a security proof provided in Ref.~\cite{Trushechkin2018_PRNGQKD} against only intercept-resend attacks.
Here, we provide a complete security analysis through the \emph{entropic uncertainty relation} (EUR).\\

Since the main goal in this section is to prove the security from the use of PRNG for basis generation, we focus on the main quantity of interest in the secrecy condition of QKD~\cite{Tomamichel2011_QLHL,Tomamichel2017_QKDProof}: the smooth min-entropy $\Hmin^{\varepsilon}(X_{P_2}|\beta P_1P_2\tilde{\theta}_{\beta}X_{P_1}'E)_{\rho_{\beta X_{P_2}P_1P_2\tilde{\theta}_{\beta}X_{P_1}'E|\tildeOmegaPE}}$.
The smooth min-entropy quantifies the maximum length of key that can be extracted from $X_{P_2}$ that is private from the adversary.
The smooth min-entropy is evaluated on the output state of the following process:
\begin{enumerate}
    \item Alice prepares decoy state BB84 states, using the basis generated from the basis generation seed $\tilde{\theta}_{\beta}$, bit value $X$ and intensity choice.
    \item The adversary performs its attack on the state, mapping systems $QE$ to $PBE'$.
    \item Bob randomly selects $P_1$ rounds and measure them in the basis generated from $\tilde{\theta}_{\beta}$, and outputs $X_{P_1}'$.
    \item Based on $X$, $P_1$, $P_2$ and $X_{P_1}'$, the decision of whether the parameter estimation step passes is made ($\tildeOmegaPE$ being the event where the parameter estimation passes).
\end{enumerate}
In decoy-state BB84, entropic uncertainty relation (EUR) is one method to tightly bound the smooth min-entropy of the bit values of single photon events.
Importantly, EUR lower bounds the smooth min-entropy by a function of the phase error rate, i.e. the mismatch of measurement outcomes in the complementary basis.
Without PRNG and assuming uniform basis choice, the phase error rate matches the bit error rate (mismatch of measurement outcomes in the original basis) which can be estimated from the experiment.
We argue that this gap between bit and phase error is small when PRNG is used, resulting in a slightly lower smooth min-entropy bound (and thus key rate).
We present these arguments in later sub-sections, starting with analysing EUR with the use of PRNG in Appendix~\ref{app:EUR_with_PRNG}, followed by the proof of the min-entropy lower bound in Appendix~\ref{app:BB84_PRNG_Secrecy}.
Detailed proof of the theorems on EUR with PRNG in presented in the final two sub-sections.

\subsection{Entropic Uncertainty Relation with PRNG}
\label{app:EUR_with_PRNG}

Entropic uncertainty relation describes the uncertainty a party has of performing incompatible measurements on a quantum state.
Let us define a measurement on a subsystem $C$ with initial basis choice $\theta$, a bijective function $f_b$ to map $\theta$ to a different basis choice, and measurement outcome $Y$ by 
\begin{equation}
     \vars{M}_{f_b(\theta) C\rightarrow Y}(\rho):=\Tr_C\,\left[\sum_{\gamma,y}\ket{y}_Y(\dyad{\gamma}_{\theta}\otimes O^{f_b(\gamma),y}_C)\rho(\dyad{\gamma}_{\theta}\otimes O^{f_b(\gamma),y}_C)\bra{y}_Y \right],
\end{equation}
where $O_C^{f_b(\theta),y}$ is the measurement operator on subsystem $C$ with basis $f_b(\theta)$.
The form of EUR of particular interest in our proof examines the following scenario:
\begin{enumerate}
    \item The basis $\theta$ is selected, and a quantum state with subsystems held by three parties, Alice ($A$), Bob ($B$), and ($E$), is generated, i.e. state $\rho_{\theta ABE}$.
    \item If Alice measures $A$ in basis $\theta$ and obtains outcome $Y$, i.e. perform $\vars{M}_{\theta A\rightarrow Y}$, Bob attempts to guess the value of $Y$ from his subsystem $B$ and $\theta$.
    \item If Alice measures $A$ in basis $f_b(\theta)$ for some bijective map $f_b$, and obtains outcome $\tilde{Y}$, i.e. perform $\vars{M}_{f_b(\theta) A\rightarrow \tilde{Y}}$, Eve attempts to guess the value of $\tilde{Y}$ from her subsystem $E$ and $\theta$.
\end{enumerate}
The EUR then quantifies the uncertainty on the outcomes $Y$ and $\tilde{Y}$ that subsystem $E$ and $B$ has respectively, which is related to the measurement operator incompatibility~\cite{Tomamichel2011_EUR}.
The more certain Eve can be on his guess of $Y$ (smaller conditional entropy), the more uncertain Bob must be on her guess of $\tilde{Y}$ (larger conditional entropy).
Here, we present a modified variant of EUR, adapted from Ref.~\cite{Tomamichel2017_QKDProof}. 
We note also that the EUR does not make any assumptions on how $\theta$ is generated.
\begin{theorem}
\label{thm:StdEUR}
    Let $\rho_{\theta ABE}$ be a tripartite quantum state and $f_b$ be a bijective function on $\theta$.
    Then,
    \begin{equation*}
        \Hmin^{\varepsilon}(Y|\theta E)_{\vars{M}_{\theta A\rightarrow Y}(\rho_{\theta ABE})}\geq q-\Hmax^{\varepsilon}(\tilde{Y}|\theta B)_{\vars{M}_{f_b(\theta) A\rightarrow \tilde{Y}}(\rho_{\theta ABE})},
    \end{equation*}
    where $c_q=\max_{\theta}\max_{y,\tilde{y}}\norm{O^{\theta,y}_A (O^{f_b(\theta),\tilde{y}}_A)^{\dagger}}_{\infty}^2$ and $q=\log_2\frac{1}{c_q}$.
\end{theorem}
\begin{proof}
We follow closely the proof of EUR in Ref.~\cite{Tomamichel2017_QKDProof}, where a similar form of EUR is proven, and taking inspiration from Ref.~\cite{Tomamichel2011_EUR}.
We begin by introducing the Stinespring dilation isometry of both measurement maps $\vars{M}_{\theta A\rightarrow Y}$ and $\vars{M}_{f_b(\theta) A\rightarrow \tilde{Y}}$ respectively,
\begin{equation}
\begin{gathered}
    V=\sum_{\gamma,y}\ket{yy}_{YY'}\otimes\dyad{\gamma}_{\theta}\otimes O_A^{\theta,y}\\
    W=\sum_{\gamma,\tilde{y}}\ket{\tilde{y}\tilde{y}}_{\tilde{Y}\tilde{Y}'}\otimes\dyad{\gamma}_{\theta}\otimes O_A^{f_b(\theta),\tilde{y}}.
\end{gathered}
\end{equation}
We can purify the state $\rho_{\theta ABE}$ to
\begin{equation}
    \ket{\psi}_{\theta\theta'ABED}=\sum_{\gamma}\sqrt{p_{\alpha}}\ket{\gamma\gamma}_{\theta\theta'}\otimes\ket{\phi_{\alpha}}_{ABED},
\end{equation}
where $\ket{\phi_{\alpha}}_{ABED}$ is a purification of $\rho_{ABE|\theta=\gamma}$, which can be performed by selecting $D$ to be a sufficiently large auxiliary system.
The post-measurement state after performing the latter measurement map $\vars{M}_{f_b(\theta) A\rightarrow \tilde{Y}}$ can thus be written as
\begin{equation}
    \rho_{\theta\theta' ABE\tilde{Y}\tilde{Y}'D}=W\dyad{\psi}_{\theta\theta'ABED}W^{\dagger}
\end{equation}
Since the state is pure, we can apply the duality relation of smooth min- and max-entropy~\cite{Konig2009_OprMinMax,Tomamichel2010_MaxMinDuality},
\begin{equation}
    \Hmax^{\varepsilon}(\tilde{Y}|\theta B)+\Hmin^{\varepsilon}(\tilde{Y}|\theta'AE\tilde{Y}'D)=0,
\end{equation}
and what remains is to bound the smooth min-entropy term, which is upper bounded by $\Hmin^{\varepsilon}(\tilde{Y}|\theta AE\tilde{Y}')$ since $\theta=\theta'$, and $D$ can be removed by the data-processing inequality for CPTP map $\Tr_D$~\cite{Tomamichel2015_ITBook}.\\

The smooth min-entropy term is evaluated on the state $[W\rho_{\theta AE}W^{\dagger}]$, where we note that it is possible to reverse the measurement of system $A$ using basis $f_b(\theta)$ and measure it using basis $\theta$.
This is expressed as a CPTP map, using the Stinespring dilation isometry of the measurement maps,
\begin{equation}
    \bigepsilon=\sum_{\gamma}\Tr_{AY'}[VW^{\dagger}\dyad{\gamma}_{\theta}\,\cdot\,\dyad{\gamma}_{\theta}WV^{\dagger}],
\end{equation}
where we have 
\begin{equation}
\begin{split}
    &\bigepsilon(\rho_{\theta AE\tilde{Y}\tilde{Y}'})\\
    =&\sum_{\gamma}\Tr_{AY'}\,\{VW^{\dagger}[\dyad{\gamma}_{\theta}\otimes\sum_{yy'}\ket{yy}_{YY'}O_A^{f_b(\gamma),y}\rho_{AE|\theta=\gamma}(O_A^{f_b(\gamma),y'})^{\dagger}\bra{y'y'}_{YY'}]WV^{\dagger}\}\\
    =&\sum_{\gamma}\Tr_{AY'}\,\{V[\dyad{\gamma}_{\theta}\otimes\sum_{y}(O_A^{f_b(\gamma),y})^{\dagger}O_A^{f_b(\gamma),y}\rho_{AE|\theta=\gamma}\sum_{y'}(O_A^{f_b(\gamma),y'})^{\dagger}O_A^{f_b(\gamma),y'}]V^{\dagger}\}\\
    =&\Tr_{AY'}[V\rho_{\theta AE}V^{\dagger}]\\
    =&\rho_{\theta YE},
\end{split}
\end{equation}
where $\rho_{\theta YE}$ is the post-measurement state on $\rho_{\theta ABE}$ with measurement map $\vars{M}_{\theta A\rightarrow Y}$ and the appropriate trace, and we note that $\sum_{y}(O_A^{f_b(\gamma),y})^{\dagger}O_A^{f_b(\gamma),y}=\mathbb{I}_A$.\\

By definition of the smooth min-entropy, there exists a state $\sigma_{\theta AE\tilde{Y}\tilde{Y}'}$ and $\omega_{\theta AE\tilde{Y}'}$ such that purified distance $P(\rho_{\theta AE\tilde{Y}\tilde{Y}'},\sigma_{\theta AE\tilde{Y}\tilde{Y}'})\leq\varepsilon$, and
\begin{equation}
\label{eqn:ModifiedEURProofHminDefn}
    \sigma_{\theta AE\tilde{Y}\tilde{Y}'}\leq 2^{-\lambda}\mathbb{I}_{\tilde{Y}}\otimes\omega_{\theta AE\tilde{Y}'},
\end{equation}
where $\lambda=\Hmin^{\varepsilon}(\tilde{Y}|\theta AE\tilde{Y}')$.
We note here that any CPTP map should not increase the purified distance, which allows us to define a state $\sigma_{\theta YE}:=\bigepsilon(\sigma_{\theta AE\tilde{Y}\tilde{Y}'})$, such that it is $\varepsilon$-close in purified distance to $\rho_{\theta YE}$,
\begin{equation}
    P(\sigma_{\theta YE},\rho_{\theta YE})\leq P(\sigma_{\theta AE\tilde{Y}\tilde{Y}'},\rho_{\theta AE\tilde{Y}\tilde{Y}'})\leq\varepsilon.
\end{equation}
We can apply the channel $\bigepsilon$ on both sides of Eq.~\eqref{eqn:ModifiedEURProofHminDefn}, giving
\begin{equation}
    \sigma_{\theta YE}\leq 2^{-\lambda}\bigepsilon(\mathbb{I}_{\tilde{Y}}\otimes\omega_{\theta AE\tilde{Y}'}).
\end{equation}
The RHS term can be expanded as
\begin{equation}
\begin{split}
    \bigepsilon(\mathbb{I}_{\tilde{Y}}\otimes\omega_{\theta AE\tilde{Y}'})=&\sum_{\gamma y\tilde{y}}\dyad{\gamma y}_{\theta Y}\otimes\mel{\tilde{y}}{\Tr_A[O^{\gamma,y}_A(O^{f_b(\gamma),\tilde{y}}_A)^{\dagger} \hat{\omega}_{AE\tilde{Y}'}^{\gamma}O^{f_b(\gamma),\tilde{y}}_A(O^{\gamma,y}_A)^{\dagger}]}{\tilde{y}}_{\tilde{Y}'}\\
    \leq&\sum_{\gamma y\tilde{y}}\dyad{\gamma y}_{\theta Y}\otimes\norm{O^{\gamma,y}_A(O^{f_b(\gamma)}_{\tilde{y}})^{\dagger}}_{\infty}^2\mel{\tilde{y}}{\Tr_A(\hat{\omega}_{AE\tilde{Y}'}^{\gamma})}{\tilde{y}}_{\tilde{Y}'}\\
    \leq&\left[\max_{\gamma}\max_{y,\tilde{y}}\norm{O^{\gamma,y}_A(O^{f_b(\gamma)}_{\tilde{y}})^{\dagger}}_{\infty}^2\right]\sum_{\gamma y}\dyad{\gamma y}_{\theta Y}\otimes\Tr_{A\tilde{Y}'}[\hat{\omega}_{AE\tilde{Y}'}^{\gamma}]\\
    =&c_q\mathbb{I}_Y\otimes\hat{\omega}_{\theta E},
\end{split}
\end{equation}
where $\hat{\omega}_{\theta E}:=\sum_{\gamma}\dyad{\gamma}_{\theta}\otimes\Tr_{A\tilde{Y}'}[\hat{\omega}_{AE\tilde{Y}'}^{\gamma}]$ and $\hat{\omega}_{AE\tilde{Y}'}^{\gamma}:=\mel{\gamma}{\omega_{\theta AE\tilde{Y}'}}{\gamma}$.
Therefore, by definition, the min-entropy is lower bounded,
\begin{equation}
    \Hmin(Y|\theta E)_{\sigma_{\theta YE}}\geq \Hmin^{\varepsilon}(\tilde{Y}|\theta AE\tilde{Y}')+q,
\end{equation}
where $q=\log_2\frac{1}{c_q}$.
Since $\sigma_{\theta YE}$ is $\varepsilon$-close in purified distance to $\rho_{\theta YE}$, the smooth min-entropy can be bounded,
\begin{equation}
    \Hmin^{\varepsilon}(Y|\theta E)_{\rho_{\theta YE}}\geq \Hmin^{\varepsilon}(\tilde{Y}|\theta AE\tilde{Y}')+q.
\end{equation}
Combining the results, we end up with
\begin{equation}
    \Hmin^{\varepsilon}(Y|\theta E)_{\rho_{\theta YE}}+\Hmax^{\varepsilon}(\tilde{Y}|\theta B)_{\rho_{\theta\tilde{Y}B}}\geq q.
\end{equation}
Noting that $\rho_{\theta YE}=\Tr_{B}[\vars{M}_{\theta A\rightarrow Y}(\rho_{\theta ABE})]$ and $\rho_{\theta\tilde{Y}B}=\Tr_{E}[\vars{M}_{f_b(\theta)A\rightarrow\tilde{Y}}(\rho_{\theta ABE})]$, we arrive at the form of the theorem.
\end{proof}
To bound the smooth min-entropy, we thus have to compute the smooth max-entropy $\Hmax^{\varepsilon}(\tilde{Y}|\theta B)_{\vars{M}_{f_b(\theta) A\rightarrow \tilde{Y}}(\rho_{\theta ABE})}$.
A typical reduction~\cite{Renes2012_MaxEnt,Tomamichel2012_BB84FiniteKey} of the smooth max-entropy is to simplify it via data-processing inequality and bound it with $\hbin(\eph)$, where $\eph=\frac{wt(\tilde{Y}_A\oplus \tilde{Y}_B)}{n}$ is termed the phase error rate and $n$ is the number of entries in $\tilde{Y}_A$ (also number of qubit of subsystem $A$), obtained from from the state $\vars{M}_{f_b(\theta) A\rightarrow \tilde{Y}_A}\circ \vars{M}_{f_b(\theta) B\rightarrow \tilde{Y}_B}(\rho_{\theta ABE|\Omega}^{\PRNG})$.
Since we never measure in the $f_b(\theta)$ basis in the protocol, the phase error is typically computed from the bit error rate, $\ebit:=\frac{wt(Y_A\oplus Y_B)}{n}$, from the state $\vars{M}_{\theta A\rightarrow Y_A}\circ \vars{M}_{\theta B\rightarrow Y_B}(\rho_{\theta ABE|\Omega}^{\PRNG})$.
Typically, it is argued that the expectation values of the phase and bit error matches, $\expE{\ebit}=\expE{\eph}$ if the basis $\theta\in\{0,1\}^{n}$ is chosen uniformly at random
and independent from subsystems $ABE$~\cite{Tomamichel2011_EUR,Tomamichel2017_QKDProof}.\\

\begin{figure}
    \centering
    \includegraphics[width=0.45\linewidth]{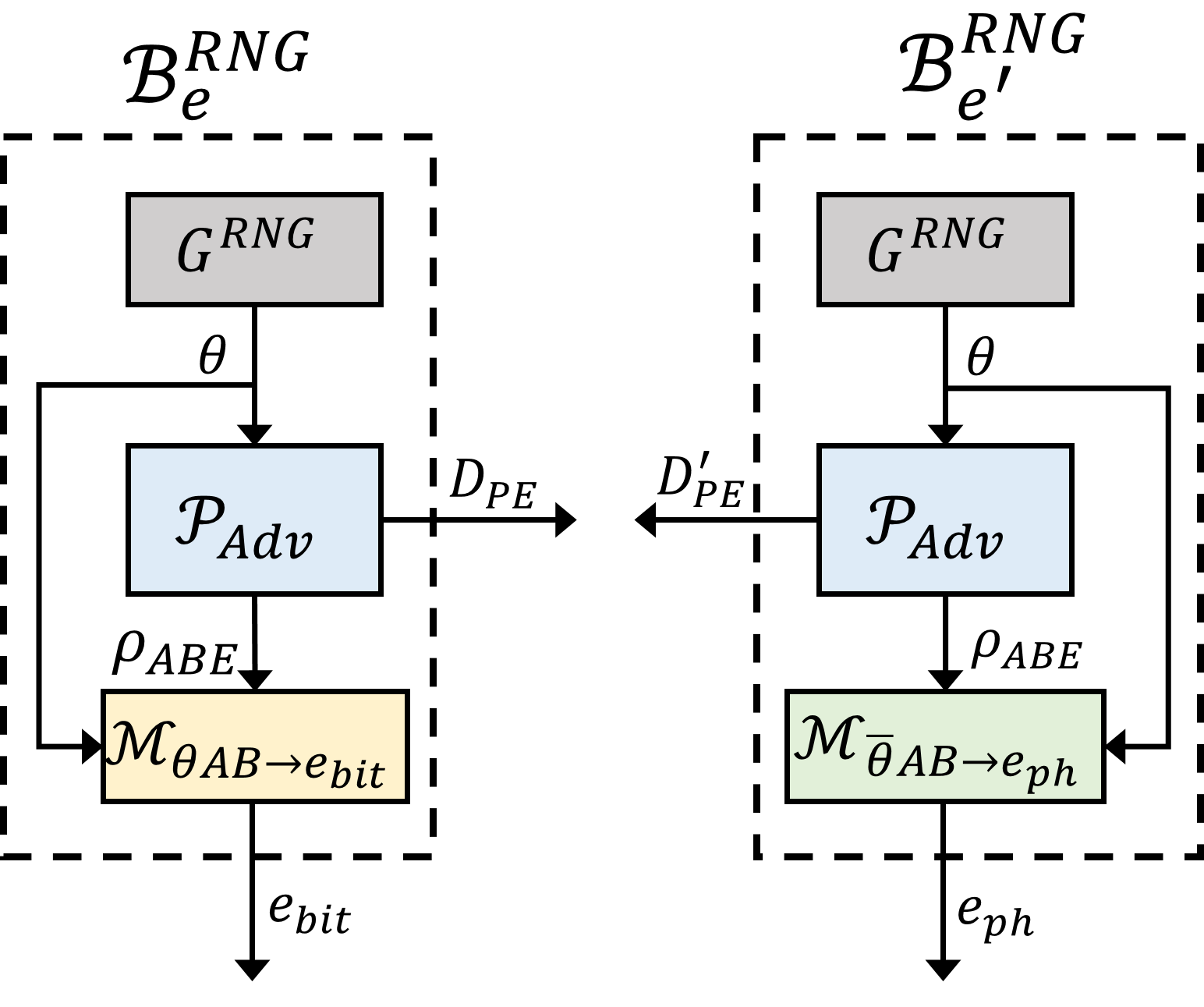}
    \caption{\label{fig:PRNGSingleAdversaryChannel}Description of quantum algorithms $\vars{B}_e^{\RNG}$ and $\vars{B}_{e'}^{\RNG}$. Each algorithm begins with a run of the random number generator $G^{\RNG}$ which outputs $\theta$, followed by the generation of quantum state $\rho_{\theta ABE}$ via a protocol $\vars{P}_{Adv}$ (e.g. QAKE protocol with adversarial attack), and ends with a measurement on subsystems $AB$ to compute the bit error $\ebit$ and phase error $\eph$ respectively. The output $\DPE$ forms a decision which defines the event set $\Omega=\{\DPE=1\}$. We label the resources required for state preparation and measurement as $t$, with the other steps (e.g. flipping of $\theta$ when computing phase error, and computing of bit/phase error) contributing resource $t_{MG}$, and overall resource requirement $t_B=t+t_{MG}$.}
\end{figure}

\begin{figure}
    \centering
    \includegraphics[width=0.35\linewidth]{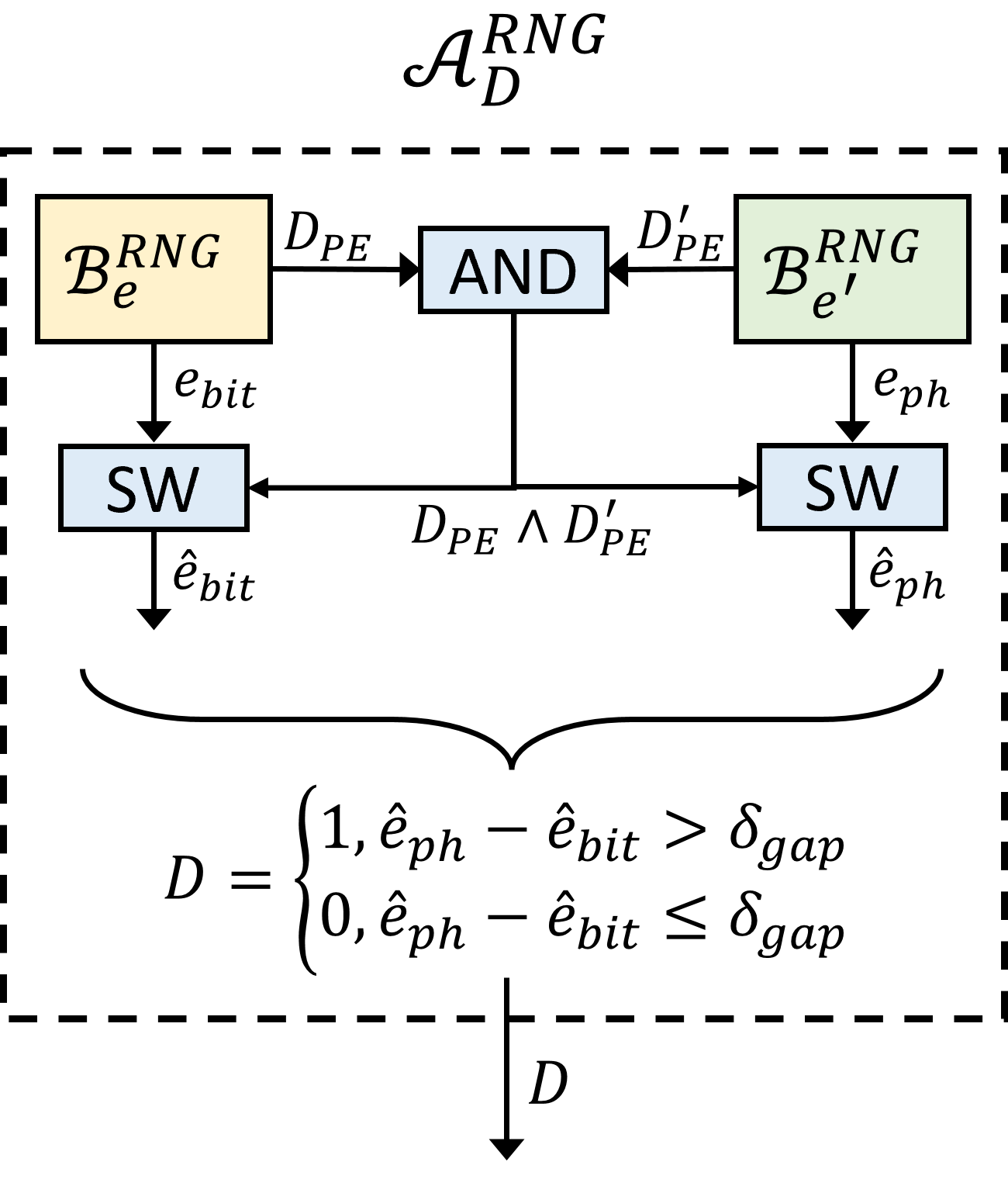}
    \caption{\label{fig:PRNGFullAdversaryChannel}Description of distinguishing adversary with output decision $D$. The adversary computes $\DPE\land \DPE'$, where $\DPE$ and $\DPE'$ are obtained from the quantum process $\vars{P}_{Adv}$. The SW function is a switch which sets $\hatebit=0$ and $\hateph=0$ respectively when $\DPE\land \DPE'=0$, and leave $\hatebit=\ebit$ and $\hateph=\eph$ when $\DPE\land \DPE'=1$. The decision $D$ is made by comparing the bit-phase error gap with a target $\deltagap$. The resource required for the adversary is defined as $t_D$, which comprises of $2t_B$ and extra steps $t_{\text{ex}}$.}
\end{figure}

This argument breaks down when a PRNG is utilised since the basis is no longer uniformly distributed.
Instead, we argue that the difference between the phase and bit error rates in the case where a quantum-secure PRNG is used (bit-phase error gap) is small.
Consider two separate instances of an adversary that is bounded by resource $t$, and prepares the state $\rho_{\theta ABE}$, which would be measured in $\theta$ and $\bar{\theta}$ respectively ($f_b$ refers to bit flip in this case) to obtain the bit and phase errors (see Fig.~\ref{fig:PRNGSingleAdversaryChannel}).
The pseudorandomness property of the PRNG ensures that a similarly small bit-phase error gap must be present when a PRNG is utilised.
Otherwise, the bit-phase error gap can be utilised by a distinguisher with $t'\geq t_D$ resources (see Fig.~\ref{fig:PRNGFullAdversaryChannel}) to distinguish a PRNG from an ideal RNG with probability greater than $\varepsilon_{\PRNG}$, where $t_D=2(t+t_{MG})+t_{ex}$.
We recall that the PRNG is chosen with $t'=2(t+t_{MG})+t_{ex}+t_{SG,PR}$.\\

We first show that the phase error rate can be lower bounded by the bit error rate with some penalty, stemming from the low bit-phase error gap for PRNG.

\begin{theorem}
\label{thm:PRNGPhaseErrBnd}
    Suppose that the basis $\theta$ is generated from a $(t',\varepsilon_{\PRNG})$-quantum secure PRNG with $t'\geq t_D$ and let $\rho_{\theta ABE}^{\PRNG}$ be the state prepared by the adversary (via protocol $\vars{P}_{Adv}$), where $A$ and $B$ are $n$-qubit states, and $\theta$ remain independent when generated by an IRNG, $\rho_{\theta ABE}^{\TRNG}=\tau_{\theta}\otimes\rho_{ABE}$. 
    Let $\Omega:=\{\DPE=1\}$ be an event that can be decided from $E$, and $\ebittol$ be the tolerance value for an upper bound on the bit error value, i.e. $\Pr[\ebit\geq \ebittol|\Omega]_{\vars{M}_{\theta A\rightarrow Y_A}\circ \vars{M}_{\theta B\rightarrow Y_B}(\rho_{\theta ABE|\Omega}^{\PRNG})}\leq\epsbit$.
    Then, an upper bound on the phase error for state $\rho_{\theta ABE|\Omega}^{\PRNG}$ after measurement is
    \begin{equation*}
    \begin{gathered}
        \Pr[\eph\geq \ephtol]_{\vars{M}_{\bar{\theta}A\rightarrow\tilde{Y}}\circ\vars{M}_{\bar{\theta}B\rightarrow\tilde{Y}'}(\vars\rho_{\theta ABE|\Omega}^{\PRNG})}\leq \epsph\\
        \ephtol=\ebittol+2\sqrt{\frac{2}{n}\ln\frac{1}{\varepsilon_{\TRNG}}}\\
        \epsph=2\varepsilon_{\TRNG}+\frac{\varepsilon_{\PRNG}}{p_{\Omega}^2}+\epsbit.
    \end{gathered}
    \end{equation*}
\end{theorem}
\begin{proof}
We first show that the bit-phase error gap is small.
Since the algorithm in Fig.~\ref{fig:PRNGFullAdversaryChannel}, $\vars{A}_D^{G^{\PRNG}_{\vars{K}},t'}$ has resource restriction $t_D$, the quantum-secure PRNG cannot be distinguished via this algorithm,
\begin{equation}
    \abs{\Pr[\vars{A}_D^{G^{\PRNG}_{\vars{K}},t_D}=1]-\Pr[\vars{A}_D^{G^{\TRNG},t_D}=1]}\leq\varepsilon_{\PRNG}.
\end{equation}
By definition of the outcome $D$, the bit-phase error gap obtained from PRNG is no much different from that of an IRNG,
\begin{equation}
    \Pr[\hateph-\hatebit>\deltagap]_{\vars{A}_D^{G^{\PRNG}_{\vars{K}},t_D}}-\Pr[\hateph-\hatebit>\deltagap]_{\vars{A}_D^{G^{\TRNG},t_D}}\leq\varepsilon_{\PRNG},
\end{equation}
What remains to evaluate the IRNG bit-phase error rate gap.\\

We first show that the expected value of the bit and phase error matches for the IRNG case.
When we have $\DPE=0$ or $\DPE'=0$, both bit and phase error would be set to 0, $\hateph=\hatebit$, and we can expand the probability
\begin{equation}
    \Pr[\hateph-\hatebit>\deltagap]_{\vars{A}_D^{G^{\TRNG},t_D}}=p_{\Omega}^2\Pr[\eph-\ebit>\deltagap|\Omega,\Omega']_{\vars{A}_D^{G^{\TRNG},t_D}}.
\end{equation}
This conditional probability matches the probability of the phase-bit error gap exceeding $\delta_{gap}$, when the phase error and bit error are drawn from the conditional state $\rho_{\theta ABE|\Omega}^{\TRNG}$ in both instances instead.
By assumption, when IRNG is utilised, $\theta$ is independent of the prepared quantum state.
Furthermore, since $\Omega$ is decided from $E$, $\rho_{\theta ABE|\Omega}^{\TRNG}=\tau_{\theta}\otimes\rho_{ABE|\Omega}^{\TRNG}$.
Let $e_i=Y_{A,i}\oplus Y_{B,i}$ and $e_i'=\tilde{Y}_{A,i}\oplus \tilde{Y}_{B,i}$ be the bit error and phase error of round $i$ of the two instances.
Since $\theta$ is uniformly random, the post-measurement state $\rho_{Y_AY_B|\Omega}$ and $\rho_{\tilde{Y}_A\tilde{Y}_B|\Omega'}$ are identical 
(after tracing out $\theta E$).
As such, noting that errors between rounds can be correlated, we have that the expected bit and phase error matches, conditioned on earlier rounds, $\expE{e_i|e_1^{i-1}}=\expE{e_i'|e_1^{'i-1}}$, where the indexing $e_i^j=e_i\cdots e_j$.\\

Here, we demonstrate that both the bit and phase error are close to some combination of the conditional expectation value $\expE{e_i|e_1^{i-1}}$ to find the bit-phase error gap.
We define
\begin{equation}
    Z_j=\sum_{i=1}^j e_i-\expE{e_i|e_1^{i-1}},
\end{equation}
where $e_1^{i-1}$ labels the first $i-1$ outcomes for error measurement.
The observed bit error can be expressed as
\begin{equation}
    \ebit=\frac{1}{n}\sum_{i=1}^{n}e_i=\frac{Z_{n}}{n}-\frac{1}{n}\sum_{i=1}^{n}\expE{e_i|e_1^{i-1}}.
\end{equation}
The variable $Z_j$ is bounded, the set $n$ is finite and
\begin{equation}
\begin{split}
    &\expE{Z_j|Z_1,\cdots,Z_{j-1}}\\
    =&Z_{j-1}+\expE{e_j|Z_1,\cdots,Z_{j-1}}-\expE{e_j|e_1^{j-1}}\\
    =&Z_{j-1},
\end{split}
\end{equation}
since $Z_1,\cdots,Z_{j-1}$ values can be computed from $e_1^{j-1}$ and vice versa, using the values of $\expE{e_i|e_1^{i-1}}$ recursively.
These properties implies that $Z_j$ is a martingale, and since $\abs{Z_j-Z_{j-1}}\leq 1$, we can apply the Azuma-Hoeffding inequality~\cite{Mitzenmacher2017_Prob},
\begin{equation}
\begin{gathered}
    \Pr[Z_{n}-Z_0\leq-n\delta ]\leq e^{-\frac{n\delta^2}{2}}\\
    \Pr[Z_{n}-Z_0\geq n\delta ]\leq e^{-\frac{n\delta^2}{2}}.
\end{gathered}
\end{equation}
Performing the same analysis for $e_{ph}$, we arrive at
\begin{equation}
\begin{gathered}
    \Pr[\sum_{i=1}^{n}\frac{\expE{e_i|e_1^{i-1}}}{n}\geq \ebit+\delta]\leq e^{-\frac{n\delta^2}{2}}\\
    \Pr[\eph\geq\sum_{i=1}^{n}\frac{\expE{e'_i|e_1^{'i-1}}}{n}+\delta]\leq e^{-\frac{n\delta^2}{2}},
\end{gathered}
\end{equation}
where a different (upper or lower) bound is utilised for the two instances.\\

Since the states are identical, the sum of conditional expectations of $e_i$ matches.
As such, the bit-phase error gap is
\begin{equation}
    \Pr[\eph-\ebit\geq 2\delta|\Omega,\Omega']_{\vars{A}_D^{G^{\TRNG},t_D}}\leq 2e^{-\frac{n\delta^2}{2}}.
\end{equation}
Defining $\deltagap=2\delta$, the bit-phase error of the algorithm is
\begin{equation}
    \Pr[\eph-\ebit\geq 2\delta|\Omega,\Omega']_{\vars{A}_{\ebit,\eph}^{G^{\PRNG}_{\vars{K}},t_D}}\leq 2e^{-\frac{n\delta^2}{2}}+\frac{\varepsilon_{\PRNG}}{p_{\Omega}^2}.
\end{equation}
by combining the results.
We note that the bit and phase error are obtained from two separate instances of bit value measurement and phase value measurement on identical $\rho_{\theta ABE|\Omega}^{\PRNG}$ state.
Given that the bit error value measured is, with probability larger than $1-\epsbit$, larger than $\ebittol$, we are guaranteed that the phase error value measured on the same state satisfies
\begin{equation}
    \Pr[\eph\geq \ebittol+2\delta|\Omega,\Omega']_{\vars{A}_{\ebit,\eph}^{G^{\PRNG}_{\vars{K}},t_D}}\leq \epsph,
\end{equation}
where $\epsph=2e^{-\frac{n\delta^2}{2}}+\frac{\varepsilon_{\PRNG}}{p_{\Omega}^2}+\epsbit$.
Since the probability no longer depends on the bit error rate $\ebit$, we can replace $\vars{A}_{\ebit,\eph}^{G^{\PRNG}}$ with simply a phase error measurement on the common input state,
\begin{equation}
    \Pr[\eph\geq \ebittol+2\delta|\Omega]_{\vars{M}_{\bar{\theta}A\rightarrow\tilde{Y}}\circ\vars{M}_{\bar{\theta}B\rightarrow\tilde{Y}'}(\vars\rho_{\theta ABE|\Omega}^{\PRNG})}\leq \epsph.
\end{equation}
Further letting $\varepsilon_{\TRNG}:=e^{-\frac{n\delta^2}{2}}$, a simple rearrangement would yield the theorem.
\end{proof}

Combining the modified EUR and PRNG results, we can obtain a lower bound on the smooth min-entropy of interest.
\begin{theorem}
\label{thm:PRNGEUR}
    Suppose that the basis $\theta$ is generated from a $(t',\varepsilon_{\PRNG})$-quantum secure PRNG, with $t'\geq t_D$, and let $\rho_{\theta ABE}^{\PRNG}$ be the state prepared by the adversary (via protocol $\vars{P}_{Adv}$), where $A$ and $B$ are $n$-qubit states,
and $\theta$ remain independent when generated by an IRNG, $\rho_{\theta ABE}^{\TRNG}=\tau_{\theta}\otimes\rho_{ABE}$. 
    Furthermore, assume the measurement operators of Alice act independently and identically on the subsystems, $O_A^{\theta,y}=\bigotimes_iO_{A_i}^{\theta_i,y_i}$ and have overlap $c_{q_1}:= \max_{\theta_i}\max_{y_i,\tilde{y}_i}\norm{O^{\theta_i,y_i}_{A_i}(O^{\bar{\theta}_i,\tilde{y}_i}_{A_i})^{\dagger}}_{\infty}^2$, with $q_1=\log_2\frac{1}{c_{q_1}}$.
    Let $\Omega:=\{\DPE=1\}$ be an event that can be decided by $E$, then the conditional smooth min-entropy for a protocol with resource utilisation $t$ (including adversary's resources) be
    \begin{equation*}
        \Hmin^{\sqrt{2\epsph}}(Y|\theta E)_{\vars{M}_{\theta A\rightarrow Y}(\rho_{\theta ABE|\Omega}^{\PRNG})}\geq n[q_1-\hbin(\ephtol)],
    \end{equation*}
    where
    \begin{equation*}
    \begin{gathered}
        \ephtol=\ebittol+2\sqrt{\frac{2}{n}\ln\frac{1}{\varepsilon_{\TRNG}}}\\
        \epsph=2\varepsilon_{\TRNG}+\frac{\varepsilon_{\PRNG}}{p_{\Omega}^2}+\epsbit,
    \end{gathered}
    \end{equation*}
    $\ebittol$ is the tolerance value for an upper bound on the bit error value, and $\Pr[\ebit\geq \ebittol|\Omega]_{\vars{M}_{\theta A\rightarrow Y_A}\circ \vars{M}_{\theta B\rightarrow Y_B}(\rho_{\theta ABE|\Omega}^{\PRNG})}\leq\epsbit$.
\end{theorem}
\begin{proof}
For the state $\rho_{\theta ABE}^{\PRNG}$, we are interested in finding the conditional smooth-min entropy, of $Y$ conditioned on $E$, for the PRNG state conditioned on event $\Omega$, $\Hmin^{\epssm}(Y|\theta E)_{\vars{M}_{\theta A\rightarrow Y}(\rho_{\theta ABE|\Omega}^{\PRNG})}$.
Applying the entropic uncertainty relation presented in Thm.~\ref{thm:StdEUR} with the bijective function $f_b(\theta)=\bar{\theta}$, we get
\begin{equation}
\begin{split}
    &\Hmin^{\epssm}(Y|\theta E)_{\vars{M}_{\theta A\rightarrow Y}(\rho_{\theta ABE|\Omega}^{\PRNG})}\\
    \geq &q-\Hmax^{\epssm}(\tilde{Y}|\theta B)_{\vars{M}_{\bar{\theta} A\rightarrow Y}(\rho_{\theta ABE|\Omega}^{\PRNG})}\\
    \geq &q-\Hmax^{\epssm}(\tilde{Y}_A|\tilde{Y}_B)_{\vars{M}_{\bar{\theta} A\rightarrow \tilde{Y}_A}\circ\vars{M}_{\bar{\theta} B\rightarrow \tilde{Y}_B}(\rho_{\theta ABE|\Omega}^{\PRNG})},
\end{split}
\end{equation}
where the second inequality is obtained by the data processing inequality~\cite{Tomamichel2015_ITBook}.
We can expand $q=\log_2\frac{1}{c_q}$ with
\begin{equation}
    c_q=\max_{\theta}\max_{y,\tilde{y}}\norm{O^{\theta,y}_A(O^{\bar{\theta},\tilde{y}}_A)^{\dagger}}_{\infty}^2,
\end{equation}
where the bijective function is a flip from $\theta$ to $\bar{\theta}$.
Since the measurement operators are assumed to be independent and identical operators with overlap
and applied independently on each subsystem $A_i$,
the norm can be simplified as
\begin{equation}
\begin{split}
    \norm{O^{\theta,y}_A(O^{\bar{\theta},\tilde{y}}_A)^{\dagger}}_{\infty}^2=&\norm{\otimes_{i=1}^{n}(O^{\theta_i,y_i}_{A_i})(\otimes_{i=1}^{n}(O^{\bar{\theta}_i,\tilde{y}_i}_{A_i}))^{\dagger}}_{\infty}^2\\
    =&\Pi_{i=1}^{n}\norm{F^{\theta_i,y_i}_{A_i}(O^{\bar{\theta}_i,\tilde{y}_i}_{A_i})^{\dagger}}_{\infty}^2
\end{split}
\end{equation}
Taking the maximum $\theta$ would mean taking the maximum $\theta_i$ for each index $i$.
As such, we have
$q=nq_1$.\\

From Thm.~\ref{thm:PRNGPhaseErrBnd}, we know that the state $\vars{M}_{\bar{\theta} A\rightarrow \tilde{Y}_A}\circ\vars{M}_{\bar{\theta} B\rightarrow \tilde{Y}_B}(\rho_{\theta ABE|\Omega}^{\PRNG})$ has a phase error upper bounded by $\ephtol$ with probability larger than $1-\epsph$.
As such, there exist a state $\sigma_{\tilde{Y}_A\tilde{Y}_B}$ that is $\epsph$-close in trace distance with a maximum phase error of $\ephtol$.
The purified distance of these states can be bounded by $\sqrt{2\epsph}$~\cite{Tomamichel2015_ITBook}.
Therefore, if we set $\epssm\geq\sqrt{2\epsph}$, we can upper bound the smooth max-entropy by the max-entropy evaluated on $\sigma$,
\begin{equation}
    \Hmax^{\epssm}(\tilde{Y}_A|\tilde{Y}_B)_{\vars{M}_{\bar{\theta} A\rightarrow \tilde{Y}_A}\circ\vars{M}_{\bar{\theta} B\rightarrow \tilde{Y}_B}(\rho_{\theta ABE|\Omega}^{\PRNG})} \leq \Hmax(\tilde{Y}_A|\tilde{Y}_B)_{\sigma_{\tilde{Y}_A\tilde{Y}_B}}.
\end{equation}
From Ref.~\cite{Renes2012_MaxEnt}, the max-entropy can be computed from the minimum number of bits required to correct for the error.
Since the two values can only differ in at most $\lfloor n\ephtol\rfloor$ locations, we have~\cite{Tomamichel2012_BB84FiniteKey}
\begin{equation}
\begin{split}
    \Hmax(\tilde{Y}_A|\tilde{Y}_B)_{\sigma_{\tilde{Y}_A\tilde{Y}_B}}\leq&\log_2\sum_{i=0}^{\lfloor n\ephtol\rfloor} {n \choose i}\\
    \leq &n\hbin(\ephtol).
\end{split}
\end{equation}
which completes the proof.
\end{proof}

\subsection{Security of decoy-state BB84 with PRNG basis}
\label{app:BB84_PRNG_Secrecy}

We can now adapt the min-entropy result to the actual decoy-state BB84 protocol.
\begin{theorem}
\label{thm:BB84PRNG_MinEnt}
    Considering an adversary with resource limitation $t$ and the use of $(t',\varepsilon_{\PRNG})$-secure PRNG with $t'\geq 2(t+t_{MG})+t_{ex}$, the min-entropy of interest, which can be described by the steps above, is lower bounded by
    \begin{equation*}
        \Hmin^{\varepsilon}(X_{P_2}|\beta P_1P_2\tilde{\theta}_{\beta}X_{P_1}'E)_{\rho_{\beta X_{P_2}P_1P_2\tilde{\theta}_{\beta}X_{P_1}'E|\tildeOmegaPE}}\geq N_{P_2,1}^{\tol}-N_{P_2,1}^{\tol}\hbin(\ephtol'),
    \end{equation*}
    where $\varepsilon\geq \sqrt{2\epsph}$, with
    \begin{gather*}
        \epsph=2\varepsilon_{\TRNG}+\frac{\varepsilon_{\PRNG}}{p_{\tildeOmegaPE}^2}+\epsserfa\\
        \ephtol'=\ebitt{1,\tol}+g(N_{P_1,1}^{\tol},N_{P_2,1}^{\tol},\epsserfa)+2\sqrt{\frac{2}{N_{P_2,1}^{\tol}}\ln\frac{1}{\varepsilon_{\TRNG}}}
    \end{gather*}
    and $g(a,b,c)=\sqrt{\frac{(b+a)(a+1)}{2a^2b}\ln\frac{1}{c}}$.
\end{theorem}
\begin{proof}
We first reduce the protocol to a worse case protocol to simplify analysis.
The choice of decoy state intensity can be delayed, by allowing Alice to first sample a random photon number count for $n_{rnd}$ rounds, $N_{\PNR}$, and select decoy choice by sampling $v$ based on $N_{\PNR}$.
In general, the photon numbers $N_{\PNR}$ can be accessible to the adversary, and any rounds $Q_i$ where there are multi-photon events could leak the basis and bit value.
As such, we let Alice announce the basis and bit value for all rounds with multi-photon events.
For single-photon events, preparing BB84 states on system $Q_i$ is identical to preparing $\ket{\Phi^+}_{A_iQ_i}$, and performing measurement on $A_i$ with the X and Z-basis operators.
This modification allows us to delay the measurement on $A_i$ to a later step.
We note that the corresponding measurement operators $O_{A_i}^{\theta_i,X_i}$ acts independently and identically on the subsystems and has overlap $q_1=1$, assuming ideal preparation of the BB84 states.
As such, we have a worse case of Alice preparing and sending out $N_{\PNR}$, $\theta_{\geq 2}$, $X_{\geq 2}$, and $Q_i$ of single-photon events instead of that in step 1 of the protocol.
In fact, WLOG, we can let the adversary prepare the single-photon states, and select $P$ at the same time, preparing $\rho_{A_{P,1}B_{P,1}E'}$.
For steps 3 and 4, Alice and Bob can operate with simply these single-photon events since $\tildeOmegaPE$ is decided only on such events.
This worse case protocol can thus be summarised as
\begin{enumerate}
    \item The adversary announces $\beta$ to Alice and Bob.
    \item Alice draws a random photon number count for $n_{rnd}$ rounds, $N_{\PNR}$, based on a probability distribution computed from the intensity and probability settings.
    \item From the basis generation seed $\tilde{\theta}_{\beta}$ and $N_{\PNR}$, Alice announces the basis $\theta_{\geq 2}$ and random bit string $X_{\geq 2}$.
    \item The adversary announces $P$ and prepares accordingly the state $\rho_{A_{P_*,1}B_{P_*,1}E'}$, where $A_{P_*,1}$ is handed to Alice and $B_{P_*,1}$ is handed to Bob, where $P_{*,1}$ is the set of single-photon events in $P$.
    \item Bob selects randomly a subset $P_1$ from set P. Define $P_2=P\setminus P_1$.
    \item Bob measures $B_{P_1,1}$ to obtain $X_{P_1,1}'$ and Alice measures $A_{P_1,1}$ to obtain $X_{P_1,1}$.
    \item Alice and Bob computes and checks the condition $\tildeOmegaPE$, and they announce $\tilde{\theta}_{\beta}$.
    \item Alice measures $A_{P_2,1}$ to obtain $X_{P_2,1}$.
\end{enumerate}
Note that $X_{P_1,0}$ would be announced in the original protocol, but this is simply a random string uncorrelated with the rest of the variables and can be removed from the min-entropy.\\

We begin by first considering another worse case where Alice truncates the number of single-photon rounds in set $P_2$ to $N_{P_2,1}^{\tol}$, which can occur when $\tildeOmegaPE$ occurs.
For $\tilde{\Omega}_{PE}$, Alice and Bob can replace the subsystem with $N_{P_2,1}^{\tol}$ Bell states, which does not affect min-entropy since it is evaluated conditioned on $\tildeOmegaPE$.
As such,
\begin{equation*}
    \Hmin^{\varepsilon}(X_{P_2}|\beta P_1P_2\tilde{\theta}_{\beta}X_{P_1}'E')_{\rho_{\beta X_{P_2}P_1P_2\tilde{\theta}_{\beta}X_{P_1}'E'|\tildeOmegaPE}}\geq \Hmin^{\varepsilon}(X_{P_2^{\trunc},1}|\beta P_1P_2\tilde{\theta}_{\beta}X_{P_1,1}'X_{P_1,1}N_{\PNR}X_{\geq 2}E')_{\rho_{|\tildeOmegaPE}},
\end{equation*}
where $\trunc$ superscript indicate truncation.\\

The state generated, $\rho_{\beta X_{P_2^{\trunc},1} P_1P_2\tilde{\theta}_{\beta}X_{P_1,1}'N_{\PNR}X_{\geq 2}E'|\tildeOmegaPE}$, matches that described in EUR, Thm.~\ref{thm:PRNGEUR}.
We note that the state can be expressed as
\begin{equation}
    \rho_{X_{P_2,1}^{\trunc} \beta\tilde{\theta}_{\beta}E|\tildeOmegaPE}=\vars{M}_{\theta_{P_2,1}A_{P_2,1}^{\trunc}\rightarrow X^{\trunc}_{P_2,1}}(\rho_{\beta\tilde{\theta}_{\beta}A_{P_2,1}^{\trunc}B_{P_2,1}^{\trunc} E|\tildeOmegaPE}),
\end{equation}
where $E=P_1P_2X_{P_1,1}'X_{P_1,1}N_{\PNR}X_{\geq 2}E'$.
We note that $\tildeOmegaPE$ can be determined from $X_{P_1,1}'X_{P_1,1}N_{\PNR}P_1P_2$, which are part of $E$.
By the fair sampling assumption, $P_1$ and $P_2$ are independent of any basis selection.
Moreover, if the basis choice $\theta_{\beta,P_2^{\trunc},1}$ is selected from an IRNG, it would be uniformly random and independent from the adversary (also from $\beta$), since $\theta_{P_1}$ and $\theta_{P_2}$ are uncorrelated, i.e. we have $\tau_{\theta_{P_2,1}}$ and uncorrelated from the rest of the state.\\

The final quantity to define is $\ebittol$, which in this case is the phase tolerance bound for the set $P_{2,1}^{\trunc}$.
This can be obtained by the modified Serfling bound~\cite{Curty2014_DecoyMDIQKD,Tomamichel2012_BB84FiniteKey}, since the choice of sets $P_1$ and $P_2$ is random.
Since $\tildeOmegaPE$ gives $N_{P_1,1}^{\tol}$ as a lower bound on the number of single photon events, and that the bit error rate is upper bounded by $e_{b,1,tol}$, we have that
\begin{equation}
    \Pr[\ebit\geq \ebitt{1,\tol}+g(N_{P_1,1}^{\tol}, N_{P_2,1}^{\tol},\epsserfa)]\leq\epsserfa. 
\end{equation}
Therefore, we can apply Thm.~\ref{thm:PRNGEUR} to arrive at the result.
\end{proof}

\section{Step-Efficient Client Authentication}
\label{app:CAProtocol}
\subsection{Protocol Description}

\begin{figure}[!h]
    \centering
    \includegraphics[width=0.5\textwidth]{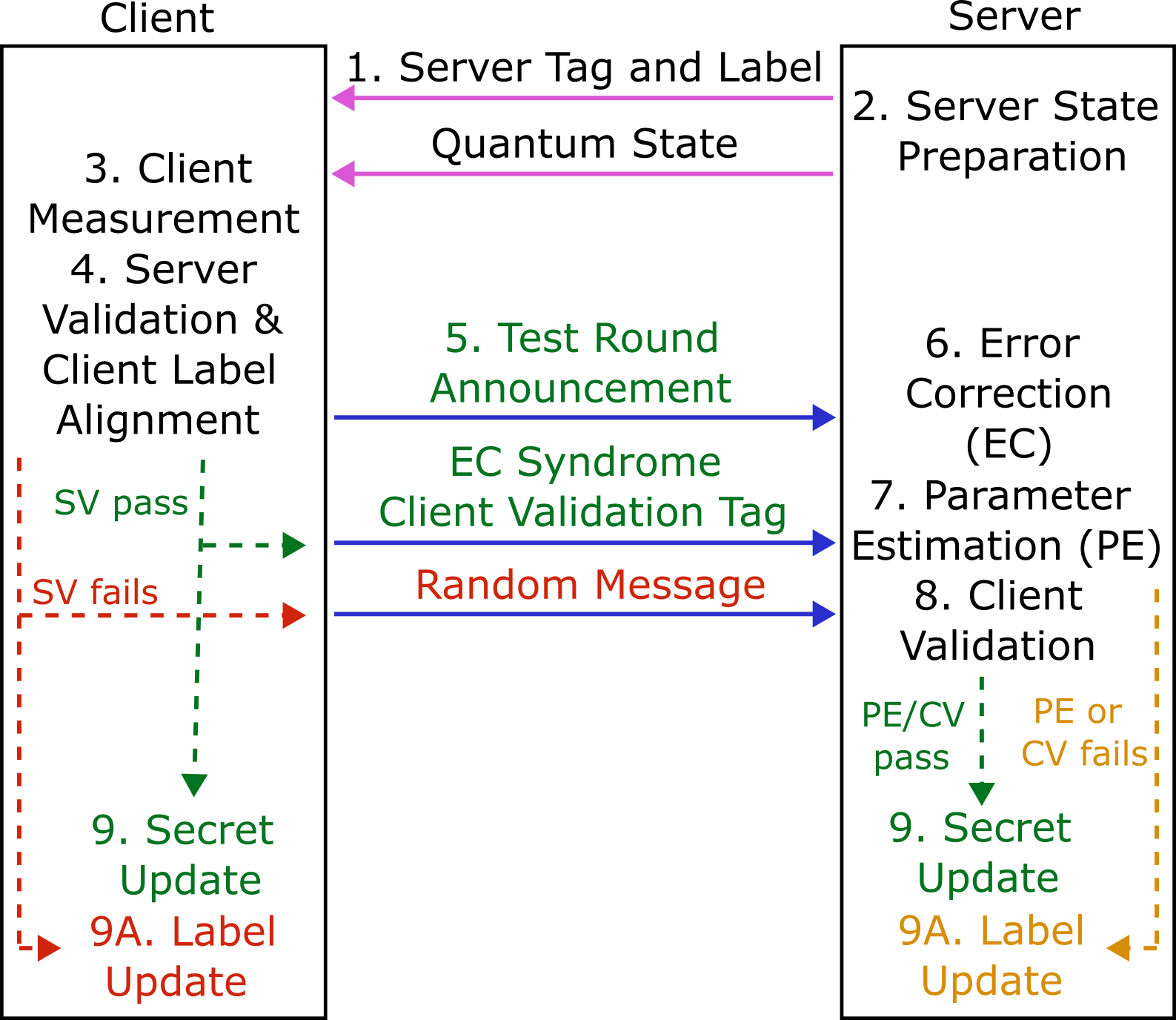}
    \caption{Summary of step-efficient CA protocol with two communication steps. In the first communication step, the server sends quantum states to the client, alongside its tag for validation by the client. After validation and client's measurement, the response would consist of the test round results, client validation tag for server to validate the client's identity, and information for error correction. At the end of the protocol, both parties would separately decide whether to update secrets or the label (indicating a protocol failure).}
    \label{fig:RoundEfficientCA}
\end{figure}

Inspired by the client authentication protocol in Ref.~\cite{Fehr2017_SID}, we developed a step-efficient client authentication (CA) protocol requiring only two communication steps.
The key modifications made to the proposed QAKE protocol with pseudorandom basis selection are:
\begin{enumerate}
    \item Removing server validation since they are not required for client authentication. The validation is instead performed at the start of the next session, where it is required to ensure updated secrets are secure.
    \item Replacing label agreement with the server utilising his own index $\alpha'$ for the session. Any index mismatch would be picked up during client or server validation, and would be corrected with an additional session.
\end{enumerate}
Incorporating these changes, alongside compression of communication rounds, leads to a protocol with two communication steps summarised in Fig.~\ref{fig:RoundEfficientCA}.\\

The client and server pre-share secrets and functions similar to the AKE protocol, except authentication key $K_2$ and hash function $h_2$, which are not necessary in this protocol.
The protocol is described in Protocol 3.

\subsection{Security Definition}

The security definition for the protocol follows from standard client authentication security requirements.
We begin with the first condition of robustness, which guarantees that the protocol should always succeed when a valid client approaches the server.
We note that since the protocol has no label agreement step, any mismatch in indices requires 2 sessions to resolve.
However, since index mismatch is unlikely in the honest case, we relax the robustness requirement to ensure that authentication passes in 2 sessions (similar to how clients has to try a second time to login when the first login fails).
As such, we formally define this two session robustness as
\begin{protocol}{Step-Efficient Client Authentication}
\textit{Goal.} Server authenticates client.
\begin{enumerate}
    \item \textbf{Server's Tag and Label}: Server sends its authentication tag $T_{\SV,\alpha'}$ and index $\alpha'$ to the client.
    \item \textbf{Server State Preparation}: The server generates a $n_{rnd}$-bit basis string using the basis seed, $\theta_{\alpha'}'=g(\tilde{\theta}_{\alpha'}')$, randomly chooses a $n_{rnd}$-bit string $X$ and a $n_{rnd}$-trit string $v\in\{0,1,2\}^{n_{rnd}}$ according to probability distribution $p_v$. The server then sends $n_{rnd}$ quantum systems $Q_i$ in phase-randomised coherent BB84 states $\left\{\rho_{Q_i}^{\theta_i',X_i,\mu_{v_i}}\right\}_{i\in[1,n_{rnd}]}$, with basis $\theta_i'$, bit value $X_i$, and intensity $\mu_{v_i}$, to the client, acting as the ``challenge".    
    \item \textbf{Server Validation and Client Label Alignment}: Based on the received label $\alpha_\rr'$, the client aligns its label. If $\alpha_\rr'>\alpha$, the client updates $\beta=\alpha_\rr'$ and sets $\DC=1$. If $\alpha_\rr'=\alpha$, the server sets $\beta=\alpha$ and checks if $\tilde{T}_{\SV,\alpha}=T_{\SV,\alpha',\rr}$. If it matches, the client sets $\DC=1$. Otherwise, the client sets $\DC=0$ and replaces any further responses with a random string. After the alignment is complete, the client sends $\DC$ and $\alpha$ to the server.
    \item \textbf{Client Measurement}: The client measures subsystems $Q_i$ using basis $\theta_{\beta}=g(\tilde{\theta}_{\beta})$, and records outcome $X_i'$. If the client detects no clicks, it declares $X_i'=\perp$. If multiple clicks are detected, the bit value $X_i'\in\{0,1\}$ is randomly selected. 
    \item \textbf{Test Round Announcement}: The client records the detection rounds, $P=\{i:X_i'\neq\perp\}$, and randomly splits it into two sets, $\abs{P_1}=\lceil f_{P_1}\abs{P}\rceil$. The client announces $P_1$, $P_2$ and $X_{P_1}'$, and the server receives the announcement $P_{1,r}$, $P_{2,r}$ and $X_{P_1,r}'$.
    \item \textbf{Error Correction}: 
    The client generates a syndrome $syn=\fsyn(X_{P_2}')$ and forwards it to the server. The server receives the syndrome $syn_\rr$ and computes the corrected bit string $\hat{X}_{P_2}'=\fsyndec(X_{P_{2,\rr}},syn_\rr)$.
    \item \textbf{Client Validation}: The client generates a tag $T_{\CV}=h_1(\Kh_1,X_{P_1}'||P_1||P_2||X_{P_2}'||syn)\oplus \KOTP_{1,\beta}$ and forwards it to the server. The server receives tag $T_{\CV,\rr}$ and generates verification tag $\tilde{T}_{\CV}=h_1(\Kh_1,X_{P_{1,\rr}}||P_{1,\rr}||P_{2,\rr}||\hat{X}_{P_2}'||syn_\rr)\oplus \KOTP_{1,\alpha'}$ and checks if $T_{\CV,\rr}=\tilde{T}_{\CV}$. If the tags matches, the server validates the client and output $D_{\CV}=1$, otherwise, he sets $D_{\CV}=0$.
    \item \textbf{Parameter Estimation}: The  server checks if: (1) the number of test rounds matches the expected number $\abs{P_{1,\rr}}=\lceil f_{P_1}\abs{P}\rceil$, (2) number of rounds in set $P_1$ exceeds a threshold, $\abs{P_1}\geq P_{1,\tol}$, (3) the lower bound on single-photon events in sets $P_{2,\rr}$ and $P_{1,\rr}$ (estimated via decoy state analysis) exceeds a threshold, $\hat{N}^{\LB}_{P_{2,\rr},1}\geq N_{P_{2,\rr},1}^{\tol}$, $\hat{N}^{\LB}_{P_{1,\rr},1}\geq N_{P_{1,\rr},1}^{\tol}$, (4) the upper bound on single-photon bit-error rate in set $P_{1,\rr}$ (estimated via decoy-state analysis) falls below a threshold, $\hatebitt{P_{1,\rr},1}^{\UB}\leq \ebitt{1,\tol}$, and (5) the upper bound on bit error rate in set $P_{2,\rr}$ (estimated via Serfling bound) falls below a threshold, $\hatebitt{P_{2,\rr}}^{\UB}\leq \ebitt{\tol}$. If all checks pass, the server sets $\DPE=1$, otherwise it sets $\DPE=0$.
    \item \textbf{Secret Update and Label Update}: The server compares the index received $\alpha_\rr$, $\DC$ and its own index $\alpha'$. The server first updates $\alpha_\rr=\alpha_\rr+1$ if $\DC=0$. If updated $\alpha_\rr>\alpha'$, the server sets $\alpha'=\alpha_\rr-1$ and sets $D_{\SA}=0$, otherwise he sets $D_{\SA}=1$ (aligned indices). The client decides whether to generate the server's tag and update the basis seed based on $\DC$, while the server decides on the success of the authentication via $\FS=\DPE\land D_{\CV}\land D_{\SA}$. When the client (resp. the server) decides to generate the tag and update the basis seed, i.e. $\DC=1$ (resp. $\FS=1$), it performs privacy amplification $\tilde{T}_{\SV,\beta}||\tilde{\theta}_{\beta}=\hPA(R,X_{P_2}')$ (resp. $T_{\SV,\alpha'}||\tilde{\theta}_{\alpha'}'=\hPA(R,\hat{X}_{P_2}')$) and sets $\alpha=\beta$. If basis seed update and server tag generation is not performed, then the client (resp. the server) would update its label, $\alpha=\beta+1$ (resp. $\alpha'=\alpha'+1$).
\end{enumerate}
\end{protocol}
\begin{definition}[$\varepsilon_{2,\rob}$-2-session-robustness]
    A CA protocol is $\varepsilon_{2,\rob}$-2-session-robustness if it passes with high probability in the absence of any adversary in either of the next two rounds, i.e. it fails to pass both rounds with low probability
    \begin{equation*}
        \Pr[F_{\textup{S},j}=0,F_{\textup{S},j+1}=0]\leq \epsrob,
    \end{equation*}
    where $F_{\textup{S},i}$ refers to the server's label in the $i$-th round, and $j\in[1,n_{sess}-1]$, where $n_{sess}$ is the maximum number of sessions.
\end{definition}
We note here that there can be other ways of defining the robustness condition.
One such method is to assume that all prior sessions of the protocol are performed in the absence of an adversary.
This is unlikely to result in any index mismatch at the current session, and thus robustness can be guaranteed with high probability with a single session.\\

The second security condition in client authentication requires that a server interacting with an invalid client (or adversary) return an authentication failure with high probability.
As such, when the authentication passes, the server can be highly confident that the client is valid.
Therefore, we formally define the client authentication condition as
\begin{definition}[$\epsCV$-client validation]
    An CA protocol has $\epsCV$-client validation if the authentication fails with high probability when the adversary attempts to impersonate the client, i.e. \begin{equation*}
        \Pr[\FS=1|\DC=\phi]\leq\epsCV.
    \end{equation*}
\end{definition}
While there is no imposition on the secrets of the protocol, they would obey additional conditions for future sessions to remain secure.
This separate security condition has the same definition as shared secrets privacy in QAKE, albeit with different ideal input and output states.

\subsection{Overall Protocol Security}
\label{app:CAProtocolSec_OverallSec}

We can similarly define a $n_{sess}$-session protocol of CA, where all sessions has to satisfy the security condition, with each session satisfying some $\varepsilon_{\secc,i}$-security condition defined by the trace distance from the ideal output state with $\Pr[\FS=1,\DC=\phi]=0$.
We follow the same idea from Appendix.~\ref{app:Multi_to_Single_Red} 
to reduce the multi-session analysis to single session, where additional shared secret preservation is necessary in the intermediate steps.
However, since there is only a single security condition, client authentication, we shall examine both conditions together, defining the overall security as
\begin{equation}
    \Delta(\rho_{\DC\FS LSE},\rhoidealint_{\DC\FS LSE})\leq\epssecint,
\end{equation}
and it remains to identify the set of ideal input and output states.\\

The main difference in the identification of ideal input/output states for this protocol is the ability of the adversary to gain partial information of secrets and partial control of the client's action.
Firstly, the client's tag can almost always be generated since there is no server check during the session itself.
This leads to the authentication key $\Kh_1$ and authentication masking key $\KOTP_{1,\beta}$ being correlated with the adversary, although this correlation is via the tag $T_{\CV}$.
This is without a corresponding trace of the authentication masking key in the session since the client's index is not updated, though one expects the corresponding trace to occur when the client performs the server check in later sessions.
Secondly, the generation of the server's tag verification value by the client, $\tilde{T}_{\SV,\alpha}$, can be sometimes influenced by the adversary when it has prior knowledge of the basis generation seed.
However, the adversary would not have direct access to this value since it is not announced.
To account for these two leakages/controls available to the adversary, we define two oracles, the first $\bigepsilon_{\Kh_1\KOTP_{1,\beta}M_{\CV}\rightarrow T_{\CV}}$, allows the adversary to choose any message $M_{\CV}$ and have the corresponding client tag generated for him.
The second $\bigepsilon_{RM_{\SV}\rightarrow \tilde{T}_{\SV,\alpha}}$ allows the adversary to select any message $M_{\SV}$, and generate the server authentication check tag for the client, who would utilise it to perform checks in later sessions when necessary.
Note that for simplicity, we will label the oracles as $\bigepsilon_{\Kh_1\KOTP_{1,\beta}E\rightarrow E'}$ and $\bigepsilon_{RE\rightarrow\tilde{T}_{\SV,\alpha}E'}$ since the adversary would decide the messages according to its system $E$, and $T_{\CV}$ is implicitly included in $E'$ since it is directly accessible to the adversary.\\

To monitor the secrets and oracles accessible, we introduce a new label $\alphaU$, which indicates the client's index if the server authentication is performed in the same session.
$\alphaU$ is updated as follows: (1) if $\DC=\phi$ within a session, leave $\alphaU$ unchanged, (2) if $\alpha_\rr'\geq\alphaU$, set $\alphaU=\alpha_\rr'$, and (3) if $\DC=1$ and $\FS\neq 1$, set $\alphaU=\alphaU+1$. 
The second scenario occurs when the client's index is updated to an index that may not fail in the next session, while the third scenario requires that the server not pass server validation since server did not participate/did not select a server authentication tag that matches the client's.
A general input state can be defined as a mixture of two variants of states 
$\rho^{in}=\sum_{jj'}p_{jjj'}\rhoinn{jjj'}_{LSE}+p_{j,j+1,j'}\rhoinn{j,j+1,j'}_{LSE}$, where
\begin{equation}
\begin{gathered}
    \rhoinn{jjj'}_{LSE}=\dyad{jjj'}_{\alpha\alphaU\alpha'}\otimes\rhoinn{jjj'}_{SE}\\
    \rhoinn{jjj'}_{SE}=\tau_{R\KOTP_{1,[j,m]}}\otimes\tilde{\tau}_{\tilde{\theta}_{[j_{\max},m]}\tilde{\theta}_{[j_{\max},m]}'}\otimes\bigepsilon_{\Kh_1\KOTP_{1,[j',j-1]}E'\rightarrow E}(\tau_{\Kh_1\KOTP_{1,[j',j-1]}}\otimes\rho_{E'}),
\end{gathered}
\end{equation} $j_{\max}=\max\{j,j'\}$, and $E'$ contains the lost secrets $\tilde{\theta}_{[1,j_{\max}-1]}\tilde{\theta}_{[1,j_{\max}-1]}'$, and the state
\begin{equation}
\begin{gathered}
    \rhoinn{j,j+1,j'}_{LSE}=\dyad{j,j+1,j'}_{\alpha\alphaU\alpha'}\otimes\rhoinn{j,j+1,j'}_{SE}\\
    \rhoinn{j,j+1,j'}_{SE}=\tau_{\KOTP_{1,[j+1,m]}}\otimes\tilde{\tau}_{\tilde{\theta}_{[\jU_{\max},m]}\tilde{\theta}_{[\jU_{\max},m]}'}\otimes\bigepsilon_{R\Kh_1\KOTP_{1,[j',j]}E'\rightarrow \tilde{T}_{\SV,j}E}(\tau_{R\Kh_1\KOTP_{1,[j',j]}}\otimes\rho_{E'}),
\end{gathered}
\end{equation}
and 
$\jU_{\max}=\max\{j+1,j'\}$, and the oracle access for $\bigepsilon_{\Kh_1\KOTP_{1,[j',j-1]}E'\rightarrow E}$ and $\bigepsilon_{R\Kh_1\KOTP_{1,\jOTP}E'\rightarrow \tilde{T}_{\SV,j}E}$
can have any order of $\bigepsilon_{\Kh_1\KOTP_{1,\beta}E\rightarrow T_{CV}E'}$ and $\bigepsilon_{RE\rightarrow\tilde{T}_{SV,\alpha}E'}$ in general.
As such, we can link the two input state
\begin{equation}
    \rhoinn{j,j+1,j'}_{SE}=\bigepsilon_{RE'\rightarrow\tilde{T}_{\SV,j}E}\circ\bigepsilon_{\Kh_1\KOTP_{1,j}E''\rightarrow E'}(\rhoinn{j+1,j+1,j'}_{SE''}),
\end{equation}
which would be useful to simplify the analysis.\\

The set of ideal output state $\vars{S}_{out}'$ has form $\rhoout=\sum_{j,\jU\in\{j,j+1\}, j'}p_{j\jU j'}\rhooutt{j\jU j'}$ with the form 
\begin{equation}
\begin{split}
    \rhooutt{j\jU j'}=&\dyad{\phi0}_{\DC\FS}\otimes\sum_{\tilde{j}'\geq j'+1}p_{\phi0,\tilde{j}'}\otimes\sigma^{\phi0,j\jU,\tilde{j}'}_{LSE}+p_{0\phi}\dyad{0\phi}_{\DC\FS}\otimes\sigma^{0\phi,j+1,j+1,j'}_{LSE}\\
    &+\dyad{1\phi}_{\DC\FS}\otimes\sum_{\tilde{j}\geq j} p_{1\phi,\tilde{j}}\sigma^{1\phi,\tilde{j},\tilde{j}+1,j'}_{LSE}+\dyad{00}_{\DC\FS}\otimes\sum_{\tilde{j}'\geq j'+1}p_{00,\tilde{j}'}\sigma^{00,j+1,j+1,\tilde{j}'}_{LSE}\\
    &+\dyad{10}_{\DC\FS}\otimes\sum_{\tilde{j}\geq j,\tilde{j}'\geq j'+1} p_{10,\tilde{j}\tilde{j}'}\sigma^{10,\tilde{j},\tilde{j}+1,\tilde{j}'}_{LSE}+p_{11}\dyad{11}_{\DC\FS}\otimes\sigma^{11,j'j'j'}_{LSE},
\end{split}
\end{equation}
where 
$\sigma_{LSE}^{\fA\fB,jj_Uj'}\in\vars{S}_{in}$.
Note that unlike the QAKE protocol where $\DC=1$ is disallowed when $\FS=\phi$, the lack of a check for server validation in the same session allows an adversary to request $T_{\SV,\alpha'}$ from the server in a separate session and pass the server validation, i.e. $\DC=1$.
However, the mismatch of $\alpha$ and $\alphaU$ guarantees that the next session with client involvement would result in $\DC=0$, and an update of $\alpha$ to indicate no confidence in the updated secrets.\\

By definition, the output state belongs to $\vars{S}_{in}$, thereby satisfying condition 3 in Thm.~\ref{thm:Multi_to_Single_Red} when reducing to single round security.
Condition 1 is trivially satisfied since all secrets remain private and indices begin with 1. 
As such, the security analysis reduces to the single-session security analysis for a single input component.
The security can be summarised as

\begin{theorem}
    The CA protocol is secure with security parameter
    \begin{equation*}
    \begin{split}
        \epssecint=& 4\sqrt{2\epsph}+2\varepsilon_2+6\epsmatch+5\epsMACa+2\epsds+\frac{1}{2}\times 2^{-\frac{1}{2}[N_{P_2,1}^{\tol}[1-\hbin(\ephtol')]-l_{\tilde{\theta}}-\log_2\abs{\vars{T}_{\SV}}\abs{\vars{T}_{\CV}}-\leakEC]}\\
        &+\sqrt{\abs{\vars{T}_{\CV}}\epsMACa-1+2^{\log_2\abs{\vars{T}_{\CV}}+\log_2\left(\frac{2}{\varepsilon_2}+1\right)-N_{P_2,1}^{\tol}[1-\hbin(\ephtol')]+\leakEC}}+\frac{1}{\abs{\vars{T}_{\CV}}}+\frac{1}{\abs{\vars{T}_{\SV}}},
    \end{split}
    \end{equation*}
    where $\epsmatch=2^{-P_{1,\tol}[1-h_{bin}(\ebittol)]}$, $\epsph$ is related to the smoothing parameter for decoy state with PRNG basis (Thm.~\ref{thm:BB84PRNG_MinEnt}), $\epsds$ accounts for errors associated with decoy state analysis, $\abs{\vars{T}_{\CV}}$ and $\abs{\vars{T}_{\SV}}$ are the size of the respective tags, $\leakEC$ is the information leakage from error correction, $N_{P_2,1}^{\tol}$ is the minimum tolerated number of single photon rounds in set $P_2$, $\ephtol'$ is the phase error tolerance inclusive of the correction due to Serfling bound and $\varepsilon_2$ is a variable to optimise over.
\end{theorem}

\begin{proof}[Sketch Proof]
We need to analyse the security for the two different input states $\rhoinn{jjj'}_{LSE}$ and $\rhoinn{j,j+1,j'}_{LSE}$.
In Thm.~\ref{thm:CAUnequalalpha}, we show that the security of $\rhoinn{j,j+1,j'}_{LSE}$ can be related to the security of $\rhoinn{jjj'}_{LSE}$.
The key observation for this proof is that only in one specific scenario where $\alpha_\rr'\geq j+1$ and $j\neq j'$, i.e. the adversary forces the client's index to update to one $\geq j+1$, the client would proceed in the protocol using the updated index -- and that matches with the case where the client begins with an index $j+1$ and a restricted adversary.
In all other cases, we have that the authentication fails with high probability, $\FS=0$, (argued in Thm.~\ref{thm:2RoundCADS}) and the output state matches the ideal case closely.\\

With the reduction, focus is placed on the input state $\rhoinn{jjj'}_{LSE}$.
The analysis can be broken in different scenarios based on the $(\DC,\FS)$ values.
Thm.~\ref{thm:2RoundCADSphi} presents the case when the adversary tries to cheat the server without the client, $\DC=\phi$, for which the security is a consequence of Thm.~\ref{thm:2RoundCADS} ensuring $\FS=0$ with high probability.
The case where the adversary tries to gain information from the client, $\FS=\phi$, is shown in Thm.~\ref{thm:2roundsFSphi}.
Since the adversary can in general trigger the client to proceed with the protocol, leakage of the basis seed, the announcement of $T_{\CV}$ and providing control to $\tilde{T}_{\SV,\tilde{j}}$ follow.
However, this is exactly the situation accounted for as part of the input state definition for $\alphaU=\alpha+1$.\\

For cases where neither $\DC$ nor $\FS$ are $\phi$, both client and server are involved in the protocol.
Similar to earlier analysis in Sec.~\ref{app:PRNG_QAKE_idealised_prot}, we can simplify the analysis by defining idealised version of $\DPE$ and $\DC$, with small penalty shown in Eqn.~\eqref{eqn:CA_idealised_prot}.
With this change, when $\DC=0$, Thm.~\ref{thm:CA_Proof_MITM_00} shows that similar to the $\DC=\phi$ case, $\tilde{F}_S=0$ with high probability, and the leakage of $T_{\SV,j'}$ and $\tilde{\theta}_{j'}$ can be accounted for by the updated index ($j'$ index no longer utilised).
When $D_C=1$ and $\tilde{F}_S=0$ ($F_S$ but with ideal checks), the adversary can easily exploit the protocol to gain access to the oracle generating $T_{\AV}$ and have control over the generation of $\tilde{T}_{\SV,\alpha}$.
However, as shown in Thm.~\ref{thm:CA_Proof_MITM_10}, this is exactly accounted for when $\alphaU$ is updated such that $\alphaU=\alpha+1$.
Finally, Thm.~\ref{thm:CA_Proof_MITM_11} shows that in the case of $D_C=\tilde{F}_S=1$, the ideal checks guarantee that $X_{P_2}'$ has high min-entropy based on Thm.~\ref{thm:BB84PRNG_MinEnt}, thereby allowing the authentication key, authentication masking keys and $R$ to remain secret. 
\end{proof}

\subsubsection{Security for $\alphaU\neq\alpha$ input state}

Before beginning our analysis, let us define a useful theorem regarding the value of $\FS$ when $\DC\neq 1$.
\begin{theorem}
\label{thm:2RoundCADS}
    For any ideal input state $\sigma_{LSE}\in\vars{S}_{in}$ where the adversary has no access to the basis seed $\tilde{\theta}_{j'}$ and oracle $\bigepsilon_{\Kh_1\KOTP_{1,j'}M_{\CV}\rightarrow T_{\CV}}$, 
    \begin{equation*}
        \Pr[\FS=1|\DC\neq1]\leq\frac{1}{\abs{\vars{T}_{\CV}}}
    \end{equation*}
    On the other hand, if the oracle and basis seed is accessible 
    \begin{equation*}
        \Pr[\FS=1|\DC\neq1]\leq\epsMACa+\epsmatch,
    \end{equation*}
    where the matching probability is
    \begin{equation*}
        \epsmatch=2^{-\abs{P_{1,\tol}}[1-\hbin(\ebittol)]},
    \end{equation*}
    noting that $\abs{\tilde{P}_1}\geq P_{1,\tol}$, where $\tilde{P}_1$ is the set $P_1$ as decided in the message $M_{\CV}$.
\end{theorem}
\begin{proof}
Let us first consider the case where the oracle and basis seed are inaccessible. 
While the adversary may have some control of the message $M_\rr$ sent to the client, since the adversary has no knowledge of $\Kh_1$ and $\KOTP_{1,j'}$, 
the uniformity property of the strongly 2-universal hash function ensures that
\begin{equation}
    \Pr[\FS=1|\DC\neq1]
    \leq\Pr[T_{\CV,\rr}=h_1(\Kh_1,M_\rr||f(M_\rr,X))\oplus \KOTP_1]
    \leq\frac{1}{\abs{\vars{T}_{\CV}}},
\end{equation}
where $f(M_{\rr},X)$ labels the other terms required in the hash described in step 7 of the main protocol, for any $T_{\CV,\rr}$ selected by the adversary.\\

When access to the basis seed and oracle are available, 
the adversary is able to perform a call to the oracle for any message $M_{\CV}$, receive the outcome $T_{\CV}$, before given access to $X$, and deciding on the message $M_\rr$ and tag $T_{\CV,\rr}$ to send to the server.
Crucially, this forces the adversary to commit to a $\tilde{P}_1$ set and $\tilde{X}_{\tilde{P}_1}$ bit values when performing the oracle call.
We can simplify the probability
\begin{equation}
\begin{split}
    &\Pr[D_{\CV}\land \DPE=1]\\
    \leq&\Pr[\DPE=1,(M_\rr||f(M_\rr,X))\neq M,T_{\CV,1,\rr}=\tilde{T}_{\CV,1}]+\Pr[\DPE=1,(M_\rr||f(M_\rr,X))= M,T_{\CV,1,\rr}=\tilde{T}_{\CV,1}]\\
    \leq&\Pr[T_{\CV,1,\rr}=\tilde{T}_{\CV,1}|(M_\rr||f(M_\rr,X))\neq M]+\Pr[\DPE=1,(M_\rr||f(M_\rr,X))= M]\\
    \leq&\epsMACa+\Pr[wt(X_{P_{1,\rr}}'\oplus X_{P_{1,\rr}})\leq\abs{P_{1,\rr}}\ebittol,\tilde{P}_1=P_{1,\rr},X_{P_{1,\rr}}'=\tilde{X}_{\tilde{P}_1},\abs{P_{1,\rr}}\geq P_{1,\tol}]\\
    \leq&\epsMACa+\Pr[wt(\tilde{X}_{\tilde{P}_1}\oplus X_{\tilde{P}_1})\leq\abs{\tilde{P}_1}\ebittol,\abs{\tilde{P}_1}\geq P_{1,\tol}]\\
    \leq&\epsMACa+\max_{\abs{\tilde{P}_1}\geq P_{1,\tol}}\sum_{i=0}^{\abs{\tilde{P}_1}\ebittol}\begin{pmatrix}     \abs{\tilde{P}_1}\\ i \end{pmatrix} 2^{-\abs{\tilde{P}_1}}\\
    \leq&\epsMACa+2^{-P_{1,\tol}[1-\hbin(\ebittol)]}.
\end{split}
\end{equation}
The first inequality splits the probability based on the equality of the messages, and the second line simplifies them by removing some conditions.
The third line notes that by the strongly 2-universal property of the hash function, with the keys not known to the adversary except via $T_{\CV}$, the probability that the adversary can guess the correct tag for input $(M_\rr||f(M_\rr,X))$ is bounded by $\epsMACa$.
The third and fourth lines further removes conditions on $\DPE=1$ and the messages being equal, and focuses mainly on terms we can use to bound the probability, where $\tilde{P}_1$ and $\tilde{X}_{\tilde{P}_1}$ are the corresponding terms from the message $M$.
The fifth line expands the probability of the bound on the Hamming weight, noting that the $X$ is randomly generated and uncorrelated to $\tilde{X}_{\tilde{P}_1}$, which has to be decided before $X$ is accessible to the adversary.
The final line uses the property of the binomial coefficient to simplify, and noting that the maximum value corresponds to $\abs{\tilde{P}_1}=P_{1,\tol}$. 
It is clear that the result upper bounds $\Pr[\FS=1|\DC\neq 1]$.
\end{proof}

\begin{figure}
    \centering
    \includegraphics[width=0.7\linewidth]{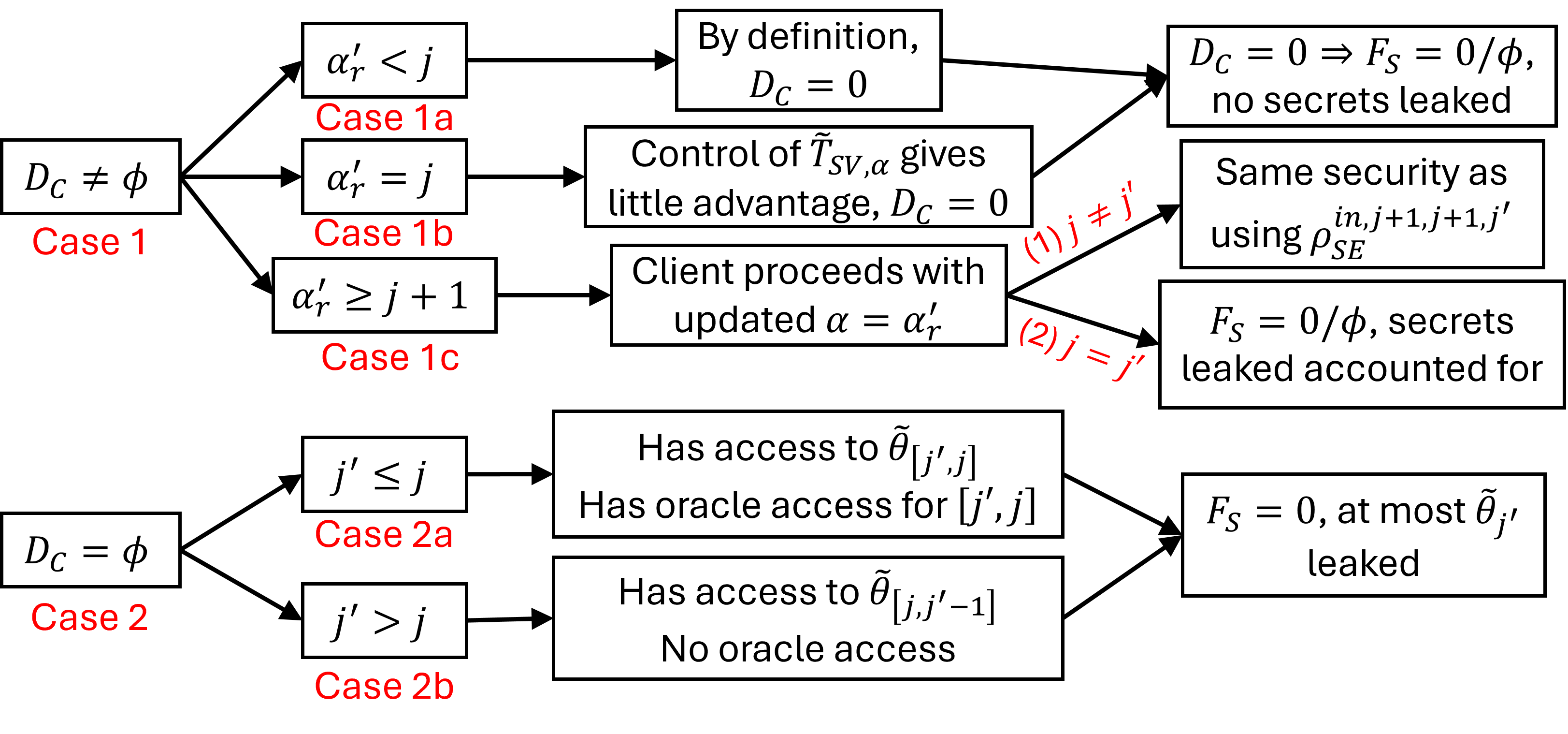}
    \caption{Schematic of the proof of Thm.~\ref{thm:CAUnequalalpha}.}
    \label{fig:CAUnequalalpha}
\end{figure}

Here, we provide a reduction of the security for input state variants with $\alphaU=\alpha+1$ to the security for input state variants with $\alphaU=\alpha$, the latter of which we analyse in the later sections.
\begin{theorem}
\label{thm:CAUnequalalpha}
    Consider the CA protocol $\vars{P}$.
    Suppose that for any valid input state $\rhoinn{j+1,j+1,j'}_{LSE}$ and adversarial strategy, there exists an ideal output state $\rhoout\in\vars{S}_{out}'$ such that
    \begin{equation*}
    \begin{gathered}
        \Delta(\vars{P}(\rhoinn{j+1,j+1,j'}_{LSE}),\rhoout)\leq\epssec.
    \end{gathered}
    \end{equation*}
    Then, for any valid input state $\rhoinn{j,j+1,j'}_{LSE}$, there exists an output state $\rhoout\in\vars{S}_{out}'$ such that
    \begin{equation*}
        \Delta(\vars{P}(\rhoinn{j,j+1,j'}_{LSE}),\rhoout)\leq\epssec+\frac{1}{\abs{\vars{T}_{\CV}}}+\frac{1}{\abs{\vars{T}_{\SV}}}+2\epsMACa+2\epsmatch.
    \end{equation*}
\end{theorem}
\begin{proof}
The protocol consists of three general steps:
\begin{enumerate}
    \item Server's state generation and message.
    \item Client's message receipt and response.
    \item Server and Client post-processing.
\end{enumerate}
where WLOG, the third step will always occur at the end of the protocol since it does not contain additional outputs to the adversary.
Moreover, since the server's state generation and message requires no external inputs, the adversary will find no advantage in delaying it.
As such, we can always take the protocol as one which follows the three step process.
Fig.~\ref{fig:CAUnequalalpha} provides a schematic of the proof for the different cases discussed below.\\

Let us consider first the case where the client is involved in the protocol, $\DC\neq\phi$ (case 1).
Furthermore, let us consider the case where $\alpha_{\rr}'<j$ (case 1a), the client sets $\DC=0$ by the protocol.
If instead $\alpha_{\rr}'=j$ (case 1b), the decision would depend on $\tilde{T}_{\SV,j}$ and $T_{\SV,\alpha',r}$.
Since the adversary has oracle access for $R$, it is able to choose a message $M_{\SV}$ and have the client generate $\tilde{T}_{\SV,j}=\hPA(R,M_{\SV})$.
Since $R$ itself is inaccessible to the adversary, the uniformity property of the hash function means that
\begin{equation}
    \Pr[\DC=1|\alpha_\rr'=j]=\Pr[\tilde{T}_{\SV,j}=T_{\SV,\alpha',r}]    \leq\frac{1}{\abs{\vars{T}_{\SV}}},
\end{equation}
and the client sets $\DC=0$ with high probability.
In both cases, if $\FS=\phi$, then no secrets are leaked and shared secrets privacy is maintained.
Otherwise, when $\DC=0$, then from Thm.~\ref{thm:2RoundCADS}, whether an oracle is accessible or not, $\Pr[\FS=1|\DC=0]\leq\epsMACa+\epsmatch$, i.e. $\FS=0$ with high probability.
When $\DC=0$ and $\FS=0$, only $\tilde{\theta}_{j'}$ may in general be leaked to the adversary from the server state prepared.
If $j+1> j'$, this information has already been assumed to be accessible to the adversary by the form of the input state.
If instead $j+1\leq j'$, since $\alpha'$ is increased by 1 when the authentication fails, $\FS=0$, and this leakage is allowed (i.e. the ideal input state includes this leakage).
Therefore, one is able to find an ideal output state with trace distance
\begin{equation}
    \Delta(\vars{P}(\rhoinn{j,j+1,j'})_{\land \DC\neq\phi},\rhoout_{\land \DC\neq\phi})\leq\frac{1}{\abs{\vars{T}_{\SV}}}+\epsMACa+\epsmatch
\end{equation}
from the actual output state.\\

In the case where $\alpha_\rr'\geq j+1$ (case 1c), the client updates $\alpha$ to $\alpha_\rr'$ and set $\DC=1$.

Note that this implicitly means that the adversary does not utilise the $\bigepsilon_{RM_{\SV,j}\rightarrow \tilde{T}_{\SV,j}}$ oracle during this session and future sessions.
After updating $\alpha$, the client proceeds with the protocol using the updated $\alpha$, and its action would be similar with one where the protocol begins with $\rhoinn{j+1,j+1,j'}_{LSE}$ instead, and for adversarial strategies which ensures $\DC=1$.
The main difference between starting with $\rhoinn{j,j+1,j'}_{LSE}$ and $\rhoinn{j+1,j+1,j'}_{LSE}$ then would be the access to the oracle $\bigepsilon_{\Kh_1\KOTP_{1,j}M_{\CV}\rightarrow T_{\CV,j}}$ if $j\geq j'$.
If we consider (1) $j\neq j'$, starting with either state makes no difference since $\KOTP_{1,j}$ is not utilised in the client validation step.
This implies that the state $\rhoout\in\vars{S}_{out}'$ for which $\Delta(\vars{P}(\rhoinn{j+1,j+1,j'}_{LSE}),\rhoout)\leq\epssec$ is also close to the case when $\rhoinn{j,j+1,j'}_{LSE}$ is utilised, i.e.
\begin{equation}
    \Delta(\vars{P}(\rhoinn{j+1,j+1,j'}_{LSE\land D_C\neq\phi,\alpha_{\rr}'\geq j+1}),\rhoout_{\land D_C\neq\phi,\alpha_{\rr}'\geq j+1})\leq\epssec.
\end{equation}
In the case where (2) $j=j'$, since the client does not generate $T_{\CV,j}$ during the protocol (recall $\alpha\geq j+1$), the adversary has to rely on the oracle to pass, $\FS=1$.
From Thm.~\ref{thm:2RoundCADS}, the probability of such an event is maximally $\epsMACa+\epsmatch$.
In this case, we therefore have $\FS=0$ or $\FS=\phi$.
Since the client accepts the round, $\DC=1$, additional oracles $\bigepsilon_{\Kh_1\KOTP_{1,\tilde{j}}M_{\CV}\rightarrow T_{\CV,\tilde{j}}}$ and $\bigepsilon_{RM_{\SV}\rightarrow \tilde{T}_{SV,\tilde{j}}}$ for $\tilde{j}$ representing the updated $\alpha$ value are provided to the adversary, and the basis seeds $\tilde{\theta}_{\tilde{j}}$ and $\tilde{\theta}_{j'}$ are leaked.
Since $\FS\neq1$ and $\DC=1$, we have that the final index is $\alpha=\tilde{j}$ and $\alphaU=\tilde{j}+1$, for which this leakage and additional oracle accesses are accounted for in the ideal output state (see form of $\rhoinn{j,j+1,j'}_{LSE}$).
\\

Let us now consider the case where $\DC=\phi$ (case 2), where the client is not involved in the protocol.

When $j'\leq j$ (case 2a), the adversary would have access to $\tilde{\theta}_{[j',j]}$ and oracle access to generate $T_{\CV,i}$ for $i\in[j',j]$.
From Thm.~\ref{thm:2RoundCADS}, this would result in the server rejecting the session with high probability, $\Pr[\FS=1|\DC=\phi]\leq\epsMACa+\epsmatch$.
Therefore, since $\FS=0$ would result in no further leakage except $\tilde{\theta}_{j'}$ which is already accessible to the adversary, we can find an ideal output state close to it.
In the other case where $j'>j$ (case 2b), the adversary has no oracle access.
By Thm.~\ref{thm:2RoundCADS}, the probability of $\FS=1$ is bounded by $\frac{1}{\abs{\vars{T}_{\CV}}}$.
From the server state generation, $\tilde{\theta}_{j'}$ can be leaked.
However, since $\alpha'$ is updated when $\FS=0$ and $j'>j$, $j_{\max}>j'$ and thus the leakage of $\tilde{\theta}_{j'}$ is accounted for in the ideal state.\\

Combining the results, the trace distance can be reduced to $\epssec+\frac{1}{\abs{\vars{T}_{\CV}}}+\frac{1}{\abs{\vars{T}_{\SV}}}+2\epsMACa+2\epsmatch$.
\end{proof}

\subsubsection{One-Sided Attack}

With accounting for all input state variant $\rhoinn{j,j+1,j'}_{LSE}$ complete, we present the security for input state variant $\rhoinn{jjj'}_{LSE}$ where only one party is involved in the protocol.
We begin with the case of $\DC=\phi$,
\begin{theorem}
\label{thm:2RoundCADSphi}
    Consider the CA protocol $\vars{P}$. Then, for any valid input state $\rhoinn{jjj'}_{LSE}$, there exists $\rhoout\in\vars{S}_{out}'$ such that
    \begin{equation*}
        \Delta(\vars{P}(\rhoinn{jjj'}_{LSE})_{\land \DC=\phi},\rhoout_{\land \DC=\phi})\leq \epsMACa+\epsmatch.
    \end{equation*}
\end{theorem}
\begin{proof}
This proof follows closely that of case 2 in Thm.~\ref{thm:CAUnequalalpha}.
Let us consider two scenarios separately: $j\leq j'$ and $j>j'$.
For the case of $j\leq j'$, the adversary has no access to any oracles but has access to $\tilde{\theta}_{[j,j'-1]}$.
From Thm.~\ref{thm:2RoundCADS}, $\Pr[\FS=1]\leq\frac{1}{\abs{\vars{T}_{\CV}}}$ and the authentication fails with high probability.
In terms of secrets leaked, $\tilde{\theta}_{j'}$ is leaked from the server state generation.
However, since $\FS=0$ prompts the server to update index $\alpha'$ and $j'\geq j$, the leakage of $\tilde{\theta}_{j'}$ would have been accounted for in the ideal state since $j_{\max}$ is increased.\\

For the second case of $j\geq j'+1$, the adversary has knowledge of $\tilde{\theta}_{[j',j-1]}\tilde{\theta}_{[j',j-1]}'$, and with oracle access from $\KOTP_{1,j'}$ to $\KOTP_{1,j-1}$, where the relevant values $\tilde{\theta}_{j'}'$ and oracle access for $\KOTP_{1,j'}$ are present.
From Thm.~\ref{thm:2RoundCADS}, the authentication will fail with high probability, $\Pr[\FS=1|\DC=\phi]\leq\epsMACa+\epsmatch$.
As such, no additional secret is leaked since $\tilde{\theta}_{j'}$ is already assumed to be accessible to the adversary prior to the session.
\end{proof}

Another one-sided attack involves $\FS=\phi$, where the adversary can potentially gain access to the two oracle variants for $(\Kh_1,\KOTP_{1,i})$ and $R$, since server authentication checks do not occur within the same session.
\begin{theorem}
\label{thm:2roundsFSphi}
    Consider the CA protocol $\vars{P}$. Then, for any valid input state $\rhoinn{jjj'}_{LSE}$, there exists $\rhoout\in\vars{S}_{out}'$ such that
    \begin{equation*}
        \Delta(\vars{P}(\rhoinn{jjj'}_{LSE})_{\land \FS=\phi},\rhoout_{\land \FS=\phi})=0.
    \end{equation*}
\end{theorem}
\begin{proof}

When the adversary attacks only the client, it is in general possible for the adversary to trigger $\DC=1$ using a tag $T_{\SV,j}$ generated from an earlier session or select a suitable $\alpha_{\rr}'>\alpha$.
Since the client proceeds with the protocol when $\DC=1$, this would result in (1) the client leaking $\tilde{\theta}_{\tilde{j}}$ due to the client measurement step and announcement of $X$ values, where $\tilde{j}$ represents the updated $\alpha$ value, (2) the client announcing a client tag $T_{\CV}$ for which the adversary has some control of the input message, and (3) the adversary being able to influence the server validation tag $\tilde{T}_{\SV,\tilde{j}}$ generated by the client.
Since $\DC=1$ and $\FS\neq 1$, the index $\alphaU$ is updated, i.e. $\alphaU=\tilde{j}+1$, for which the ideal state accounts for all three control/leakage available to the adversary in the form of leaked basis seed and oracle access to $\bigepsilon_{\Kh_1\KOTP_{1,\tilde{j}}\rightarrow T_{\CV}}$ and $\bigepsilon_{RM_{\SV}\rightarrow\tilde{T}_{\SV,\tilde{j}}}$.
In other words, the output state matches an ideal output state with the secrets of the form $\rhoinn{\tilde{j},\tilde{j}+1,j'}_{SE}$.\\

On the other hand, if $\DC=0$, none of the three events occur since the client responds with random strings.
As such, the output state would match an ideal output state in $\vars{S}_{out}'$.
\end{proof}

\subsubsection{Man-in-the-Middle Attack}

The final attack variant to consider is one where both parties are involved.
We begin the analysis for a man-in-the-middle attack by replacing the decoy state parameter estimation with actual parameter checks and the authentication check with a matching check, similar to that in Appendix~\ref{app:PRNG_QAKE_idealised_prot}.
The $\tildeDPE$ decision modification would be the the same (note $\DPE$ can have slight differences), and this would result in the same $2\epsds$ penalty.
The authentication check to be replaced is the client validation step, where we replace $D_{\CV}$ with $\tilde{D}_{\CV}$, defined as
\begin{equation}
    \tilde{D}_{\CV}=\begin{cases} 1 & \substack{\{\alpha>\alpha', M=(\hat{X}_{P_2},P_{2,\rr},P_{1,\rr},X_{P_1,\rr}',S_\rr),T_{CV,\alpha'}=T_{\CV,\rr}\}\\
    \lor\{\beta=\alpha',\DC=1,T_{\CV}=T_{\CV,\rr},(X_{P_2}',P_1,P_2,X_{P_1}',syn)=(\hat{X}_{P_2},P_{2,\rr},P_{1,\rr},X_{P_1,\rr}',syn_\rr)\}} \\ 0 & otherwise\end{cases},
\end{equation}
where $\tilde{D}_{\CV}=1$ includes two scenarios: (1) 
the client validation tag with the index $\alpha'$ comes from the client's response or (2) the client validation tag is generated from the adversary's oracle.
In either case, the same condition of matching message and an untouched transmission ($\DC=1$ ensures a tag is generated for $\beta=\alpha'$) gives $D_{\CV}=1$.
The probability that $D_{\CV}=1$ when $\tilde{D}_{\CV}=0$ would be bounded by $\epsMACa$, since an mismatched message would not allow the adversary to guess the correct tag, by the $\epsMACa$-strong 2-universal property of the hash function.
As such, the overall penalty can be summarised as
\begin{equation}
\label{eqn:CA_idealised_prot}
    \Delta(\vars{P}(\rhoinn{jjj'}_{LSE}),\rhooutt{jjj'})\leq\Delta(\vars{P}'(\rhoinn{jjj'}_{LSE}),\rhopoutt{jjj'})+2\epsMACa+2\epsds,
\end{equation}
and we can simplify the analysis, noting that the output state alters to $\tildeDPE$ and $\tilde{D}_{\CV}$ as well.\\

When both client and server are involved in the protocol, WLOG, the server's state generation is the first step of the protocol, since it does not take in any input from the adversary and thus can always be shifted to before the client's response.
It can also be argued that the client's and server's index and secrets update would be the final steps, since they do not result in any output accessible to the adversary at the end of the session.
The most general description of the protocol thus has fixed protocol steps, namely
\begin{equation}
\begin{split}
    \vars{P}'=&\vars{A}_{RX'P_2\rightarrow \tilde{\alpha}\tilde{T}_{\SV,\tilde{\alpha}}\tilde{\theta}_{\alpha}}\circ\vars{B}_{\alpha'\alpha_\rr\tildeFS Xsyn_\rr R\rightarrow \tilde{\alpha}'T_{\SV,\tilde{\alpha}'}\tilde{\theta}_{\alpha'}'}\circ\vars{C}_{\alpha'\alpha_\rr XX_{P_1,\rr}'P_{1,\rr}P_{2,\rr}syn_\rr T_{\CV,\rr}E''\rightarrow\tildeDPE\tildeFS}\\
    &\circ\bigepsilon_{\DC X_{P_1}'P_1P_2synT_{\CV}E'\rightarrow X_{P_1,\rr}'P_{1,\rr}P_{2,\rr}syn_\rr T_{\CV,\rr}\alpha_\rr E''}\circ\vars{A}_{\alpha \alpha_\rr'T_{\SV,\alpha,\rr}\tilde{T}_{\SV,\alpha}B\tilde{\theta}_{\beta}\Kh_1\KOTP_{1,\beta}\rightarrow \DC X'X_{P_1}'P_1P_2synT_{\CV}}\\
    &\circ\bigepsilon_{T_{\SV,\alpha'}QE\rightarrow \alpha_\rr'T_{\SV,\alpha',\rr}BE'}\circ\vars{B}_{X\tilde{\theta}_{\alpha'}'\rightarrow Q},
\end{split}
\end{equation}
where $\vars{A}$ and $\vars{B}$ refer to quantum channels performed at the server and client respectively and $\bigepsilon$ are quantum channels associated with adversary strategy. 
Note that we collate all decision for $\tilde{F}_S$ and $\tilde{D}_{PE}$ into a single channel $\vars{C}$ separate from the rest of the steps and information for $\tildeDPE$ and $\tildeFS$ (same as $F_S$ except that checks $\tildeDPE$ and $\tilde{D}_{\CV}$ are ``ideal") checks are captured within $E''$, such as the original values of the message sent by the client or the message utilised in the oracles.\\

\begin{figure}
    \centering
    \includegraphics[width=0.65\linewidth]{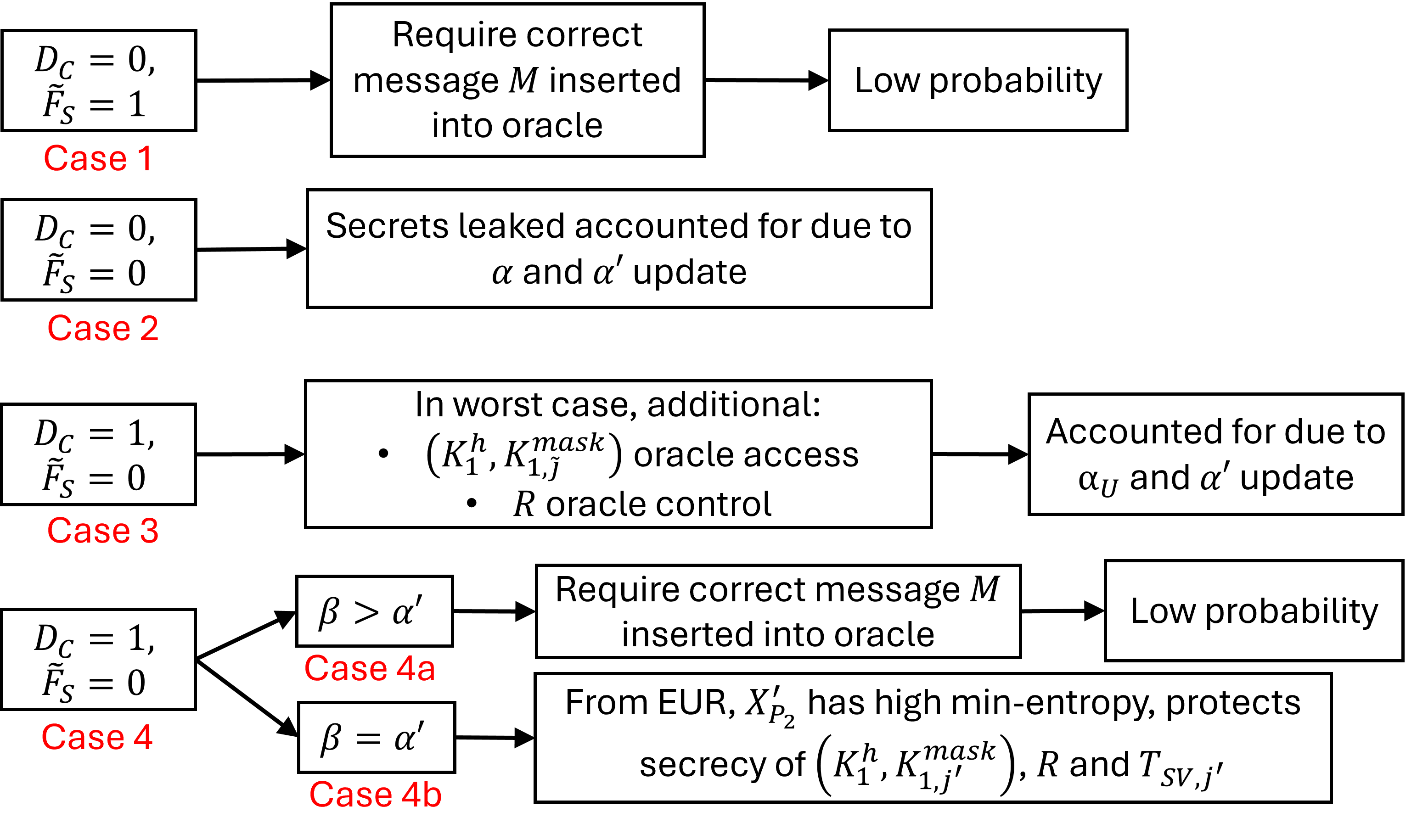}
    \caption{Schematic of proofs of Thm.~\ref{thm:CA_Proof_MITM_00} (case 1 and 2), Thm.~\ref{thm:CA_Proof_MITM_10} (case 3) and Thm.~\ref{thm:CA_Proof_MITM_11} (case 4).}
    \label{fig:CA_Proof_MITM}
\end{figure}

The security analysis for man-in-the-middle attack can then be split into multiple cases.
Fig.~\ref{fig:CA_Proof_MITM} provides a summary of the arguments of the various cases.
The first case we consider for the analysis is the $\DC=0$ case.
\begin{theorem}
\label{thm:CA_Proof_MITM_00}
Consider the CA protocol $\vars{P}$, then there exists $\sigma_{LSE}^{00,j+1,j+1,\tilde{j}'}\in\vars{S}_{in}$ such that
\begin{equation*}
    \Delta\left(\vars{P}'(\rhoinn{jjj'}_{LSE})_{\land D_C=0},\sum_{\tilde{j}'\geq j'+1}p_{00,\tilde{j}'}\sigma_{LSE}^{00,j+1,j+1,\tilde{j}'}\right)\leq\epsmatch.
\end{equation*}
\end{theorem}
\begin{proof}
The trace distance can be split into two components based on the $\tilde{F}_S$ value by triangle inequality, where the $\tilde{F}_S=1$ case (case 1) can be further upper bounded by $\Pr[\tildeFS=1,\DC=0]$.

By definition of $\tilde{F}_S$, we can further upper bound the probability by $\Pr[\tilde{D}_{\CV}=1,\tildeDPE=1,\DC=0]$.
With these conditions, the only case where it can occur is when 
$\alpha>\alpha'$, and the adversary correctly matches the messages.
Following the proof of Thm.~\ref{thm:2RoundCADS}, since $X$ is uncorrelated to the value $\tilde{X}_{\tilde{P}_1}$ committed during the oracle call, the probability that the two values matched would be bounded by $\epsmatch$.\\

For the case of $\tilde{F}_S=1$ (case 2), the definition of $\tilde{D}_{\CV}$ necessarily requires $\alpha_\rr'\leq j$, leading the index update to be $\tilde{\alpha}=j+1$.
The server on the other hand updates based on $\alpha_\rr=\tilde{j}'> j'$ announced to it, otherwise, it updates to $j'+1$.
As such, in general, the output indices are a mixture of $(j+1,j+1,\tilde{j}')$ where $\tilde{j}'\geq j'+1$, matching the index of the $\sigma_{LSE}^{j+1,j+1,\tilde{j}'}$ as required. 
Furthermore, since $D_C=0$ implies the client responds with random strings, the only value that are leaked are $T_{\SV,j'}$ and $\tilde{\theta}_{j'}$.
The updating of the indices to $\alpha=j+1$ and $\alpha'>j'$ means that these values are no longer utilised (i.e. these leaks are already accounted for).
Note that any oracle access to $(\Kh_1,\KOTP_{1,i})$ for indices $i<\tilde{j'}$ can be discarded according to Thm.~\ref{thm:OTPWegCarSec} since the corresponding masking key are no longer utilised, and the basis seeds for the same indices are traced out.
As such, the output state for the secret matches the ideal states $\sigma_{SE}^{j+1,j+1,\tilde{j}'}$, resulting in 0 trace distance.\\

Combining both cases 1 and 2, the resulting trace distance is bounded by $\epsmatch$.
\end{proof}

The next case to consider
is $\DC=1,\tildeFS=0$ (case 3), which can occur without requiring complex attacks from the adversary since the client does not perform an authentication check of the server within the same session.
It would, however, trigger an update of $\alphaU$, indicating that server authentication check would fail in the next session involving a client.
\begin{theorem}
\label{thm:CA_Proof_MITM_10}
    Consider the CA protocol $\vars{P}$. Then, there exists $\sigma_{LSE}^{10,\tilde{j},\tilde{j}+1,\tilde{j}'}\in\vars{S}_{in}$ such that
    \begin{equation*}
        \Delta\left(\vars{P}'(\rhoinn{jjj'}_{LSE})_{\land D_C=1,\tilde{F}_S=0},\sum_{\tilde{j}\geq j,\tilde{j}'\geq j'+1}p_{10,\tilde{j}\tilde{j}'}\sigma_{LSE}^{10,\tilde{j},\tilde{j}+1,\tilde{j}'}\right)=0.
    \end{equation*}
\end{theorem}
\begin{proof}

To achieve $D_C=1$, the adversary is free to select $\alpha_\rr'=\tilde{j}\geq j$ and since $\tilde{F}_S=0$, the adversary is also free to select $\alpha_\rr=\tilde{j}'-1\geq j'$.
As such, the updated indices would be $\alpha=\tilde{j}$, $\alpha'=\tilde{j}'$ with $\alphaU=\tilde{j}+1$ due to the condition $\DC=1$ and $\tilde{F}_S\neq 1$.
This matches with the indices of the ideal state as required.\\

What remains is to prove that the output state matches the ideal state in subsystems $SE$ as well.
For any combination of $j$ and $j'$, we can always consider a more powerful adversary that has access to the client's basis seed $\tilde{\theta}_{\tilde{j}}$, which allows the adversary to prepare quantum states that results in a chosen $X'$ to be received by the client.
This has two results: (1) the adversary is able to partially select the message used to generate $T_{\CV}$ from authentication masking key $\KOTP_{1,\tilde{j}}$ and (2) the adversary is able to partially select the message used by the client to generate $\tilde{T}_{\SV,\tilde{j}}$ (and updated $\tilde{\theta}_{\tilde{j}}$).
We can proceed with providing the adversary more power -- to be able to fully select the messages for both (1) and (2), without requiring that the messages be correlated.
This corresponds with the adversary being provided with an oracle access to $\bigepsilon_{\Kh_1\KOTP_{1,\tilde{j}}\rightarrow T_{\CV}}$ and one access to $\bigepsilon_{RM_{\SV}\rightarrow \tilde{T}_{\SV,\tilde{j}}}$, and where $\tilde{\theta}_{\tilde{j}}$ can be correlated to the adversary (but for which the adversary has no direct access to).
The additional access to these two oracles and correlation of $\tilde{\theta}_{\tilde{j}}$ (part of $E$) are accounted for in the states $\sigma_{LSE}^{\tilde{j},\tilde{j}+1,\tilde{j}'}$ when $\alphaU=\alpha+1$.
Therefore, the output state (in the worst case) corresponds to a state in $\vars{S}_{in}$ and thus the trace distance is 0.

\end{proof}

The final case to consider is the case when $\DC=\tildeFS=1$ (case 4).
A critical component of security is provided by $\tilde{D}_{\CV}$, where the requirement for the message selection limits the success to the case of $\beta=\alpha'$, along with an unmodified message.
This allows us to reduce the analysis to a protocol that is similar to that of Thm.~\ref{thm:BB84PRNG_MinEnt}, where the entropic uncertainty relation is applied.
\begin{theorem}
\label{thm:CA_Proof_MITM_11}
    Consider the CA protocol $\vars{P}$. Then, there exists $\sigma_{LSE}^{j'j'j'}\in\vars{S}_{in}$ such that
    \begin{multline*}
        \Delta(\vars{P}'(\rhoinn{jjj'}_{LSE})_{\land D_C=\tilde{F}_S=1},p_{11}\sigma_{LSE}^{j'j'j'})\leq 4\sqrt{2\epsph}+2\varepsilon_2+\frac{1}{2}\times 2^{-\frac{1}{2}[N_{P_2,1}^{\tol}[1-\hbin(\ephtol')]-l_{\tilde{\theta}}-\log_2\abs{\vars{T}_{\SV}}\abs{\vars{T}_{\CV}}-\leakEC]}\\
        +\sqrt{\abs{\vars{T}_{\CV}}\epsMACa-1+2^{\log_2\abs{\vars{T}_{\CV}}+\log_2\left(\frac{2}{\varepsilon_2}+1\right)-N_{P_2,1}^{\tol}[1-\hbin(\ephtol')]+\leakEC}}+\epsmatch,
    \end{multline*}
    where $\varepsilon_2>0$ is a parameter to be optimised.
\end{theorem}
\begin{proof}
By definition, when $\tildeFS=1$, the announced index value $\alpha_\rr\leq \alpha'$, which leaves $\alpha'$ unchanged throughout the protocol.
The definition of $\tilde{D}_{\CV}=1$ requires additionally that the client's updated index be $\beta\geq \alpha'$.
Let us consider the case of $\beta>\alpha'$ (case 4a), which necessarily requires $\alpha>\alpha'$ for $\tilde{D}_{\CV}=1$.
In this case, the adversary has oracle access to $(\Kh_1,\KOTP_{1,\alpha'})$, where the message inserted to the oracle has to match the server's received message.
We can thus follow the proof of Thm.~\ref{thm:2RoundCADS} to show that the probability of the message matching and thus this event occuring ($\tildeFS=1$) is upper bounded by $\epsmatch$.
Therefore, we can split off 
case 4a of $\beta\geq\alpha>\alpha'$, where the trace distance between the output state and ideal output would be bounded by $\epsmatch$.\\

What remains is the case where $\beta=\alpha'$ (case 4b), which has $\alpha\leq\alpha'$, where $\tilde{D}_{\CV}=1$ ensures that the messages from the client are not tampered with.
Since $\tildeFS=\DC=1$, the final indices are $\alpha=\alphaU=\alpha'=j'$, which corresponds to the index of the ideal state $\sigma_{LSE}^{j'j'j'}$.
Let us consider one such state $\sigma_{LSE}^{j'j'j'}\in\vars{S}_{in}$, defined by 
\begin{equation}
    \sigma_{LSE}^{j'j'j'}
    =\tilde{\Pi}^{j'j'j'}_{\alpha\alphaU\alpha'}\otimes\tau_{R\KOTP_{1,[j',m]}}\otimes\tilde{\tau}_{\tilde{\theta}_{[j',m]}\tilde{\theta}_{[j',m]}'}\otimes\rho_{E\land \DC=\tildeFS=1\land \beta=j'}.
\end{equation}
The output state can be written as
\begin{multline}
    \vars{P}'(\rhoinn{jjj'}_{LSE|j\leq j'})_{\land \DC=\tildeFS=1\land \beta=j'}=\tilde{\Pi}^{j'j'j'}_{\alpha\alphaU\alpha'}\otimes\tau_{\KOTP_{1,[j'+1,m]}}\otimes\tilde{\tau}_{\tilde{\theta}_{[j'+1,m]}\tilde{\theta}'_{[j'+1,m]}}\otimes\rho_{R\Kh_1\KOTP_{1,j'}\tilde{\theta}_{j'}\tilde{\theta}'_{j'}T_{\SV,j'}\tilde{T}_{\SV,j'}E''\land\Omega},
\end{multline}
where $\Omega:=\{\alpha_\rr\leq j',\alpha_\rr'=j',\tildeDPE=\tilde{D}_{\CV}=\DC=1\}$ (which matches $\{\DC=\tildeFS=1,\beta=j'\}$) and
the $\tilde{D}_S$ and $\DC$ values has been absorbed into $E''$.
We can simplify the trace distance as
\begin{equation}
\begin{split}
    &\Delta(\vars{P}'(\rhoinn{jjj'}_{LSE|j\leq j'})_{\land\Omega},\rhooutt{jjj'}_{\land\Omega})\\
    \leq&\Delta(\rho_{R\Kh_1\KOTP_{1,j'}\tilde{\theta}_{j'}\tilde{\theta}'_{j'}T_{\SV,j'}\tilde{T}_{\SV,j'}E''\land\Omega},\tau_{R}\otimes\tilde{\tau}_{\tilde{\theta}_{j'}\tilde{\theta}_{j'}'}\otimes\tilde{\tau}_{T_{\SV,j'}\tilde{T}_{\SV,j'}}\otimes\rho_{E''\land\Omega})\\
    \leq&\Delta(\rho_{R\Kh_1\KOTP_{1,j'}\tilde{\theta}_{j'}\tilde{\theta}'_{j'}T_{\SV,j'}\tilde{T}_{\SV,j'}E''\land\Omega},\tau_{R}\otimes\tilde{\tau}_{\tilde{\theta}_{j'}\tilde{\theta}_{j'}'}\otimes\tilde{\tau}_{T_{\SV,j'}\tilde{T}_{\SV,j'}}\otimes\rho_{\Kh_1\KOTP_{1,j'}E''\land\Omega})\\
    &+\Delta(\rho_{\Kh_1\KOTP_{1,j'}E''\land\Omega},\tau_{\Kh_1\KOTP_{1,j'}}\otimes\rho_{E''\land\Omega})\\
    \leq&\Delta(\rho_{R\Kh_1\KOTP_{1,j'}\tilde{\theta}_{j'}'T_{\SV,j'}E''\land\Omega},\tau_{R\tilde{\theta}_{j'}'T_{\SV,j'}}\otimes\rho_{\Kh_1\KOTP_{1,j'}E\land\Omega})+\Delta(\rho_{\Kh_1\KOTP_{1,j'}E''\land\Omega},\tau_{\Kh_1\KOTP_{1,j'}}\otimes\rho_{E\land\Omega})\\
    \leq &p_{\tildeOmegaPE}\{2\epssma+\frac{1}{2}\times 2^{-\frac{1}{2}[\Hmin^{\epssma}(\hat{X}_{P_2}'|\Kh_1\KOTP_{1,j'}E'')_{\rho_{\land\Omega'|\tildeOmegaPE}}-l_{\tilde{\theta}}-\log_2\abs{\vars{T}_{\SV}}]}\\
    &+2(\varepsilon_{sm,2}+\varepsilon_2)+\sqrt{\abs{\vars{T}_{\CV}}\epsMACa-1+2^{\log_2\abs{\vars{T}_{\CV}}+\log_2\left(\frac{2}{\varepsilon_2}+1\right)-\Hmin^{\varepsilon_{sm,2}}(\hat{X}_{P_2}'|E'')_{\rho_{\land\Omega'|\tildeOmegaPE}}}}\},
\end{split}
\end{equation}
where the first inequality considers a specific state in the class of ideal output states where $T_{SV,j}$ and $\tilde{T}_{\SV,j}$ are equal and independent from the adversary.
The second inequality applies the triangle inequality, while the third inequality invokes $\tilde{D}_{\CV}=1$, which indicates that the variables $\hat{X}'_{P_2}=X_{P_2}'$ and thus the hashed outcomes $(\tilde{\theta}_{j'},\tilde{T}_{\SV,j'})=(\tilde{\theta}_{j'}',T_{\SV,j'})$, allowing the trace distance to reduce to one set of variables (more formally, this is simulated by introducing a ``copy" channel from one set to another and removing the CPTP ``copy" channel).
The fourth inequality applies the quantum leftover hash lemma (Thm.~\ref{thm:Modified_QLHL}),
noting again that $\Omega$ implies $\hat{X}'_{P_2}=X_{P_2}'$, and that $\Omega=\Omega'\land\tildeOmegaPE$, which is defined later.\\

Here, we analyse the first min-entropy term, which would also acts as a lower bound the second term.
Let us list explicitly the conditions of $\Omega$,
\begin{enumerate}
    \item $\alpha_\rr\leq j'$, $\alpha_\rr'=j'$, $\DC=1$: Conditions associated with index selection, and passing the initial client check.
    \item $(P_{2,\rr},P_{1,\rr},X_{P_1,\rr}',syn_\rr)=(P_1,P_2,X_{P_1}',syn)$: Message sent from client to server is unchanged.
    \item $T_{\CV}=T_{\CV,\rr}$: The tag sent from the client to server is unchanged.
    \item $X_{P_2}'=\hat{X}_{P_2}$: The server's corrected bitstring matches with the client's bitstring.
    \item $\frac{wt(X_{P_{1,\rr}}\oplus X'_{P_1,r}[\{P_{1,\rr}\}])}{N_{P_{1,\rr}}}\leq e_{b,tol}$, $\abs{P_{1,\rr}}\geq P_{1,\tol}$: Remaining parameter estimation checks.
    \item $\tildeOmegaPE$: Standard parameter estimation checks.
\end{enumerate}
where we list the events for the first five events from $\Omega_1$ to $\Omega_5$ as being part of $\Omega'$.
We note here that $\tildeOmegaPE$ estimates directly on the $(P_1,P_2,X_{P_1}')$ values instead of the received values, since $\OmegaPE\land\Omega_2=\tildeOmegaPE\land\Omega_2$.
We first simplify by noting that $\hat{X}_{P_2}'$ is computed from $X_{P_{2,\rr}}$ and the syndrome $syn_\rr$ received.
Therefore, by data processing inequality\cite{Tomamichel2015_ITBook},
we have that
\begin{equation}
    \Hmin^{\epssm}(\hat{X}_{P_2}'|\Kh_1\KOTP_{1,j'}E'')_{\rho_{\land\Omega'|\tildeOmegaPE}}\geq \Hmin^{\epssm}(X_{P_{2,\rr}}|\Kh_1\KOTP_{1,j'}syn_\rr E'')_{\rho_{\land\Omega'|\tildeOmegaPE}},
\end{equation}
noting that $syn_\rr$ is in general able to be determined from $E''$ (adversary stores a copy of the $syn_\rr$ it sends to the server).
Further imposing that $P_{2,\rr}=P_2$ from the $\Omega_2$ condition, we can simplify the min-entropy before removing the conditions of $\Omega_2$ to $\Omega_5$~\cite{Tomamichel2017_QKDProof}, i.e.
\begin{equation}
    \Hmin^{\epssm}(X_{P_{2,\rr}}|\Kh_1\KOTP_{1,j'}syn_\rr E'')_{\rho_{\land\Omega'|\tildeOmegaPE}}\geq \Hmin^{\epssm}(X_{P_2}|\Kh_1\KOTP_{1,j'}syn_\rr E'')_{\rho_{\land\Omega_1|\tildeOmegaPE}}.
\end{equation}
Here, we can reverse the CPTP map used to generate $syn_\rr E''$, leaving us with
\begin{equation}
    \Hmin^{\epssm}(X_{P_2}|\Kh_1\KOTP_{1,j'}syn_\rr E'')_{\rho_{\land\Omega_1|\tildeOmegaPE}}\geq \Hmin^{\epssm}(X_{P_2}|\Kh_1\KOTP_{1,j'}X_{P_1}'P_1P_2synT_{\CV}E')_{\rho_{\land\Omega_1|\tildeOmegaPE}}.
\end{equation}
We can remove the syndrome $syn$ and the tag $T_{\CV}$ utilising the min-entropy chain rule~\cite{Tomamichel2015_ITBook}.
Furthermore, with the removal of $T_{\CV}$, $\Kh_1$ and $\KOTP_{1,j'}$ is uncorrelated with any other variables in the state and can be removed, leaving us with
\begin{equation}
    \Hmin^{\epssm}(X_{P_2}|\Kh_1\KOTP_{1,j'}X_{P_1}'P_1P_2synT_{\CV}E')_{\rho_{\land\Omega_1|\tildeOmegaPE}}\geq \Hmin^{\epssm}(X_{P_2}|\tilde{\theta}X_{P_1}'P_1P_2E')_{\rho_{|\tildeOmegaPE}}-\leakEC-\log_2\abs{\vars{T}_{\CV}},
\end{equation}
where we also remove the $\Omega_1$ condition, and included additional $\tilde{\theta}$ conditioning.
Here, the protocol steps are simply:
\begin{enumerate}
    \item The server prepares decoy state BB84, using the basis generated from $\tilde{\theta}$ and bit value $X$.
    \item The adversary performs its attack on the state, mapping $QE$ to $PBE'$.
    \item The client randomly select $P_1$ rounds and measure them in the basis generated from $\tilde{\theta}$ and outputs $X_{P_1}'$.
    \item Based on $X$, $P_1$, $P_2$ and $X_{P_1}'$, the parameter estimation decision is made (event $\tildeOmegaPE$),
\end{enumerate}
This is exactly the same protocol steps as that in Thm.~\ref{thm:BB84PRNG_MinEnt}, and we can utilise the same analysis to demonstrate that
\begin{equation}
    \Hmin^{\epssm}(\hat{X}_{P_2}'|\Kh_1\KOTP_{1,j'}E'')_{\rho_{\land\Omega'|\tildeOmegaPE}}\geq N_{P_2,1}^{\tol}[1-\hbin(\ephtol')]-\leakEC-\log_2\abs{\vars{T}_{\CV}},
\end{equation}
with $\epssm=\sqrt{2\epsph}$.
The second min-entropy term can follow a similar analysis, but since $\Kh_1$ and $\KOTP_{1,j'}$ is traced out, $T_{\CV}$ would be a random string, and can be simply removed without incurring the $\log_2\abs{\vars{T}_{\CV}}$ penalty.
As such, the overall trace distance is
\begin{multline}
    \Delta(\Pi_{\DC\tildeFS}^{11}\vars{P}'(\rhoinn{jjj'}_{LSE})\Pi_{\DC\tildeFS}^{11},\rhopoutt{jjj'}_{\land \Omega}) \leq 4\sqrt{2\epsph}+2\varepsilon_2+\frac{1}{2}\times 2^{-\frac{1}{2}[N_{P_2,1}^{\tol}[1-\hbin(\ephtol')]-l_{\tilde{\theta}}-\log_2\abs{\vars{T}_{\SV}}\abs{\vars{T}_{\CV}}-\leakEC]}\\
    +\sqrt{\abs{\vars{T}_{\CV}}\epsMACa-1+2^{\log_2\abs{\vars{T}_{\CV}}+\log_2\left(\frac{2}{\varepsilon_2}+1\right)-N_{P_2,1}^{\tol}[1-\hbin(\ephtol')]+\leakEC}}
    .
\end{multline}
\end{proof}

\end{document}